\documentclass[11pt]{article}
\usepackage{jheppub}
\usepackage{amsmath,amssymb,amsthm,graphicx,amscd}
\usepackage[matrix,arrow,curve]{xy}
\usepackage{hyperref}
\usepackage{color}
\usepackage[usenames,dvipsnames,svgnames,table]{xcolor}
\usepackage{multirow}
\usepackage{tabularx}
\usepackage{bm}
\usepackage{dsfont}

\usepackage{tikz}

\usetikzlibrary{calc,decorations.pathmorphing,shapes}

\newcounter{sarrow}
\newcommand\xrsquigarrow[1]{%
\stepcounter{sarrow}%
\mathrel{\begin{tikzpicture}[baseline= {( $ (current bounding box.south) + (0,-0.5ex) $ )}]
\node[inner sep=.5ex] (\thesarrow) {$\scriptstyle #1$};
\path[draw,<-,decorate,
  decoration={snake,amplitude=2.5pt,segment length=2mm,pre=lineto,pre length=4pt}]
    (\thesarrow.south east) -- (\thesarrow.south west);
\end{tikzpicture}}%
}

\def\be{\begin{equation}}
\def\ee{\end{equation}}
\def\nn{\nonumber}
\def\arg{\operatorname{arg}}
\def\lm{\limits}

\def\Tr{{\rm Tr}\,}

\def\CA{{\mathcal A}}
\def\CT{{\mathcal T}}
\def\CP{{\mathcal P}}
\def\CW{{\mathcal W}}
\def\CB{{\mathcal B}}
\def\CK{{\mathcal K}}

\def\CM{{\mathcal M}}
\def\IR{{\mathbb{R}}}

\def\FOmega{\underline{\overline{\Omega}}}

\newcommand{\twid}{\widetilde}
\newcommand{\Hilb}{\mathcal{H}}

\newcommand{\cl}{\operatorname{cl}}
\newcommand{\conj}{\overline}
\newcommand{\inup}{\textcolor{red}{\Upsilon}}
\newcommand{\indown}{\textcolor{red}{\Delta}}
\newcommand{\oup}{\textcolor{blue}{\Upsilon}}
\newcommand{\odown}{\textcolor{blue}{\Delta}}
\newcommand{\oQ}{\textcolor{blue}{Q}}
\newcommand{\osQ}{\textcolor{blue}{\mathcal{Q}}}
\newcommand{\inX}{\textcolor{magenta}{X}}
\newcommand{\incdat}{\textcolor{red}{\operatorname{Incoming-Data}}}
\newcommand{\outdat}{\textcolor{blue}{\operatorname{Outgoing-Data}}}
\newcommand{\intdat}{\textcolor{blue}{\operatorname{Internal-Data}}}
\newcommand{\formgamma}{\mathbb{Z} [ \! [ \twid{\Gamma} ] \! ]}

\newcommand{\lift}[1]{\bm{#1}}
\newcommand{\liftnb}[1]{\bm{#1}}
\newcommand{\la}{\bm{a}}
\newcommand{\lb}{\bm{b}}
\newcommand{\lac}{\bm{\conj{a}}}
\newcommand{\lbc}{\bm{\conj{b}}}
\newcommand{\lte}{\bm{e}}
\newcommand{\ltf}{\bm{f}}

\newcommand{\N}{{\mathcal N}}

\newtheorem{theorem}{Theorem}[section]
\newtheorem{lemma}[theorem]{Lemma}
\newtheorem{proposition}[theorem]{Proposition}
\newtheorem{corollary}[theorem]{Corollary}

\newenvironment{definition}[1][Definition]{\begin{trivlist}
\item[\hskip \labelsep {\bfseries #1}]}{\end{trivlist}}

\newenvironment{remark}[1][Remark]{\begin{trivlist}
\item[\hskip \labelsep {\bfseries #1}]}{\end{trivlist}}
\newenvironment{claim}[1][Claim]{\begin{trivlist}
\item[\hskip \labelsep {\bfseries #1}]}{\end{trivlist}}

\makeatletter
\newcommand{\labitem}[2]{%
\def\@itemlabel{#1:}
\item
\def\@currentlabel{#1}\label{#2}}
\makeatother

\newcommand{\ti}[1]{\textit{#1}}

\title{Wild Wall Crossing and BPS Giants}

\abstract{We show that the BPS spectrum of pure $SU(3)$ four-dimensional
super Yang-Mills with ${\cal N} = 2$ supersymmetry exhibits a surprising phenomenon:
there are regions of the Coulomb branch where the growth of the BPS degeneracies
with the charge is \emph{exponential}.
We show this using spectral networks and independently using
wall-crossing formulae and quiver methods.  The computations using
spectral networks provide a very nontrivial example of how these networks
determine the four-dimensional BPS spectrum. We comment on some physical implications
of the wild spectrum: for example, exponentially
many field-theoretic BPS states with large charge are gigantic.
Finally, we exhibit some surprising, thus far unexplained,
regularities of the BPS spectrum.}

\author[1]{Dmitry Galakhov,}
\author[2]{Pietro Longhi,}
\author[3]{Tom Mainiero,}
\author[4]{Gregory W. Moore,}
\author[5]{and Andrew Neitzke}
\affiliation[1]
{Institute for Theoretical and Experimental Physics,\\ Moscow, Russia,}
\affiliation[1,2,4]
{NHETC and Department of Physics and Astronomy, Rutgers University,\\
Piscataway, NJ 08855--0849, USA}
\affiliation[3]
{Department of Physics, University of Texas at Austin,\\
Austin, TX 78712, USA}
\affiliation[5]
{Department of Mathematics, University of Texas at Austin,\\
Austin, TX 78712, USA}
\emailAdd{galakhov@physics.rutgers.edu}
\emailAdd{longhi@physics.rutgers.edu}
\emailAdd{mainiero@physics.utexas.edu}
\emailAdd{gmoore@physics.rutgers.edu}
\emailAdd{neitzke@math.utexas.edu}

\date{\today}

\begin{document}

\maketitle

\section{Introduction \& Conclusion}

One good reason to investigate the BPS spectra of
four-dimensional $\N=2$ field theories is that one might
discover new phenomena in field theory.  This paper
demonstrates an example of such a new phenomenon.

In the past few years there has been much progress
in understanding the BPS spectra of $\N=2, d=4$ theories.
For recent reviews see \cite{Moore:2012yp,MOORE_FELIX,Cecotti:2012se}.
These methods have been particularly powerful when applied to the theories of class
$S[A_1]$.  As one example, the
spectrum generator technique of \cite{GMN2} gives an algorithm
which can --- in principle --- give the BPS spectrum of any
theory of class $S[A_1]$ anywhere on its Coulomb branch.
Advances in quiver technology have also been very
effective in investigating this class of theories \cite{CV,DIACONESCU}.
In contrast, theories associated to higher rank gauge groups, such as theories of class
$S[A_{K-1}]$ for $K > 2$, have been less explored.

It has been noted by various authors that
theories of class $S[A_{K-1}]$ for $K > 2$
could have higher spin BPS states, beyond the familiar hypermultiplets and vectormultiplets
which occur in theories of class $S[A_1]$.
One result of this paper is that this expectation is indeed correct: higher spin BPS multiplets do
occur at some points of the Coulomb branch in one explicit theory of class $S[A_2]$,
namely the pure $d=4$, $\N=2$, $SU(3)$ theory.

In addition, we find a much more surprising phenomenon:  theories of class $S$ can have
\ti{wild BPS spectra}, i.e. at some points of the Coulomb branch,
the number of BPS states with mass $\le M$ grows \ti{exponentially} with $M$.
The main result of this paper is two independent demonstrations, in
Sections \ref{sec:spectral-wild-SU3} and \ref{sec:wc-wild-SU3}, that
wild spectra appear in the pure $d=4, \N=2, SU(3)$ theory.

As explained in Section \ref{sec:physical-estimates} below,
this exponential growth is physically a bit surprising.
Indeed, the existence of a conformal fixed point defining the 4d theory,
plus dimensional analysis, implies that the degeneracy of
BPS states at energy $E$ in finite volume $V$ cannot grow faster than $\exp[\mathrm{const} \times V^{1/4} E^{3/4}]$.
On the other hand, here we are finding that the spectrum of BPS 1-particle
states grows like $\exp [\mathrm{const} \times E]$.
The resolution of this puzzle must lie in the difference between BPS 1-particle states
and states in the finite volume Hilbert space; we propose that the size of the objects
represented by the BPS 1-particle states grows with $E$, so that for any fixed $V$,
most of the BPS 1-particle states simply do not fit into the finite-volume Hilbert space.
Indeed, in Section \ref{sec:physical-estimates},
using Denef's picture of BPS bound states, we demonstrate directly that
their size does indeed grow with $E$.
The invalid exchange of large $E$ and large $V$ limits when
accounting for field theory entropy should perhaps serve
as a cautionary tale.

\medskip

Here is the fundamental idea which we use to find wild BPS degeneracies.
Suppose we have an $\N=2$ theory and a point of the Coulomb branch
in which the spectrum contains two BPS hypermultiplets,
of charges $\gamma$ and $\gamma'$, and no bound states thereof --- i.e.
we have the BPS degeneracies $\Omega(\gamma) = 1$, $\Omega(\gamma') = 1$,
$\Omega(a \gamma + b \gamma') = 0$ for all other $a, b \ge 0$.
Then suppose we move on the Coulomb branch to a point where the central charges
$Z_\gamma$ and $Z_{\gamma'}$ have the same phase.  Such a point lies on a wall
of marginal stability.  On the other side of the wall, the spectrum includes
bound states with charge $a \gamma + b \gamma'$ for various $a$, $b$.  Their precise
degeneracies can be determined by the Kontsevich-Soibelman wall-crossing formula,
and indeed depend \ti{only} on the integer $m = \langle{\gamma,\gamma'}\rangle$.
For this reason we call the collection of BPS states thus generated an ``$m$-cohort.''

The cases $m = 1$ and $m = 2$ occur already in the theories of class $S[A_1]$.
For $m=1$ an $m$-cohort contains only a single bound state; for $m=2$ an $m$-cohort contains an
infinite set of hypermultiplets plus a single vector multiplet.
In either case, at any rate, one does not get wild degeneracies.  In contrast, for $m>2$
the wall-crossing formula shows that an $m$-cohort does contain wild degeneracies.
Indeed, even if one restricts attention to charges of the form $n(\gamma + \gamma')$, one already
has exponential growth.
This is explained and made precise in Proposition \ref{prop:asymp}, Section \ref{subsec:3kron}, and
Section \ref{sec:asymptotics} below.  With this in mind, for any $m>2$, we will say that a theory contains
``$m$-wild degeneracies'' if its BPS spectrum contains an $m$-cohort.

The BPS degeneracies arising in $m$-cohorts have been studied at some length in the mathematics literature
because they
arise as Donaldson-Thomas invariants attached to the $m$-Kronecker
quiver in one region of its stability parameter space. The latter
have been intensively studied in
\cite{REINEKE,REINEKE03,REINEKE08,REINEKE-09,WEIST,WEIST12}.  One interesting feature noted there
is that for $m>2$, the phases of the central charges of BPS states in an $m$-cohort are
\ti{dense} in some arc of the circle.

This discussion motivates two approaches to the problem of exhibiting
wild degeneracies in a physical theory.
Our first approach goes via the ``spectral networks'' of \cite{GMN5,GMN6}:
rather than studying the wall-crossing directly, we make a guess about the
kind of spectral networks which \ti{could} arise from wall-crossing involving
two hypermultiplets with arbitrary $m = \langle{\gamma,\gamma'}\rangle$.
For $m = 1$ the network we draw looks like a saddle,
which motivates an equine terminology:  our networks are built from constituents
we call ``horses'' (defined in Section. \ref{sec:herds}, Figure \ref{fig:horse}, and detailed in Appendix \ref{app:herd-appendix}), glued together to form ``$m$-herds.''
See Figure \ref{fig:herds} for some examples.
We show moreover that $m$-herds indeed occur in physical spectral
networks at some particular points of the Coulomb branch of the $SU(3)$ theory:  see Figure
\ref{fig:3-herd-in-su3} for the evidence.
The general rules of spectral networks, combined with
Proposition 3.1 and Proposition 3.2 below, lead
to the following formula for the BPS spectrum for charges of
the form $n (\gamma + \gamma'):= n \gamma_c$ in the wild region.
 We first form a generating
function $P_m(z)$ related to the BPS spectrum by
\be
P_m(z)  = \prod_{n=1}^{\infty} (1-(-1)^{mn} z^n)^{n \Omega(n\gamma_c)/m}.
\ee
Then, Proposition 3.1 states that $P_m(z)$ is a solution of the algebraic equation \eqref{eq:P},
which we reproduce here:
	\begin{equation}
		P_{m} = 1 + z \left( P_{m} \right)^{(m-1)^2}.
		\label{eq:P-intro}
	\end{equation}
This equation had been identified previously by Kontsevich and
Soibelman \cite{KS_MOTIVIC_I} and by Gross and Pandharipande
\cite{GROSS}, as the one governing the generating function of BPS degeneracies of an $m$-cohort,
for charges of the form $n(\gamma + \gamma')$.
It follows that if we have an $m$-herd ($m>2$) somewhere in our theory, then our theory does contain at least
the part of an $m$-cohort corresponding to charges of the form
$n(\gamma + \gamma')$.  In particular, if the theory contains an $m$-herd, then it does
contain wild degeneracies.  Since we have found $m$-herds at some points of the Coulomb branch in
the pure $SU(3)$ theory, we conclude that we indeed have wild degeneracies in that theory.

The algebraic equation \eqref{eq:P-intro} is an instance of a more general phenomenon.
It has been observed by Kontsevich that the generating functions of Donaldson-Thomas invariants
are often solutions of algebraic equations. In fact, for the Kronecker
quiver this has been proved \cite{KontsevichCommun}.
 Our analysis via spectral networks produces the algebraic
equation \eqref{eq:P-intro} directly.  Moreover, we expect that this will happen more generally, as we explain in Appendix
\ref{app:algeq}; thus spectral networks seem to be a natural framework for explaining
Kontsevich's observation.

Our second method of demonstrating the existence of wild spectra
uses wall-crossing more directly.  Namely, in Section \ref{subsec:PathCoulomb}
we exhibit a path on the Coulomb branch
which begins in a strong coupling chamber
with a finite set of BPS states, and leads to a
wall-crossing between two hypermultiplet charges $\gamma$, $\gamma'$
with $\langle \gamma, \gamma' \rangle = 3$.
As we have discussed above, the existence of such a path directly
implies the existence of wild spectra.  In fact this gives more than we got from the
spectral network:  it shows that there is a whole $3$-cohort in the spectrum.
In Section \ref{sec:wild-SU3} we perform some nontrivial checks of this
statement by factorizing the spectrum generator derived from
the known finite spectrum in a strong coupling chamber.
In Section \ref{subsec:3kron} we also check numerically the exponential growth of
the BPS degeneracies for sequences of charges of the
form $n (a \gamma + b\gamma')$, $n\to \infty$,  for various values of $a,b$.

In Section \ref{subsec:RelateQuivers} we discuss the
behavior of the ``BPS quivers'' of the $SU(3)$ theory
along the path found in Section \ref{subsec:PathCoulomb}.  It turns out that
the Kronecker 3-quiver is in fact a subquiver of the BPS quiver, after
one has performed suitable mutations and   made a suitable choice of half-plane to define simple roots.
  We similarly
argue that for \ti{all} $m \ge 3$ (not only $m=3$) there are Kronecker $m$-subquivers and corresponding
$m$-wild spectra on the Coulomb branch of the $SU(3)$ theory.

In the course of our investigations we also studied the protected spin characters
(a.k.a. ``refined BPS degeneracies'') for the $m$-Kronecker quiver in the
wild region. Our main tool was the   ``motivic'' Kontsevich-Soibelman
formula \cite{KS_MOTIVIC_I,KS_MOTIVIC_II}. While investigating these
spin degeneracies   we discovered some
beautiful but strange systematics.  Some
of these were  previously discovered by Weist and Reineke
in \cite{WEIST12} and \cite{REINEKE03}, respectively, but
some are new. We
collect them in Section \ref{sec:moonshine}. Perhaps the most notable
new observation is that the spin degeneracies appear, (on the basis of
numerical data), to obey a universal scaling law. See
equations \eqref{eq:scalingfunction} and \eqref{eq:Poisson}.

In Section \ref{sec:OpenProblems} we discuss a few  open problems and questions
raised by the present work. Appendix \ref{app:PSC-tables} reviews
definitions of protected spin characters and presents some data.
The remaining appendices address more technical points of spectral networks.

\section*{Acknowledgements}

We thank Ofer Aharony, Tom Banks,  Frederik Denef, Emanuel Diaconescu,
Maxim Kontsevich,
Jan Manschot and Steve Shenker
for useful discussions and  correspondence.
The work of DG, PL, and GM is supported by the DOE under grant
DE-FG02-96ER40959.  GM also gratefully acknowledges
partial support from a Simons Fellowship and hospitality
of the Aspen Center for Physics (NSF Grant 1066293)
during part of this research.
The work of AN is supported by the NSF under grant numbers
DMS-1006046 and DMS-1151693.
The work of GM and AN is jointly supported by an NSF Focused Research Group
award DMS-1160461 and DMS-1160591.
The work of DG was partly supported by Ministry of Education and Science of the Russian Federation under contract 8207, by NSh-3349.2012.2, by RFBR grants 13-02-00457, 12-02-31535-mol-a.
The work of TM was partly supported by the NSF under an NSF Research Training Group award
DMS-0636557.

\section{Brief Review of Spectral Networks} \label{sec:sn-review}

In this section we give a brief review of the spectral network machinery and its use for computing BPS spectra in  $\mathcal{N}=2, \, d=4$ theories of class $S$.
For a more complete discussion we refer the reader to \cite{GMN5}.
For a more informal (but incomplete) review see  \cite{MOORE_FELIX}.

\subsection{The Setting}
Recall that the $\mathcal{N}=2, \, d=4$ theories of class $S$ are
specified by three pieces of data \cite{GMN2,Gaiotto:2009we}:
\begin{enumerate}
	\item A Lie algebra $\mathfrak{g}$ of ADE type (as in \cite{GMN5} the following discussion assumes $\mathfrak{g}=A_{K-1}$),

	\item a compact Riemann surface $C$ with punctures at points $\mathfrak{s}_{1},\cdots, \mathfrak{s}_{n} \in C$,

	\item a collection of defect operators $D$ located at the punctures.
\end{enumerate}
To shed some light on this collection of data, we note that such theories can be constructed via a partial twist (preserving eight supercharges) of the $\mathcal{N}=(2,0),\, d=6$ theory $S[\mathfrak{g}]$ defined on $\mathbb{R}^{3,1} \times C$.  The defect operators $D$ are codimension-2 defects located at $\mathbb{R}^{3,1} \times \{ \mathfrak{s}_{1} \}, \, \cdots , \, \mathbb{R}^{3,1} \times \{ \mathfrak{s}_{n} \}$.
A four-dimensional $\mathcal{N}=2$ field theory is produced after integrating out the degrees of freedom along $C$ and is labeled $S[\mathfrak{g},C,D]$.

We now present some useful definitions.

\begin{definition}[Definitions]\
\begin{enumerate}
	\item The \textit{Coulomb branch} $\CB$ of $S[\mathfrak{g},C,D]$ is the set of tuples $(\phi_{2}, \cdots, \phi_{K})$ of holomorphic $r$-differentials $\phi_{r}$ with singularities at $\mathfrak{s}_{1},\cdots, \mathfrak{s}_{n} \in C$ prescribed by the defect operators $D$.

	\item Let $u = (\phi_{2}, \, \cdots, \, \phi_{r} ) \in \CB$ and denote the holomorphic cotangent bundle of $C$ as $\CT^* C$.  Then the \textit{spectral cover} is a $K$-sheeted branched cover $\pi_{u}: \Sigma_{u} \rightarrow C$, where $\Sigma_{u}$ is the subvariety\footnote{$\Sigma_{u}$ is also called the Seiberg-Witten curve.}
	\begin{equation}
	\Sigma_{u} := \{\lambda \in \CT^*C: \lambda^{K} + \sum_{r=2}^{K} \phi_{r} \lambda ^{K-r} = 0\} \subset \CT^* C,
	\end{equation}
	and the projection $\pi_{u}$ is the restriction of the standard projection $\CT^* C \rightarrow C$.

	\item As $\Sigma_{u} \subset \CT ^* C$ it carries a natural holomorphic $1$-form which is just the restriction of the tautological (Liouville) 1-form.  In the spirit of its tautological nature we abuse notation and denote this 1-form $\lambda_{u}$.
	\end{enumerate}
\end{definition}

Often we will work over a fixed $u \in \CB$; so eventually the index $u$ will be dropped where there is no ambiguity.

\subsubsection{Spectral Cover Crash Course}

Let us make some observations about the spectral cover. First, the fibers are given by
\begin{equation*}
	\pi_{u}^{-1}(z) = \{ \lambda(z) \in \CT_{z}^{*} C: \lambda(z)^{K} + \sum_{r=2}^{K} \phi_{r}(z) \lambda^{K - r}(z) = 0 \},
\end{equation*}
i.e. the roots of the defining polynomial of $\Sigma_{u}$ at the point $z$.  Generically, $\pi_{u}^{-1}(z)$ consists of $K$ distinct roots, although at particular values of $z$ (branch points) two or more roots may coincide.  In fact, letting $C'=C-\{ \text{branch points} \}$,
$\pi_{u}|_{C'}$ is a $K$-fold (unramified) cover of $C'$.

If $\pi_{u}|_{C'}$ is a non-trivial cover, the roots do not fit together into global holomorphic $1$-forms on $C$ as they undergo monodromy around branch points.  However, restricted to the complement of a choice of branch cuts on $C$, the cover is trivializable: a projection of $K$ distinct sheets onto the complement.  Each sheet is the graph traced out by a root of the defining polynomial; such roots are distinct holomorphic differential forms.  A choice of trivialization of the restricted cover is a bijective map between the set of $K$ sheets and the set $\{1, 2, \dots, K\}$, or equivalently, a labeling of the roots of the defining polynomial from $1$ to $K$.

\begin{definition}[Definitions] \
	\begin{enumerate}
	\item Make a suitable choice of branch cuts for the branched cover $\pi_{u}: \Sigma_{u} \rightarrow C$.  The complement of these branch cuts in $C$ will be denoted by $C^{c}$.

	\item A choice of trivialization of $\pi^{-1}(C^{c}) \rightarrow C^{c}$ will be denoted by a labeling of the roots of the defining polynomial for $\Sigma$, i.e. a labeling $\lambda_{i} \in H^{0}(C^{c}; K),\, i = 1, \dots, K$, where each $\lambda_{i}$ (a holomorphic $1$-form on $C^{c}$) is a distinct root of the defining polynomial for $\Sigma$.  Note that this gives us a labeling of sheets: the  $i$th sheet is the graph of $\lambda_{i}$ in $\CT^{*}C$.  If we wish, we can extend the $\lambda_{i}(z)$ to branch points $z$ to speak of ``collisions'' of sheets.

	\item For later convenience, we define
	\begin{align*}
	\lambda_{ij} := \lambda_{i} - \lambda_{j} \in H^{0}(C^{c};K).
	\end{align*}

	\end{enumerate}
\end{definition}

As in \cite{GMN5} we will assume that all branch points are simple, i.e.
at most two sheets of $\Sigma$ collide at any $z$.

\begin{definition}
	A branch point of type  $ij$ ($i,j \in 1, \cdots, K$) is a point $z \in C$ where the $i$th and $j$th sheets of $\Sigma_{u}$ collide, i.e, $\lambda_{i}(z) = \lambda_{j}(z)$.
\end{definition}

The data of the full spectral cover can be recovered after trivializing by specifying the monodromy around all branch points, and all closed cycles of $C$.  In this paper, we assume \textit{simple ramification}: in a neighborhood around each branch point, the spectral cover looks like the branched cover $z \mapsto z^2$ of the disk.  Thus, for a simple closed curve surrounding a branch point of type $ij$, there is a $\mathbb{Z}/2\mathbb{Z}$ monodromy
\begin{equation*}
	\lambda_{i} \leftrightarrow \lambda_{j}.
\end{equation*}
Monodromy around an arbitrary closed cycle of $C$ may permute the sheets in a more complicated fashion.

\subsubsection{BPS objects in $S[A_{K-1}, C, D]$}
Theories of class $S$ admit a zoo rich in BPS species, each of which has a different classical description from the point of view of the six-dimensional geometry of $\mathbb{R}^{3,1} \times C$.  Our ultimate interest in this paper is in the 4D (vanilla) BPS states,
but the power behind the spectral network machine draws heavily on the symbiosis between these different species; so we take a moment to project each of them into the spotlight.

\subsubsection*{BPS Strings and ``vanilla'' 4D BPS states}
4D BPS states in the four-dimensional $\mathcal{N}=2$ theory arise from extended objects in the 6D description: BPS strings.
In the effective IR description, at a point $u \in \CB$, BPS strings wrap closed paths $p$ on the branched cover $\Sigma_{u} \subset \CT^*C \rightarrow C$.  The resulting states are classified by their homology classes $\gamma  = [p] \in H_{1}(\Sigma_{u}; \mathbb{Z})$ in the sense that there is a natural grading of the Hilbert space of BPS strings as
\begin{equation*}
	\Hilb^{\operatorname{BPS}}(u) = \bigoplus_{ \gamma \in H_{1}(\Sigma_{u}; \mathbb{Z})} \Hilb(\gamma;u),
\end{equation*}
commuting with the action of the super-Poincar\'{e} group.
\begin{definition}
	The \textit{charge lattice} of 4D BPS states at a point $u \in \CB$ is $\Gamma_{u} = H_{1}(\Sigma_{u}; \mathbb{Z})$. It is equipped with an antisymmetric pairing $\langle \cdot, \cdot \rangle: \Gamma_{u} \times \Gamma_{u} \rightarrow \mathbb{Z}$ given by the intersection form on $H_{1}(\Sigma_{u}; \mathbb{Z})$.
\end{definition}

To count the number of BPS states of a particular charge $\gamma$ we recall a major celebrity of this paper: the second helicity supertrace (a.k.a. the ``BPS degeneracy" or ``BPS index")
\begin{align*}\label{eq:eq:SecHelSTR}
	\Omega(\gamma; u) &= -\frac{1}{2} \Tr_{\Hilb(\gamma;u)} (2 J_{3})^2 (-1)^{2J_{3}},
\end{align*}
where $J_{3}$ is any generator of the rotation subgroup of the massive little group.  This index is piecewise constant on $\CB$, jumping across real codimension-1 walls of marginal stability  on $\CB$ where two BPS states with linearly independent charges $\gamma, \gamma' \in \Gamma_{u}$ have central charges of the same phase: $\arg \left(Z_{\gamma} \right) = \arg \left(Z_{\gamma'} \right)$.   To compute this index we will not rely on its definition as a supertrace, but instead utilize the geometric methods of the spectral network machine.

\begin{remark}[Remarks] \
	\begin{enumerate}
		\item On $\CB$ there may be (complex) codimension-1 loci where a cycle of $\Sigma_{u}$ degenerates.  Let $\CB^{*} = \CB - \text{\{degeneration loci\}}$.  Then the collection $\widehat{\Gamma} = \{\Gamma_{u} \}_{u \in \CB^{*}}$ forms a local system of lattices $\widehat{\Gamma} \rightarrow \CB^{*}$.  This local system is often equipped with a non-trivial monodromy action.

		\item As mentioned previously, we will often drop the subscript $u \in \CB$ as we will often be working over a single point on the Coulomb branch, or choosing a local trivialization of the local system of lattices on some open set.

		\item Strictly speaking, the lattice of charges $\Gamma_u$
		is not quite $H_{1}(\Sigma_{u}; \mathbb{Z})$ \cite{GMN2}; in the theory we consider
		in this paper, though, $\Gamma_u$ is just
		a sublattice of $H_{1}(\Sigma_{u}; \mathbb{Z})$, and for our considerations there
		is no harm in replacing $\Gamma_u$ by  $H_{1}(\Sigma_{u}; \mathbb{Z})$.
		(If we considered the theory with ${\mathfrak g} = \mathfrak{gl}(K)$ instead of ${\mathfrak g} = \mathfrak{sl}(K)$
		then the charge lattice would be literally $H_{1}(\Sigma_{u}; \mathbb{Z})$.)

		\item From the four-dimensional point of view, $\Gamma_{u}$ is the lattice of electric/magnetic and flavor charges in the IR effective abelian gauge theory defined at $u$.

	\end{enumerate}
\end{remark}
Fix $u \in \CB$. The central charge and mass of a string $p: S^{1} \rightarrow \Sigma$ are\footnote{The integral $\int_{p} \lambda$ is only a function of the class $[p] \in H_{1}(\Sigma; \mathbb{Z})$; hence, the central charge reduces to a function $\Gamma \rightarrow \mathbb{C}$.}
\begin{align*}
	Z_{p} = \frac{1}{\pi} \int_{p} \lambda, \\
	M_{p} = \frac{1}{\pi} \int_{p} |\lambda|.
\end{align*}
With this, the BPS condition $|Z_{p}| = M_{p}$ is given by
\begin{equation}
	 \int_{p} \lambda = e^{i \vartheta} \int_{p} |\lambda|
	 \label{eq:BPS_cond_int}
\end{equation}
 for some $\vartheta = \arg(Z_p) \in \mathbb{R}/2\pi \mathbb{Z}$.  The value of $\vartheta$ specifies which four-dimensional BPS subalgebra is preserved.  We can rewrite this condition in more useful form; indeed, let $v_p$ denote a vector field along the path $p$, then (\ref{eq:BPS_cond_int}) is true iff
\begin{equation}
	\operatorname{Im} \left[e^{-i \vartheta} \langle \lambda, v_{p} \rangle \right] = 0.
	\label{eq:BPS_cond_vanilla}
\end{equation}

\subsubsection*{Solitons}
The theory $S[\mathfrak{g},C,D]$ is equipped with a special set of BPS surface defect operators $\{\mathbb{S}_{z} \}_{z \in C}$ parameterized (in the UV\footnote{In the six-dimensional UV description, the operator $\mathbb{S}_{z}$ attached to a point $z \in C$ is a surface defect which intersects $C$ at a single point.}) by points on $C$.  In the IR, for fixed $u \in \CB$, the operator $\mathbb{S}_{z}$ possesses finitely many massive vacua labeled by the set $\pi^{-1}(z)$ (with $\pi = \pi_{u}$).
 Letting $z \in C'$, then \textit{solitons} are BPS states\footnote{After insertion of $\mathbb{S}_{z}$ there are four remaining supercharges.  A BPS soliton preserves two supercharges.} bound to the defect $\mathbb{S}_{z}$, which interpolate between two different vacua.  Classically, they are given by oriented paths $s:[0,1] \rightarrow \Sigma$ with endpoints $s(0),\, s(1) \in \pi^{-1}(z)$; furthermore, each such path satisfies a BPS condition that we will now describe.

Consider a soliton path $s$ such that, after choosing a trivialization,
$s$ only runs along sheets $i$ and $j$ and such that
 the projection $s_{C} := \pi \circ s$ is a connected open path on $C$.  Let $v_{s_{C}}$ be a vector field along the path $s_{C}$. Then, the BPS condition is the differential equation
\begin{align}
	\operatorname{Im} \left[ e^{-i \vartheta} \langle \lambda_{ij}, v_{s_{C}} \rangle \right] &= 0
	\label{eq:BPS_soliton_cond}
\end{align}
for some fixed angle $\vartheta$.
For more complicated solitons that travel along more than two sheets, we can break the soliton up into a concatenation of
partial solitons running along various pairs of sheets; each partial soliton involved in the concatenation must
satisfy (\ref{eq:BPS_soliton_cond}) where $ij$ is replaced by the relevant pair of sheets, and $\vartheta$ is the same for each partial soliton.
Hence, the BPS condition for solitons leads to a system of $\binom{K}{2}$ differential equations on $C'$ (one for each disjoint pair of sheets).
For such a soliton $s$, broken into partial solitons $\{s^{r}\}_{r = 1}^{L}$ as decribed above, its central charge and mass are
\begin{equation}
	\begin{aligned}
		Z_{s} &= \sum_{r = 1}^{L} \frac{1}{\pi} \int_{s^{r}_{C}} \lambda_{ij}\\
		M_{s} &= \sum_{r = 1}^{L} \frac{1}{\pi} \int_{s^{r}_{C}} |\lambda_{ij}|.
	\end{aligned}
	\label{eq:soliton_mass}
\end{equation}
We can now identify the angle as the
phase of the central charge, $\vartheta = \arg(Z_{s})$, and indeed the BPS condition is equivalent
to $M_{s} = |Z_{s}|$.

As with 4D BPS states, solitons also carry a charge, but now given by a relative homology class as they are open paths.

Let $z \in C'$; choose a labeling of the $K$ points in $\pi^{-1}(z) \in C'$.
\begin{definition}\
	\begin{enumerate}
		\item Let $z_{l} \in \pi^{-1}(z)$ denote the pre-image of $z \in C'$ on the $l$th sheet.  Then a soliton is of type $ij$ if it is given by a path that begins on $z_{i}$ and ends on $z_{j}$.  We also refer to such solitons as $ij$-\textit{solitons}.

		\item $\Gamma_{ij}(z,z)$ is the set of charges of $ij$-solitons, i.e.
					\begin{equation*}
						\Gamma_{ij}(z,z) := \left \{[a] \in  H_{1}(\Sigma, \{z_{i} \} \cup \{ z_{j} \}; \mathbb{Z}): \text{$a$ is a $1$-chain with $\partial a = z_{j} - z_{i}$} \right \}.
				\end{equation*}

		\item The total set of soliton charges is
			\begin{equation*}
				\Gamma(z,z) :=  \bigsqcup_{i,j = 1}^{K} \Gamma_{ij}(z,z).
			\end{equation*}
	\end{enumerate}
\end{definition}

\begin{remark}[Remarks] \
	\begin{enumerate}

		\item A soliton $s$ can be extended by ``parallel transporting" its endpoints.  Indeed, let $s$ be a soliton of type $ij$ with $s(0),\, s(1) \in \pi^{-1}(z)$.  Now, given a path $q:[0,1] \rightarrow C'$ from $z$ to $z'$, let $q\{n\}$ denote the lift of $q$ to the $n$th sheet of $\Sigma$ defined by lifting the initial point $q(0)$ to sheet $n$; then one can define the transported path,
		\begin{align}
			\mathbb{P}_{q} s &= q\{j\} \star s \star q^{-}\{i\}
			\label{eq:soliton_trans}
		\end{align}
		where $\star$ denotes concatenation of paths, and $q^{-}\{i\}$ is $q\{i\}$ with reversed orientation.  The resulting path on $\Sigma$ has endpoints in $\pi^{-1}(z')$.  If $s$ is an $ij$ soliton, then the path $\mathbb{P}_{q}s$ is a soliton iff $q$ satisfies (\ref{eq:BPS_soliton_cond}) for the same pair of sheets $ij$.

\item $\bigcup_{z \in C'} \Gamma(z,z) \rightarrow C'$ is a local system over $C'$: for any path $q:[0,1] \rightarrow C$ there is a
parallel transport map $P_{q}: \Gamma(q(0),q(0)) \rightarrow \Gamma(q(1),q(1))$, induced by
the map $\mathbb{P}_{q}$ defined above, and only depending on the homotopy class of $q$ relative to the endpoints. (Henceforth we
abbreviate this as ``rel endpoints.'')

		\item  If there is an extension of an $ij$-soliton through a branch cut emanating from an $ij$ branch point, it becomes a $ji$-soliton.  More generally, if a soliton passes through any branch cut, its type is permuted according to the permutation of sheets across the branch cut.

	\end{enumerate}
\end{remark}

Just as with 4D vanilla BPS states, for each $a_{z} \in \Gamma(z,z)$, there is an index $\mu(a_{z}) \in \mathbb{Z}$ that counts BPS solitons of charge $a_{z}$.  Again, this can be defined as a supertrace over an appropriate BPS subspace, however, we will compute it via
geometric methods.  Using the parallel transport map described in the remarks above, this BPS index is also stable along extensions of solitons at generic $z \in C'$; \footnote{In the sense that if $s$ is an $ij$ soliton, and $q$ is a
sufficiently small path on $\Sigma$ satisfying (\ref{eq:BPS_soliton_cond}), then $\mu([s]) = \mu(P_{q}[s])$} it jumps only at points $z \in C'$ where solitons of different types exist and interact.  This motivates the following
(notation-simplifying) definitions.

\begin{definition}[Definitions]\
	\begin{enumerate}
		\item A soliton $s: [0,1] \rightarrow \Sigma$ is said to be of \textit{phase} $\theta$ if it satisfies the BPS condition\footnote{Thought of as a system of equations on each ``partial soliton" as described above.} (\ref{eq:BPS_soliton_cond}) for $\vartheta = \theta$.

		\item A point $z \in C$ is said to \textit{support} a soliton of phase $\vartheta$ if there exists a soliton $s$ with  $[s] \in \Gamma(z,z)$ and $\mu([s]) \neq 0$.  A path $p$ on $C$ supports a family of solitons of phase $\vartheta$ if each point on $p$ supports a soliton fitting into a 1-parameter family of solitons of phase $\vartheta$.  When the phase $\vartheta$ is clear from context, occasionally we will just say that $p$ supports a family of solitons.

		\item Let $p \subset C$ be a path on $C$ supporting a family of solitons of phase $\vartheta$
extending a soliton $s_{0}$ with charge $a_{z} = [s_{0}] \in \Gamma(z,z)$.  With an abuse of notation, occasionally
$a$ will denote any one of the parallel transports of $a_{z}$ along the path $p$.

		\item Let $z \in p$ and $a_{z} \in \Gamma(z,z)$.  If the index $\mu(a_{z})$ is constant for any soliton in the family generated by parallel transports of $a_{z} \in \Gamma(z,z)$ along  $p \subset C$, then we will denote the index by $\mu(a,p) \in \mathbb{Z}$.
		\end{enumerate}
	\end{definition}

\subsubsection*{Framed 2D-4D States}
We consider one final BPS construction: the framed 2D-4D states.  Given $\vartheta \in \mathbb{R}/(2 \pi \mathbb{Z}),\, z_1,\,z_{2} \in C$, and $\wp$ a path on $C$ from $z_1$ to $z_2$, one can associate two surface defects
 $\mathbb{S}_{z_{1}}$ and $\mathbb{S}_{z_{2}}$, along with a supersymmetric interface $L_{\wp, \vartheta}$ between these two surface defects.\footnote{From the four-dimensional perspective, $L_{\wp, \vartheta}$ is a line defect extended along $\mathbb{R}^{0,1}$ and living on the interface between $\mathbb{S}_{z_{1}}$ and $\mathbb{S}_{z_{2}}$.} The interface is supersymmetric in the sense that it preserves two out of the four supercharges preserved by the surface defects; the parameter $\vartheta$ controls which two are preserved.  Framed 2D-4D states are the vacua of the theory after insertion of this defect.

Geometrically, such a state is represented by a path $f: [0,1] \rightarrow \Sigma$ such that there exists a finite subdivision of times
\begin{align*}
	[0,1] &= [0,t_{1}] \cup [t_{1}, t_{2}] \cup \cdots \cup [t_{N-1}, 1]
\end{align*}
and, with respect to this subdivision:
\begin{itemize}
	\item $f|_{[0,t_{1}]}$ and $f|_{[t_{N}, 1]}$ have images in $\pi^{-1}(\wp)$ (in particular, the path begins on a lift of $z_{1}$ and ends on a lift of $z_{2}$).

	\item If $1 < i < N -2$, then $f_{[t_{i}, t_{i + 1}]} $ is either a soliton of phase $\vartheta$, or has image in $\pi^{-1}(\wp)$.
\end{itemize}

 When $f$ is projected to $C$ the resulting path looks like $\wp$ with finitely many diversions to solitons (and back) along the way. In \cite{GMN3}, such a path $f$ is referred to as a \textit{millipede with body $\wp$ and phase $\vartheta$}.

   Similar to solitons, we can classify framed 2D-4D states by their values in a set of charges given by relative
   homology classes $[f]$, for $f$ a millipede; as the geometric description above suggests, now the relative
   cycles can have boundaries on pre-images of two different points on $C$.

\begin{definition}
Let $\wp:[0,1] \rightarrow C$ with $\wp(0) = z$ and $\wp(1) = w$; with a choice of labeling of sheets above $\pi^{-1}(z)$ and $\pi^{-1}(w)$, let $z_{i}$ (resp. $w_{i}$) be a point on the $i$th sheet in $\pi^{-1}(z)$ (resp. $\pi^{-1}(w)$).  Then, the set of charges of framed 2D-4D states corresponding to $\wp$ is
	\begin{align*}
		\Gamma(z,w) &:=  \bigsqcup_{i,j = 1}^{K} \left \{[a] \in H_{1}(\Sigma, \{z_{i}\} \cup \{z_{j}\}; \mathbb{Z}): \text{$a$ is a 1-chain with $\partial a = w_j - z_i$ } \right\}.
	\end{align*}
\end{definition}

Furthermore, for each $a \in \Gamma(z,w)$ we define the counting index $\FOmega(L_{\wp,\vartheta}, a)$ that, once again, can be defined via a supertrace over an appropriate Hilbert space, but we will only utilize its interpretation from a geometric perspective.

\begin{remark}
	It is believed that the theory obtained after insertion of the defect $L_{\wp, \vartheta}$ only depends on the homotopy class (rel boundary) of $\wp$.  This homotopy invariance is the key ingredient that ties the story of spectral networks together.
\end{remark}

\subsubsection{Adding a Little Twist} \label{sec:twist}
Before proceeding to the definition of the $\mathcal{W}_{\vartheta}$ networks, we make an important technical detour.  As discussed in \cite{GMN5}, the indices $\mu(a)$ and $\FOmega(L_{\wp,\vartheta},a)$ are only well-defined up to a sign, due to potential integer shift ambiguities in the fermion number operators that enter their definitions.  To correct these ambiguities globally over all regions of parameter space, it suffices to construct (geometrically motivated) $\mathbb{Z}/2\mathbb{Z}$ extensions of $\Gamma$ and $\Gamma(z,w)$.  First, a bit of notation that will be used throughout this section and part of Appendix \ref{app:herd-appendix}.

\begin{definition}
	Let $S$ be a real surface, then $\xi^{S}: \twid{S} \rightarrow S$ is the unit tangent bundle projection to $S$.
\end{definition}
  The map $\xi^{S}_{*}: H_{1}(\twid{S}; \mathbb{Z}) \rightarrow H_{1}(S; \mathbb{Z})$ has a kernel generated by the homology class that has a representative winding once around some fiber.

\begin{definition}
Let $H \in H_{1}(\twid{\Sigma};\mathbb{Z})$ denote the homology class represented by a 1-chain that winds once around a fiber of $\twid{\Sigma} \rightarrow \Sigma$, then
\begin{align*}
	\twid{\Gamma} := H_{1}(\twid{\Sigma}; \mathbb{Z}) / \left(2H \right).
\end{align*}
	We abuse notation and denote the image of $H$ in $\twid{\Gamma}$ by $H$ again.
\end{definition}
It follows that $\twid{\Gamma}$ is a $\mathbb{Z}/2\mathbb{Z}$ extension of $\Gamma$, i.e there is an exact sequence of abelian groups,
\begin{equation*}
	0 \rightarrow \mathbb{Z}/2\mathbb{Z} \rightarrow \twid{\Gamma} \rightarrow \Gamma \rightarrow 0.
\end{equation*}
Similarly, for framed states and solitons we define extended charge sets.  First we pass through an intermediate construction.
\begin{definition}\
	 Let $\twid{\pi}: \twid{\Sigma} \rightarrow \twid{C}$ be the restriction of $d\pi: T\Sigma \rightarrow TC$ to the unit tangent bundle.   For fixed $\twid{z},\, \twid{w} \in \twid{C}$, choose a labeling of sheets above $\pi^{-1}(z)$ and $\pi^{-1}(w)$; let $z_{i}$ (resp. $w_{i}$) be a point on the $i$th sheet in $\pi^{-1}(z)$ (resp. $\pi^{-1}(w)$), then define
			\begin{equation}
				\begin{aligned}
				G_{ij}(\twid{z},\twid{w}) & := \left\{[a] \in H_{1}(\twid{\Sigma}, \{\twid{z}_{i}\} \cup \{\twid{w}_{j}\}; \mathbb{Z}) : \text{$a$ is a 1-chain with $\partial a = \twid{w}_{j}- \twid{z}_{i}$} \right \}, \\
				G(\twid{z}, \twid{w}) & :=  \bigsqcup_{i,j = 1}^{K} G_{ij}(\twid{z},\twid{w}).
				\end{aligned}
				\label{eq:G_def}
			\end{equation}
\end{definition}

\begin{remark}
	 $G(\twid{z}, \twid{w})$ is equipped with an $H_{1}(\twid{\Sigma};\mathbb{Z})$ action given by the addition of a closed cycle (at the level of chains).
\end{remark}

This allows us to make the following definition,
\begin{definition}
		\begin{align}
			\twid{\Gamma}(\twid{z},\twid{w}) &:= G(\twid{z},\twid{w})/\langle 2 H \rangle.
			\label{eq:twid_gamma_def}
		\end{align}
\end{definition}

Sometimes it is useful to view $\twid{\Gamma}(\twid{z},\twid{w})$ as a disjoint union of quotients of $G_{ij}$:

\begin{definition}
		\begin{align*}
			\twid{\Gamma}_{ij}(\twid{z},\twid{w}) &:= G_{ij}(\twid{z},\twid{w})/\langle 2 H \rangle.
		\end{align*}
\end{definition}
So we may write,
\begin{align*}
	\twid{\Gamma}(\twid{z},\twid{w}) &:= \bigsqcup_{i,j =1}^{K} \twid{\Gamma}_{ij}(\twid{z},\twid{w}).
\end{align*}

\begin{remark}
	\item $\twid{\Gamma}(\twid{z},\twid{w})$ is equipped with a $\twid{\Gamma}$ action, descending from addition of a closed cycle
with a relative cycle.  For $\gamma \in \twid{\Gamma}$ and $a \in \twid{\Gamma}(\twid{z},\twid{w})$ we will denote
this action by $\gamma: a \mapsto a + \gamma = \gamma + a$.  In fact, for any ordered pair $ij$, $\twid{\Gamma}_{ij}(\twid{z}, \twid{w})$  is a torsor for $\twid{\Gamma}$.
\end{remark}

$\twid{\Gamma}(\twid{z},\twid{w})$ carries an extra $\mathbb{Z}/2\mathbb{Z}$'s worth of ``winding" information in the sense that $\twid{\Gamma}(\twid{z}, \twid{w}) \overset{\mathrm{proj}}{\rightarrow} \Gamma(z,w)$ is a principal $\mathbb{Z}/2\mathbb{Z}$ bundle, with $\mathrm{proj}$ given by forgetting lifts\footnote{More precisely, $\operatorname{proj}$ is the map descending from the induced map on relative homology $(\xi^{\Sigma})_{*}: H_{1} \left(\twid{\Sigma}, \pi^{-1}(\twid{z}) \cup \pi^{-1}(\twid{w}); \mathbb{Z} \right) \rightarrow H_{1} \left(\Sigma, \pi^{-1}(z) \cup \pi^{-1}(w); \mathbb{Z} \right)$ where $z = \left(\pi \circ \xi^{\Sigma} \right)(\twid{z})$ and $w = \left(\pi \circ \xi^{\Sigma} \right)(\twid{w})$.}, and the $\mathbb{Z}/2 \mathbb{Z}$ action given by adding $H$.

Now, a soliton is a smooth curve on $\Sigma$; furthermore, the tangent vectors at the endpoints (which lie on disjoint sheets) of a soliton are oppositely oriented in the sense that their pushforwards to $C$ are oppositely oriented.

\begin{definition}
 	Let $\twid{z} \in \twid{C}$, then $-\twid{z} \in \twid{C}$ is the unit tangent vector pointing in the opposite direction to $\twid{z}$.
\end{definition}

\begin{remark}
	To every soliton (represented by a smooth path) there is a natural lifted charge in $\twid{\Gamma}(\twid{z}, - \twid{z})$ that descends from the relative homology class of the soliton's tangent framing lift.
\end{remark}

We introduce one final piece of technology.  First, note that for each $\twid{z} \in \twid{C'}$ there is a disjoint union of $K$ lattices inside of the set $\twid{\Gamma}(\twid{z},\twid{z})$:
\begin{equation*}
	\bigsqcup_{i = 1}^{K} \twid{\Gamma}_{ii}(\twid{z}, \twid{z}) \subset \twid{\Gamma}(\twid{z},\twid{z}).
\end{equation*}
Any representative of an element in $\twid{\Gamma}_{ii}(\twid{z}, \twid{z})$ has zero boundary, hence, is actually a cycle.  Indeed,there is a canonical ``basepoint forgetting" isomorphism of lattices $\twid{\Gamma}_{ii}(\twid{z}, \twid{z}) \cong \twid{\Gamma}$ for each $i = 1, \cdots, K$, descending from the identity map at the level of chain representatives.  This allows us to define the \textit{closure} map.
\begin{definition}
	\begin{align*}
		\cl: \bigcup_{i = 1}^{K} \twid{\Gamma}_{ii}(\twid{z}, \twid{z}) \rightarrow \twid{\Gamma}
	\end{align*}
	is the map which acts on each component by the ``basepoint-forgetting" map described above.
\end{definition}

Now, due to the sign ambiguity in $\mu$ and $\FOmega$ then, na\"{i}vely, only their absolute values are well-defined: i.e. we have functions,
\begin{align*}
	\mu_{ \geq 0} &: \bigcup_{z \in C'} \Gamma(z,z) \rightarrow \mathbb{Z}_{\geq 0}\\
	\FOmega_{ \geq 0} ( \wp, \cdot) &: \bigcup_{(z,w) \in C' \times C'} \Gamma(z,w) \rightarrow \mathbb{Z}_{\geq 0}.
\end{align*}
However, with our ``lifted charge" definitions, we can lift $\mu_{\geq 0}$ to a function $\mu: \bigcup_{\twid{z} \in \twid{C'}} \twid{\Gamma}(\twid{z}, -\twid{z}) \rightarrow \mathbb{Z}$ such that $\forall a \in \bigcup_{\twid{z} \in \twid{C'}} \twid{\Gamma}(\twid{z}, -\twid{z})$,
\begin{equation}
	\begin{aligned}
		\left| \mu(a) \right| &= \mu_{\geq 0}( \xi_{*}^{\Sigma} a)\\
		\mu(a + H) &= - \mu(a).
		\label{eq:mu_resolved}
	\end{aligned}
\end{equation}
Similarly, fixing a path $\wp$ on $C$, the framed BPS degeneracies lift to well-defined functions $\FOmega(L_{\wp,\vartheta}, \cdot): \bigcup_{(\twid{z},\twid{w}) \in \twid{C'} \times \twid{C'}} \twid{\Gamma}(\twid{z}, \twid{w}) \rightarrow \mathbb{Z}$ such that $\forall a \in \bigcup_{(\twid{z},\twid{w}) \in \twid{C'} \times \twid{C'}} \twid{\Gamma}(\twid{z}, \twid{w})$,
\begin{equation}
	\begin{aligned}
		\left| \FOmega \left(L_{\wp,\vartheta},a \right) \right| &=  \FOmega_{\geq 0} \left (L_{\wp,\vartheta}, \xi_{*}^{\Sigma} a \right)\\
		\FOmega \left(L_{\wp,\vartheta}, a + H \right) &= -\FOmega \left(L_{\wp,\vartheta}, a \right).
	\end{aligned}
	\label{eq:fomega_resolved}
\end{equation}

\subsection{The $\mathcal{W}_{\vartheta}$ Networks} \label{sec:wtheta}
Using (\ref{eq:BPS_soliton_cond}), we can produce a concrete picture of (the projections to $C$ of) $ij$-solitons on the curve $C$.  This motivates the following definitions.

\begin{definition}
	Fix $\vartheta \in \mathbb{R}/2\pi \mathbb{Z}$, for each (ordered) pair of sheets $ij$ we define a (real) oriented line field $l_{ij}(\vartheta)$ on $C^{c}$ given at every $z \in C^{c}$ by
	\begin{equation*}
		l_{ij,z}(\vartheta) :=  \left \{v \in T_{z} C :
	\operatorname{Im} \left[ e^{-i \vartheta} \left \langle \lambda_{ij}, v \right \rangle \right]=0 \right \},
	\end{equation*}
	with $v$ positively oriented if $\operatorname{Re} \left[ e^{-i \vartheta}  \left \langle \lambda_{ij}, v \right \rangle \right]>0$.
\end{definition}

Given an integral curve $p$ of $l_{ij}(\vartheta)$, the orientation of $l_{ij}(\vartheta)$ tells us how to lift the curve back to a curve $p_{\Sigma}$ on $\Sigma$.

\begin{definition}
	Any integral curve $p$ (on $C'$) of $l_{ij}(\vartheta)$ has a lift to a curve $p_{\Sigma}$ on $\Sigma$ defined as the union of $p\{i\}$ (the lift of $p$ to the $i$th sheet), and $p^{-}\{j\}$ (the lift of $p$ to the $j$th sheet, reversing orientation).
 \end{definition}

\begin{remark}[Remarks] \
	\begin{itemize}
		\item Fix $z_{*} \in C^{c}$ and take a neighborhood $U$ of $z_{*}$ that does not contain any branch cuts of type $ij$.  Then for each ordered pair $ij$ we can define local coordinates $w_{ij}: U \rightarrow \mathbb{C}$ by
		\begin{equation}
			w_{ij}(z) = \int_{z_{*}}^{z} \left(\lambda_{i} - \lambda_{j} \right).
			\label{eq:wij_coord}
		\end{equation}
		 In these coordinates, the integral curves of $l_{ij}(\vartheta)$ are precisely the straight lines of inclination $\vartheta$.

		\item Note that the line field $l_{ji}(\vartheta)$ is just $l_{ij}(\vartheta)$ with reversed orientation.

	\item On a cycle surrounding a branch point of type $ij$, the monodromy action induces $\lambda_{ij}  \mapsto \lambda_{ji} = -\lambda_{ij}$; hence, $l_{ij}(\vartheta) \mapsto l_{ji}(\vartheta)$ (i.e., the line field orientation reverses when passing through a branch cut extending from a branch point.)

	\end{itemize}
\end{remark}

We can finally define the (real) codimension-1 networks of interest.
\begin{definition}
	\begin{equation*}
		\CW_{\vartheta} = \bigcup_{\text{ordered pairs $ij$}} \left\{p: \text{$p$ is an integral curve of $l_{ij}(\vartheta)$ and $p$ supports a soliton of phase $\vartheta$}  \right\} \subset C'.
	\end{equation*}
\end{definition}
The network $\CW_{\vartheta}$ is composed of individual integral curve segments, which may interact and join each other at vertices on $C'$.
\\
\begin{definition}[Definitions]\
	\begin{enumerate}
		\item An individual integral curve segment on $\CW_{\vartheta}$ is called a \textit{street}.\footnote{In \cite{GMN5} these were also referred to as $S$-walls.}  A \textit{street of type $ij$} is a street that is an integral curve of $l_{ij}(\vartheta)$.

		\item A joint is a point on $C'$ where two or more streets of different types meet.
	\end{enumerate}
\end{definition}

The upshot of all these constructions is that now we have a solidified picture of solitons via a network on $C'$.  Indeed, we can lift $\CW_{\vartheta}$ to a graph on $\operatorname{Lift}(\CW_{\vartheta}) \subset \Sigma$ by taking the union of the lifts (as defined above) $p_{\Sigma}$ of each street $p$.
  Then an $ij$ soliton of phase $\vartheta$ traces out a path supported on $\operatorname{Lift}(\CW_{\vartheta})$, and begins and ends on points $z_{i},\,z_{j}$ that are lifts to the $i$th and $j$th streets (respectively) of a point $z$ on a street of type $ij$.  In particular, an $ij$ street of $\CW_{\vartheta}$ represents the endpoints of a set of solitons of type $ij$.

 From a constructive viewpoint, however, the reader may feel unsatisfied as we have not yet defined how to determine the condition that $p \subset \CW_{\vartheta}$ actually \textit{supports} a soliton of phase $\vartheta$, i.e. $\mu(a,p) \neq 0$, for some $a$ the charge of a soliton of phase $\vartheta$. To fill this void we remark that there are exactly three integral curves of $l_{ij}(\vartheta) \cup l_{ji}(\vartheta)$ emerging from each $ij$-branch point.  On each such integral curve $p$ there is a family of solitons represented by ``small" paths:
  for $z \subset p$ and $z_{i}, z_{j} \in \pi^{-1}(z)$ lifts of $z$ to sheet $i$ and sheet $j$ (respectively), there is a soliton supported on $p_{\Sigma}$ traveling from $z_{i} \in \Sigma$, through the ramification point on $\Sigma$, to $z_{j} \in \Sigma$.  Such solitons become arbitrarily light as $z$ approaches the branch point.  Furthermore, as argued in \cite{GMN5}, letting $a$ be the (lifted) charge of any soliton in this family, we assign
 \begin{align}
	\mu(a ,p) &= +1.
	\label{eq:simpleton_input}
\end{align}
\begin{definition}[Terminology]
	The ``light" solitons described in the previous paragraph will be called \textit{simpletons}.
\end{definition}

We defer the problem of determining the soliton indices $\mu$ on all other streets until the appropriate definitions are developed in the next section; for now it will suffice to say that, with this condition, the soliton indices on all other streets can be determined via a set of algebraic equations.

\subsubsection{$\CK$-walls and Degenerate Networks} \label{sec:degen_net}
Of particular interest in this paper will be $\CW_{\vartheta}$ networks of a very special type.

\begin{definition}
	 A street $p \subset \CW_{\vartheta}$ is \textit{two-way} if it consists of a coincident $ij$-street and a $ji$-street.  Equivalently, $p$ supports $ij$-solitons and $ji$-solitons.  A street that is not two-way is called \textit{one-way}.  A network that contains a two-way street is said to be \textit{degenerate}.
\end{definition}

We adopt the following convention in order to keep track of the individual directions of the constituent one-way streets of a two-way street.

\begin{definition}[Convention]
	Let $p$ be a two-way street consisting of coincident $ij$ and $ji$-streets, then we will say $p$ is of type $ij$ and assign it the orientation of its constituent $ij$-street. (Or, equivalently, we will say $p$ is of type $ji$ and assign it the orientation of its constituent $ji$-street.)
\end{definition}

As described in \cite{GMN5}, sec. 6.2, for generic values of $\vartheta$, the network $\CW_{\vartheta}$ will only contain one-way streets due to a bifurcation behavior of integral curves near branch points.  However, at critical values $\vartheta_{c} \in \mathbb{R}/\mathbb{Z}$, an $ij$ street will collide with a $ji$ street and the network $\CW_{\vartheta_{c}}$ will contain two-way streets.  Now we make an important claim:

\begin{center}
	$\CW_{\vartheta}$ contains a two-way street $\Rightarrow \exists$ a homologically non-trivial closed loop on $\Sigma$ satisfying (\ref{eq:BPS_cond_vanilla}) for some phase $\vartheta \in \mathbb{R}/2\mathbb{Z}$.
	\label{claim_2way}
\end{center}

To see this, fix a point $z \in p \subset C$ on any two-way street $p$; without loss of generality we will say $p$ is of type $ij$.  Then $z$ supports a soliton of type $ij$ and a soliton of type $ji$, both of the same phase $\vartheta$; the concatenation of these two paths yields a closed loop $l$ satisfying (\ref{eq:BPS_cond_vanilla}) for the phase $\vartheta$.  Moreover, this loop is homologically non-trivial.  Indeed, the period of $l$ is just the sum of the periods of the two solitons forming it.  However, both have periods (central charges) of the same phase; so the sum must be nonzero.

Thus, via the claim, a degenerate network automatically leads to a possible 4D BPS state of charge $[l] \in \Gamma$; in fact, there are possible BPS states of charges $n [l],\, n \in \mathbb{Z}_{>0}$.  All that remains is to determine the BPS indices $\Omega(n [l])$ which, as expressed more explicitly below, are computable from the soliton data supported on $\CW_{\vartheta}$.

In practice, degenerate networks can be found by looking for discontinuous changes in the topology of $\CW_{\vartheta}$ as $\vartheta$ is varied.  Indeed, if a region $R \subset \mathbb{R}/\mathbb{Z}$ does not contain any degenerate networks then, as $\vartheta$ is varied continuously in $R$, the network $\CW_{\vartheta}$ also varies continuously (in the sense described in \cite{GMN5}).
However, if the region $R$ contains a single critical angle $\vartheta_{c}$, the bifurcation of integral curves near a branch point induces a discontinuous change in the topology of $\CW_{\vartheta}$ as $\vartheta$ is varied\footnote{However, there may be an accumulation point  of critical angles as in the picture of the vector multiplet when $K=2$ (see \cite{GMN5}).  Around such an accumulation point the topology of $\CW_{\vartheta}$ rapidly changes, and there is no open region containing the accumulation point where the topology smoothly varies.
Even ``worse," as we will see, the critical angles can densely fill an open interval.} past $\vartheta_{c}$. (If we consider the parameter space of $\vartheta$ and the
Coulomb branch then the locus where degenerate networks appear defines \textit{$\CK$-walls}.)

\subsubsection{Formal Variables} \label{sec:form_var}
In order to construct the generating functions that keep track of various BPS degeneracy indices, it is helpful to construct spaces of formal variables with some algebraic structure.
\begin{definition}
		$\formgamma$ is the commutative ring of formal series generated by formal variables $X_{\gamma},\, \gamma \in \twid{\Gamma}$ such that
		\begin{align*}\\
			X_{0} &= 1,\\
			X_{H} &= -1,\\
			X_{\gamma} X_{\gamma'} &= X_{\gamma + \gamma'}.
		\end{align*}
\end{definition}

To define an algebraic structure for formal variables in 2D-4D/soliton charges, we note that there is a partially defined ``addition" operation.

\begin{definition}[Definitions]\
	\begin{enumerate}
			\item Let $a \in \twid{\Gamma}(\twid{z}_{1},\twid{z}_{2})$, then
			\begin{align*}
				\operatorname{end}(a) &= \twid{z}_{2}\\
				\operatorname{start}(a) &= \twid{z}_{1}.
			\end{align*}

			\item Let $a, b \in \bigcup_{\twid{z}, \twid{w} \in \twid{C}} \twid{\Gamma}(\twid{z},\twid{w})$, then if $\text{end}(a) = \text{start}(b)$ there is a well-defined operation (concatenation of paths) $a + b \in \bigcup_{\twid{z}, \twid{w} \in \twid{C}} \twid{\Gamma}(\twid{z},\twid{w})$ descending from the usual addition of relative homology cycles.
		\end{enumerate}
\end{definition}
With this we can define the space of interest.

\begin{definition} The \emph{homology path algebra}
	$\CA$ is the non-commutative $\formgamma$-algebra of formal series generated by formal variables $X_{a}$, for every  $a \in \bigcup_{\twid{z}, \twid{w} \in \twid{C}} \twid{\Gamma}(\twid{z},\twid{w})$; such that
	\begin{enumerate}
		\item For $\gamma \in \twid{\Gamma}$,
		\begin{align*}
			X_{\gamma} X_{a} = X_{a + \gamma} = X_{\gamma + a},
		\end{align*}

		\item for any $a,\, b \in \bigcup_{\twid{z}, \twid{w} \in \twid{\Sigma}} \twid{\Gamma}(\twid{z},\twid{w})$
		\begin{align*}
			X_{a} X_{b} &=
			\left\{
			\begin{array}{ll}
				X_{a + b}, & \text{if $\operatorname{end}(a) = \operatorname{start}(b)$}\\
				0, & \text{otherwise}
			\end{array}
			\right. .
		\end{align*}
	\end{enumerate}
\end{definition}

There are two important $\formgamma$-subalgebras of $\CA$.

\begin{definition}[Definition]\
	\begin{enumerate}
			\item $\CA_{S}$ is the $\formgamma$-subalgebra generated by formal variables in soliton charges\\
			 $a \in \bigcup_{\twid{z} \in \twid{C}} \twid{\Gamma}(\twid{z},-\twid{z})$.

		\item $\CA_{C}$ is the (commutative) $\formgamma$-subalgebra generated by formal variables in\\ $a \in \bigcup_{\twid{z} \in \twid{C}} \bigsqcup_{i = 1}^{K} \twid{\Gamma_{ii}}(\twid{z},\twid{z})$.
	\end{enumerate}
\end{definition}

The closure map can be easily extended to $\CA_{C}$.
\begin{definition}
	\begin{align*}
		\cl:& \CA_{C} \rightarrow \formgamma
	\end{align*}
	is the linear extension of the map
	\begin{align*}
		\cl(X_{a}) &= X_{\cl(a)}.
	\end{align*}
\end{definition}

We now define generating functions for BPS indices.
\begin{definition}\
	 For each path $\wp$ from $z \in C$ to $w \in C$ that represents an interface $L_{\wp,\vartheta}$,
we associate the framed generating function
		\begin{align*}
			F(\wp,\vartheta) &:= \sum_{a_{*} \in \Gamma(z,w)} \FOmega(L_{\wp,\vartheta},a) X_{a} \in \CA,
		\end{align*}
	where $a \in \twid{\Gamma}(\twid{z}, \twid{w})$ is a lift of the charge $a_{*} \in \Gamma(z,w)$ such that $\twid{z}$ and $\twid{w}$
	are the unit tangent vectors at the ends of $\wp$.
	\end{definition}

	For each street $p$ of type $ij$, we associate two \textit{soliton generating functions}: $\Upsilon(p)$, that encodes the indices of solitons of type $ij$, and $\Delta(p)$, that encodes the indices of solitons of type $ji$.

\begin{definition}
	 Let $z \in p \subset C$, then we define
				\begin{align}
					\Upsilon_{z}(p) &:= \sum_{a_{*} \in \Gamma_{ij}(z,z)} \mu(a) X_{a} \in \CA_{S}
					\label{eq:upsilon_def} \\
					\Delta_{z}(p) &:= \sum_{b_{*} \in \Gamma_{ji}(z,z)} \mu(b) X_{b} \in \CA_{S},
			\label{eq:delta_def}
				\end{align}
				where $a \in \twid{\Gamma}_{ij}(\twid{z}, -\twid{z})$, $b \in \twid{\Gamma}_{ji}(-\twid{z}, \twid{z})$ denote respective lifts  of $a_{*} \in \Gamma_{ij}(z,z)$ and $b_{*} \in \Gamma_{ji}(z,z)$, for $\twid{z} \in \twid{C}$ the unit tangent vector agreeing with the orientation of $p$ at the point $z \in C$.\footnote{As $\twid{\Gamma}(\twid{z}, -\twid{z})$ is a principal $\mathbb{Z}/2\mathbb{Z}$ bundle over $\Gamma(z,z)$, there are two possible lifts of $a_{*}$ related by addition of $H$.  Via $X_{H} = -1$ along with (\ref{eq:mu_resolved}) and (\ref{eq:fomega_resolved}), the definition of $\Upsilon_{z}(p)$ is independent of the choice of lift.  This argument also applies to $\Delta_{z}(p)$.}
\end{definition}

\begin{definition}
	From the soliton generating functions on a street $p$, we can define the
	\textit{street factor},
				\begin{equation*}
					\begin{aligned}
						Q(p) &:= \cl \left[ 1 + \Upsilon_{z}(p) \Delta_{z}(p) \right]\\
						&= 1 + \sum_{a_{*} \in \Gamma_{ij}(z,z), b_{*} \in \Gamma_{ji}(z,z)} \mu(a) \mu(b) X_{\cl(a+b)} \in \formgamma.
					\end{aligned}
					\label{eq:def-Q}
				\end{equation*}
	where $z \in C$ is any point on $p$.
\end{definition}

\begin{remark}
	As the notation suggests, $Q(p)$ is independent of the choice of point $z$.  This follows as the index $\mu(a)$ is constant as any charge $a$ is parallel transported along any path supported on $p \subset C$.  By the same reasoning, for any $z,z' \in p$, $\Upsilon_{z}(p)$ and $\Upsilon_{z'}(p)$ are related by applying an appropriate parallel transport map \footnote{For this reason, the point $z$ in soliton generating functions is often dropped as in the calculations of Appendix \ref{app:herd-appendix}.} (similarly for $\Delta_{z}(p)$ and $\Delta_{z'}(p)$).
\end{remark}

And now for the punchline.

\subsubsection{Computing $\Omega(n \gamma_{c})$} \label{sec:comp_Omega}
The power of the spectral network machine can be summarized with the following squiggly arrows:
\begin{center}
	\begin{tabularx}{\linewidth}{XcccX}
		Jumping of Framed 2D-4D Spectrum + Homotopy Invariance of $L_{\wp,\vartheta}$ & $\overset{(A)}{\xrsquigarrow{\hspace{4mm}}}$ & Soliton Spectrum & $\overset{(B)}{\xrsquigarrow{\hspace{4mm}}}$ & (Vanilla) 4D spectrum.
	\end{tabularx}
\end{center}

To understand $(A)$: the framed generating function $F(\wp,\vartheta)$ is piecewise constant in the sense that as the endpoints of $\wp$ are varied on $C - \CW_{\vartheta}$, then $F(\wp,\vartheta)$ does not vary in $\CA$; however, if an endpoint of $\wp$ is varied across a street of $\CW_{\vartheta}$, then $F(\wp, \vartheta)$ will jump in a manner depending on the spectrum of solitons located on that street.  Indeed, $F(\wp, \vartheta)$ is the sum of the charges of ``millipedes," and as the ``body" $\wp$ of each such millipede crosses the street $p$, then the millipede can gain an extra leg by detouring along a soliton supported along $p$; hence, the spectrum of 2D-4D states (represented by millipedes) will jump.  To reproduce the soliton spectrum we utilize the homotopy invariance of the operator $L_{\wp,\vartheta}$ to equate the different jumps of $F(\wp, \vartheta)$ across different, but homotopic (rel endpoints), paths $\wp$.  The resulting equations are equivalent to conditions on the soliton generating
functions.  These conditions, combined with the simpleton input data (\ref{eq:simpleton_input}),
  allow us to completely determine the soliton generating functions, which encapsulate the soliton spectrum.

To describe $(B)$, let $\Gamma_c \subset \Gamma$ be the lattice of charges $\gamma$ with $e^{-i \vartheta_{c}} Z_\gamma
\in \IR_-$; then the degenerate network $\CW_{\vartheta_{c}}$ captures all of the 4D BPS states carrying charges
$\gamma \in \Gamma_c$.
Their spectrum can be extracted from the generating functions $Q(p)$. But, first we have to deal with a technical point.

\begin{definition}[Definitions]\
	\begin{enumerate}
		\item	For every curve $q$ on a surface $S$, there is a canonical ``lift" $\widehat{q}$ to a curve on $\twid{S}$, given by the tangent framing.

		\item	For each $\gamma \in \Gamma$, we define another lift $\twid{\gamma} \in \tilde\Gamma$
			by the following rule.
			First, represent $\gamma$ as a union of smooth closed curves $\beta_m$ on $\Sigma$.
			Then $\twid{\gamma}$ is the sum of $\widehat{\beta}_m$, shifted by $\left( \sum_{m \le n} \delta_{mn} + \# (\beta_m \cap \beta_n) \right) H$
			(of course, because we work modulo $2H$, all that matters here is whether this sum is odd or even.)
	\end{enumerate}
\end{definition}
One can check directly (see Appendix \ref{app:sign-rule})
that $\tilde\gamma$ so defined is independent of the choice of how we represent
$\gamma$ as a union of $\beta_m$; this requirement is what forced us to add the tricky-looking shift.

Then, for each street $p$, we factorize $Q(p)$ as a product:
\begin{definition}
	\begin{equation} \label{eq:Q-exp}
		Q(p) = \prod_{\gamma \in \Gamma_c} (1 - X_{\twid{\gamma}})^{\alpha_\gamma(p)}.
	\end{equation}
	This representation determines the coefficients $\alpha_\gamma(p)$.
\end{definition}

\begin{definition}
Let $\lift{p}_{\Sigma} \in C_{1}(\Sigma; \mathbb{Z})$ be the one-chain corresponding to the lift $p_{\Sigma}$, then we define\footnote{Note that the sum over streets in (\ref{eq:L_def}) reduces to a sum over two-way streets; indeed, $Q(p) \neq 1$ iff $p$ is two-way.}
	\begin{align}
		L(\gamma) &:= \sum_{\text{streets p}} \alpha_\gamma(p) \lift{p}_{\Sigma} \in C_{1}(\Sigma; \mathbb{Z}).
		\label{eq:L_def}
	\end{align}
\end{definition}
Now, as shown in \cite{GMN5}, the magic of this definition is that $L(\gamma)$
 is actually a 1-cycle satisfying the BPS condition (\ref{eq:BPS_cond_int}) for $\vartheta = \vartheta_{c}$. \footnote{This last comment follows from the fact that $\int_{p_{\Sigma}} \lambda = \int_{p} \lambda_{ij} \in e^{i \vartheta_{c}} \mathbb{R}_{<0} $ for any street $p$ of type $ij$.} Let us make the further assumption that $\Gamma_{c}$ is a rank-1 lattice, which holds automatically off of the walls of marginal stability on $\CB$, then it follows that both $\gamma$ and $[L(\gamma)]$ are multiples of a choice of generator $\gamma_{c} \in \Gamma_{c}$.  With this in mind, the journey to the end of the squiggly arrow $(B)$ follows by analyzing the jumping of $F(\wp,\vartheta)$, but now as $\vartheta$ is varied across the critical angle $\vartheta_{c}$ (fixing $\wp$).  The resulting analysis (see \cite{GMN5}, sec. 6) leads us to the desired result:
\begin{align}
	[L(\gamma)] &=  \Omega(\gamma) \gamma, \qquad \gamma \in \Gamma_c,
	\label{eq:Omega_L_rel}
\end{align}
from which all BPS indices of 4D BPS states with central charge phase $\vartheta_{c}$ can be computed.

\subsubsection*{Abstract Spectral Networks}
It is possible to abstract the properties of the $\CW_{\vartheta}$ networks in order to draw networks on $C$ that do not necessarily come from integral curves of (\ref{eq:BPS_soliton_cond}).  It is not necessary to give a precise list of the properties here, and we instead refer the interested reader to Section 9 of \cite{GMN5}. There, the abstracted networks are particularly useful for defining the ``non-abelianization map" between moduli spaces of flat $GL(1)$-bundles on $\Sigma$, and flat $GL(K)$-bundles on $C$.  In this paper, however, our interest in abstract spectral networks will be in constructions of \textit{potential} $\CW_{\vartheta}$ networks.  Indeed, the $m$-herds mentioned in the introduction, and introduced in Section \ref{sec:herds}, are examples of abstract networks on an arbitrary curve $C$.  By searching the parameter space of the pure $SU(3)$ theory, where $C = S^{1} \times \mathbb{R}$ and $K=3$, it turns out that a large subset of $m$-herds actually arise as $\CW_{\vartheta}$ networks
at various points on the Coulomb branch.

\section{Spectral network analysis of a wild point on the Coulomb branch} \label{sec:spectral-wild-SU3}

\begin{figure}[t!]
	\begin{center}
		 \includegraphics[scale=0.55]{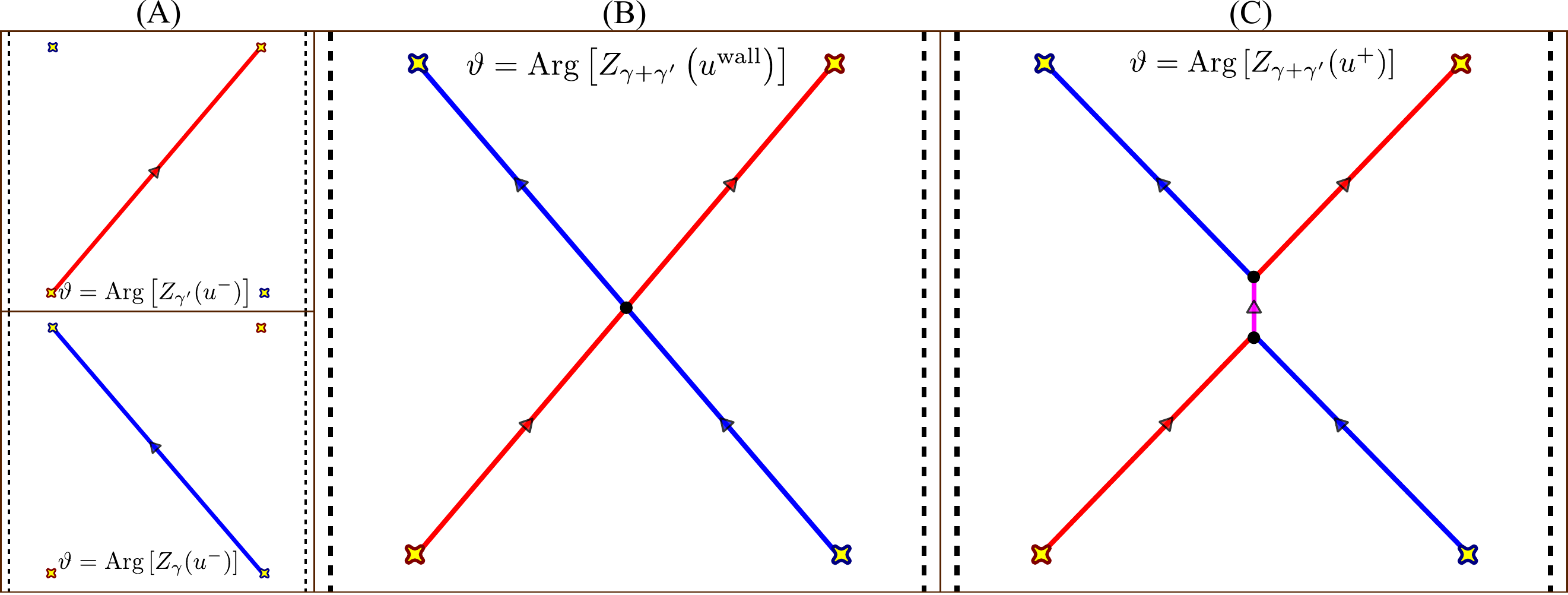}
		\caption{A hypothetical wall-crossing of two hypermultiplets with charges $\gamma,\, \gamma'$ such that $\langle \gamma, \gamma' \rangle =1$.  Streets of type $12$ are shown in red, $23$ in blue, and $13$ in fuchsia; only two-way streets are depicted.  Arrows denote street orientations according to the convention described in Section \ref{sec:degen_net}. Yellow crosses denote branch points.  Arrows denote the direction of solitons of type $12,\,23,$ or $13$ (according to the street).  The black dotted lines are identified to form the cylinder. $(A)$: The two hypermultiplet networks at a point $u^{-}$ just ``before" the wall of marginal stability. $(B)$: The hypermultiplet networks at a point $u^{\text{wall}}$ on the wall of marginal stability and at phase $\vartheta = \arg\left[Z_{\gamma}(u^{\text{wall}}) \right] = \arg\left[Z_{\gamma'}(u^{\text{wall}}) \right] = \arg\left[Z_{\gamma + \gamma'}(u^{\text{wall}}) \right]$. $(C)$: Slightly ``after" the wall at a point $u^{+}$, a BPS bound state of charge $\gamma + \gamma'$ is born and a two-way street of type $13$ ``grows" as one proceeds away from the wall.
\label{fig:1-herd_motivation}}
	\end{center}
\end{figure}

\begin{figure}[t!]
	\begin{center}
		 \includegraphics[scale=0.55]{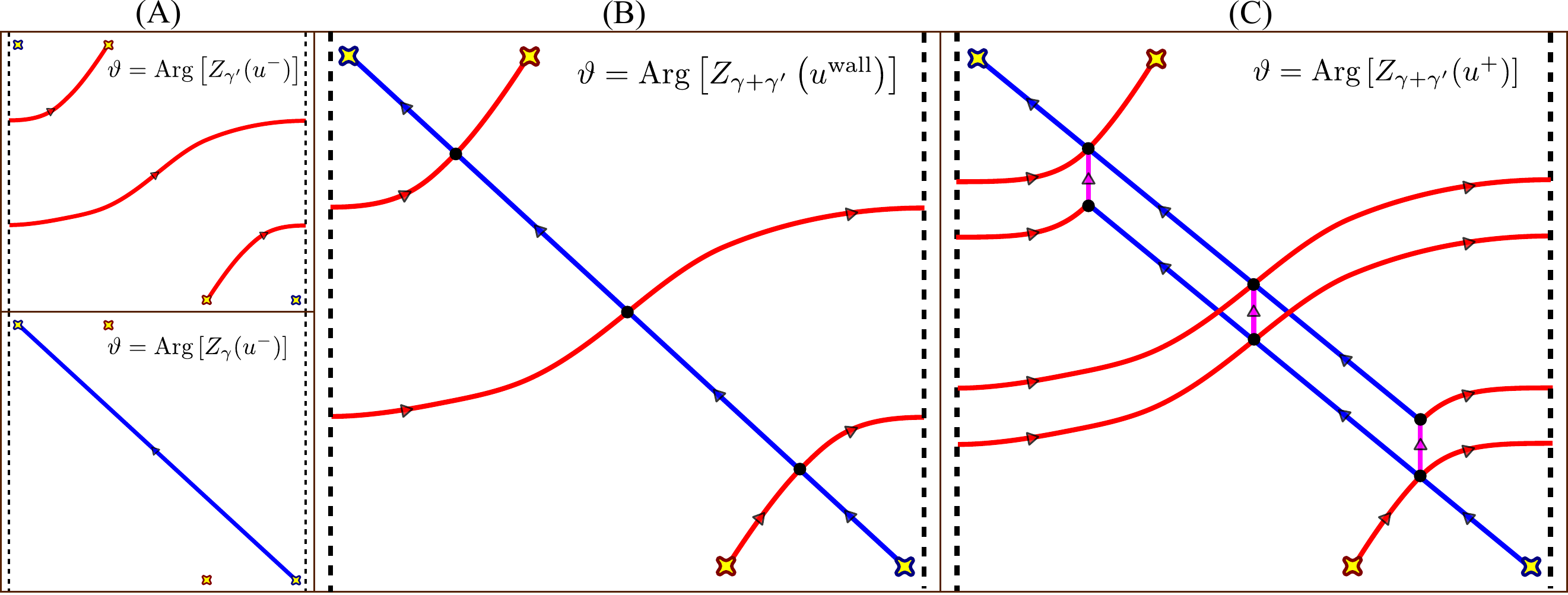}
		\caption{A hypothetical wall-crossing of two hypermultiplets with charges $\gamma,\, \gamma'$ such that $\langle \gamma, \gamma' \rangle =3$.  The story is similar to that described in the caption of Fig.~\ref{fig:1-herd_motivation}. \label{fig:3-herd_motivation}}
	\end{center}
\end{figure}

\subsection{Horses and Herds} \label{sec:herds}
We begin by describing a sequence of spectral networks that may arise in the hypothetical wall-crossing between two BPS hypermultiplets of charges $\gamma,\,\gamma' \in \Gamma$ such that $\langle \gamma, \gamma' \rangle = m$.  Indeed, assume at some point on the Coulomb branch there are two BPS states
(occurring at different phases) such that the degenerate network associated to each state has a single two-way street given by a simple curve passing through two branch points of the same type (frame $(A)$ of Figs. \ref{fig:1-herd_motivation}-\ref{fig:3-herd_motivation}); such spectral networks are associated to BPS hypermultiplets.  Now, assume that there exists a marginal stability wall on the Coulomb branch associated to the (central charge phase) crossing of these two hypermultiplets (and no other BPS states).  On the other side of the wall, a possible bound state of charge $\gamma + \gamma'$ may be formed (where $\gamma,\, \gamma'$ are the charges of the original hypermultiplets).  Figs.
\ref{fig:1-herd_motivation}-\ref{fig:3-herd_motivation} depict three hypothetical snapshots along a path passing through the wall of marginal stability for the cases $m= 1,3$; frame $(C)$ depicts a guess at the appearance of the degenerate network associated to the bound state of charge $\gamma + \gamma'$.  After drawing such pictures for progressively higher $m$, and given a sufficient dose of mildly-confused staring, one will begin to notice that the (two-way streets of) networks associated to the bound state of charge $\gamma + \gamma'$ can be decomposed into $m$-components that look like ``extended" saddles; as they are the generalization of saddles we have no choice but to call each such component a ``horse."

\begin{figure}
	\begin{center}
		 \includegraphics[scale= 0.75]{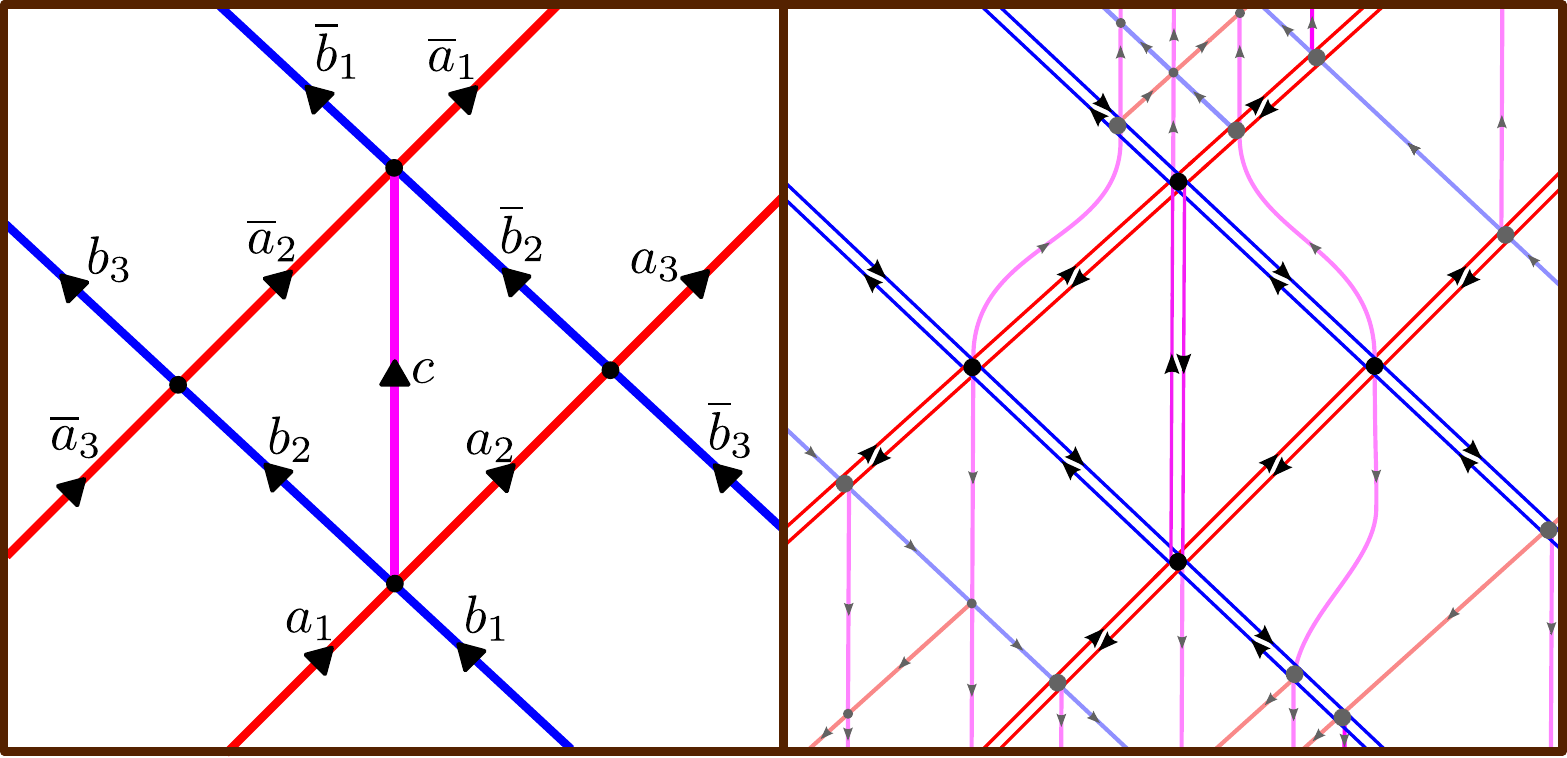}
		\caption{\textit{Left Frame}: Two-way streets of a horse on some open disk $U$; the solid streets depicted are capable of being two-way; one-way streets are not shown.  The sheets of the cover $\Sigma \rightarrow C$ are (locally) labeled from $1$ to $K \geq 3$.  Red streets are of type $12$, blue streets are of type $23$, and fuchsia streets are of type $13$.  We choose an orientation for this diagram such that all streets ``flow up."  \textit{Right Frame}: A relatively simple example of a horse with one-way streets shown as partially transparent and two-way streets resolved (using the ``British resolution", cf. Appendix \ref{app:six-way} or \cite{GMN5}).  One can imagine horses with increasingly intricate ``backgrounds" of one-way streets. \label{fig:horse}}
	\end{center}
\end{figure}

\begin{figure}
	\begin{center}
		\includegraphics[scale=0.25]{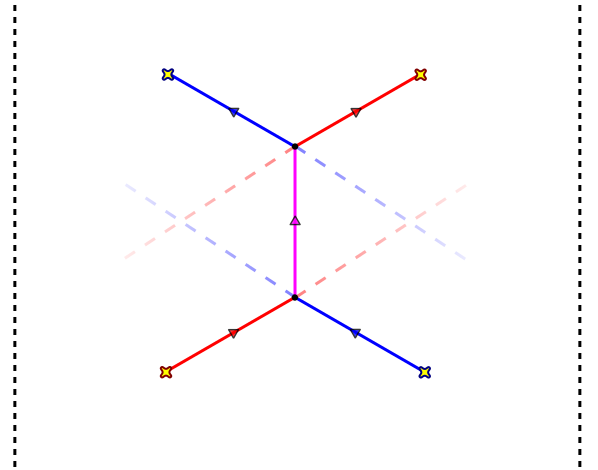}
		\includegraphics[scale=0.25]{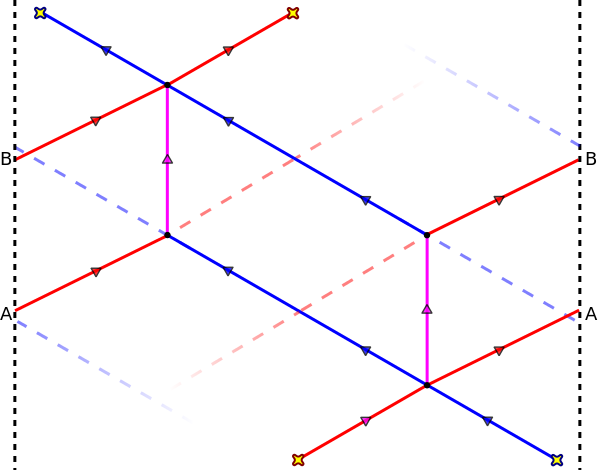}
		\includegraphics[scale=0.25]{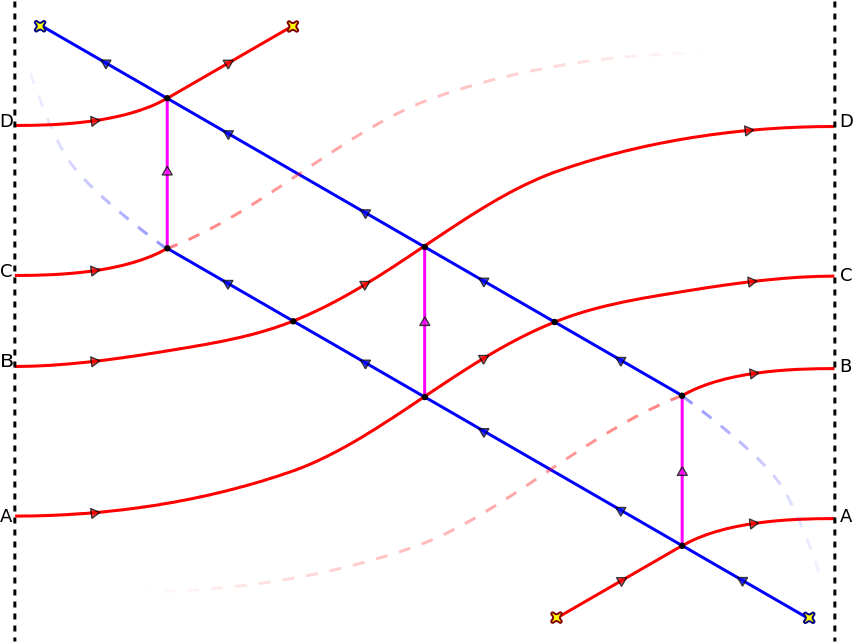}
		\includegraphics[scale=0.25]{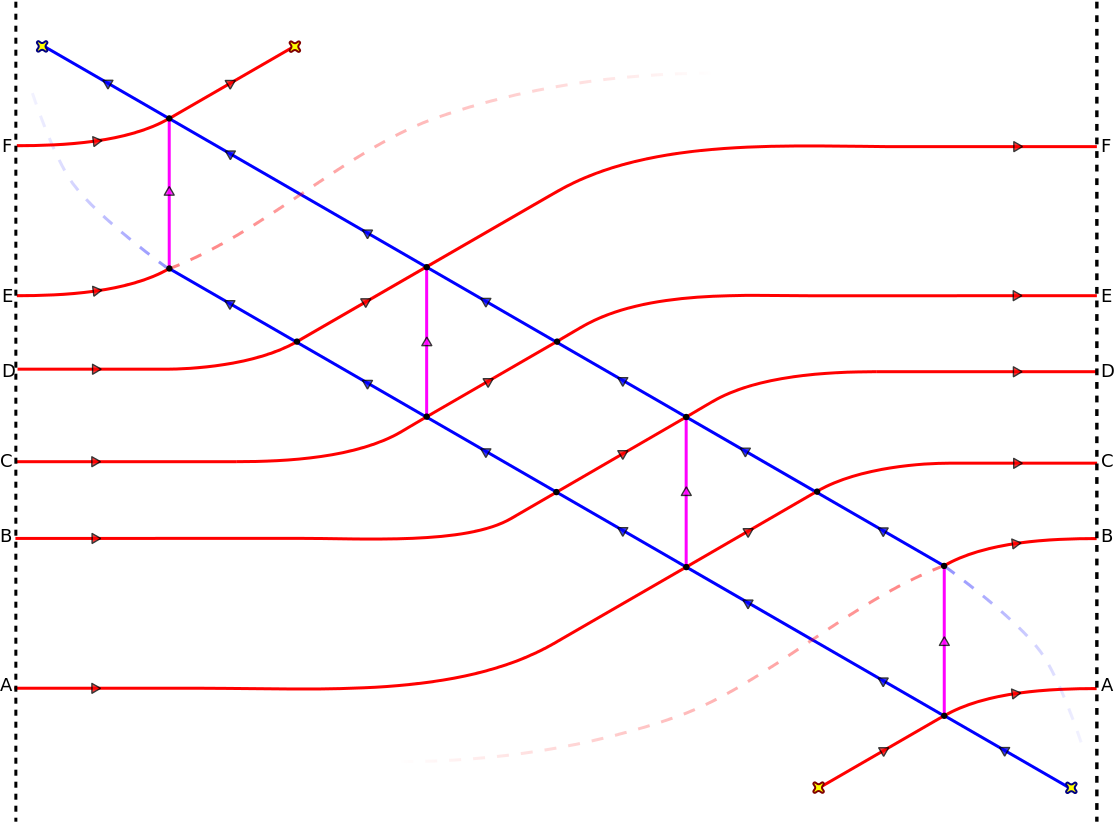}
	\end{center}
	\caption{The first four herds on the cylinder. Solid streets are two-way; dotted, transparent streets are streets of Fig.~\ref{fig:horse} that happen to be only one-way as indicated by Prop.~\ref{prop_Q}.  The black dotted lines are identified to form the cylinder and capital Latin letters are placed on either side to aid in the identification of streets.  Top row (from left to right): The 1-herd (saddle) and 2-herd. The
 middle row shows a 3-herd and the bottom row shows a 4-herd. \label{fig:herds}}
\end{figure}

\begin{definition}[Definitions]\
	\begin{enumerate}
	 \item A \textit{horse street} $p \in \{a_{1},a_{2},a_{3},b_{1},b_{2},b_{3},c,\conj{a_{1}},\conj{a_{2}},\conj{a_{3}},\conj{b_{1}}, \conj{b_{2}},\conj{b_{3}}\}$ is one of the streets of Fig.~\ref{fig:horse} (left frame).

	 \item Let $N$ be a spectral network (subordinate to some branched cover $\Sigma \rightarrow C$) and $U \subset C'$ be an open disk region.  Then $U \cap N$ is a \textit{horse} if a subset of its streets can be identified with Fig.~\ref{fig:horse} in a way such that:
	 	\begin{enumerate}
	 		\item every two-way street is a horse street,

	 		\item there is always a two-way street identified with the street labeled $c$.
	 	\end{enumerate}
	\end{enumerate}
\end{definition}

We can reconstruct the two-way streets of the full spectral network by gluing $m$ horses back together.  This leads us to the following working definition (a more complete definition is provided in Appendix \ref{app:herd-appendix}), which we extend to any curve $C$.

\begin{definition}[Working Definition]
	Given a collection of $m$ horses, let $p^{(l)}$ denote a horse street on the $l$th horse ($l=1,\cdots, m$). A spectral network on a curve $C$ is an $m$-herd if its two-way streets are generated by gluing together $m$ horses using the following relations:
	\begin{equation}
			\begin{aligned}
			a_{1}^{(l)} &= a_{3}^{(l-1)}\\
			b_{1}^{(l)} &= b_{3}^{(l-1)}\\
			\conj{a_{1}}^{(l)} &= \conj{a_{3}}^{(l+1)}\\
			\conj{b_{1}}^{(l)} &= \conj{b_{3}}^{(l+1)},
			\end{aligned}
			\label{eq:horse_glue}
		\end{equation}
		and such that $a_{1}^{(1)},\, b_{1}^{(1)},\, \conj{a_{1}}^{(m)},\,$ and $\conj{b_{1}}^{(m)}$ are connected to four distinct branch points.
\end{definition}

\begin{remark}
	It can be shown from our definition that a $1$-herd (which consists of a single horse) is just a saddle.  Indeed, a small computation will show that $Q(p)$ is nontrivial ($Q(p) \neq 1$) only  for $p=a_{1}, \, \overline{a_{1}}, \, b_{1}, \, \overline{b_{1}},\,$ and $c$; this leads us to the picture of a saddle extending from four branch points (pictured in the top left corner of Fig.~\ref{fig:herds}).
\end{remark}

An advantage of the decomposition into horses is computability: a horse should be thought of as a scattering machine which takes inflowing solitons, and regurgitates outgoing solitons as well as all spectral data ``bound" to the horse.\footnote{See Appendix \ref{app:horse_machine} for a the precise and explicit description of the horse as a scattering machine.}  The combinatorial problem of computing the BPS degeneracies $\Omega(n \gamma_{c}),\,n \geq 1$, using spectral network machinery, is then greatly simplified and explicit results can be obtained for all $m \geq 1$.  In fact, we have the following.
\begin{proposition} \label{prop_Q}\
	Let $N$ be an $m$-herd, then $Q(p)$ for all two-way streets $p$ on $N$ are given in terms of powers of a single generating function $P_{m}$ satisfying the algebraic equation
	\begin{equation}
		P_{m} = 1 + z \left( P_{m} \right)^{(m-1)^2},
		\label{eq:P}
	\end{equation}
	where $z = (-1)^{m} X_{ \twid{\gamma} + \twid{\gamma}'}$ for some $\gamma, \gamma' \in H_{1}(\Sigma;\mathbb{Z})$ such that $\langle \gamma, \gamma' \rangle = m$.  In particular, adopting the notation $Q(p,l) := Q(p^{(l)})$,
	\begin{equation}
	\begin{aligned}
		 P_{m} &= Q(c,l)\\
		\left( P_{m} \right)^{m-l} &= Q(a_{2},l) = Q(b_{2},l) = Q(a_{3},l) = Q(b_{3},l)\\
		\left( P_{m} \right)^{l-1} &= Q(\conj{a_{2}},l) = Q(\conj{b_{2}},l) = Q(\conj{a_{3}},l) = Q(\conj{b_{3}},l)\\
		\left( P_{m} \right)^{m-l+1} &= Q(a_{1},l) = Q(b_{1},l)\\
		\left( P_{m} \right)^{l} &= Q(\conj{a_{1}},l) = Q(\conj{b_{1}},l).
		\end{aligned}
		\label{eq:Q_red}
	\end{equation}
	for $l=1,\cdots,m$.\\
\end{proposition}
\begin{proof}
	See Appendix \ref{app_prop_Q_pf} for the full calculational proof.
\end{proof}

The precise cycle $\gamma_{c} = \gamma + \gamma'$ that appears depends on the embedding of $N$ in $C$ as a graph.  Further, as shown at the end of Appendix \ref{app_prop_Q_pf}, there are cycles representing $\gamma$ and $\gamma'$ that look like the charges of simple ``saddle-connection" hypermultiplets.  Indeed, the cycle representing either $\gamma$ or $\gamma'$  projects down to a path on $C$ that runs between two distinct branch points of the same type.  These are precisely the (hypothetical) hypermultiplets whose wall-crossing motivated the construction of $m$-herds.\footnote{The representative cycles discussed here, however, do not live entirely on $\operatorname{Lift}(N) \subset \Sigma$.  Roughly speaking representatives of $\gamma,\, \gamma'$ are given by the lifts of paths running along the $a_{i},\conj{a_{i}}$ and $b_{i}, \, \conj{b_{i}} $ respectively, but these do not define closed paths on $\Sigma$ without running through at least one street of type $13$.}

\begin{remark}[Remarks] \
\begin{itemize}
	\item A street $p$ is two-way iff $Q(p) \neq 1$.  Thus, (\ref{eq:Q_red}) states that on the first ($l=1$) and last ($l=m$) horses, some streets depicted in Fig.~\ref{fig:horse} are only one-way.

	\item When $m = 1,2$, (\ref{eq:P}) has easily derivable solutions:
	\begin{align}
		P_{1} &= 1 + z, \label{eq:P_1}\\
		P_{2} &= (1 - z)^{-1}. \label{eq:P_2}
	\end{align}
	For a saddle ($m = 1$), this result, combined with (\ref{eq:Q_red}), states that there are five two-way streets; each such two-way street $p$ is equipped with a generating function $Q(p) = 1 + z$, as originally derived in \cite{GMN5}.
\end{itemize}
\end{remark}

\subsection{Connection with Kontsevich-Soibelman, Gross-Pandharipande} \label{sec:GP}

The algebraic equation (\ref{eq:P}) and relevant solutions appear in a conjecture by Kontsevich and Soibelman (KS) \cite{KS_MOTIVIC_I}, later proven by Reineke  \cite{REINEKE-09} and generalized by Gross-Pandharipande (GP) \cite{GROSS}.  A series solution of (\ref{eq:P}) can be obtained using the Lagrange formula for reversion of series and the result for $m>1$  is \cite{KS_MOTIVIC_I}:
\begin{equation}
	\begin{aligned}
	P_{m} &= \sum_{n=0}^{\infty} \frac{1}{1+(m^2-2m) n} \binom{(m-1)^2 n}{n} z^n, \\
	 &= \exp \left[\sum_{n=1}^{\infty} \frac{1}{(m-1)^2 n} \binom{(m-1)^2 n}{n} z^{n} \right].
	 \end{aligned}
\label{eq:P_sol}
\end{equation}
To describe the connection between our result and that of KS and GP, we review the generalized conjecture of GP, briefly adopting their notation in \cite{GROSS}.

The algebraic equation (\ref{eq:P}) appears in \cite{GROSS}. \footnote{A different, but related, algebraic equation on the quantity $(P_{m})^{m}$ was originally stated by Kontsevich and Soibelman in \cite{KS_MOTIVIC_I}.}  There, the object of study is a group of (formal 1-parameter familes of) automorphisms of the torus $\mathbb{C}^* \times \mathbb{C}^*$ generated by $\theta_{(a,b),f}$ that are defined by
\begin{equation*}
	\begin{array}{lr}
		\theta_{(a,b),f}(x) = f^{-b} \cdot x, \, &	\theta_{(a,b),f}(y) = f^{a} \cdot y
	\end{array}
\end{equation*}
where $x$ and $y$ are coordinate functions on the two factors of $\mathbb{C}^* \times \mathbb{C}^*, \, (a,b) \in \mathbb{Z}^2, \,$ and $f$ is a formal series of the form
\begin{equation*}
	f = 1+ x^{a} y^{b} \left[ t f_{1}(x^{a} y^{b})+ t^2 f_{2}(x^{a} y^{b}) + \cdots \right], \, f_{i}(z) \in \mathbb{C}[z].
\end{equation*}
Alternatively we may say $f \in \mathbb{C}[x,x^{-1},y,y^{-1}][[t]]$ (i.e. $f$ is a formal power series in $t$ with coefficients Laurent polynomials in $x$ and $y$). Such automorphisms preserve the holomorphic symplectic form
\begin{equation*}
	\omega = (xy)^{-1} dx \wedge dy.
\end{equation*}
Now, letting
\begin{equation*}
	\begin{array}{lr}
		S_{q} = \theta_{(1,0),(1+t x)^{q}}, &  T_{r} = \theta_{(0,1),(1+t y)^{r}},
	\end{array}
\end{equation*}
we consider the commutator
\begin{equation}
	T_{r}^{-1} \circ S_{q} \circ T_{r} \circ S_{q}^{-1} = \prod^{\rightharpoonup} \theta_{(a,b), f_{(a,b)}}
	\label{eq:GP_comm}
\end{equation}
where the product on the right hand side is over primitive vectors $(a,b) \in \mathbb{Z}^{2}$ (i.e. $\gcd(a,b)=1$) such that $a,b>0$, and the order of the product is taken with increasing slope $a/b$ from left to right.  The conjecture of Gross-Pandharipande involves the slope $1$ term of (\ref{eq:GP_comm}).

\begin{remark}[Conjecture (Gross-Pandharipande)] \ \\
	For arbitrary $(q,r)$, the slope 1 term $\theta_{(1,1), f_{(1,1)}}$ in (\ref{eq:GP_comm}) is specified by
	\begin{equation*}
	f_{1,1} = \left( \sum_{n=0}^{\infty} \frac{1}{(q r - q - r) n + 1} \binom{(q-1) (r-1) n}{n} t^{2 n} x^{n} y^{n} \right)^{q r}.
	\end{equation*}
\end{remark}

The case $q=r$ was first conjectured by KS, and later proven by Reineke.  Now, letting
\begin{equation}
	\CP_{q,r} = \sum_{n=0}^{\infty} \frac{1}{(q r - q - r) n + 1} \binom{(q-1) (r-1) n}{n} t^{2 n} x^{n} y^{n},
	\label{eq:GP_P_sol}
\end{equation}
For general $q,r$, Gross and Pandharipande noted that $\CP_{q,r}$ satisfies the equation
\begin{equation}
	t^2 x y \left( \CP_{q,r} \right)^{(q-1) (r-1)} - \CP_{q,r} + 1 = 0;
	\label{eq:GP_P}
\end{equation}
so that $f_{1,1}$ is an algebraic function (over $\mathbb{Q}(t,x,y)$).

In the case $q=r=m$, the equation (\ref{eq:GP_P}) and solution (\ref{eq:GP_P_sol}) bear striking similarity to (\ref{eq:P}) and (\ref{eq:P_sol}), which motivates identifying
$t^2 xy = z$ in hopes of identifying $\CP_{m,m}$ with $P_{m}$.

To motivate the identification $t^2 xy = z$, we turn our attention back to the original motivation for our definition of $m$-herds:  they are expected to arise after two hypermultiplets of charges $\gamma,\, \gamma'$, with $\langle \gamma, \gamma' \rangle = m$, cross a wall of marginal stability.  If $m$-herds do arise in this manner, then in the resulting wall-crossing formula we should expect the $P_{m}$ to be related to the generating function for the KS transformations attached to the charges $n(\gamma + \gamma'),\, n>0$.  We now go about unpacking the identification of such a wall crossing formula with (\ref{eq:GP_comm}).

	Assume on one side of the wall $\arg(Z_{\gamma}) < \arg(Z_{\gamma'})$, then the wall crossing formula reads (see Section \ref{sec:specgen})
\begin{equation}
	\CK_{\gamma} \CK_{\gamma'} = \CK_{\gamma'} \left[ \prod_{(a,b) \in \mathbb{Z}^2} \left(\CK_{a \gamma + b \gamma'} \right)^{\Omega(a \gamma + b \gamma')} \right] \CK_{\gamma}
	\label{eq:hm_wcf}
\end{equation}
where all products are taken in order of increasing central charge phase (when read from left to right) and the $\CK_{\alpha}$ are transformations on a twisted Poisson algebra of functions on the torus $T = \Gamma \otimes_{\mathbb{Z}} \mathbb{C}^{\times}$, i.e. the space of functions generated by polynomials in formal variables $Y_{\alpha}, \alpha \in \Gamma$ equipped with twisted product given by
\begin{equation}
	Y_{\alpha} Y_{\beta}= (-1)^{\langle \alpha, \beta \rangle} Y_{\alpha + \beta}. \label{eq:twisted-product-rule}
\end{equation}
$T$ is equipped with a holomorphic symplectic form induced by the symplectic pairing on $\Gamma$; it is equivalently given by the holomorphic Poisson bracket
\begin{equation}\label{eq:formal-Y}
	\left \{ Y_{\alpha}, Y_{\beta} \right \} = \langle \alpha, \beta \rangle Y_{\alpha} Y_{\beta}.
\end{equation}
Now, the $\CK_{\alpha}$ are symplectomorphisms that act as
\begin{equation}\label{eq:formal-K}
	\CK_{\alpha}: Y_{\beta} \mapsto (1-Y_{\alpha})^{\langle \alpha, \beta \rangle} Y_{\beta}.
\end{equation}
For $\langle \gamma, \gamma' \rangle = m$, it follows that
\begin{equation}
	\begin{array}{lr}
		\begin{aligned}
		\CK_{\gamma} &: Y_{\gamma}  \mapsto Y_{\gamma}, \\
		\CK_{\gamma} &: Y_{\gamma'}  \mapsto (1 - Y_{\gamma})^{m} Y_{\gamma'},
		\end{aligned} &
		\begin{aligned}
		\CK_{\gamma'} &: Y_{\gamma}  \mapsto (1 - Y_{\gamma'})^{-m} Y_{\gamma},\\
		\CK_{\gamma'} &: Y_{\gamma'}  \mapsto Y_{\gamma'}.
		\end{aligned}
	\end{array}
	\label{eq:K_action}
\end{equation}
We identify the torus $\mathbb{C}^{\times} \times \mathbb{C}^{\times}$ of Gross-Pandharipande by the subtorus of $T$ generated by
\begin{equation*}
	\begin{aligned}
	x & := Y_{\gamma}\\
	y & := Y_{\gamma'};
	\end{aligned}
\end{equation*}
then by (\ref{eq:K_action}) we have\footnote{To make the identification with $S_{m}$ and $T_{m}$ we evaluate the formal (time) parameter at $t = -1$.  Alternatively, we could set $-tx = Y_{\gamma}$ and $-ty = Y_{\gamma'}$.}
\begin{equation*}
	\begin{aligned}
	\CK_{\gamma} &= \theta_{(1,0),(1 - x)^{m}} = S_{m}\\
	\CK_{\gamma'} &= \theta_{(0,1),(1 - y)^{m}} = T_{m}.
	\end{aligned}
\end{equation*}
Furthermore, noting that
\begin{equation}
	x^{a} y^{b} = (-1)^{\langle a \gamma, b \gamma' \rangle} Y_{a \gamma + b \gamma'} = (-1)^{mab} Y_{a \gamma + b \gamma'},
	\label{eq:sign_origin}
\end{equation}
then
\begin{equation*}
	\begin{aligned}
	\CK_{a \gamma + b \gamma'}&: x = Y_{\gamma}  \mapsto (1 - Y_{a \gamma + b \gamma'})^{- m b}  Y_{\gamma} = (1 - (-1)^{m a b} x^{a} y^{b})^{m b} x \\
	& : y = Y_{\gamma'}  \mapsto (1 - Y_{a \gamma + b \gamma'})^{-m a} Y_{\gamma'} ^{m a} = (1 - (-1)^{m a b} x^{a} y^{b})^{-m a} y;
	\end{aligned}
\end{equation*}
giving the identification
\begin{equation*}
	\CK_{a \gamma + b\gamma'} = \theta_{(a,b), (1 - (-1)^{m a b} x^{a} y^{b})^{m}}.
\end{equation*}

  On the right hand side of (\ref{eq:hm_wcf}) $\arg(Z_{\gamma}) > \arg(Z_{\gamma'})$ and so the phase ordered product is equivalent to ordering by increasing slope $a/b$ from left to right.  This completes the identification of (\ref{eq:hm_wcf}) with (\ref{eq:GP_comm}).  Matching the slope 1 terms in both equations,
  \begin{align*}
 	 \theta_{(1,1),f_{1,1}} = \prod_{n \geq 1} \left(\CK_{n \gamma_{c} } \right)^{\Omega(n \gamma_{c})},
  \end{align*}
 where $\gamma_{c} := \gamma + \gamma'$; in terms of generating functions, this is equivalent to the statement\footnote{To see this, let $g_{n} = (1 - (-1)^{mn} (xy)^{n})^{m}$, then $\CK_{n \gamma_{c}} = \theta_{(n,n), g_{n}} = \theta_{(1,1),(g_{n})^{n}}$; furthermore, as $\theta_{(1,1),(g_{n})^{n}}$ fixes the product $xy$: $\theta_{(1,1),(g_{n})^{n}} \circ \theta_{(1,1),(g_{l})^{l}} = \theta_{(1,1),(g_{n})^{n} (g_{l})^{l}}$.}
 \begin{equation*}
 	f_{1,1} = \prod_{n \geq 1} \left[(1 - (-1)^{m n} z^{n})^{m} \right]^{n \Omega(n \gamma_{c})}.
 \end{equation*}
Equivalently, as $f_{1,1} = (\CP_{m,m})^{m^2}$,
 \begin{equation}
 	\left( \CP_{m,m} \right)^{m} = \prod_{n \geq 1} (1 - (-1)^{m n} z^{n})^{n \Omega(n \gamma_{c})}.
 	\label{eq:index_pred}
 \end{equation}
Now assume that the generating function $P_{m}$, derived in the context of spectral networks, \textit{is} the generating function $\CP_{m,m}$, derived in the context of wall crossing; then, given the exponents $\{\alpha_{n}\}_{n \geq 1}$ of the factorization of $P_{m}$ (see (\ref{eq:P_decomp})), (\ref{eq:index_pred}) predicts spectral network techniques will show $\Omega(n \gamma_{c}) = m \alpha_{n} / n$.    As we will see, this prediction is confirmed with Prop.~\ref{prop_l}.

\subsection{Herds of horses are wild (for $m \geq 3$)}

\begin{definition}
	For each two-way street $p$, define the sequence of exponents $\left\{\alpha_{n}\left(p,l\right)\right\}_{n \geq 1} \subset \mathbb{Z}$ via
	\begin{equation}
		Q(p,l) = \prod_{n=1}^{\infty} \left( 1 - (-1)^{m n} z^{n} \right)^{\alpha_{n} \left( p, l \right)}.
		\label{eq:Q_herd_exp}
	\end{equation}
	We also define the sequence of integers $\left\{\alpha_{n}\right\}_{n \geq 1}$ via
	\begin{equation}
		P_m = \prod_{n=1}^{\infty} (1 - (-1)^{m n} z^{n})^{\alpha_{n}}.
		\label{eq:P_decomp}
	\end{equation}
\end{definition}
By Prop. \ref{prop_Q}, we can express all $\alpha_{n}(p,l)$ as multiples of $\alpha_n$. \footnote{The radius of convergence $R$ of the series in equation (\ref{eq:P}) is $\log R = -c_m$, where $c_m$ is given in equation (\ref{eq:c_m}); in particular $R < 1$. Therefore, the product expansion is only a formal expansion and is not absolutely convergent; otherwise, it would predict that all the singularities of $d \log P$ sit on the unit circle.}

\begin{remark}
	The choice of signs $(-1)^{m n}$ follows from our convention of factorization, defined by $(\ref{eq:Q-exp})$, in terms of formal variables in the image of $Y_{\gamma} \mapsto X_{\twid{\gamma}}$ (which forms an embedding of the twisted algebra of $Y_{\gamma},\, \gamma \in \Gamma$, as subalgebra of $\formgamma$ as detailed in Appendix \ref{app:sign-rule}).  By Prop. \ref{prop_Q}, $z^{n} = (-1)^{m n} X_{n \twid{\gamma}_{c}}$ for some $\gamma_{c} \in \Gamma$, leading to the choice of signs in $(\ref{eq:Q_herd_exp})$.
\end{remark}

\begin{proposition} \label{prop_l}
	\begin{equation*}
		[L(n \gamma_{c}) ] = m \alpha_{n} \gamma_{c} \in H_{1}(\Sigma; \mathbb{Z}).
	\end{equation*}
\end{proposition}
\begin{proof}[Proof (sketch)]
	A rough argument goes as follows.  Note that, using Prop. \ref{prop_Q} and the definition of $L(n \gamma_{c})$ in (\ref{eq:L_def}), we have
	\begin{equation}
		\begin{aligned}
		L(n \gamma_{c}) &= \sum_{l=1}^{m} \sum_{p^{(l)}} \alpha_{n}(p,l) \liftnb{p}^{(l)} \\
		 &= \alpha_{n} \sum_{l=1}^{m} \left\{ \liftnb{c}^{(l)} + (m-l) \left( \lift{a_{2}}^{(l)} + \lift{a_{3}}^{(l)} + \lift{b_{2}}^{(l)} + \lift{b_{3}}^{(l)} \right) \right. \\
		 & + (l-1) \left( \lift{\conj{a_{2}}}^{(l)} + \lift{\conj{a_{3}}}^{(l)}  + \lift{\conj{b_{2}}}^{(l)} + \lift{\conj{b_{3}}}^{(l)} \right) + \\
		& \left. + (m-l+1) \left(\lift{a_{1}}^{(l)} + \lift{b_{1}}^{(l)} \right) +  l \left( \lift{\conj{a_{1}}}^{(l)} + \lift{\conj{b_{1}}}^{(l)} \right) \right\}.
		\end{aligned}
		\label{eq:ln}
	\end{equation}
	Each term in this sum can be split up into a sum of words of the form
	\begin{align*}
		\la_{1}^{(1)} + \lb_{1}^{(1)} + \left(\cdots \right) + \lac_{1}^{(m)} + \lbc_{1}^{(m)},
	\end{align*}
	Each such word represents a closed cycle on the lift of the $m$-herd to a graph on $\Sigma$, and is homologous\footnote{This homological equivalence can be shown using explicit calculations of the form shown in Appendix \ref{app:prop_l_pf}. For the reader that wishes to avoid excruciating detail: sufficient staring at some simple examples will suffice.} to $\gamma_{c}$. As $\la_{1}^{(1)},\, \lb_{1}^{(1)},\,\lac_{1}^{(m)},\, \lbc_{1}^{(m)}$ all come with multiplicity $m$ in (\ref{eq:ln}), then there are $m$ such words and the proposition follows.  A full proof, using brute-force homology calculations, can be found in Appendix \ref{app:prop_l_pf}.
\end{proof}

Via (\ref{eq:Omega_L_rel}), the immediate result of Prop. \ref{prop_l} is that
\begin{equation*}
	\Omega(n \gamma_{c}) = \frac{m \alpha_{n}}{n} ;
\end{equation*}
so all that remains is to compute $\alpha_{n}$.   For the cases $m = 1,2$: using  (\ref{eq:P_1}) and (\ref{eq:P_2}) we immediately have\footnote{The case $m=1$ (i.e. the saddle) was also computed in \cite{GMN5}.}
\begin{equation}
	\alpha_{n} =
	\left\{
	\begin{array}{rr}
		\delta_{n,1},   & \text{if $m = 1$}\\
		-\delta_{n,1},   & \text{if $m = 2$}
	\end{array}
	\right.
	\Rightarrow
	\Omega(n \gamma_{c}) =
	 \left\{
	\begin{array}{rr}
		\delta_{n,1},  & \text{if $m = 1$}\\
		-2\delta_{n,1},  & \text{if $m = 2$}
	\end{array}
	\right. .
	\label{eq:Omega_simple}
\end{equation}
 More generally, we can find an explicit form for $\alpha_{n}$ by taking the $\log$ of both sides of (\ref{eq:P_decomp}), matching powers of $z$, and applying M\"{o}bius inversion to derive
\begin{equation*}
	\alpha_{n} = \frac{1}{n}\sum_{d|n} (-1)^{m d + 1}  \mu \left(\frac{n}{d} \right) \frac{1}{(d-1)!} \left[\frac{d^{d}}{dz^{d}} \log(P_m)\right]_{z=0},
\end{equation*}
where $\mu$ is the M\"{o}bius mu function.  Using (\ref{eq:P_sol}),
\begin{equation*}
	\alpha_{n} = \frac{1}{(m-1)^2 n} \sum_{d|n} (-1)^{m d + 1} \mu \left(\frac{n}{d} \right) \binom{(m-1)^2 d}{d}, \,  m \geq 2.
\end{equation*}

\begin{corollary}
For $m \geq 2 $,
\begin{equation}
\Omega(n \gamma_{c}) = \frac{m}{(m-1)^2 n^2}\sum_{d|n} (-1)^{m d + 1}  \mu \left(\frac{n}{d} \right)  \binom{(m-1)^2 d}{d}.
\label{eq:Omega_sol}
	\end{equation}
\end{corollary}

This agrees with the result of Reineke\footnote{Reineke showed (in our notation)  $\Omega(n \gamma_{c}) = \frac{1}{(m -2) n^2} \sum_{d|n} (-1)^{md + 1} \mu(n/d) \binom{(m-1)^2d -1}{d}$.  To translate between results, we use the observation that $\binom{(m-1)^2d}{d} = \frac{(m-1)^2}{m (m-2)} \binom{(m-1)^2d -1}{d}$.}
in the last section of \cite{REINEKE-09}.  A table of the values of $\Omega(n \gamma_{c})$ is
provided in Appendix \ref{app_omega_table} for $1 \leq n,m \leq 7$.  From this explicit result,
we can deduce the large $n$ asymptotics for the non-trivial\footnote{In the case $m =2$, using the identity $\sum_{d|n} \mu(d) = \delta_{n,1}$ in (\ref{eq:Omega_sol}) reproduces the result $\Omega(n \gamma_{c}) = -2 \delta_{n,1}$ of (\ref{eq:Omega_simple}).} case $m \geq 3$.

\begin{proposition} \label{prop:asymp} \
Let $m \geq 3$, then as $n \rightarrow \infty$,
	\begin{equation}
	\Omega(n \gamma_{c}) \sim (-1)^{m n + 1} \left( \frac{1}{m-1} \sqrt{\frac{m}{2\pi (m-2)}} \right)   n^{-5/2} e^{c_{m} n},
	\label{eq:Omega_asymp}
	\end{equation}
		where $c_{m}$ is the constant
	\begin{equation}
		c_{m} = (m-1)^2 \log \left[ (m-1)^2 \right] -m (m-2) \log \left[ m (m-2) \right].  \label{eq:c_m}
	\end{equation}
	\end{proposition}
\begin{proof}
Restricting $n$ to be an element of an infinite subsequence of primes, the sum over divisors simplifies and the claimed asymptotics (restricted to this subsequence) follow immediately using Stirling's asymptotics and (\ref{eq:Omega_sol}).  See Appendix \ref{app:proof_asymp} for a full proof.
\end{proof}

\subsection{Herds in the pure $SU(3)$ theory}

Now, finally, let us exhibit some points of the Coulomb branch of the pure $SU(3)$ theory where
$m$-herds actually occur in spectral networks $\CW_\vartheta$.

In the pure $SU(3)$ theory, the curve $C$ is $\mathbb{CP}^1$ with two defects.  It is natural to view it topologically as the cylinder $\mathbb{R} \times S^{1}$.
Moreover, the spectral curve \eqref{eq:su3-spectral-curve} has 4 branch points.  Thus, the pictures of actual spectral networks in
this theory look much like the ``hypothetical'' spectral networks we considered in Figures \ref{fig:1-herd_motivation}, \ref{fig:3-herd_motivation}.

In particular, consider the parameters
\be\label{eq:MNpoint}
u_2 = -3, \quad u_3 = \frac{95}{10}
\ee
(in the notation of \eqref{eq:su3-spectral-curve}.)  At this point, in accordance with the discussion of Section \ref{sec:herds},
we consider two charges $\gamma$, $\gamma'$ supporting BPS hypermultiplets,
represented simply by paths connecting pairs of branch points across the cylinder,
as in the left side of Figure \ref{fig:3-herd_motivation}.
  In particular they have $\langle \gamma, \gamma' \rangle = 3$.
By numerically computing the appropriate contour integrals we find that these charges have
$Z_\gamma = 7.244 - 9.083 i$, $Z_{\gamma'} = 20.980 - 40.148 i$.

Now, our proposal in Section \ref{sec:herds} was that when we have two such hypermultiplets, there will be a wall of marginal stability in
the Coulomb branch when $Z_\gamma$ and $Z_{\gamma'}$ become aligned, and on one side of that wall, the spectral network at the phase
$\vartheta = \arg Z_{\gamma + \gamma'}$ will contain a 3-herd.
So, we plot the spectral network at phase $\vartheta = \arg Z_{\gamma + \gamma'}$, and find
Figure \ref{fig:3-herd-in-su3}.   Comparing with Figure \ref{fig:herds}, we see that
the two-way streets in this network make up a 3-herd as desired.\footnote{In particular, our point \eqref{eq:MNpoint} is on the
side of the wall where the 3-herd exists.  The wall of marginal stability where the 3-herd disappears
can be reached by moving $u_3$ in the negative real direction.}

\begin{figure}[h]
\begin{center}
\includegraphics[width=0.45\textwidth]{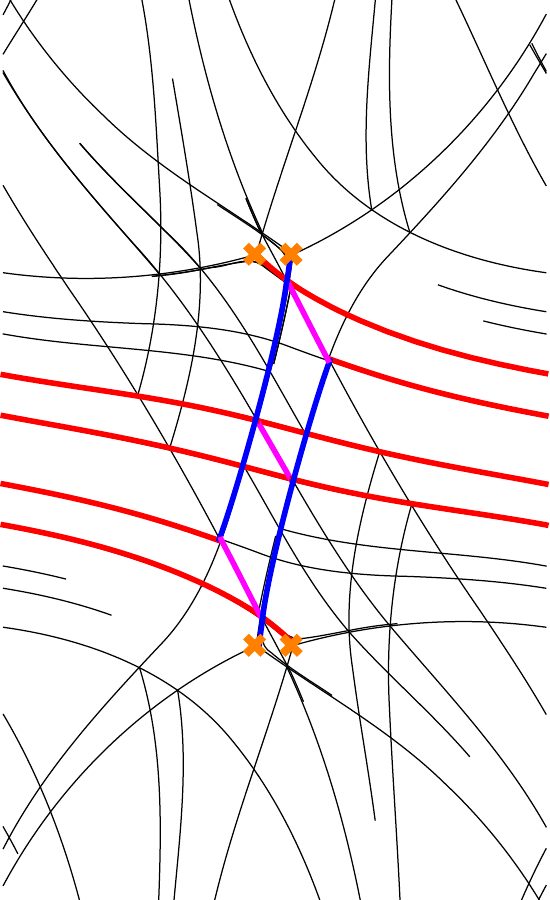}
\caption{The spectral network $\CW_\vartheta$
which occurs in the pure $SU(3)$ theory at the point \eqref{eq:MNpoint} of the Coulomb branch.
The phase $\vartheta$ has been chosen very close to the critical phase $\vartheta = \arg Z_{\gamma + \gamma'}$.
Here we represent the cylinder $C$ as the periodically identified plane, i.e., the left and right sides of the figure
should be identified.  Streets which become two-way
at $\vartheta = \arg Z_{\gamma + \gamma'}$ are shown in thick red, blue
and fuchsia.  We do not show the whole network but only a cutoff version of it, as described in \cite{GMN5}.\label{fig:3-herd-in-su3}}
\end{center}
\end{figure}

Moving $u_3$ in the positive real direction,
we have similarly found a $4$-herd, a $5$-herd and a $6$-herd.  It is
natural to conjecture that one can similarly obtain $m$-herds for any $m$
in this way. Of course, at a fixed point in the Coulomb branch it is
in general possible that there could be $m$-herds for many different
values of $m$ at different values of $\vartheta$.

In any case, the existence of $3$-herds in the pure $SU(3)$ theory is already
enough to show that the analysis of the last few sections is not only a
theoretical exercise:  the wild BPS degeneracies we found there indeed occur
in the $\N=2$ supersymmetric pure $SU(3)$ Yang-Mills theory!

\section{Wild regions for pure $SU(3)$ theory from wall-crossing} \label{sec:wc-wild-SU3}

In the previous section we exhibited an example of a class of spectral networks
that lead to the $m$-wild degeneracies of slope $(1,1)$.
An explicit point on the Coulomb branch of the pure $SU(3)$ theory
which produces such a spectral network for $m=3$ was given in
equation (\ref{eq:MNpoint}) above.

In the present section we start anew, and use wall crossing
and quiver techniques to give an alternative demonstration
that wild degeneracies exist on the Coulomb branch of the pure $SU(3)$ theory.

\subsection{Strong Coupling Regime of the Pure $SU(3)$ Theory} \label{subsec:strong-cplg}

The spectral curve $\Sigma$ of pure $SU(3)$ SYM theory is
\be \label{eq:su3-spectral-curve}
	\lambda^{3} -\frac{u_{2}}{z^{2}}\,\lambda +\left( \frac{1}{z^{2}}+\frac{u_{3}}{z^{3}}+\frac{1}{z^{4}} \right)=0.
\ee
It is a branched three-sheeted covering of the cylinder $C$, with six ramification points.
There are four branch points corresponding to two-cycles of $S_{3}$, and there are also ramifications at the irregular singularities at $0,\infty$, with associated permutations of the sheets given by three-cycles.

In the strong coupling region, i.e. at small values of the moduli $u_2$, $u_3$, the BPS spectrum is finite; so the spectral network evolves in a rather simple fashion.
As a concrete example we choose $u_2=0.7$, $u_3=0.4i$; then varying $\vartheta$ from 0 to $\pi$ we encounter six degenerate networks containing
finite webs, which are depicted in Figure \ref{fig:1}.\\
\begin{figure}[htbp]
\begin{center}
\includegraphics[width=0.26\textwidth]{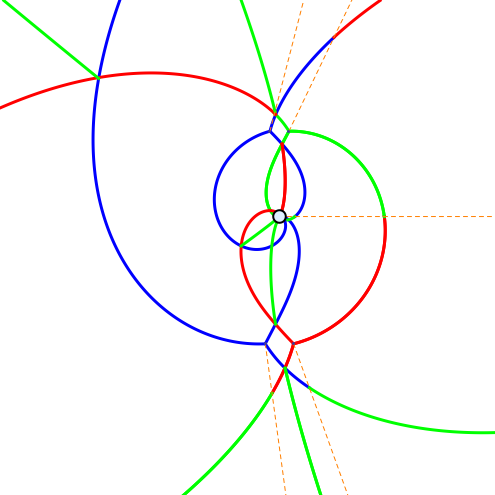}\includegraphics[width=0.26\textwidth]{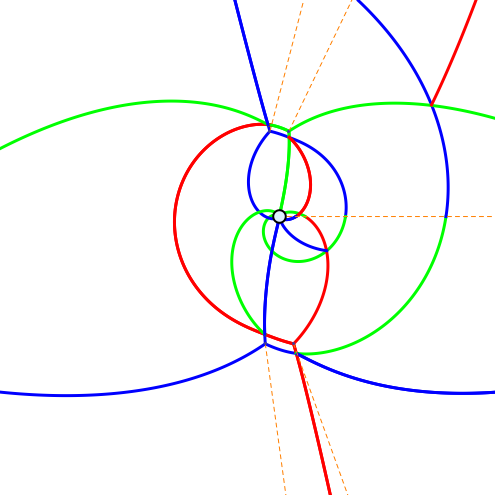}\includegraphics[width=0.26\textwidth]{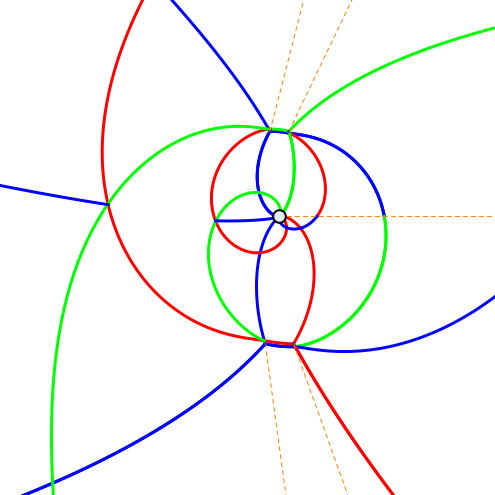}\\
\includegraphics[width=0.26\textwidth]{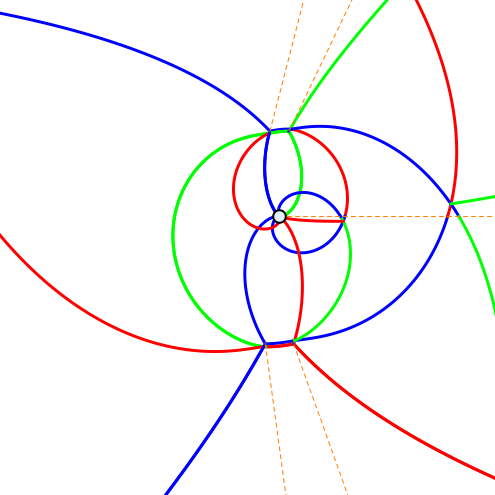}\includegraphics[width=0.26\textwidth]{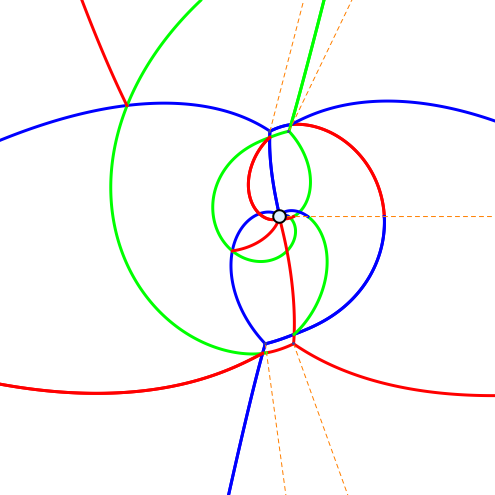}\includegraphics[width=0.26\textwidth]{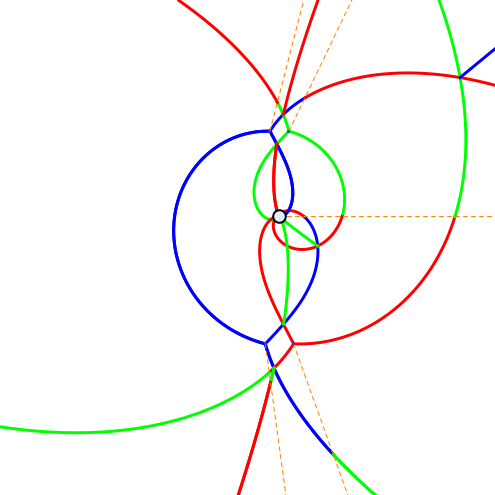}
\caption{The six hypermultiplets in the strong coupling chamber: from the top left, the flips corresponding to $\gamma_{1},\gamma_{2},\gamma_{1}+\gamma_{3},\gamma_{2}+\gamma_{4},\gamma_{3},\gamma_{4}$. $Arg \, Z_{\gamma_{1}}<Arg \, Z_{\gamma_{2}}<Arg \, Z_{\gamma_{3}}<Arg \, Z_{\gamma_{4}}$.
Here we represent the cylinder $C$ as the punctured plane.\label{fig:1}}
\end{center}
\end{figure}
We assign to these cycles the charges $\gamma_1$, $\gamma_2$, $\gamma_2+\gamma_4$, $\gamma_1+\gamma_3$, $\gamma_3$, $\gamma_4$, Figure \ref{fig:1-bis} shows the charge assignments with the basis cycles resolved.\\
\begin{figure}[htbp]
\begin{center}
\includegraphics[width=0.4\textwidth]{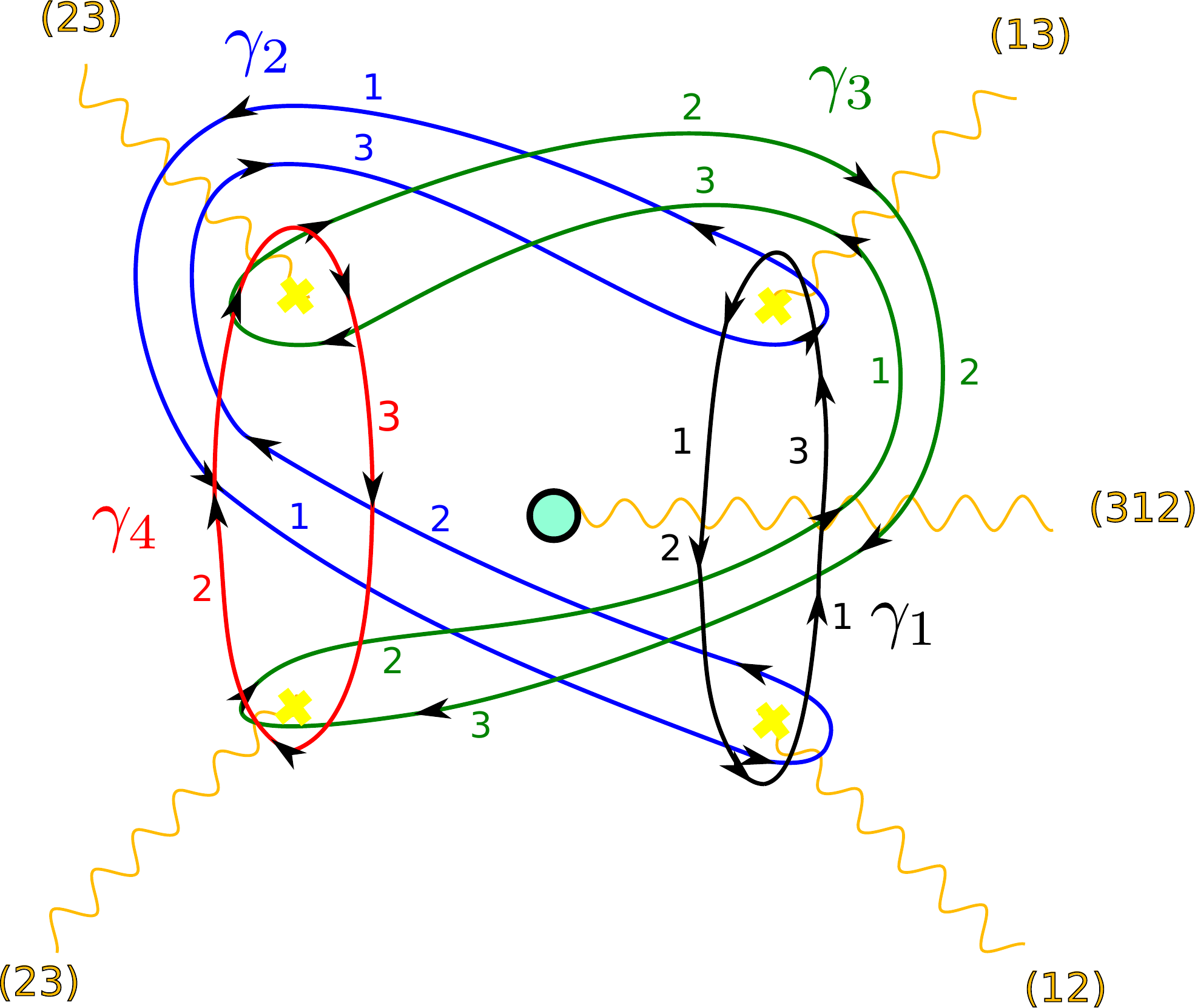}
\caption{The labeling of finite networks. We only show the four basis hypermultiplets $\gamma_{1},\dots,\gamma_{4}$. The trivialization is indicated by the branch cuts (wavy lines, the associated permutations of sheets are also specified), the sheets on which the cycles run are indicated explicitly. Here we represent the cylinder $C$ as the punctured plane.\label{fig:1-bis}}
\end{center}
\end{figure}
The mutual intersections of cycles can be read off Figure \ref{fig:1-bis}, and are summarized by the following pairing matrix $P_{ij}=\langle\gamma_{i},\gamma_{j}\rangle$
\be
P=\left(\begin{array}{cccc}
0  & -2 & 1 & 0 \\
2  & 0 & -2 & 1 \\
-1  & 2 & 0 & -2 \\
0  & -1 & 2 & 0
\end{array}\right).\label{eq:pairing-matrix}
\ee
For a video showing the evolution of the spectral network through an angle of $\pi$, see \cite{video-strong-cplg}.

\subsection{A path on the Coulomb branch}\label{subsec:PathCoulomb}
We now consider a straight path on the Coulomb branch of the pure $SU(3)$ theory, parameterized by
\be
\begin{array}{c}
 u_2(t)=(u_2^{(f)}-u_2^{(i)})t+u_2^{(i)}, \\
 u_3(t)=(u_3^{(f)}-u_3^{(i)})t+u_3^{(i)}, \label{eq:path}
\end{array}
\ee
with $t\in\left[0,1\right]$ and
\be
\begin{array}{c}
 u_{2}^{(i)}=0.7,\quad u_{3}^{(i)}=0.4 i\quad \mbox{(strong coupling chamber)}\\
 u_{2}^{(f)}=0.56-0.75i,\quad u_{3}^{(f)}=2+1.52i\quad \mbox{(wild chamber)}
\end{array}
\ee
As discussed above, the spectrum in the strong coupling chamber is known (see for example \cite{GMN5}) to consist of six hypermultiplets. As we move along our path we cross several walls of marginal stability, with consequent jumps of the BPS spectrum. In order to study the evolution of the BPS spectrum, we must track explicitly the evolution of central charges. Variation of the moduli also induces changes in the geometry of the Seiberg-Witten curve $\Sigma$, therefore in computing the central charges at different points one must take care of deforming the cycles in a way compatible with the flat parallel transport of the local system $\widehat{\Gamma} \to \CB^{*}$.
Starting from the point studied in Section \ref{subsec:strong-cplg}, the evolution of branch points can be tracked on $C$, as shown in Figure \ref{fig:BPflow}.
\begin{figure}[htbp]
\begin{center}
\includegraphics[width=0.26\textwidth]{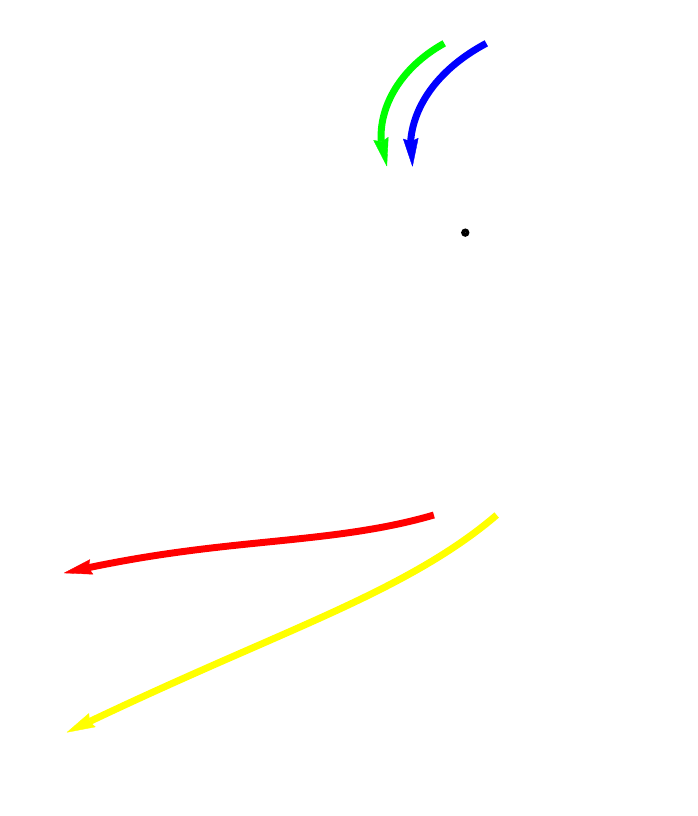}
\caption{The picture shows the projection of the Seiberg-Witten curve on $C$. The four arrows show the progression of the four branch points as we vary $u_{2,3}$ along the path of equation (\ref{eq:path}). The black dot is the singularity at $z=0$. The central charges have been computed numerically using Mathematica
and, as a check, they evolve smoothly along the path (see \cite{video-charge-motion}). \label{fig:BPflow}}
\end{center}
\end{figure}

\subsection{Cohorts in pure $SU(3)$}\label{sec:su3coh}
As the moduli cross walls of marginal stability, the BPS spectrum jumps according to a regular pattern. At a wall MS$(\gamma,\gamma')$ for two populated hypermultiplets, with $\langle\gamma,\gamma'\rangle=m>0$, the KS wall crossing formula predicts
\be\label{eq:cohWCF}
	\CK_{\gamma'}\CK_{\gamma} = : \prod^{
}_{a,b\geq 0}\CK^{\Omega(a \gamma+b\gamma')}_{a \gamma+b\gamma'} :
\ee
where the normal ordering symbols $:\ \ :$ on the right hand side indicate that factors are ordered according to the phases of central charges, phase-ordering on the right hand side is the opposite of that on the left-hand side.
We refer to the spectrum on the right hand side as the \emph{cohort} generated by $\gamma,\gamma'$, and will occasionally denote it by ${\cal C}_{m}(\gamma,\gamma')$. An important fact to note about cohorts, following from the linearity of the central charge homomorphism, is that
\be
	\arg Z_{\gamma'} < \arg Z_{a \gamma+ b \gamma'} < \arg Z_{\gamma},\qquad \forall a,b\geq 0 \label{eq:coh-boundaries}
\ee
for moduli corresponding to the right hand side of (\ref{eq:cohWCF}).

Quite generally, the wall-crossing of two hypermultiplet states with pairing $m$ can be analyzed in terms of the corresponding $m$-Kronecker quiver (see \cite{DENEF,CV}), from this perspective the degeneracies of an $m$-cohort correspond to Euler characteristics of moduli spaces of (semi)stable quiver representations.

Cohort structures ${\cal C}_{m}$ with $m=1,2$ are known exactly. Examples of such cohorts have been encountered a number of times in the literature \cite{KS,GMN1,GMN2,GMN3,GROSS}, and are common in $A_{1}$ theories of class ${\cal S}$. For later convenience, we recall the structure of the $m=2$ cohort in figure \ref{fig:m_2-cohort}.

\begin{figure}[htbp]
\begin{center}
\includegraphics[width=0.40\textwidth]{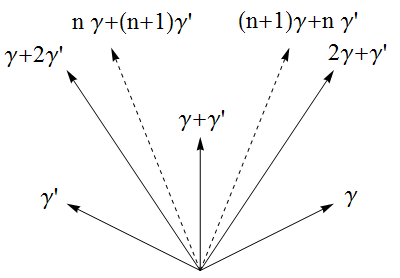}
\caption{The populated BPS rays of the $m=2$ cohort (a schematic depiction of central charges in the complex plane). The state with charge $\gamma_{}+\gamma'_{}$ is a BPS vectormultiplet ($\Omega=-2$), surrounded by two infinite towers of hypermultiplets ($\Omega=1$), represented by dashed arrows.\label{fig:m_2-cohort}}
\end{center}
\end{figure}

As we start moving along our path on the Coulomb branch, from $t=0$ to $t=1$, several cohorts are created. The first wall of marginal stability is ${\rm MS}(\gamma_{1}+\gamma_{3},\gamma_{2}+\gamma_{4})$, with $\langle \gamma_{1}+\gamma_{3},\gamma_{2}+\gamma_{4} \rangle = 2$, thus a ${\cal C}_{2}$ cohort is generated. As we proceed along the path, other BPS states undergo wall-crossing, generating  other ${\cal C}_{2}$ cohorts. As shown in Fig.~\ref{fig:spectrum-evo}, first $\gamma_{1}$ generates a cohort with $\gamma_{2}$, then $\gamma_{3},\gamma_{4}$ generate a similar cohort, finally another $m=2$ cohort is generated by wall crossing of $\gamma_{1}$ and $\gamma_{2}+\gamma_{4}$. At this point, \emph{i.e.} within a chamber around $t=0.95$, the spectrum can be schematically summarized as the union of four ${\cal C}_{2}$ cohorts
\begin{align}
	{\cal C}_{2}(\gamma_{2}+\gamma_{4},\gamma_{1}+\gamma_{3}) \cup %
	{\cal C}_{2}(\gamma_{2},\gamma_{1}) \cup %
	{\cal C}_{2}(\gamma_{4},\gamma_{3}) \cup %
	{\cal C}_{2}(\gamma_{1},\gamma_{2}+\gamma_{4}) \label{eq:cohorts}
\end{align}
consisting of four vectormultiplets, and infinite towers of hypermultiplets.

\begin{figure}[htbp]
\begin{center}
\includegraphics[width=0.28\textwidth]{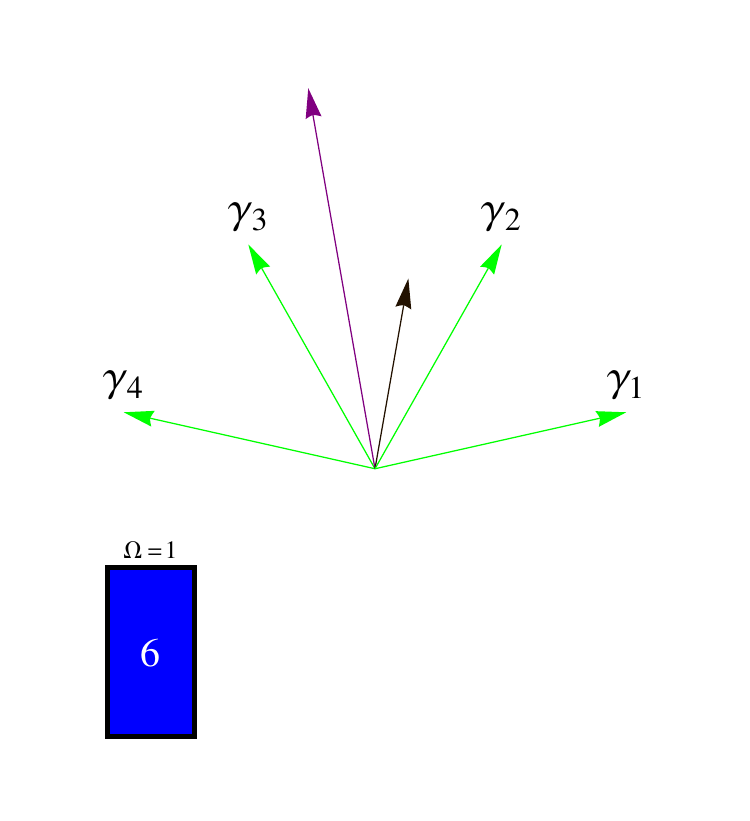}\includegraphics[width=0.28\textwidth]{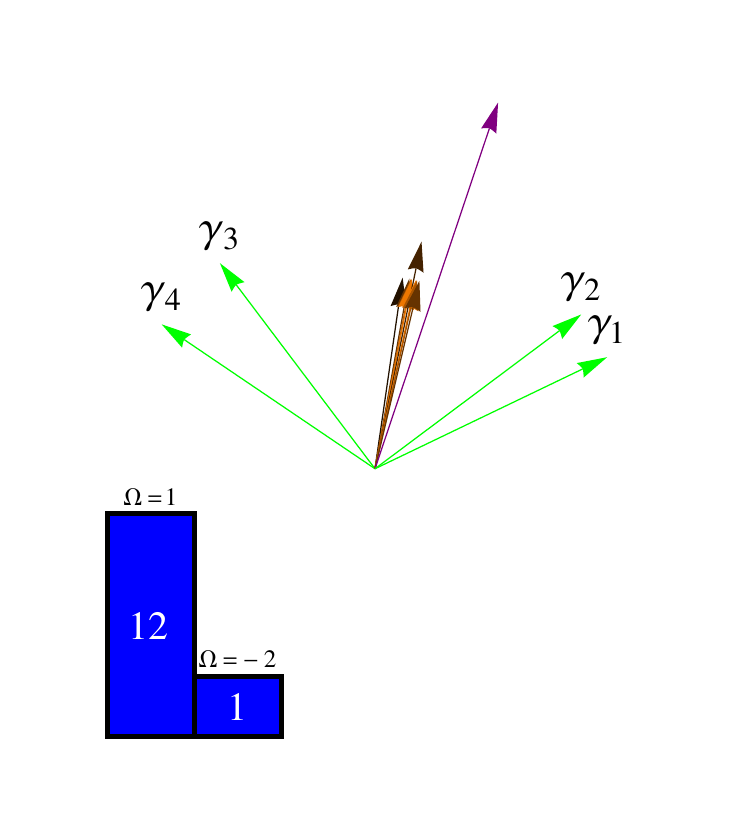}\includegraphics[width=0.28\textwidth]{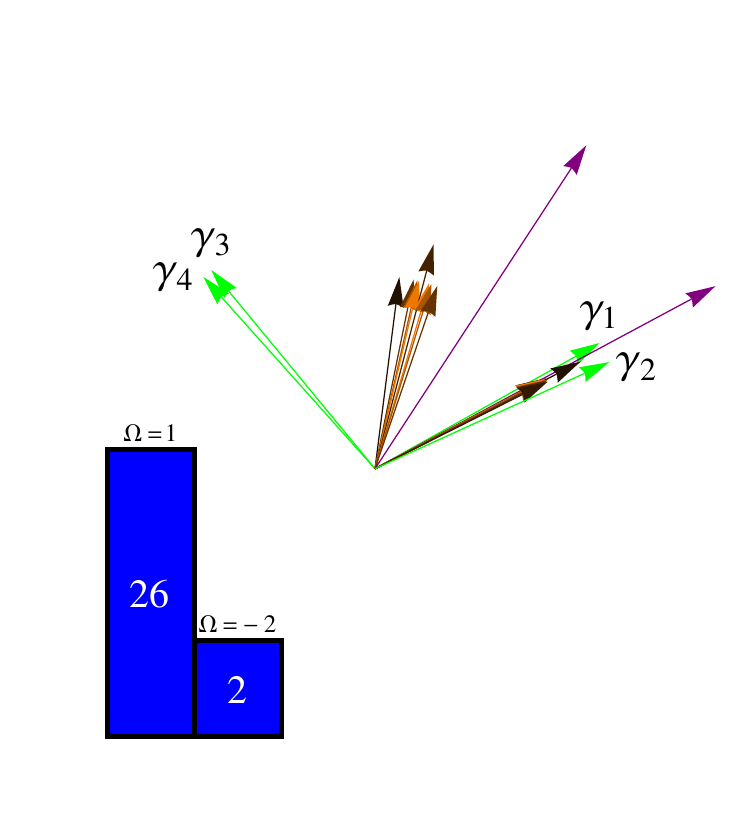}\\
\includegraphics[width=0.28\textwidth]{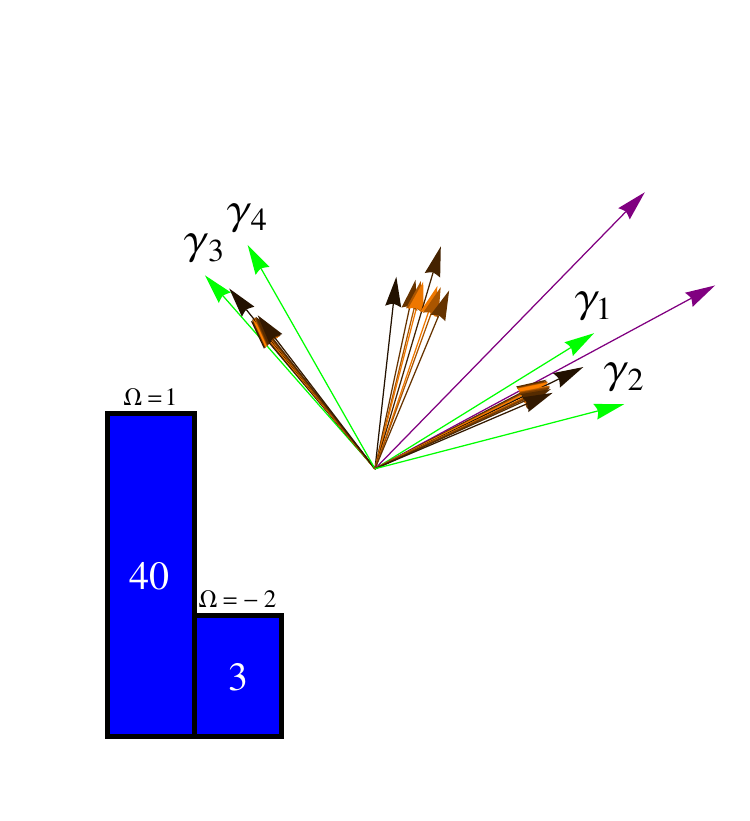}\includegraphics[width=0.28\textwidth]{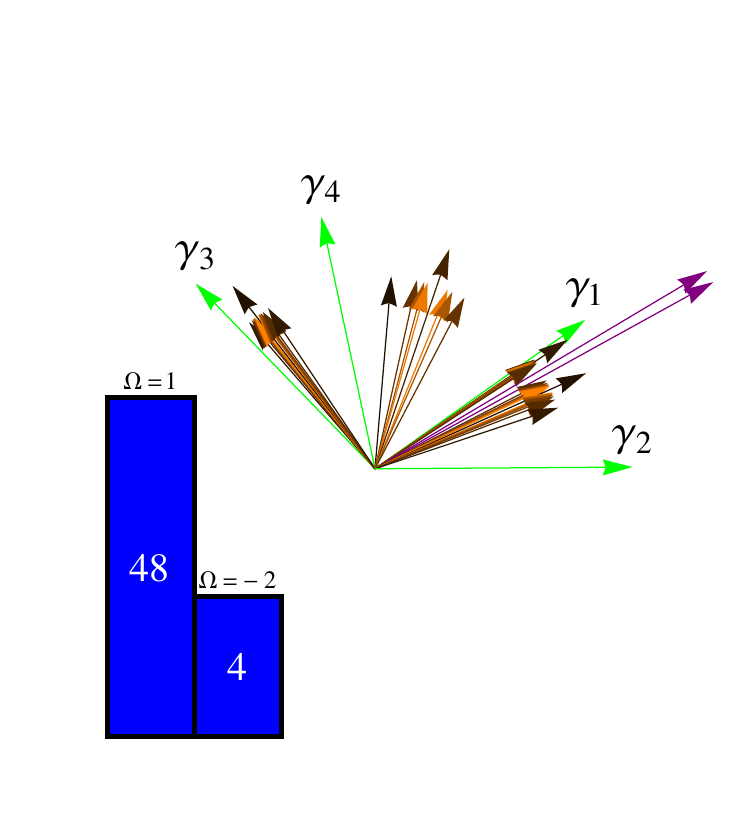}\includegraphics[width=0.28\textwidth]{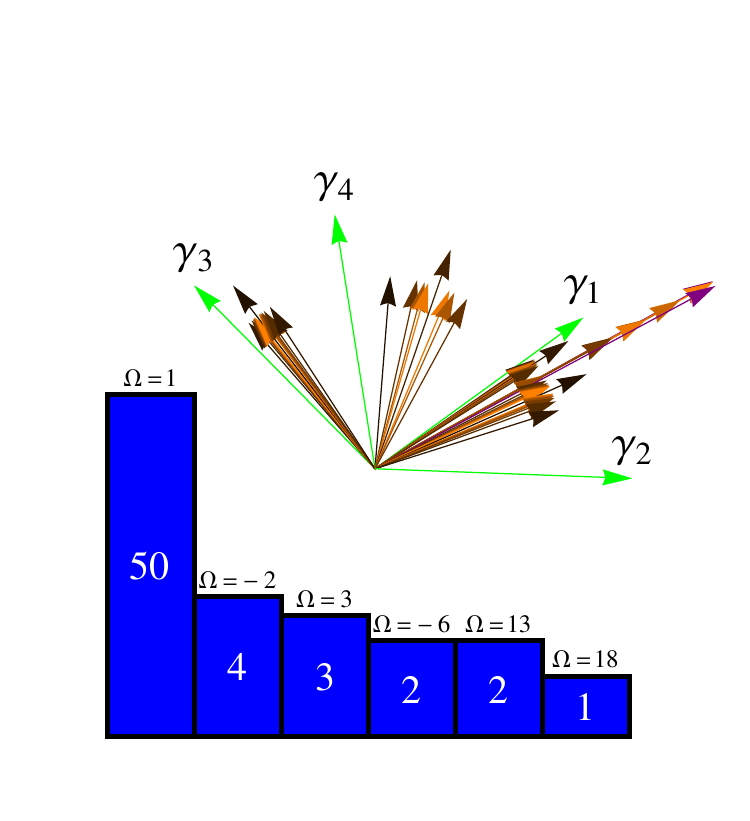}
\caption{The evolution of the spectrum is illustrated. Green arrows represent the basis hypermultiplets, the purple arrows are $\gamma_{2}+\gamma_{4}$ and $2\gamma_{1}+\gamma_{2}$, the two states that generate the $m=3$ cohort. For other charges, increasing length denotes higher $|\Omega|$ and lighter shades denote larger charges. First picture: the strong coupling chamber. Second picture: the states $\gamma_{1}+\gamma_{3}$ and $\gamma_{2}+\gamma_{4}$ have crossed and created a ${\cal C}_2$ cohort. Third picture: $\gamma_{2}$ and $\gamma_{1}$ cross and create another cohort. Fourth picture: the cohort generated by $\gamma_{3},\gamma_{4}$. Fifth picture: $\gamma_{2}+\gamma_{4}$ and $\gamma_{1}$ have crossed and created a cohort. In the sixth picture $\gamma_{2}+\gamma_{4}$ and $2\gamma_{1}+\gamma_{2}$ have crossed, generating wild degeneracies.\label{fig:spectrum-evo}. For a video showing the full evolution of the spectrum along our path, see \cite{video-spectrum-charges}.}

\end{center}
\end{figure}

Proceeding further along our path, we encounter another wall of marginal stability: $\gamma_{2}+\gamma_{4}$ undergoes wall-crossing with $2\gamma_{1}+\gamma_{2}$ generating a new cohort with $m=3$. This phenomenon has not been studied before, and deserves a detailed analysis. We anticipate here that this cohort contains distinctive new features, such as a wealth of higher spin states and a cone of densely populated BPS rays.

It is worth stressing that merely finding a point on the Coulomb branch where $Z_{\gamma_{2}+\gamma_{4}}$ approaches $Z_{2\gamma_{1}+\gamma_{2}}$ is hardly sufficient to claim that such wall-crossing happens. In addition one must make sure that such rays are populated. This is certainly the case in our example. Another important requirement is the absence of populated rays between $Z_{\gamma_{2}+\gamma_{4}}$ and $Z_{2\gamma_{1}+\gamma_{2}}$, as we approach their mutual wall of marginal stability. We claim that there aren't any, based on two independent facts. First, at values of the moduli just before $MS(\gamma_{2}+\gamma_{4},2\gamma_{1}+\gamma_{2})$, the spectral network shows simple, smooth evolution for  $\arg Z_{2\gamma_{1}+\gamma_{2}}<\vartheta<\arg Z_{\gamma_{2}+\gamma_{4}}$, see \cite{video-before-wall-focus}. Second, our explicit path on the Coulomb branch -- together with property (\ref{eq:coh-boundaries}) of cohorts -- guarantees that all boundstates created so far fall outside of the cone
bounded by the central charges of $2\gamma_{1}+\gamma_{2},\,\gamma_{2}+\gamma_{4}$: indeed if a populated boundstate were between $\gamma_{2}+\gamma_{4}$ and $2\gamma_{1}+\gamma_{2}$, it would have to be one of the following
\begin{itemize}
\item a boundstate of $2\gamma_{1}+\gamma_{2}$ with a charge counterclockwise of $\gamma_{2}+\gamma_{4}$
\item a boundstate of $\gamma_{2}+\gamma_{4}$ with a charge clockwise of $2\gamma_{1}+\gamma_{2}$
\item a boundstate of two charges lying respectively counterclockwise of $\gamma_{2}+\gamma_{4}$ and clockwise of $2\gamma_{1}+\gamma_{2}$
\item a boundstate due to one of the antiparticles
\end{itemize}
All these possibilities are clearly ruled out by our explicit choice of path. Our analysis relies on the numerical evaluation of central charges at several points on the Coulomb branch, video \cite{video-charge-motion} shows the smooth evolution of central charges of basis hypermultiplets along the path, ensuring that integration contours have been adapted suitably. Another important check is the following: at fixed $u_{2},\,u_{3}$ we tune the spectral network to the phase of central charges (as predicted numerically), and we check that there are indeed degenerate networks.

\subsection{Wall-crossings with intersections $m>3$}\label{subsec:higher-m}
So far we have encountered an MS wall of two hypermultiplets with intersection pairing $3$, but there is nothing special about $m=3$. The path proposed in (\ref{eq:path}) can be extended through walls of marginal stability with $m=4,5$, and higher. The strategy is simply to look for a direction on the Coulomb branch, along which the ray $\gamma_{2}+\gamma_{4}$ sweeps across the infinite tower of hypermultiplets with charges $(n+1)\gamma_{1}+n \gamma_{2}$.

For example, moving along a straight line from $(u_{2}^{(f)},u_{3}^{(f)})$ to
\be
u_{2}^{(4)} =0.56-0.75 i , \qquad u_{3}^{(4)} = 2.00 +1.99 i
\ee
induces wall-crossing of $\gamma_{2}+\gamma_{4}$ with $3\gamma_{1}+2\gamma_{2}$, with intersection $\langle\gamma_{2}+\gamma_{4},3\gamma_{1}+2\gamma_{2}\rangle = 4$. In this chamber the spectrum gains a new $m=4$ cohort, described by the $4$-Kronecker quiver.

Proceeding further, along a straight segment, to
\be
u_{2}^{(5)} =0.56-0.75 i , \qquad u_{3}^{(5)} = 2.00 +2.52 i
\ee
we cross the marginal stability wall of $\gamma_{2}+\gamma_{4}$ and $4\gamma_{1}+3\gamma_{2}$, with intersection $\langle\gamma_{2}+\gamma_{4},4\gamma_{1}+3\gamma_{2}\rangle = 5$ generating an $m=5$ cohort.

In the same spirit, we have checked numerically that there is a path along which
$\gamma_{2}+\gamma_{4}$ crosses all hypermultiplets with charges $(m-1)\gamma_{1}+(m-2)\gamma_{2}$, with pairings
\be
   \langle \gamma_{2}+\gamma_{4},(m-1)\gamma_{1}+(m-2)\gamma_{2} \rangle = m
\ee
hence generating an infinite tower of cohorts. The situation gets very complicated, as these cohorts will widen and start overlapping with each other, inducing further wild wall crossing.\footnote{As explained in the next section, the spectrum is best studied via the \emph{spectrum generator} technique introduced in \cite{GMN2}. This technique is straightforwardly applicable whenever comparing two points on the Coulomb branch, such that the lattice basis vectors have corresponding central charges all contained within a half-plane. When instead one or more of the central charges exit the half-plane, one needs to account for that by suitably modifying the spectrum generator. While moving from strong coupling into these \emph{wilder} regions, we actually incur in such a situation.}. It is worth stressing that, by the same reasoning outlined for the wall-crossing of $\gamma_{2}+\gamma_{4}$ with $2\gamma_{1}+\gamma_{2}$, there are no populated states between $\gamma_{2}+\gamma_{4}$ and $(m-1)\gamma_{1}+(m-2)\gamma_{2}$ 
immediately before the point where they cross.  This crucial fact guarantees that in this 
region of the Coulomb branch $m$-cohorts are generated, for arbitrarily high $m$.

Finally, we remark that a natural question arises as to whether 
analogous wall-crossings happen where the integer $m$ is negative. 
In fact, there is a simple physical argument that such wall-crossings cannot 
happen on Coulomb branches of physical theories, it goes as follows.
Let us consider two charges $\gamma_{1},\,\gamma_{2}$ with 
$\langle\gamma_{1},\gamma_{2}\rangle<0$; we would like to investigate
whether there could be a chamber of the Coulomb branch, bounded by 
$MS(\gamma_{1},\gamma_{2})$, where 
\begin{itemize}
\item ${\rm{arg}}Z_{\gamma_{2}}>{\rm{arg}}Z_{\gamma_{1}}$
\item $\Omega(\gamma)=1$ for $\gamma\in\{\pm\gamma_{1},\pm\gamma_{2}\}$
\item $\Omega(\gamma)=0$ for all other combinations $\gamma=a\gamma_{1}+b\gamma_{2}$.
\end{itemize}
If these conditions were realized, we would be in a situation in which the spectrum 
generator (defined below eq. (\ref{eq:phase-ordered-product})) contains a factor 
$\CK_{\gamma_{2}}\CK_{\gamma_{1}}$, and we stress that there would be
no other $\CK$ factors between $\CK_{\gamma_{2}}$ and $\CK_{\gamma_{1}}$.\\
We claim that this cannot happen: under sufficiently general conditions, 
near a wall $MS(\gamma_{1},\gamma_{2})$ we expect Denef's multicenter equations 
(for the case under consideration, they are reported below in (\ref{eq:Denef})) to provide a 
reliable description of the boundstates.
It is immediately evident from such description  that, in the case of negative 
$\langle\gamma_{1},\gamma_{2}\rangle=m$, on the side of $MS(\gamma_{1},\gamma_{2})$ 
where ${\rm{arg}}Z_{\gamma_{2}}>{\rm{arg}}Z_{\gamma_{1}}$, there will be stable 
boundstates of $\gamma_{1}$ with $\gamma_{2}$ populating rays between those of 
$Z_{\gamma_{1}}$ and $Z_{\gamma_{2}}$. In particular, inside the spectrum generator, 
the factors $\CK_{\gamma_{2}}$ and $\CK_{\gamma_{1}}$ are \emph{necessarily} separated 
by other factors $\CK_{a\gamma_{1}+b\gamma_{2}}$, for $a,b> 0$, violating the conditions
formulated above.

Nevertheless, it makes sense to ask what the prediction of the KSWCF would be. 
To learn something interesting, it is actually sufficient to consider the motivic version of the primitive 
WCF (see \cite{Diaconescu-Moore}). From such formula, the protected spin character 
(see appendix \ref{app:PSC-tables}) associated to $\gamma_{1}+\gamma_{2}$ has the 
simple expression
\be
	\Omega(\gamma_{1}+\gamma_{2};y):=Tr_{\mathfrak{h}_{m}} (y)^{2J_{3}}(-y)^{2I_{3}} =\frac{y^{m}-y^{-m}}{y-y^{-1}}
\ee
corresponding (not uniquely)\footnote{Albeit necessarily involving exotic representations.} to the following exotic representations of $so(3)\oplus su(2)_{R}$
\be
	\mathfrak{h}_{m} = %
	\left\{\begin{array}{lr} 
	\left(\frac{1}{2},\frac{1}{2}\right)\oplus(1,0) & m=-1 \\
	\left(0,\frac{1}{2}\right) & m=-2 \\
	\left( \frac{-m-2}{2},\frac{1}{2} \right) \oplus  \left( \frac{-m-3}{2},0 \right) & m\leq-3
	\end{array}\right..
\ee
Since the no-exotics theorem is in fact fairly well established for pure $SU(K)$ gauge theories 
\cite{DIACONESCU}, this further supports the argument that such wall-crossings cannot occur on the Coulomb branch.

\section{Some Numerical Checks on the  $m=3$ Wild Spectrum }
\label{sec:DegsFSWCF}

The discussion of Section \ref{sec:wc-wild-SU3} is sufficient to
prove that there are wild degeneracies on the Coulomb branch
of the pure $SU(3)$ theory. However, since this phenomenon is
somewhat novel, we have checked the results using
the ``spectrum generator'' in some relevant regions of the
Coulomb branch.  This section explains those checks.

\subsection{The spectrum generator technique} \label{sec:specgen}

According to the KSWCF, the phase-ordered product
\be\label{eq:phase-ordered-product}
	A(\sphericalangle)=\ : \prod_{{{\gamma}},\,{{\rm  arg}Z_{\gamma} \in \sphericalangle}} \CK_{\gamma}^{\Omega(\gamma)} :
\ee
is invariant across walls of marginal stability provided no occupied BPS rays cross into or out
of the angular sector $\sphericalangle$. Considering an angle of $\pi$ corresponds to a choice of the ``half plane of particles''. Once this choice is made, $A(\pi)$ is called a \footnote{Several equivalent choices are related by how one chooses the half-plane in the complex plane of central charges.} \emph{spectrum generator} and denoted $\mathbb{S}$ \cite{GMN2}.

The idea of the ``spectrum generator technique'' is that if - through some means or other -   one can compute $A(\pi)$, then, by factorization one
can deduce the spectrum (after computing the phase ordering of the $Z_\gamma$ at that point).  For example
in \cite{GMN2} an algorithm is given for computing $A(\pi)$ without an \emph{a priori} knowledge of the spectrum.
Here our strategy will be a little different. We will derive the spectrum generator in the strong coupling chamber, where the spectrum can
be easily read off from the spectral network or from quiver techniques. We then use wall-crossing to argue that $A(\pi)$
is unchanged along a particular path in the Coulomb branch (described in Section \ref{sec:wc-wild-SU3})
 to the wild region. Then we factorize the spectrum generator at points along that path.

An effective technique for factorizing $\mathbb{S}$ is the following. Let $\{\gamma_{i}\}_{i=1,\dots,k}$ be a basis for the lattice of charges $\Gamma$, and let $\gamma =\sum{a_{i}\gamma_{i}}$, with $a_{i}>0$. Define the height $|\gamma|:=\sum_{i} a_{i}$, and $\mathbb{S}^{(r)}=: \prod_{\gamma, |\gamma|\leq r} \CK_{\gamma}^{\Omega(\gamma)}  :$\footnote{Recall that the ordering depends crucially on the position $u$ on the Coulomb branch, hence we should really write $\mathbb{S}^{(r)}(u)$. To lighten the notation we do not indicate the $u$-dependence.}.
The full spectrum generator $\mathbb{S}$ can then be factorized by studying its action on the basis formal variables\footnote{i.e., it is sufficient to work with formal variables corresponding to a choice of simple roots for the lattice of charges. The choice of simple roots must be consistent with the choice of half-plane that comes with the spectrum generator.} $Y_{\gamma_{i}}$ for increasing values of $r$, by employing
\be \label{eq:decomposition-technique}
Y_{-\gamma_{i}}(\mathbb{S}-{\tilde{\mathbb{S}}^{(r)})}Y_{\gamma_{i}}=-\sum\lm_{|\gamma'|=r+1} \langle \gamma_{i},\gamma' \rangle\Omega(\gamma') Y_{\gamma'}+\ldots
\ee
where $\tilde{\mathbb{S}}$ represents the factorization of the spectrum generator under study. The ellipses contain terms with $Y_{\gamma}, \lvert \gamma \rvert > r+1$.

\subsection{Factorizing the spectrum generator  } \label{sec:wild-SU3}

The spectrum in the strong coupling region can be obtained via spectral network techniques, as discussed in Section \ref{subsec:strong-cplg}.
According to the results presented there, the spectrum generator is
\be
	\mathbb{S} = \CK_{\gamma_{4}} \CK_{\gamma_{3}}  \CK_{\gamma_{2}+\gamma_{4}}\CK_{\gamma_{1}+\gamma_{3}} \CK_{\gamma_{2}} \CK_{\gamma_{1}},
\ee
in agreement with \cite{GMN5,CV}.

We now fix a point on our path
\be
 u_{2}= 0.56 - 0.73 i,\quad u_{3}=1.94 + 1.49 i,
\ee
corresponding to the situation exhibited in (\ref{eq:cohorts}) immediately before the wall ${\rm MS}(\gamma_{2}+\gamma_{4},2\gamma_{1}+\gamma_{2})$.  The central charges corresponding to the simple roots are
\be
\begin{array}{ll}
 Z_{\gamma_{1}}=8.42972 + 6.00549 i &  Z_{\gamma_{2}}= 4.83278 - 0.0226871 i\\
 Z_{\gamma_{3}}=-7.30679 +  7.50651 i &  Z_{\gamma_{4}}= -0.504898 + 2.53401 i\, ,
\end{array}
\ee
the factorization of the spectrum generator up to $|\gamma|=21$ reads\footnote{Color code: The
 factors in blue come from the hypermultiplets of
    the strong coupling chamber. The factors in red come from  vectormultiplets. The remaining
    factors in black are hypermultiplets created by the wall-crossing from the strong coupling chamber.}
\be
\begin{split}
& {\color{Blue}{\CK_{\gamma _3}}}\CK_{2 \gamma _3+\gamma _4}\CK_{3 \gamma _3+2 \gamma _4}\CK_{4 \gamma _3+3 \gamma _4}\CK_{5 \gamma _3+4 \gamma _4}\CK_{6 \gamma _3+5 \gamma _4}\CK_{7 \gamma _3+6 \gamma _4}\CK_{8 \gamma _3+7 \gamma _4}\CK_{9 \gamma _3+8 \gamma _4} \\
& \CK_{10 \gamma _3+9 \gamma _4}\CK_{11 \gamma _3+10 \gamma _4}{\color{Red}{\CK_{\gamma _3+\gamma _4}^{-2}}}\CK_{10 \gamma _3+11 \gamma _4}\CK_{9 \gamma _3+10 \gamma _4}\CK_{8 \gamma _3+9 \gamma _4}\CK_{7 \gamma _3+8 \gamma _4}\CK_{6 \gamma _3+7 \gamma _4} \\
& \CK_{5 \gamma _3+6 \gamma _4}\CK_{4 \gamma _3+5 \gamma _4}\CK_{3 \gamma _3+4 \gamma _4}\CK_{2 \gamma _3+3 \gamma _4}\CK_{\gamma _3+2 \gamma _4}{\color{Blue}{\CK_{\gamma _4}}\color{Blue}{\CK_{\gamma _1+\gamma _3}}}\CK_{2 \gamma _1+\gamma _2+2 \gamma _3+\gamma _4} \\
& \CK_{3 \gamma _1+2 \gamma _2+3 \gamma _3+2 \gamma _4}\CK_{4 \gamma _1+3 \gamma _2+4 \gamma _3+3 \gamma _4}\CK_{5 \gamma _1+4 \gamma _2+5 \gamma _3+4 \gamma _4}{\color{Red}{\CK_{\gamma _1+\gamma _2+\gamma _3+\gamma _4}^{-2}}}\CK_{4 \gamma _1+5 \gamma _2+4 \gamma _3+5 \gamma _4} \\
& \CK_{3 \gamma _1+4 \gamma _2+3 \gamma _3+4 \gamma _4}\CK_{2 \gamma _1+3 \gamma _2+2 \gamma _3+3 \gamma _4}\CK_{\gamma _1+2 \gamma _2+\gamma _3+2 \gamma _4}{\color{Blue}{\CK_{\gamma _1}}}\CK_{2 \gamma _1+\gamma _2+\gamma _4}\CK_{3 \gamma _1+2 \gamma _2+2 \gamma _4} \\
& \CK_{4 \gamma _1+3 \gamma _2+3 \gamma _4}\CK_{5 \gamma _1+4 \gamma _2+4 \gamma _4}\CK_{6 \gamma _1+5 \gamma _2+5 \gamma _4}\CK_{7 \gamma _1+6 \gamma _2+6 \gamma _4}{\color{Red}{\CK_{\gamma _1+\gamma _2+\gamma _4}^{-2}}}\CK_{6 \gamma _1+7 \gamma _2+7 \gamma _4} \\
& \CK_{5 \gamma _1+6 \gamma _2+6 \gamma _4}\CK_{4 \gamma _1+5 \gamma _2+5 \gamma _4}\CK_{3 \gamma _1+4 \gamma _2+4 \gamma _4}\CK_{2 \gamma _1+3 \gamma _2+3 \gamma _4}\CK_{\gamma _1+2 \gamma _2+2 \gamma _4}{\color{Blue}{\CK_{\gamma _2+\gamma _4} }}\\
& \CK_{2 \gamma _1+\gamma _2}\CK_{3 \gamma _1+2 \gamma _2}\CK_{4 \gamma _1+3 \gamma _2}\CK_{5 \gamma _1+4 \gamma _2} \CK_{6 \gamma _1+5 \gamma _2}\CK_{7 \gamma _1+6 \gamma _2}\CK_{8 \gamma _1+7 \gamma _2}\CK_{9 \gamma _1+8 \gamma _2} \\
& \CK_{10 \gamma _1+9 \gamma _2}\CK_{11 \gamma _1+10 \gamma _2}{\color{Red}{\CK_{\gamma _1+\gamma _2}^{-2}}}\CK_{10 \gamma _1+11 \gamma _2}\CK_{9 \gamma _1+10 \gamma _2}\CK_{8 \gamma _1+9 \gamma _2}\CK_{7 \gamma _1+8 \gamma _2}\CK_{6 \gamma _1+7 \gamma _2} \\
& \CK_{5 \gamma _1+6 \gamma _2}\CK_{4 \gamma _1+5 \gamma _2}\CK_{3 \gamma _1+4 \gamma _2}\CK_{2 \gamma _1+3 \gamma _2}\CK_{\gamma _1+2 \gamma _2}{\color{Blue}{\CK_{\gamma _2}}}\label{eq:before-wall}
\end{split}
\ee
The spectrum exhibits four $m=2$ cohorts, as expected from the discussion of Section \ref{sec:su3coh}: they include four vectormultiplets (with $\Omega=-2$), accompanied by infinite towers of hypermultiplets.

On the other side of the $m=3$ wall, at
\be
 u_{2}= 0.56 - 0.75 i,\, u_{3}=2.00 + 1.52 i,
\ee
central charges read
\be
\begin{array}{ll}
Z_{\gamma_{1}}=8.52337 + 6.18454 i &  Z_{\gamma_{2}}= 4.89813 - 0.18347  i\\
Z_{\gamma_{3}}=-7.43876 +  7.53531 i &  Z_{\gamma_{4}}= -0.410809 + 2.59321 i .
\end{array}
\ee
The spectrum generator, up to $|\gamma|=21$, is
\be
\begin{split}
& {\color{Blue}{\CK_{\gamma _3}}}\CK_{2 \gamma _3+\gamma _4}\CK_{3 \gamma _3+2 \gamma _4}\CK_{4 \gamma _3+3 \gamma _4}\CK_{5 \gamma _3+4 \gamma _4}\CK_{6 \gamma _3+5 \gamma _4}\CK_{7 \gamma _3+6 \gamma _4}\CK_{8 \gamma _3+7 \gamma _4}\CK_{9 \gamma _3+8 \gamma _4}\\
& \CK_{10 \gamma _3+9 \gamma _4}\CK_{11 \gamma _3+10 \gamma _4}{\color{Red}{\CK_{\gamma _3+\gamma _4}^{-2}}}\CK_{10 \gamma _3+11 \gamma _4}\CK_{9 \gamma _3+10 \gamma _4}\CK_{8 \gamma _3+9 \gamma _4}\CK_{7 \gamma _3+8 \gamma _4}\CK_{6 \gamma _3+7 \gamma _4}\\
& \CK_{5 \gamma _3+6 \gamma _4}\CK_{4 \gamma _3+5 \gamma _4}\CK_{3 \gamma _3+4 \gamma _4}\CK_{2 \gamma _3+3 \gamma _4}\CK_{\gamma _3+2 \gamma _4}{\color{Blue}{\CK_{\gamma _4}\CK_{\gamma _1+\gamma _3}}}\CK_{2 \gamma _1+\gamma _2+2 \gamma _3+\gamma _4}\\
& \CK_{3 \gamma _1+2 \gamma _2+3 \gamma _3+2 \gamma _4}\CK_{4 \gamma _1+3 \gamma _2+4 \gamma _3+3 \gamma _4}\CK_{5 \gamma _1+4 \gamma _2+5 \gamma _3+4 \gamma _4}{\color{Red}{\CK_{\gamma _1+\gamma _2+\gamma _3+\gamma _4}^{-2}}}\CK_{4 \gamma _1+5 \gamma _2+4 \gamma _3+5 \gamma _4}\\
& \CK_{3 \gamma _1+4 \gamma _2+3 \gamma _3+4 \gamma _4}\CK_{2 \gamma _1+3 \gamma _2+2 \gamma _3+3 \gamma _4}\CK_{\gamma _1+2 \gamma _2+\gamma _3+2 \gamma _4}{\color{Blue}{\CK_{\gamma _1}}}\CK_{2 \gamma _1+\gamma _2+\gamma _4}\CK_{3 \gamma _1+2 \gamma _2+2 \gamma _4}\\
& \CK_{4 \gamma _1+3 \gamma _2+3 \gamma _4}\CK_{5 \gamma _1+4 \gamma _2+4 \gamma _4}\CK_{6 \gamma _1+5 \gamma _2+5 \gamma _4}\CK_{7 \gamma _1+6 \gamma _2+6 \gamma _4}{\color{Red}{\CK_{\gamma _1+\gamma _2+\gamma _4}^{-2}}}\CK_{6 \gamma _1+7 \gamma _2+7 \gamma _4}\\
& \CK_{5 \gamma _1+6 \gamma _2+6 \gamma _4}\CK_{4 \gamma _1+5 \gamma _2+5 \gamma _4}\CK_{3 \gamma _1+4 \gamma _2+4 \gamma _4}\CK_{2 \gamma _1+3 \gamma _2+3 \gamma _4}\CK_{\gamma _1+2 \gamma _2+2 \gamma _4}{\color{OliveGreen}{\CK_{2 \gamma _1+\gamma _2}}}\\
& {\color{OliveGreen}{\CK_{6 \gamma _1+4 \gamma _2+\gamma _4}\CK_{10 \gamma _1+7 \gamma _2+2 \gamma _4}^3\CK_{4 \gamma _1+3 \gamma _2+\gamma _4}^3{\CK_{8 \gamma _1+6 \gamma _2+2 \gamma _4}^{-6}}\CK_{10 \gamma _1+8 \gamma _2+3 \gamma _4}^{68}\CK_{6 \gamma _1+5 \gamma _2+2 \gamma _4}^{13}}}\\
& {\color{OliveGreen}{\CK_{8 \gamma _1+7 \gamma _2+3 \gamma _4}^{68} \CK_{6 \gamma _1+6 \gamma _2+3 \gamma _4}^{18}\CK_{2 \gamma _1+2 \gamma _2+\gamma _4}^3{\CK_{4 \gamma _1+4 \gamma _2+2 \gamma _4}^{-6}}{\CK_{8 \gamma _1+8 \gamma _2+4 \gamma _4}^{-84}}\CK_{6 \gamma _1+7 \gamma _2+4 \gamma _4}^{68} }}\\
& {\color{OliveGreen}{\CK_{4 \gamma _1+5 \gamma _2+3 \gamma _4}^{13}\CK_{6 \gamma _1+8 \gamma _2+5 \gamma _4}^{68}\CK_{6 \gamma _1+9 \gamma _2+6 \gamma _4}^{18}\CK_{2 \gamma _1+3 \gamma _2+2 \gamma _4}^3{\CK_{4 \gamma _1+6 \gamma _2+4 \gamma _4}^{-6}} }}\\
& {\color{OliveGreen}{\CK_{4 \gamma _1+7 \gamma _2+5 \gamma _4}^3\CK_{2 \gamma _1+4 \gamma _2+3 \gamma _4}}}{\color{Blue}{\CK_{\gamma _2+\gamma _4}}}\CK_{3 \gamma _1+2 \gamma _2}\CK_{4 \gamma _1+3 \gamma _2}\CK_{5 \gamma _1+4 \gamma _2}\CK_{6 \gamma _1+5 \gamma _2}\CK_{7 \gamma _1+6 \gamma _2} \\
& \CK_{8 \gamma _1+7 \gamma _2}\CK_{9 \gamma _1+8 \gamma _2}\CK_{10 \gamma _1+9 \gamma _2}\CK_{11 \gamma _1+10 \gamma _2}{\color{Red}{\CK_{\gamma _1+\gamma _2}^{-2}}}\CK_{10 \gamma _1+11 \gamma _2}\CK_{9 \gamma _1+10 \gamma _2}\CK_{8 \gamma _1+9 \gamma _2} \\
& \CK_{7 \gamma _1+8 \gamma _2}\CK_{6 \gamma _1+7 \gamma _2}\CK_{5 \gamma _1+6 \gamma _2}\CK_{4 \gamma _1+5 \gamma _2}\CK_{3 \gamma _1+4 \gamma _2}\CK_{2 \gamma _1+3 \gamma _2}\CK_{\gamma _1+2 \gamma _2}{{\color{Blue}{\CK_{\gamma _2}}}},
\end{split}\label{eq:after-wall}
\ee
where $\CK$-factors in green are those of the newborn $m=3$ cohort. Notice the large values of $\Omega$.

Both formulae (\ref{eq:before-wall}), (\ref{eq:after-wall}) can be recast in more suggestive forms by adopting the notation\footnote{We adopt the following conventions: a product of noncommutative factors $\prod_{k \nearrow a}^{b}$ indicates that values of $k$ increase from left to right between $a$ and $b$, while $\prod_{k \searrow a}^{b}$ denotes decreasing values of $k$ from left to right.}
\be
	\Pi^{(n,m)}(a,b):=\left( \prod_{k\nearrow n}^{\infty}\CK_{(k+1) a+ k b} \right) \,\CK_{a+b}^{-2} \, \left(\prod_{\ell\searrow m}^{\infty}\CK_{\ell a+ (\ell+1) b}\right)
\ee
Expression (\ref{eq:before-wall}) is then simply the truncation to $|\gamma|=21$ of (cf. (\ref{eq:cohorts}))
\be
	\Pi^{(0,0)}(\gamma_{3},\gamma_{4})\ \Pi^{(0,1)}(\gamma_{1}+\gamma_{3}, \gamma_{2}+\gamma_{4}) \ \Pi^{(0,0)}(\gamma_{1},\gamma_{2}+\gamma_{4}) \ \Pi^{(1,0)}(\gamma_{1},\gamma_{2})
\ee
Similarly, for (\ref{eq:after-wall}) we have
\be
\begin{split}
	& \Pi^{(0,0)}(\gamma_{3},\gamma_{4})\ \Pi^{(0,1)}(\gamma_{1}+\gamma_{3}, \gamma_{2}+\gamma_{4}) \ \Pi^{(0,1)}(\gamma_{1},\gamma_{2}+\gamma_{4}) \\
	& \qquad \qquad  \Xi(2\gamma_{1}+\gamma_{2},\gamma_{2}+\gamma_{4})\ \Pi^{(2,0)}(\gamma_{1},\gamma_{2})
\end{split}
\ee
where $\Xi(2\gamma_{1}+\gamma_{2},\gamma_{2}+\gamma_{4})$ represents the contribution from the full ${\cal C}_{3}(2\gamma_{1}+\gamma_{2},\gamma_{2}+\gamma_{4})$ cohort, which we now analyze in greater detail.

\subsection{Exponential growth of the BPS degeneracies }\label{subsec:3kron}

We now focus on the part of BPS spectrum within the cohort ${\cal C}_{3}(\gamma_{2}+\gamma_{4},2\gamma_{1}+\gamma_{2})$. Let
\be
\CK_{(a,b)}\equiv \CK_{a(2\gamma_{1}+\gamma_{2})+b(\gamma_{2}+\gamma_{4})} , \qquad {a,b\in\mathbb{Z}},
\ee
then, up to $a+b=15$, $\Xi(2\gamma_{1}+\gamma_{2},\gamma_{2}+\gamma_{4})$ reads
\be
\begin{split}
&\CK_{(1,0)}\CK_{(3,1)}\CK_{(8,3)}\CK^{-6}_{(10,4)}\CK^3_{(5,2)}\CK^{13}_{(7,3)}\CK^{68}_{(9,4)}\CK^{465}_{(10,5)} \CK^{-84}_{(8,4)}\CK^{18}_{(6,3)}\\
&\quad \CK^{-6}_{(4,2)} \CK^3_{(2,1)}\CK^{2530}_{(9,5)}\CK^{399}_{(7,4)}\CK^{68}_{(5,3)}\CK^{4242}_{(8,5)}\CK^{34227}_{(9,6)}\CK^{-478}_{(6,4)}\CK^{13}_{(3,2)}\CK^{4242}_{(7,5)}\\
&\qquad \CK^{-32050}_{(8,6)}\CK^{68}_{(4,3)}\CK^{399}_{(5,4)}\CK^{2530}_{(6,5)}\CK^{16965}_{(7,6)} \CK^{118668}_{(8,7)}\CK^{18123}_{(7,7)}\CK^{-2808}_{(6,6)}\CK^{465}_{(5,5)}\\
&\qquad \quad \CK^{-84}_{(4,4)}\CK^{18}_{(3,3)}\CK^{-6}_{(2,2)}\CK^3_{(1,1)} \CK^{118668}_{(7,8)}\CK^{16965}_{(6,7)}\CK^{2530}_{(5,6)}\CK^{399}_{(4,5)} \CK^{-32050}_{(6,8)}\CK^{68}_{(3,4)}\\
&\qquad \qquad\CK^{4242}_{(5,7)}\CK^{34227}_{(6,9)}\CK^{-478}_{(4,6)}\CK^{13}_{(2,3)}\CK^{4242}_{(5,8)}\CK^{68}_{(3,5)}\CK^{399}_{(4,7)}\CK^{2530}_{(5,9)}\CK^{465}_{(5,10)}\CK^{-84}_{(4,8)}\\
&\qquad \qquad\quad\CK^{18}_{(3,6)}\CK^{-6}_{(2,4)}\CK^3_{(1,2)}\CK^{68}_{(4,9)}\CK^{13}_{(3,7)}\CK^{-6}_{(4,10)}\CK^3_{(2,5)}\CK_{(3,8)}\CK_{(1,3)}\CK_{(0,1)}\label{eq:C3-spectrum}
\end{split}
\ee
The BPS degeneracies appearing in (\ref{eq:C3-spectrum}) look rather \emph{wild} at first sight. One way of looking at them is to consider sequences of charges $(a_{0}+n a,b_{0}+n b)$ approaching different ``slopes'' $a/b$ for $n\to \infty$, and study the asymptotics of $\Omega$ for large $n$. As illustrated in figure \ref{fig:logOmega}, the BPS index grows exponentially with $n$, the asymptotic exponential behavior depends entirely on $a/b$ and not on $a_{0}, b_{0}$.

\begin{figure}[htbp]
\begin{center}
\includegraphics[width=0.44\textwidth]{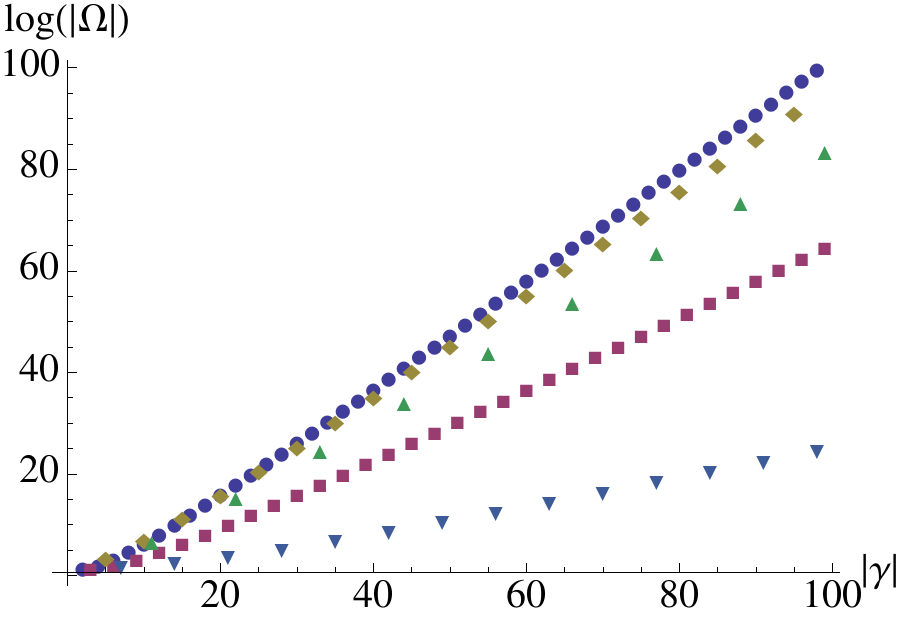} \includegraphics[width=0.44\textwidth]{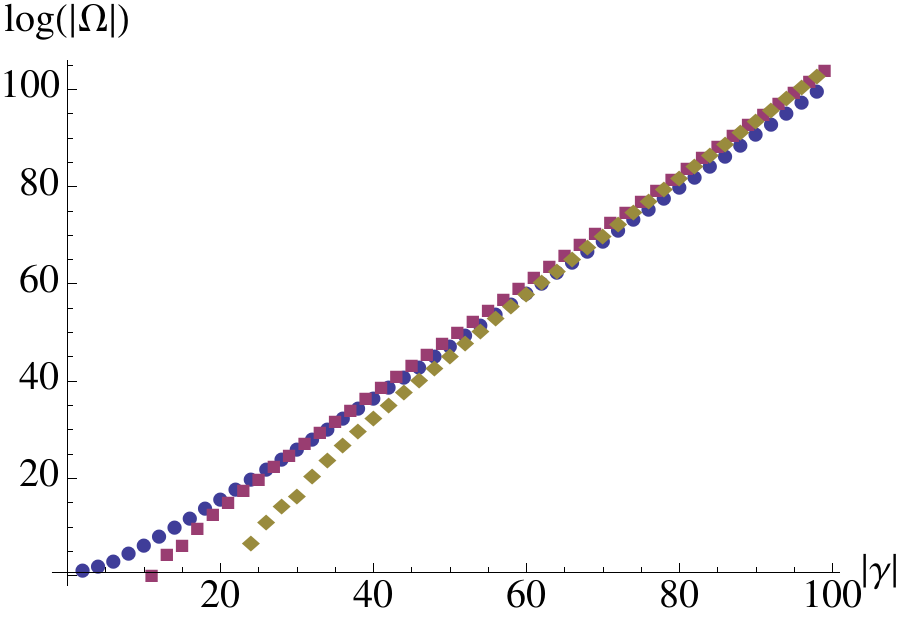}
\caption{Left: values of $\log\Omega( a n,b n)$ for several slopes $a/b$: $1$ (circles), $3/2$ (diamonds), $7/4$ (up-triangles), $2$ (squares), $5/2$ (down-triangles). Right: sequences of type $(a_{0}+ a n, b_{0}+ b n)$ have the same asymptotics; here we show $a=b=1$ with $a_{0}-b_{0} = 0,5,10$.\label{fig:logOmega}}
\end{center}
\end{figure}

According to the positivity conjecture discussed below equation (\ref{eq:spin-splitting}), BPS indices count dimensions of Hilbert subspaces, as stated more precisely in (\ref{eq:positivity-consequence}). Such exponential growth in the number of states may seem surprising in the context of a gauge theory.  We will return to the physical implications below, in Section \ref{sec:physical-estimates}.

\section{Relation to quivers}\label{subsec:RelateQuivers}

In addition to spectral networks, one
 alternative route to the BPS spectrum is the dual description in terms of quiver
quantum mechanics \cite{DENEF,Denef-Moore,CV}.
The problem of counting BPS states gets mapped into that of counting cohomology classes of moduli spaces of quiver representations. These classes are organized into Lefschetz multiplets, which correspond to the $\mathfrak{so}(3)$ multiplets. The PSC $\Omega(\gamma,u;y)$ is then given by the Poincar\'e polynomial associated to a certain quiver representation.

The basic observation here is that an isolated wall-crossing of hypermultiplets with
charges $\gamma, \gamma'$ such that $\langle \gamma, \gamma' \rangle = m$ will produce
the spectrum of the Kronecker $m$-quiver in the wild stability region.

\subsection{Derivation of the Kronecker $m$-quivers from the strong coupling regime}

Here we briefly describe how the quiver description fits in our study of the BPS spectrum of this theory. We start in the strong coupling chamber: we choose a half-plane as shown in the first frame of figure \ref{fig:subquiver}, the corresponding BPS quiver is shown in the second frame of the same figure. As we move along the path (\ref{eq:path}), we come to the situation shown in the third frame of figure \ref{fig:subquiver}: three MS walls have been crossed, and the corresponding $m=2$ cohorts are indicated (this corresponds to the situation shown in the fifth frame of figure \ref{fig:spectrum-evo} above.). Note that no walls of the second kind\footnote{In the physics literature, a wall of the second kind is, roughly speaking, the locus on the moduli space where the central charge of a populated state \emph{exits} the half-plane associated with the quiver under study. When this happens, the quiver description changes by a mutation, for more details, see \cite{CV}.} have been crossed, hence the same BPS
quiver is still valid.\\
\begin{figure}[h!]
        \centering
                \includegraphics[width=0.27\textwidth]{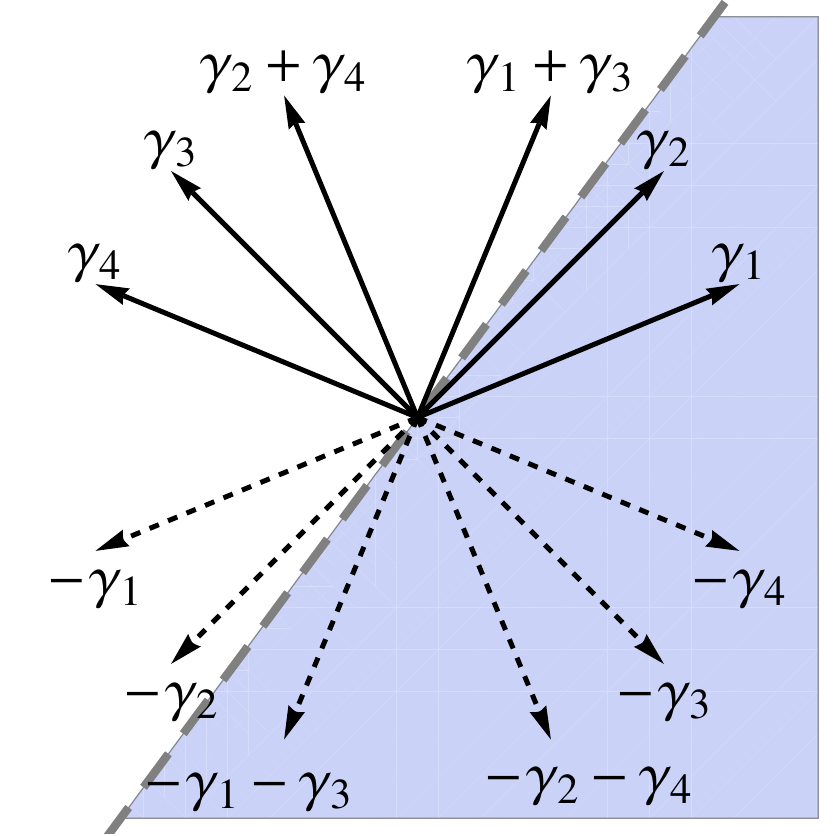} \hfill \includegraphics[width=0.25\textwidth]{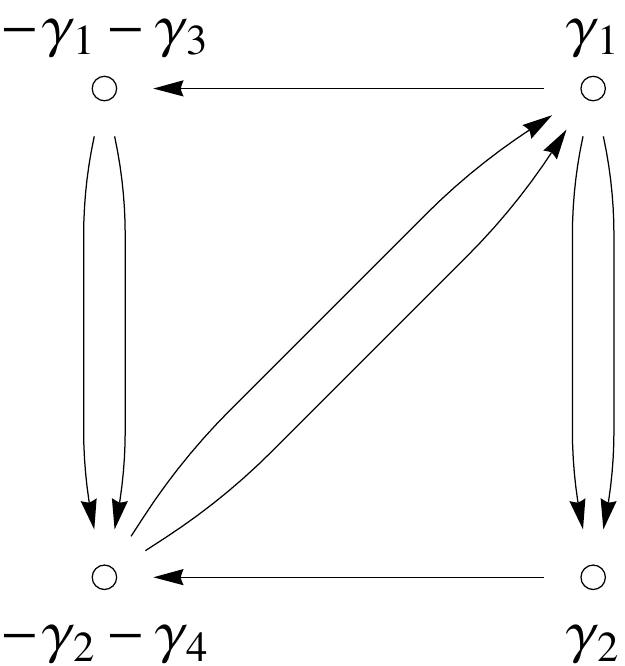} \hfill \includegraphics[width=0.27\textwidth]{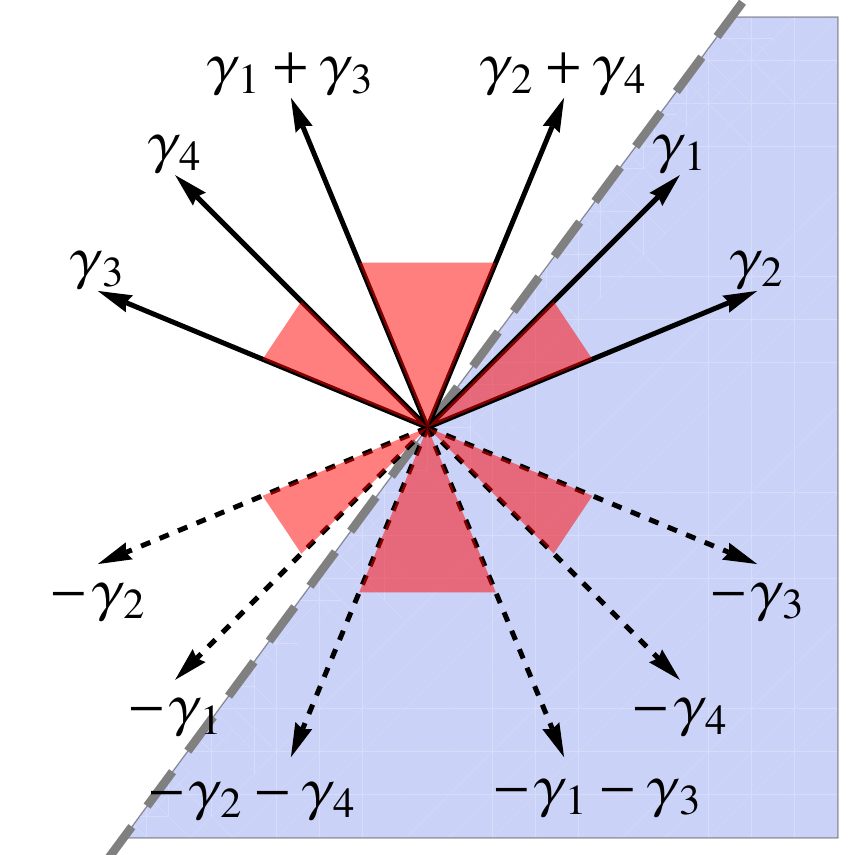}
        \caption{Left: the disposition of charges and choice of half plane in the strong coupling chamber. The depiction of the central charges is schematic. Center: the quiver at strong coupling. Right: central charges and cohorts after crossing the first three MS walls along our path.\label{fig:subquiver}}
\end{figure}
Now, while keeping the moduli fixed, we rotate the half-plane clockwise inducing a mutation on the quiver, as shown in the first two frames of figure \ref{fig:subquiver-bis}. We then proceed a little further along our path on ${\cal B}$, until we cross the wall $MS(\gamma_{1},\gamma_{2}+\gamma_{4})$, again this does not involve crossing walls of the second kind, and the same quiver is still valid. The charge disposition and cohorts are shown in the third frame of figure \ref{fig:subquiver-bis}.\\
\begin{figure}[h!]
        \centering
                \includegraphics[width=0.27\textwidth]{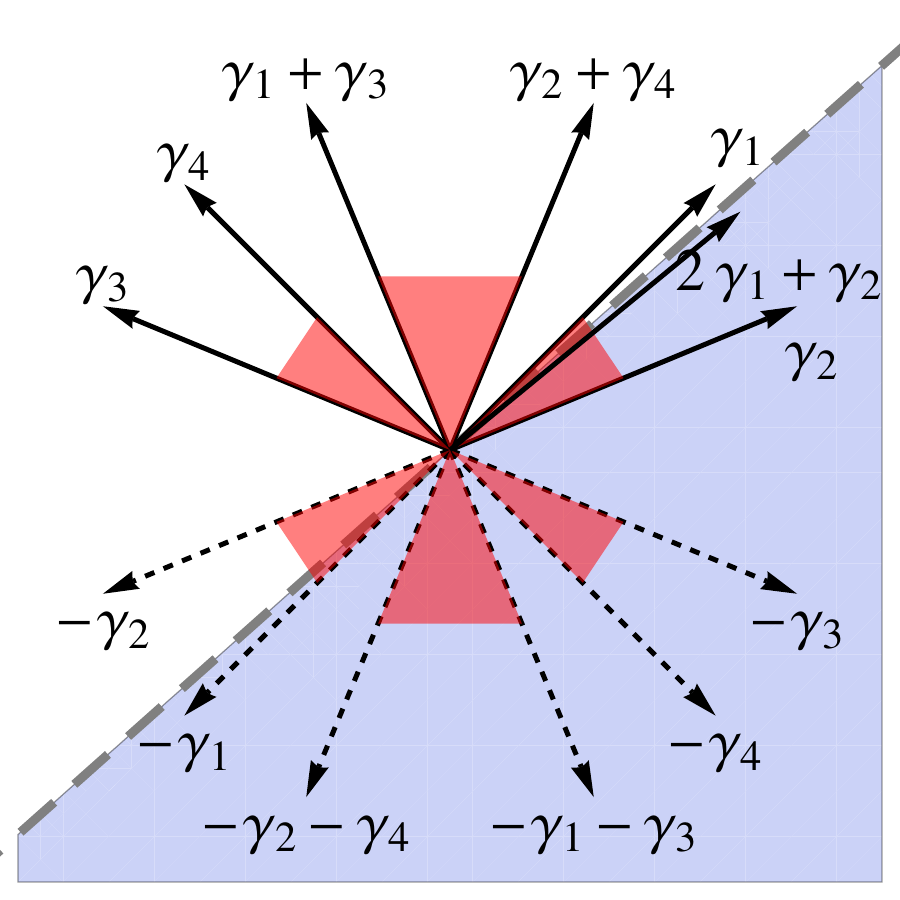} \hfill \includegraphics[width=0.25\textwidth]{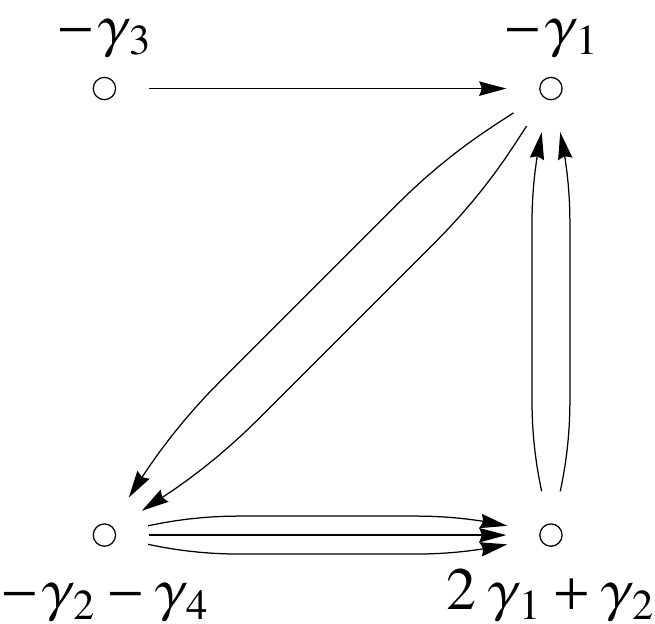} \hfill \includegraphics[width=0.27\textwidth]{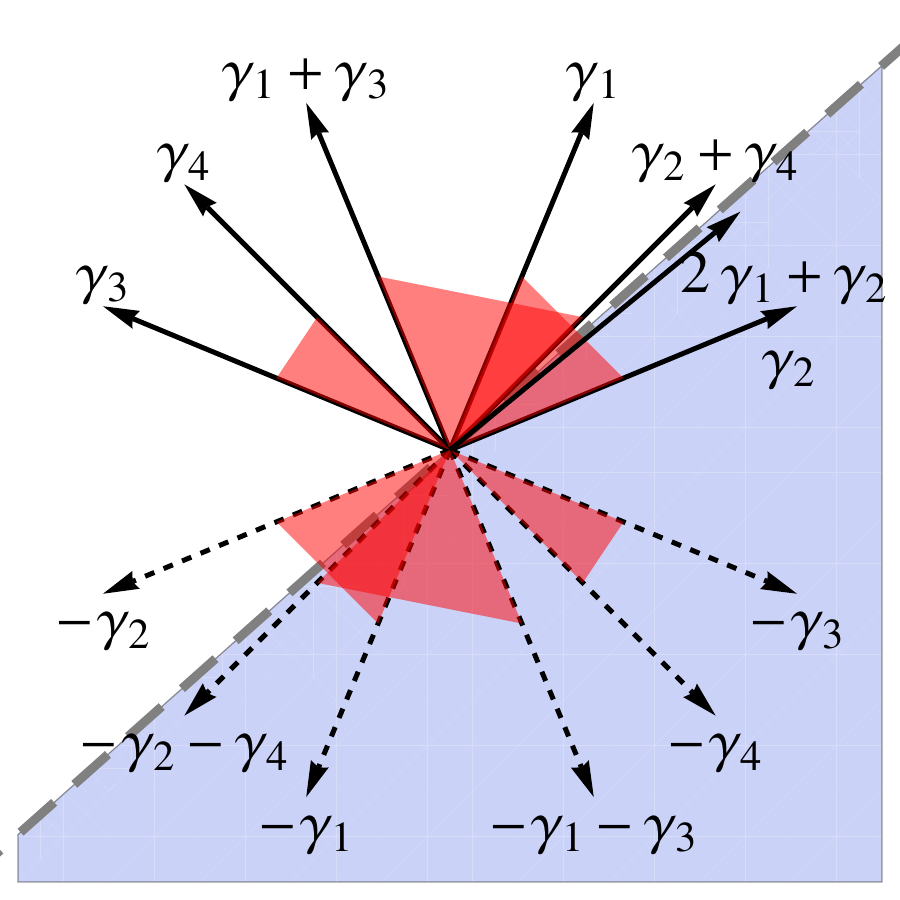}
        \caption{Left: a clockwise rotation of the half-plane past the ray $Z_{\gamma_{1}}$. Center: the corresponding BPS quiver. Right: after proceeding further on ${\cal B}$ we cross $MS(\gamma_{1},\gamma_{2}+\gamma_{4})$\label{fig:subquiver-bis}}
\end{figure}
Finally, we rotate the half-plane counterclockwise, as shown in figure \ref{fig:subquiver-tris}, inducing an inverse mutation on the node $-\gamma_{2}-\gamma_{4}$, which results in the desired BPS quiver.\\
\begin{figure}[h!]
        \centering
                \hfill \includegraphics[width=0.27\textwidth]{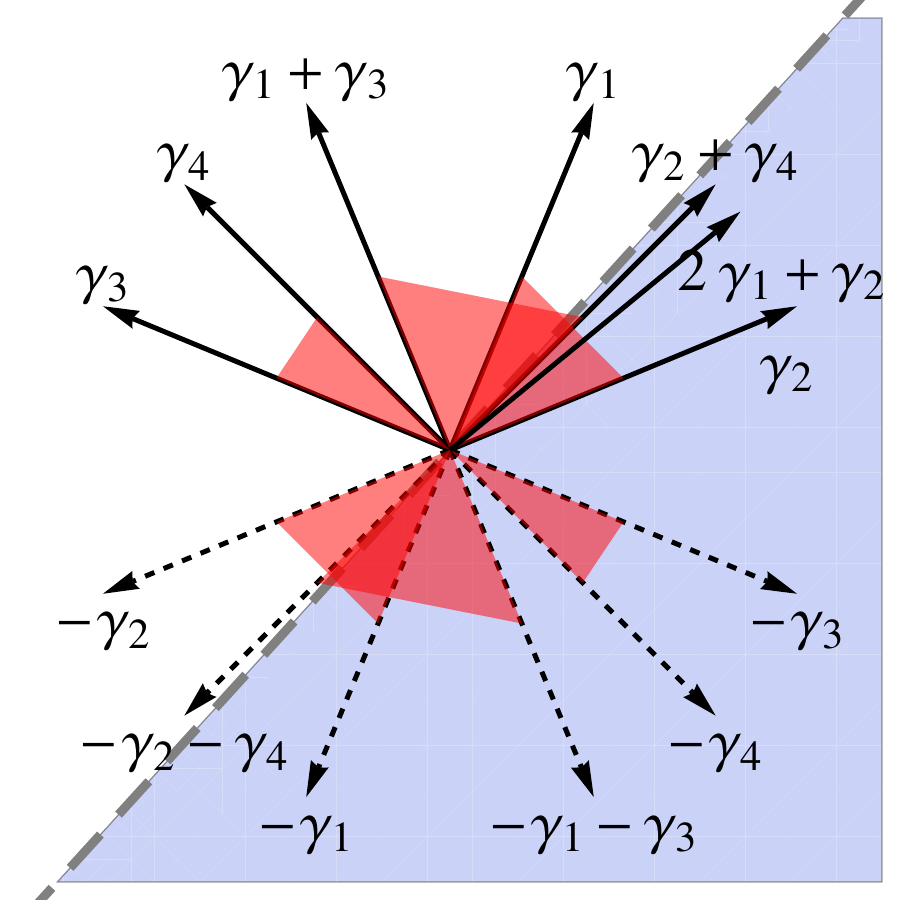} \hfill \includegraphics[width=0.29\textwidth]{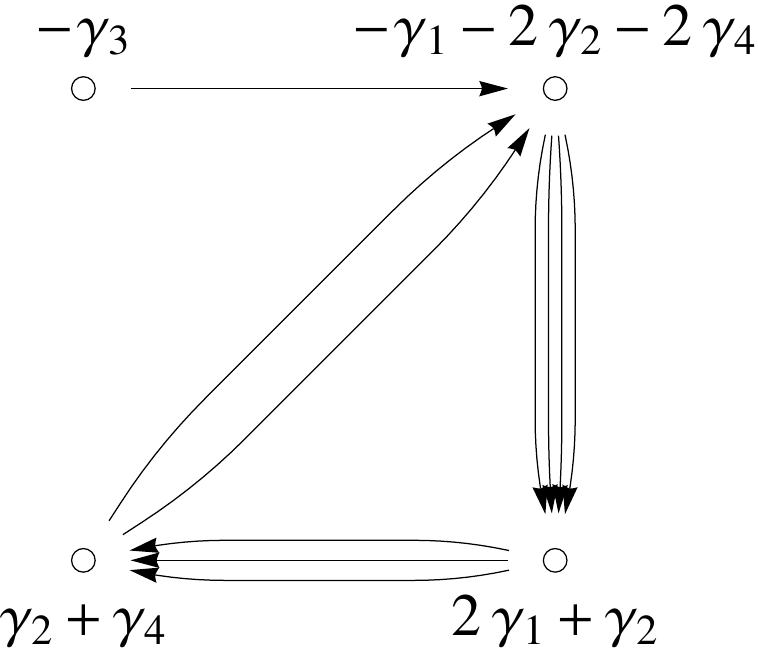} \hfill
        \caption{Left: a counterclockwise rotation past $Z_{\gamma_{2}+\gamma_{4}}$. Right: the corresponding BPS quiver.\label{fig:subquiver-tris}}
\end{figure}
The two lower nodes of the quiver we just obtained manifestly exhibit the $3$-Kronecker quiver involved in wild wall-crossing as a subquiver. In particular, it offers a convenient starting point for studying stable quiver representations on both sides of $MS(\gamma_{2}+\gamma_{4},2\gamma_{1}+\gamma_{2})$: states with charge $a(\gamma_{2}+\gamma_{4})+b(2\gamma_{1}+\gamma_{2})$ correspond to particularly simple dimension vectors, in which the two upper nodes decouple leaving the pure $3$-Kronecker quiver. We will not pursue the stability analysis in this paper, let us stress however that, since we have been working with stability parameters constrained by special geometry on the Coulomb branch (as opposed to working in $\mathbb{C}^{4}$), it should be possible to perform such analysis on both sides of the above-mentioned MS wall, thus recovering the related wild degeneracies.\\

The above construction generalizes easily to higher $m$. Consider indeed the situation in frame three of Figure \ref{fig:subquiver-bis}: here one could rotate the half-plane clockwise up until crossing the ray of $\gamma^{(j+1,j)}:=(j+1)\gamma_{1}+j\gamma_{2}$, resulting in a sequence of mutations leading to the quiver of Figure \ref{fig:subquiver-higher-m}. Then, without crossing walls of the second kind, one can move on ${\cal B}$ on a continuation of our path, as discussed in Section \ref{subsec:higher-m}, until getting past $MS((j+1)\gamma_{1}+j\gamma_{2},\gamma_{2}+\gamma_{4})$, the same quiver description still holds.\\
\begin{figure}[h!]
        \centering
                 \hfill \includegraphics[width=0.27\textwidth]{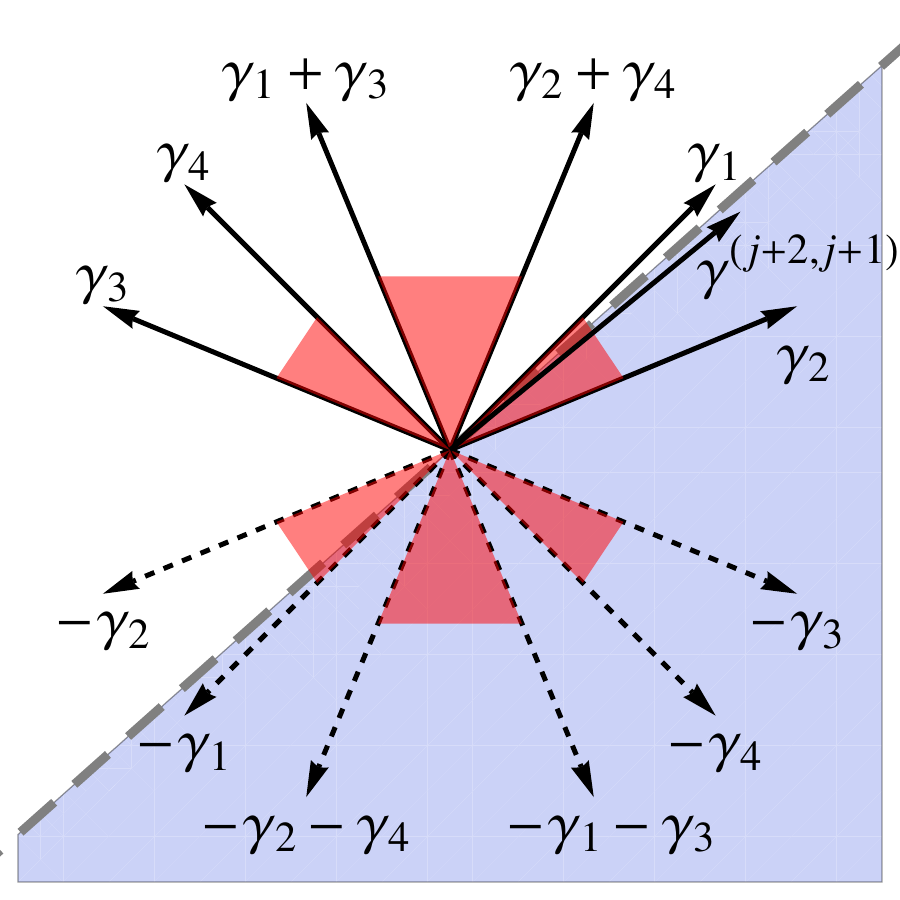} \hfill  \includegraphics[width=0.27\textwidth]{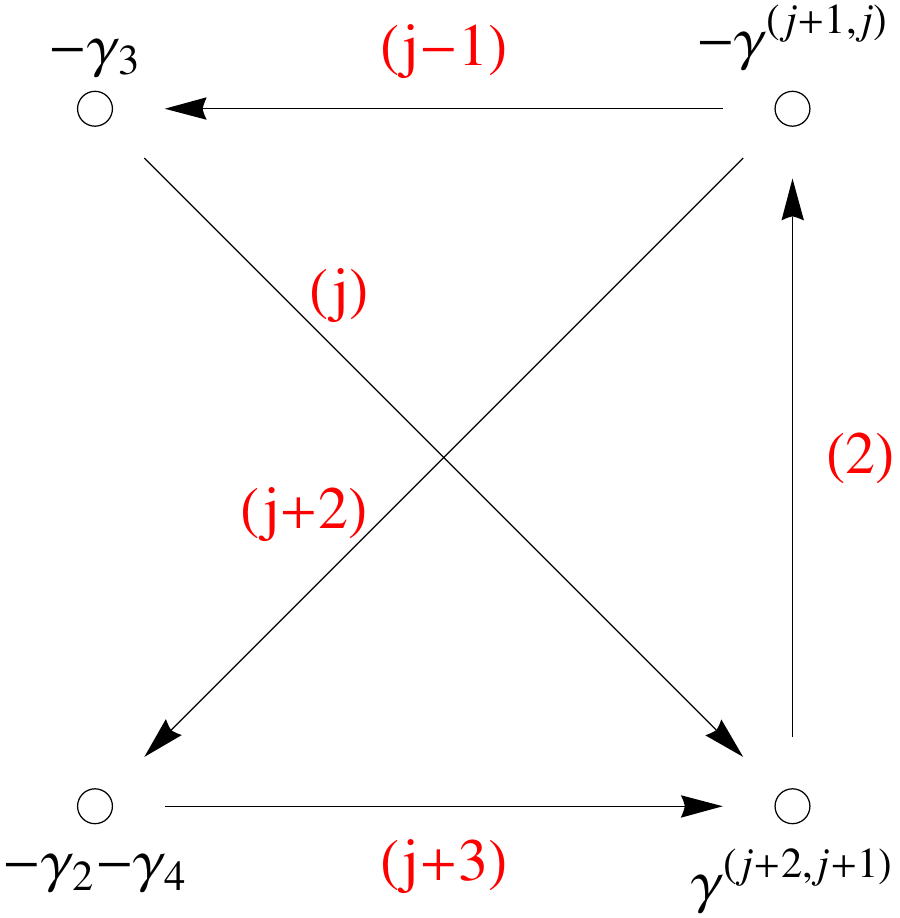}  \hfill \includegraphics[width=0.25\textwidth]{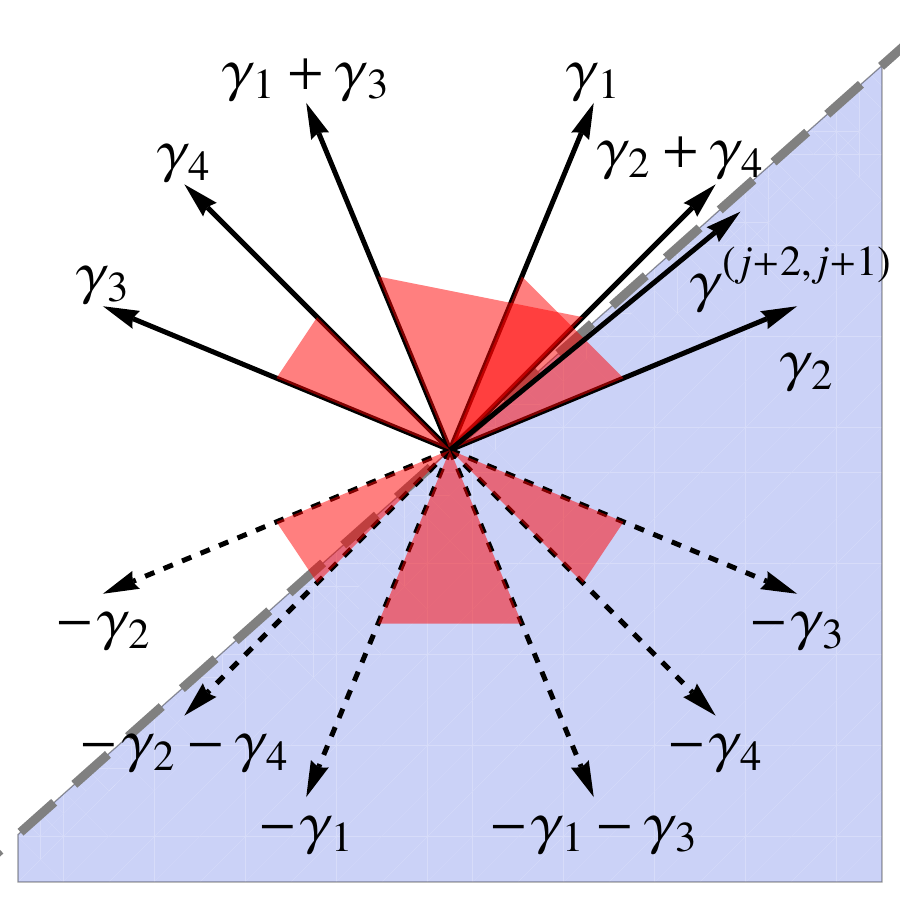}
        \caption{Left: a clockwise rotation of the half-plane past the ray $Z_{(j+1)\gamma_{1}+j\gamma_{2}}$. Center: the corresponding BPS quiver, arrow multiplicities are indicated in red. Right: after proceeding further on ${\cal B}$ we cross $MS((j+1) \gamma_{1}+j\gamma_{2},\gamma_{2}+\gamma_{4})$\label{fig:subquiver-higher-m}}
\end{figure}
At this point, a counterclockwise rotation of the half-plane, corresponding to an inverse mutation on $-\gamma_{2}-\gamma_{4}$ yields the quiver given in Figure \ref{fig:subquiver-higher-m-bis}. Again the two lower nodes exhibit the Kronecker subquiver of interest.\\
\begin{figure}[h!]
        \centering
                \hfill \includegraphics[width=0.29\textwidth]{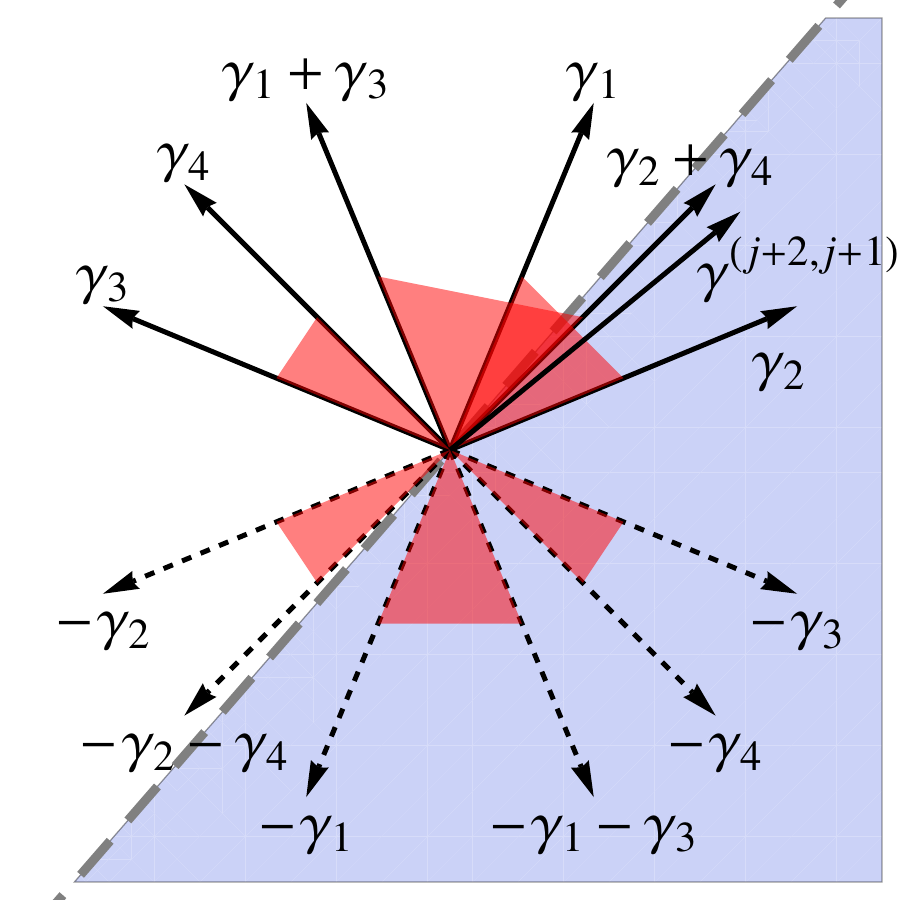} \hfill \includegraphics[width=0.31\textwidth]{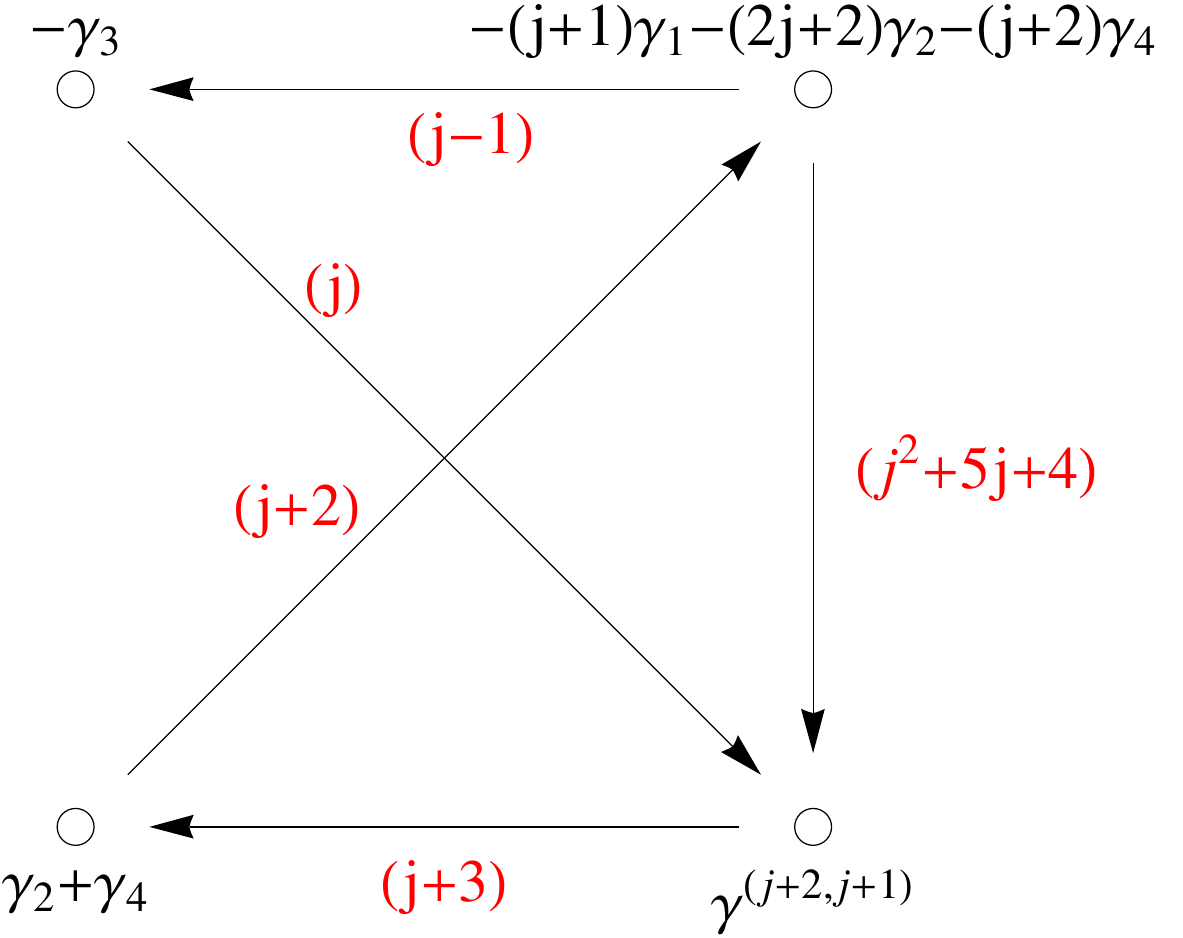} \hfill
        \caption{Left: a counter-clockwise rotation of the half-plane past the ray $Z_{\gamma_{2}+\gamma_{4}}$. Right: the corresponding BPS quiver, with arrow multiplicities indicated in red.\label{fig:subquiver-higher-m-bis}}
\end{figure}

\subsection{A nontrivial symmetry of BPS degeneracies} \label{sec:mocm}

One very nice application of the quiver approach is that it reveals
an intriguing symmetry of BPS degeneracies which would be very hard
to discover using spectral networks.

Our previous analysis of the ${\cal C}_{3}$ spectrum has focused on sequences of states $(n a+a_{0}) \gamma_{1}+(n b+b_{0}) \gamma_{2}$ with fixed slope  $a/b$ as $n\to \infty$. In this section we will instead consider sequences of states with the same BPS index.

In full generality, given two hypermultiplets with charges $\gamma,\gamma'$ such that $\langle\gamma,\gamma'\rangle=m>0$, we know already from the semi-primitive WCF that, across the wall $MS(\gamma,\gamma')$, a new hypermultiplet of charge $\gamma+m\gamma'$ will be a stable boundstate. The constituents $\gamma,\gamma'$, as well as their CPT conjugates will also be stable. Now, note that $\langle - \gamma',\gamma+m\gamma'\rangle = m$, moreover we have the following relation between stability parameters
\be
	{\rm sign}\,\left({\rm Im}\frac{Z_{\gamma}}{Z_{\gamma'}}\right) \, \equiv \, {\rm sign}\,\left({\rm Im}\frac{Z_{-\gamma'}}{Z_{\gamma+m\gamma'}}\right).
\ee
Thus, any boundstate of $\gamma,\gamma'$ can \emph{equivalently} be described as a boundstate of $-\gamma',\gamma+m\gamma'$. Such change of \emph{simple roots} for the $K(m)$ quiver simply corresponds\footnote{In the mathematics literature this correspondence is a known isomorphism among Kronecker moduli spaces, see for example \cite{WEIST12}, Remark 3.2.} to a change of duality frame by
\be
	g_{m}=\left(\begin{array}{cc}
	0 & 1 \\
	-1 & m
	\end{array}\right) \quad \in Sp(2,\mathbb{Z})
\ee
in a basis where $\gamma,\gamma'$ are represented by column vectors $(1,0),(0,1)$ respectively.
That is, there is a mutation of the quiver corresponding to the change of basis $g_m$.
Since this is detectable by the semiprimitive wall crossing formula there should
be a halo interpretation, to which we return in Section \S \ref{sec:OpenProblems}, Remark 4.

The above is essentially an observation of \cite{WEIST12} and it immediately implies
some remarkable identities for BPS indices. The group
\be
	{\cal R}=\langle h,h' | h^{2}=1, {h'}^{2}=1 \rangle \, = \, \mathbb{Z}_{2} \star \mathbb{Z}_{2}
\ee
  has an action on $\mathbb{Z}\gamma\oplus\mathbb{Z}\gamma'$ by
\be \label{eq:Z2-generators}
	h=\left(\begin{array}{cc}
	0 & 1 \\
	1 & 0
	\end{array}\right) ,\quad %
	h'=\left(\begin{array}{cc}
	-1 & m \\
	0 & 1
	\end{array}\right),\quad %
	g_{m}= h\, h',
\ee
then the BPS indices must have the   symmetry:
\be
	\Omega(g\cdot\gamma) = \Omega(\gamma),\quad\forall g\in{\cal R}.  \label{eq:g3orbits}
\ee
In other words, the spectrum can be organized into orbits of ${\cal R}$.

{\bf Remarks}
\begin{itemize}
\item The identity (\ref{eq:g3orbits}) extends to the protected spin character
\be
	\Omega(g \cdot\gamma;y) = \Omega(\gamma;y).
\ee
\item Consider for example $m=3$, we call the \emph{slope} of $(a,b)$ the ratio $a/b$. The eigenvalues of $g_{3}$ are
\be
	\xi_{\pm} = \frac{3\pm \sqrt 5}{2}
\ee
corresponding to the slopes limiting the cone of dense states of Fig.~\ref{fig:cone}. All $g_{3}$ orbits are confined to lie either inside or outside of the cone, and asymptote to the limiting rays.
\item The only orbits falling outside of the cone are those of ``pure'' hypermultiplets i.e. states with $\Omega=1$. All the other orbits are contained within the cone.
\item In the pure $SU(2)$ theory the limiting rays of the $g_{2}$ cone collapse into a single ray, which coincides with the slope of the gauge boson. In that context, the $g_{2}$ action has an interpretation in terms of a half-turn around the strong coupling chamber, combined with the residual $R$-symmetry, in a similar spirit to the approach of \cite{BILAL}. One can check that, in a suitable duality frame $g_{2}$ is a square root of the monodromy at infinity, up to an overall factor.

For the $m=1$ Kronecker quiver, the $g_{1}$ action simply recovers the whole spectrum.

\end{itemize}

\subsection{Asymptotics of BPS degeneracies}\label{sec:asymptotics}

For physical reasons we are often interested in the asymptotics of BPS degeneracies
for large charges.
There is no known simple closed formula for the degeneracies $\Omega(a\gamma_{1}+b\gamma_{2})$ of the $3$-Kronecker quiver. In this section we discuss some aspects of the large $a,b$ asymptotics.

The Poincar\'e polynomial for quivers without closed loops has been found explicitly in a closed form, as a sum over constrained partitions of corresponding quiver dimension vectors \cite{REINEKE}.
Unfortunately Reineke's formula does not lend itself well to an evaluation of the large charge asymptotics. On the other hand, use of localization techniques allows one to estimate asymptotic behavior of the Euler characteristic for moduli spaces of $m$-Kronecker quiver representations \cite{WEIST}.

Weist conjectured the following. Consider a state $N\gamma+M\gamma'$ with $\langle \gamma,\gamma'\rangle=m$,
 in a wild region of the Coulomb branch. The  corresponding BPS index equals the Euler characteristic of the moduli space of the quiver with $m$ arrows between two nodes with spaces $\mathbb{C}^N$ and $\mathbb{C}^M$ in
  a wild region of stability parameters. Now consider a sequence of dimension vectors $N=a n+a_0$, $M=b n+b_0$, with $a,b,a_{0},b_{0}$ fixed. Weist conjectured that the asymptotic behavior of the Euler characteristic has the form
\be
\begin{split}
 & \log \vert \Omega(N \gamma + M \gamma')\vert \mathop{\sim}_{n\to\infty} n \, C_{a,b}(m)  \\
 & C_{a,b}(m) = \frac{\sqrt{m a b-a^2-b^2}}{\sqrt{m-2}}\left[(m-1)^2\log(m-1)^2-(m^2-2m)\log\left(m^2-2m\right)\right]\label{eq:Weist}
\end{split}.
\ee
Note that $C_{1,1}(m)= c_{m}$ of equation (\ref{eq:c_m}).

\subsection{Numerical check of Weist's conjecture}
In section \ref{sec:DegsFSWCF} we obtained BPS degeneracies by using an algorithmic approach, based on the KSWCF (\ref{eq:cohWCF}).
The results are in agreement with \cite{GROSS}: in particular we found a sequence of degeneracies of slope 1 behaving as predicted by Reineke in \cite{REINEKE-09}, as well as a highly populated -- suggesting dense -- cone of ``wild'' BPS states in the complex $Z_{\gamma}$-plane. The region outside such cone is populated by hypermultiplets only, falling in sequences approaching the boundaries of the cone, as shown in figure \ref{fig:cone}.

\begin{figure}[htbp]
\begin{center}
\includegraphics[width=0.37\textwidth]{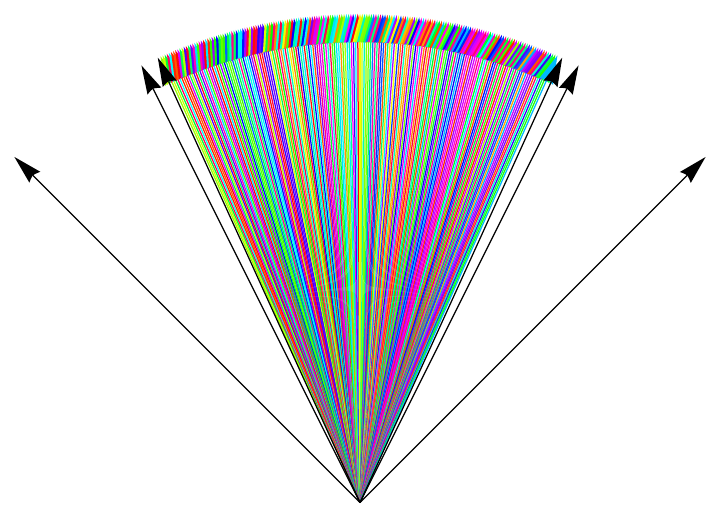}
\caption{Schematic picture of BPS states charges for 3-Kronecker quiver. A dense cone is bounded by rays of slopes $a_{}/b_{}=(3\pm \sqrt{5})/2$. Only hypermultiplets fall out of the cone.\label{fig:cone}}
\end{center}
\end{figure}

Let $\gamma_{a,b}(n)=(n a+a_0)\gamma_1+(n b+b_0)\gamma_2$. Denoting by
\be
\eta_{a,b}:= \log\vert \Omega(\gamma_{a,b}(n)) \vert
\sqrt{\frac{m-2}{{m a_{}b_{}-a_{}^2-b_{}^2}}} ,  \label{eq:eta-coeff}
\ee
Weist's conjecture says that $\eta_{a b}/ n \rightarrow c_m$, $\forall \gamma \in$ dense cone as $n$ grows ($c_{m}$ is defined by formula (\ref{eq:c_m})). In Fig.~ \ref{fig:logOmega} we already noticed this kind of behavior, to some extent. In order to establish a more precise match between our data and Weist's conjecture, it is convenient to plot the behavior of $\eta_{a,b}/n$ versus the $|\gamma|$ filtration level, as in Fig. \ref{fig:weist}.

\begin{figure}[htbp]
\begin{center}
\includegraphics[width=0.53\textwidth]{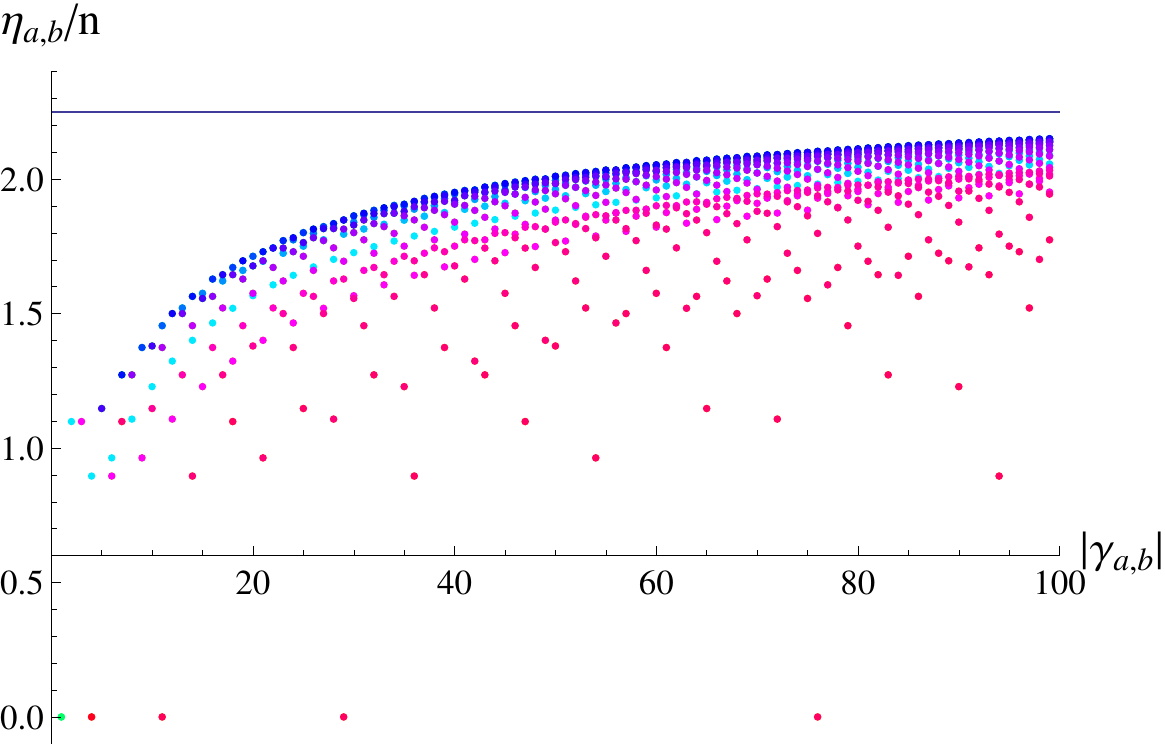}
\caption{The data shown is for $m=3$ with $(a_0,b_0)=(0.0)$. The straight horizontal line represents the Weist coefficient $c_{3}=4\log 4-3\log 3$. For generic values of $(a,b)$ the BPS degeneracies indeed approach the Weist asymptotics at large $|\gamma|$.\label{fig:weist}}
\end{center}
\end{figure}

Different colors depict different slopes from the red for $a_{}\gg b_{}$ or $a_{}\ll b_{}$ to the blue for $a_{}\sim b_{}$. As the graph
shows, the speed of convergence actually depends on the slope,
so the degeneracies for BPS states nearest to the cone boundaries approach Weist's asymptotics in the worst way.
Note that there are some charges that do not obey the general asymptotic behavior. These give the horizontal data points at
the bottom of Fig.~\ref{fig:weist}. These charges indeed lie outside the dense cone.
\footnote{Note that, because of $g_m$ symmetry, the figure would look rather different
if we plotted the degeneracies as a function of $n$ using $\gamma_{a,b}(n)$.}

\section{Physical estimates and expectations}\label{sec:physical-estimates}

\subsection{An apparent paradox}\label{subsec:ApparentParadox}

In this section we first present a physical argument which seems to lead
to a very general bound on the behavior of the BPS index in any supersymmetric field theory.
The purported bound, however, is explicitly violated by the ``wild'' degeneracies we have just
found in the pure $SU(3)$ theory.  Thus, na\"{i}vely, there is a paradox.  We
first explain the paradox in more detail, and then explain how
this paradox is resolved.\footnote{We thank T. Banks and S. Shenker for crucial remarks on this matter.}

At very large energy our effective theory should approach a UV conformal fixed point.
So consider a $d$-dimensional CFT put in a box of volume $V$ and heated up to temperature $T$.
Since we have only two dimensionful parameters and we assume the energy and the
entropy of the system to be extensive quantities, simple dimensional analysis is
enough to predict their form up to dimensionless constants (which will depend on the theory):
\be
\begin{split}
E(T,V)=\alpha V T^{d},\\
S(T,V)=\beta V T^{d-1}.
\end{split}
\ee
Eliminating the temperature dependence we derive the scaling of the entropy with the energy:
\be\label{eq:entropy}
S(E,V)=\kappa V^{1/d} E^{(d-1)/d}.
\ee
This provides an estimate for the behavior of the number of microstates of energy $E$ supported in a volume $V$,
and gives the correct asymptotic dependence for $E \to \infty$.

In order to excite massive states we can increase the temperature, thus taking into account heavier BPS states.
The BPS index, being a signed sum over the states in the theory, cannot exceed the overall number of states.
\footnote{  In fact, the  data for the
 Kronecker $m$-quiver suggest that in this case all the summands contributing to
 the BPS index have the same sign, so the BPS index actually counts the   number of states up to an overall sign.}
 Thus we come to the following chain of inequalities (here we take $d=4$ and set $E=\vert Z_\gamma \vert$):
\be
\vert \Omega(\gamma) \vert  = \left| \mathop{\rm Tr}\nolimits_{{\mathfrak{h}}^{\gamma}}(-1)^{2J_3}\right|\leq
\frac{1}{4}   \mathop{\rm Tr}\nolimits_{{\cal H}_{BPS,\gamma}} 1
  \leq \frac{1}{4} \mathop{\rm Tr}\nolimits_{{\cal H},E} 1
  =\frac{1}{4}  e^{S(E)}\sim e^{\kappa V^{\frac{1}{4}} E^{\frac{3}{4}}} \label{eq:FT-bound},
\ee
where the last estimate assumes large $E$. Thus the observed behavior
$\log|\Omega(\gamma)|\sim E$ for large $\gamma$ in the pure $SU(3)$ theory  seems to give a contradiction.

The resolution of this paradox comes from taking into account the fact that our bound applies only to the theory in a finite volume. If the size of BPS states becomes large enough and they do not fit into the box of finite volume, then they do not contribute to the na\"{i}ve counting of degrees of freedom. So we should instead consider a ``truncated BPS index'' $\check \Omega_V$, counting only the states which fit into a box of size $V$; we should expect this index to satisfy the inequality
\be
\vert \check\Omega_V \vert =\vert \mathop{\rm Tr}\nolimits_{{\cal H}_{BPS},M=|E|,R\leq V^{\frac{1}{3}}}(-1)^{2J_3} \vert \lesssim  e^{\kappa V^{\frac{1}{4}} E^{\frac{3}{4}}}
\ee
with $R$ the average size of a BPS state.

The rest of this section is devoted to arguing that the above scenario
is indeed correct. We will use the semiclassical picture of BPS states
given by the Denef equations, reviewed in Section \ref{subsec:DenefEquations},
to give a lower bound for the average size of the semiclassical BPS states.
 The resolution of the paradox is spelled out
in some more detail in Section \ref{subsec:ResolveParadox}.
We give some supporting evidence for the validity of the use of the Denef
equations for describing the exponentially large number of BPS states
in Section \ref{subsec:QC-QQ}.

\subsection{Denef equations}\label{subsec:DenefEquations}

In order to estimate the size of the BPS states arising in the theory, we refer
to the interpretation \cite{Denef-Moore} of those BPS states that arise from
wall-crossing as multi-centered   solutions similar to those arising in
 ${\cal N}=2$ supergravity \cite{DENEF}. We assume Denef's multicentered
picture has a good $\alpha'\to0$ limit and can be applied to field theory.
Suppose we have a set of elementary BPS states with charges
$\{\gamma_A\}_{A=1}^{n}$ placed at corresponding points
$\{r_A\}_{A=1}^{n}$ of $\mathbb{R}^3$.  This configuration is again BPS
only if the following set of equations is satisfied:
\be
\sum\lm_{\substack{B=1\\ B\neq A}}^n \frac{\langle \gamma_A,\gamma_B\rangle}{|r_A-r_B|}
= 2{\rm Im}(e^{-i\vartheta} Z_{\gamma_A}), \label{eq:denef-eqs}
\ee
where $\vartheta=\mathop{\rm arg} \sum\lm_{A=1}^n Z_{\gamma_A}$.

Now let us consider a BPS state of total charge $M\gamma_1 + N \gamma_2$,
with $\langle \gamma_1, \gamma_2 \rangle =m$.   Let   us,
\emph{for the moment}, suppose that the dominant contribution to the
entropy comes from a boundstate of $M$ elementary constituents of
charge $\gamma_1$ and $N$ elementary constituents of charge $\gamma_2$.

In the case where the charges are of the form
\be\label{eq:MaxDecomp}
\{  \underbrace{\gamma_1, \dots, \gamma_1}_{M}, \underbrace{\gamma_2, \dots, \gamma_2}_{N}\}
\ee
the equations simplify to
\be\label{eq:Denef}
\begin{split}
\sum_{a=1}^N\frac{m}{r_{ia}} & = \kappa_{1}:=2{\rm Im}(e^{-i \vartheta} Z_{\gamma_1}) , \qquad 1\leq i \leq M  \\
\sum_{i=1}^M\frac{-m}{r_{ia}} & = \kappa_{2}:=2{\rm Im}(e^{-i \vartheta} Z_{\gamma_2}) , \qquad 1\leq a \leq N  \\
\end{split}
\ee
We can view the index  $a$ as running  over ``electrons'' and $i$ over ``magnetic monopoles,'' in
an appropriate duality frame.

Now we are interested in the size of the boundstate. Therefore we consider the sum over
the first equation in \eqref{eq:Denef}. (Doing the analogous sum over the second equation
produces an equivalent result.) The result is that
\be\label{eq:SumDenef}
\sum_{i,a} \frac{1}{r_{ia} } =  \frac{NM}{\vert M Z_{\gamma_1} + N Z_{\gamma_2} \vert}
\left( \frac{2 {\rm Im}(\bar Z_{\gamma_2} Z_{\gamma_1} )}{m}\right)
\ee
We can rewrite this equation nicely in terms of the \emph{harmonic average}
of the distances $r_{ia}$:
\be\label{eq:HarmonicIdentity}
\langle r_{ia} \rangle_h= \left( \frac{m}{2 {\rm Im}(\bar Z_{\gamma_2} Z_{\gamma_1} )}\right)
\vert M Z_{\gamma_1} + N Z_{\gamma_2} \vert.
\ee
On the other hand, we can use the well-known inequality
that the harmonic average is a lower bound for the
ordinary average,  $\langle r_{ia} \rangle_h\leq \langle r_{ia} \rangle$,
to conclude that
\be\label{eq:LowerBound}
\left( \frac{m}{2 {\rm Im}(\bar Z_{\gamma_2} Z_{\gamma_1} )}\right)
\vert M Z_{\gamma_1} + N Z_{\gamma_2} \vert \leq \langle r_{ia} \rangle.
\ee
Equation \eqref{eq:LowerBound} is a key result. It shows that if $N$ or $M$ goes
to infinity then the size of the average BPS molecule grows linearly with $N$ or $M$, respectively.

We have shown that boundstates of total charge $M\gamma_1 + N \gamma_2$
with constituents \eqref{eq:MaxDecomp} become large when $N,M \to \infty$.
However, other partitions of $N,M$ can and do lead to BPS boundstates.
In general, given a pair of partitions
\be\label{eq:OtherPartitions}
M=\sum\lm_{j=1}^{M} l_j j,\quad N=\sum\lm_{k=1}^N s_k k
\ee
there can be other boundstates where there are $l_j$ centers
of charge $j \gamma_1$ and $s_k$ centers of charge $k \gamma_2$.
In order to deal with these cases, let us introduce, for any set of
charges $\{ \gamma_A \}$, the moduli space ${\cal M}(\{ \gamma_A \} ) $   of
solutions   to the Denef equations (\ref{eq:denef-eqs}). If there
are $n$ centers it is a subspace of $\IR^{3n}$. Note that
the moduli space for charges
\be\label{eq:gencd}
\{\underbrace{\gamma_1, \dots, \gamma_1}_{l_1},
\underbrace{2\gamma_1, \dots, 2\gamma_1}_{l_2}, \dots,
\underbrace{ \gamma_2, \dots, \gamma_2}_{s_1},
\underbrace{2\gamma_2, \dots , 2\gamma_2}_{s_2}, \dots   \}
\ee
is in fact a subspace of the moduli space for \eqref{eq:MaxDecomp}, where certain collections
of centers $r_i$ and $r_a$ have (separately)  collided. Nevertheless, the
identity \eqref{eq:HarmonicIdentity} applies uniformly throughout the moduli space
and hence applies to all possible partitions. As an extreme example,
the moduli space $\CM(\{ M \gamma_1, N \gamma_2 \})\cong \IR^3 \times S^2$, where the
$\IR^3$ is the center of mass and the $S^2$ has a radius
\be
R_{12} = \left( \frac{m}{2 {\rm Im}(\bar Z_{\gamma_2} Z_{\gamma_1} )}\right)
\vert M Z_{\gamma_1} + N Z_{\gamma_2} \vert.
\ee
In any case, we can conclude that for any partition of charges
such as \eqref{eq:gencd} the average BPS state has a size bounded
below by a linear expression in $N$ and $M$. We call these large
semiclassical BPS states \emph{BPS giants}.

\subsection{Resolution and Revised Bound}\label{subsec:ResolveParadox}

The giant BPS states resolve the paradox explained in Section \ref{subsec:ApparentParadox} above.
Indeed  we can adapt the bound (\ref{eq:FT-bound}) by adjusting the volume of the box $V$ in such a way that states of mass $E$ fit in a volume
$V_E:= R_E^3:=E^3$. From our estimate of the sizes of BPS molecules
we know that the average size indeed scales linearly with $E$. Therefore
the new bound is
\be
\log \vert\Omega(E)\vert \sim \alpha E \lesssim  \kappa E^{{3}/{4} } V_E^{{1}/{4}} \sim  \kappa' E^{{3}/{2}} \label{eq:paradox-resolved}
\ee
and is indeed satisfied.

In equation \eqref{eq:paradox-resolved} $\kappa'$ is a dimensionful constant, it scales as
$\kappa'\sim({\rm length})^{\frac{3}{2}}$. Let us give a
physical interpretation for this scale. If we consider a
sequence of charges $N(a\gamma_1+b \gamma_2)$, with
$N\to \infty$ and $\gamma_{p}:=a\gamma_1+b \gamma_2$ primitive,
 then the size of an average BPS molecule
behaves as $R\sim r_0 N$, where $r_0$ is the  size of a state
with   charge  $\gamma_{p}$.
The energy behaves as $E = |Z_{0}| N$, where $Z_{0}$ is a central charge of
the state with charge $\gamma_{p}$. Thus we can
give a formula accounting for the scaling dimension of $\kappa'$ in (\ref{eq:paradox-resolved}) by using
\be
 V_{E} =  R_E^{3}\sim (r_{0} N)^{3} \sim (r_{0} E / |Z_{0}|)^{3}
\ee
to deduce
\be
\begin{split}
&  E^{3/4}V_E^{1/4}\sim   \left(\frac{r_0}{|Z_0|}\right)^{3/4} E^{3/2}, \\
& \Rightarrow \quad \kappa'\sim   \left(\frac{r_0}{|Z_0|}\right)^{3/4}.
\end{split}
\ee

We remark that

\begin{enumerate}

\item The length scale $ \left(r_0/|Z_0|\right)^{1/2}$ is a function
of the moduli, since both $r_0$ and $Z_0$ are functions of the moduli.

\item The coefficient $\alpha$ in (\ref{eq:paradox-resolved}) is
\be
	\alpha=\frac{C_{a,b}(m)}{|Z(a\gamma_{1}+b\gamma_{2})|}
\ee
for the series of charges above eq. (\ref{eq:eta-coeff}). As we noted in Section \ref{subsec:higher-m} there are points on the Coulomb branch with arbitrarily high $m$ and
\be
	C_{a,b}(m)\sim\sqrt{ab}\log m^{2}
\ee
for large $m$. Hence, somewhat surprisingly, the coefficient of the logarithmic growth is unbounded on the Coulomb branch.

\end{enumerate}

\subsection{Discussion of validity of the semiclassical picture }\label{subsec:QC-QQ}

 In this section we will address
the question of how reliable the semiclassical approximation is.
We will review some supporting evidence for the reliability of the
semiclassical pictures
based on the relation of an  exact result for BPS degeneracies
 $\Omega$ to certain symplectic volumes.

As a side remark we note that numerical data for the $3$-Kronecker quiver strongly suggest (cf. the discussion about positivity below (\ref{eq:spin-splitting}))
 that the BPS index actually measures the number of states
\be
|\Omega(\gamma)|=\mathop{\rm Tr}\nolimits_{\mathfrak{h}_{\gamma}} 1. \label{eq:positivity-consequence}
\ee
This relation is not essential to our argument but it does nicely
simplify the considerations.

Let us recall the symplectic structure on Denef moduli space $\CM(\{\gamma_A \})$.
Overall translation acts on this space and the reduced
space $\overline{{\cal M}}(\{\gamma_A\})= {\cal M}(\{\gamma_A\})/\mathbb{R}^{3}$
  is generically $2n-2$ dimensional. Moreover, the
reduced space admits a symplectic form \cite{DEBOER}:
\be
\omega=\frac{1}{4}\sum\lm_{i<j}\langle\gamma_i,\gamma_j\rangle
\frac{\epsilon_{abc} dr_{ij}^a\wedge dr_{ij}^b r_{ij}^c}{r_{ij}^3}.
\ee
In the semiclassical approximation we identify a subspace of the space of
BPS states with a set of BPS field configurations. We expect that
the dimension of a
subspace corresponding to a charge decomposition  can be estimated, in the semiclassical approximation,
by the symplectic volume
\be
{\rm Vol}(\{\gamma_A\}):=\frac{1}{(n-1)!}\int\lm_{\overline{ \cal M} }\left(\frac{\omega}{2\pi}\right)^{n-1} . \label{eq:sym-vol}
\ee
where $n$ is the number of centers.

Now, thanks to a result of Manschot, Pioline, and Sen \cite{MANSCHOT, Manschot:2011xc},
 in the example of the $m$-Kronecker quiver the protected spin character in the wild chamber
can in fact be expressed exactly as a sum over two partitions
\eqref{eq:OtherPartitions}
so that
\be
\begin{split}
\label{eq:MPS-WCFy}
& \Omega(M\gamma_1 + N \gamma_2;y) = \\
& \qquad\qquad =\sum\lm_{\{l_j\},\{s_k \}} g_{\rm ref}(\{l_j\},\{s_k\};y)\prod\lm_{j,k}\frac{1}{l_j!j^{l_j}s_k!k^{s_k}}\left(\frac{y-y^{-1}}{y^j-y^{-j}}\right)^{l_j}\left(\frac{y-y^{-1}}{y^k-y^{-k}}\right)^{s_k}
\end{split}
\ee
where $g_{\rm ref}$ refers to an equivariant Dirac index on the
space of solutions to   Denef's equations with distinguishable centers described by
charge partitions $\{l_j\}$, $\{s_k \}$. If we specialize to the index at $y=1$
\footnote{In the conventions of \cite{MANSCHOT} we take $y\to 1$ rather than $y\to -1$ to get the index.}
then  $g_{\rm ref}$ has a very nice interpretation as the symplectic volume (\ref{eq:sym-vol})
of the moduli space of solutions to Denef's equations (up to a sign):
\begin{eqnarray}\label{eq:MPS-WCF}
\Omega(M\gamma_1 + N \gamma_2)=\sum\lm_{\{l_j\},\{s_k \}}(-1)^{m M N+1-\sum\lm_j l_j-\sum\lm_k s_k}
{\rm Vol}(\{l_j\},\{s_k\})\prod\lm_{j,k}\frac{1}{l_j!j^{2 l_j}s_k!k^{2 s_k}}.
\end{eqnarray}
where   ${\rm Vol}(\{l_j\},\{s_k\})$  is \eqref{eq:sym-vol} for
the charges \eqref{eq:gencd}.

We will take this relation of the exact number of BPS states to symplectic volumes
as sufficient evidence for the validity of our resolution. There are, however,
some further interesting aspects of this formula which we will comment on
in the following Sections \ref{sec:halo} and \ref{sec:maximalpart} below.

\subsubsection{A toy example: the Hall halo} \label{sec:halo}

A very nice exactly solvable example of BPS configurations is provided by the   Hall halo of \cite{DENEF}.
 Consider a configuration of $N$ electric particles and a single
 magnetic monopole of charge $m$. This corresponds to the case $(M,N)=(1,N)$ in
 the notation above.  In this case the equations
 \eqref{eq:Denef} imply that the $N$ electric particles all
 lie on a single sphere centered on the magnetic particle and of
 radius:
\be\label{eq:Hall}
R_{12} = \left( \frac{m}{2 {\rm Im}(\bar Z_{\gamma_2} Z_{\gamma_1} )}\right)
\vert   Z_{\gamma_1} + N Z_{\gamma_2} \vert.
\ee
Now, in this case Denef argued that to get the spin character
 we  can just apply the usual quantum mechanics of
Landau levels on a sphere with a magnetic monopole inside. Counting
the  ground states gives the corresponding protected
spin character   \cite{SOLIDSTATE}
\be\label{eq:LL-DEG}
\Omega(y)=(-y)^{-(m-N)N}\frac{\prod\lm_{j=1}^m(1-y^{2j})}{\prod\lm_{j=1}^N(1-y^{2j})\prod\lm_{j=1}^{m-N}(1-y^{2j})},
\ee
in perfect agreement with Reineke's general formula (see eq. (5.3) of \cite{DENEF}).

There are two interesting lessons we can draw from \eqref{eq:LL-DEG}:

\begin{enumerate}

\item First, naive physical intuition suggests that the large size of BPS states
is due to large angular momentum. This example shows that in fact this is not
necessarily the case.   In this case the size of the configuration is given by
formula (\ref{eq:Hall}).  Nevertheless this configuration contains representations of
many different spins.

\item
Second, we can derive the number of states in a multiplet by taking $y\rightarrow -1$.
Then $\Omega=\frac{m!}{N!(m-N)!}$.
In the limit $N\ll m$ the number of allowed states is much
greater than the number of populated states, so quantum statistics
does not play an important role, and the semiclassical approximation should work. Indeed,
\be\label{eq:LeadingTerm}
\Omega=\frac{m!}{N!(m-N)!}\mathop{\sim}_{N\ll m} \frac{ m^N}{N!}+ \cdots
\ee
This confirms the semiclassical expectation that the number of states should be
counted by the symplectic volume since the volume is proportional to $m^{N}$.
Note however that, for fixed $N$ the binomial coefficient is really a polynomial in $m$
and \eqref{eq:LeadingTerm} is only the leading term at large $m$. Since $1/m$ plays
 the role of $\hbar$ we   can interpret the subleading
terms as quantum corrections to the naive semiclassical reasoning.

\end{enumerate}

\subsubsection{Estimating the contribution of the maximal partition} \label{sec:maximalpart}

Let us consider the contribution to the BPS degeneracy of the maximal
partition \eqref{eq:MaxDecomp} in the formula \eqref{eq:MPS-WCF}. The symplectic volume for this partition is
\be
	\text{Vol}((N,M),\kappa_{1},\kappa_{2}, m):=\frac{1}{(N+M-1)!}\int _{\overline\CM} \left(\frac{\omega}{2\pi}\right)^{{N+M-1}}
\ee
where we used the fact that there are $n = N+M$ centers. We would like to estimate this volume when $N,M$ become large.

Rescaling both $\kappa_{1,2}$ in \eqref{eq:Denef} by $\lambda\in \IR$
together with $r_{ij}\mapsto r_{ij}\lambda^{-1}$
shows that solutions for rescaled values of  $\kappa_{1,2}$
are obtained by simply rescaling the distances. Therefore the ratio
$$H((N,M),\kappa_{1}/\kappa_{2}):=\text{Vol}((N,M),\kappa_{1},\kappa_{2}, m)/m^{N+M-1}$$
only depends on the ratio $\kappa_{1}/\kappa_{2}$ and on $N,M$.
For simplicity, let us specialize to   $M=N-1$.  in the limit $N\to \infty$ we have
\be
\begin{split}
	& \lim_{N\to\infty}\frac{1}{N}\log\left( \text{Vol}((N-1,N),\kappa_{1},\kappa_{2}, m) \right)\nn \\
	& \sim   \log m^2 + F\left(\frac{\kappa_{1}}{\kappa_{2}}\right).
\end{split}
\ee
Note that the second piece is independent of $m$.

There are two important lessons we can draw from this computation:

\begin{enumerate}

\item
This behavior nicely coincides with the Weist coefficient, but only in
the large $m$ limit when:
\be
C_{1,1}(m)\sim\log m^{2} + {\cal O}(m^{-1})
\ee
The fact that we must take $m\to \infty$
 is not terribly surprising in view of the Hall halo example
discussed above.

\item  It is interesting to note that at finite values of $m$ the maximal partition
does \emph{not} fully account for the exponential growth coefficient,
even in the large charge regime.  Indeed, as pointed out in \cite{WEIST12} we should take into account many other  partitions to derive even the leading asymptotic behavior of the BPS index. One important (and subtle) aspect of \eqref{eq:MPS-WCF}
is that the different symplectic volumes are weighted with \emph{signs}.
This might imply some subtlety in applying the semiclassical pictures we have used,
and should be understood better. In the meantime, as we discuss further in Remark 5
of Section \ref{sec:OpenProblems}: in the formula \eqref{eq:MPS-WCF}, considering the case where
the BPS ray lies in the dense cone, there can be
striking cancelations between volumes of different partitions.

\end{enumerate}

\section{Spectral Moonshine}\label{sec:moonshine}

In the course of these investigations we noticed some
unusual and very intriguing features in our data.
We mention these here, leaving a deeper analysis and
conceptual understanding of these features to future work.

\subsection{Scaling behavior of the spin degeneracies}

An interesting pattern of the spectrum
emerges when we consider the distribution
of spin multiplets within $\mathcal{H}_{\gamma}^{\rm BPS}$, the subspace of BPS states with gauge charge $\gamma$.
For the definitions of the protected spin character and the spin decompositions see
Appendix \ref{app:PSC-tables}.

Let $\delta_{\gamma}(j)$ be the number of times a spin-$j$
multiplet\footnote{Meaning a representation  $\rho_{hh}\otimes(j,0)$ of $\mathfrak{so}(3)\oplus\mathfrak{su}(2)_{R}$.}
 occurs within $\mathcal{H}_{\gamma}^{\rm BPS}$, as in (\ref{eq:spin-decomposition}). Numerical
 data suggests that
 \emph{all states within the dense cone} exhibit a common $\delta$-distribution, as
 shown in Fig.~\ref{fig:mb} (the data are in Appendix \ref{app:PSC-tables}).\\
\begin{figure}[htbp]
\begin{center}
\includegraphics[width=0.44\textwidth]{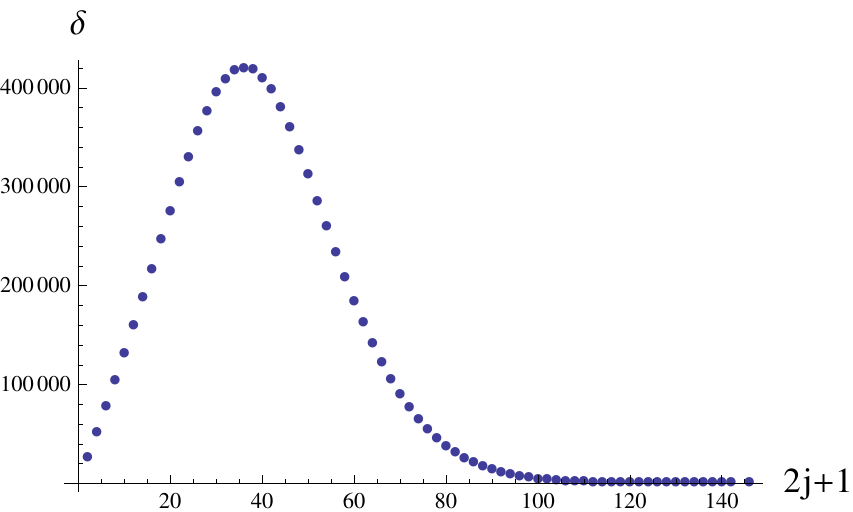} \includegraphics[width=0.44\textwidth]{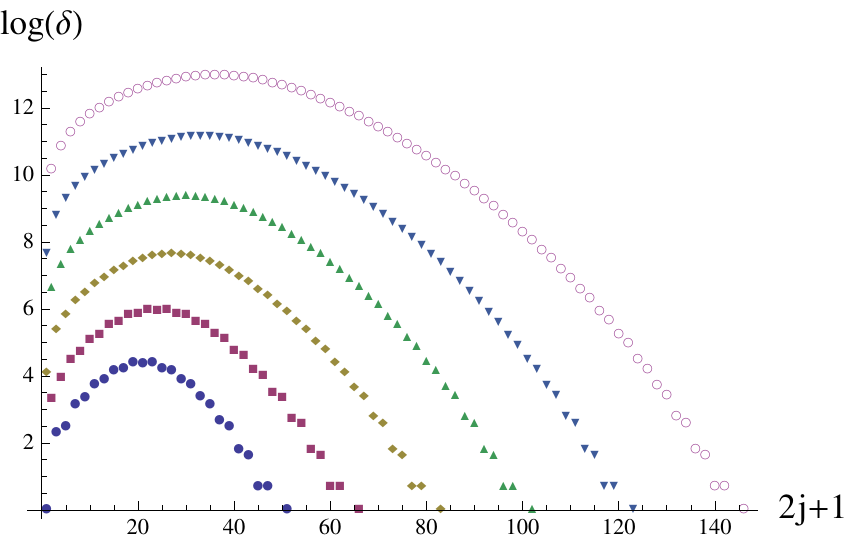}
\caption{On the left: the distribution $\delta(j)$ for $\gamma=12(\gamma_{1}+\gamma_{2})$. On the right: the same distribution for several states $\gamma=n(\gamma_{1}+\gamma_{2})$. This feature extends to other slopes as well, indeed all states within the dense cone exhibit such ``Poisson'' behavior.\label{fig:mb}}
\end{center}
\end{figure}\\

More precisely, letting $\gamma_{n}$ denote the sequence of charges $(na+a_{0})\gamma_{1}+(nb+b_{0})\gamma_{2}$,
data collected by computer experiments strongly suggest that there are
functions $\kappa_1, \kappa_2,\kappa_3,\rho,\alpha$%
\footnote{The $\kappa_{1},\kappa_{2}$ employed here have nothing to do with those of section \ref{sec:maximalpart}.} of $a,b,a_0,b_0$ such that,
if we define $j_n(s)$ by
\be
s = (2j_n(s)+1)/(\rho \vert \gamma_n \vert),
\ee
then the limit
\be\label{eq:scalingfunction}
u(s):=\lim_{n\to\infty} \, \kappa_3 |\gamma_{n}|^{-\kappa_{1}} e^{-\kappa_{2} \, |\gamma_n|} \delta_{\gamma_n}\left( j_n(s) \right).
\ee
exists and is given by
\be\label{eq:Poisson}
u(s) = s^{\alpha} e^{-\alpha (s-1)}
\ee
(Recall that $\vert \gamma_n \vert = n(a+b) + a_0 + b_0 $).
The numerical evidence further suggests that
for $m=3$, $\alpha\approx 2$, regardless of the slope\footnote{This estimate is based on data collected for $|\gamma|<30$.}.

If we assume that the above scaling law holds and the limiting behavior
to the scaling function is sufficiently rapid, then one can relate
the parameters $\kappa_1, \kappa_2$ of the scaling law to the
leading terms in the $n\to \infty$ asymptotic expansion
\be
\log \vert \Omega(\gamma_{n})\vert   \sim   \kappa_2 \vert \gamma_n \vert  +(\kappa_{1}+2) \log  \vert \gamma_n \vert + {\cal O}(1)
\ee
where ${\cal O}(1)$ has a finite limit as $n\to \infty$. Indeed, comparing with
the Weist asymptotics (\ref{eq:Weist}) we learn that $(a+b) \kappa_2 = C_{a,b}(m)$.
Similarly, comparing with known asymptotics of $\Omega(\gamma_n)$ we can learn
about the $a,b,a_0,b_0$ dependence of $\kappa_1$.

%

Regardless of the validity of the scaling law,  it is
worthwhile stressing that the sub-leading behavior of $\log \vert \Omega(\gamma_{n})\vert$
 is interesting in its
own right. It is often assumed that, at large $n$, the BPS index is a
continuous function of the slope $a/b$, however -- somewhat surprisingly --
the sub-leading correction exhibits a dependence on $a_{0},b_{0}$ too. To see
this, consider two different sequences approaching slope $1$, namely $\gamma^{(n)}=(n,n)$
and $\tilde\gamma^{(n)}=(n,n+1)$, we have
\be \label{eq:subleading-discontinuity}
	\begin{split}
	& \log  \vert \Omega(\gamma_{n})\vert = n C_{1,1}(m) -\frac{5}{2} \log n + {\cal O}(1)\\
	& \log  \vert \Omega(\gamma_{n})\vert = n C_{1,1}(m) -2 \log n + {\cal O}(1) ,
	\end{split}
\ee
where we used the known result\footnote{Cf theorem 6.6 of \cite{WEIST}}
\be
	\Omega(n,n-1) = \frac{1}{(3n+2)(2n+1)}{{4n+2}\choose{n+1}}.
\ee
The subleading dependence on $a/b,a_{0},b_{0}$ exhibited in (\ref{eq:subleading-discontinuity}) also occurs at the other slopes  in the same $g_{m}$ orbit.

\subsection{Partitions and relation to modular functions}

Interesting features of the pattern of spin decompositions lie in the tail of the distribution. First of all, for certain sequences $\gamma^{(\alpha)},\,\alpha=1,2,\ldots$ such that $|\gamma^{(\alpha)}|\stackrel{\alpha\to\infty}{\longrightarrow}\infty$, we observe the asymptotic behavior of $J_{\rm max}({\gamma}) := {\rm max}\{j | \delta_{\gamma}(j)\neq 0\}$, in particular
\be
\begin{split}
& J_{\rm max} ({(n,n)})=\frac{n^2+1}{2} , \qquad J_{\rm max} (({n+1,n}))=\frac{n^2+n}{2} .
\end{split}
\ee
We can compare this behavior, as well as that of all other sequences in our data, with a prediction deriving from Kac's theorem (see e.g. \cite{REINEKE08}) about the dimensionality of the quiver varieties.
More precisely, for the Kronecker $m$-quiver $K(m)$, Kac's theorem asserts that the dimension of the quiver variety $M_{(a,b)}(K(m))$ for indecomposable representations with dimension vectors $\gamma=(a,b)$ is
\be
	{\rm dim} M_{(a,b)}(K(m)) = m a b -a^{2} -b^{2}+1,
\ee
therefore, noting that the Lefschez multiplet of maximal spin must be
\be
	j_{{\rm max}}(\gamma=(a,b)) = \frac{1}{2} {\rm dim} M_{(a,b)}(K(m))
\ee
we find, as a nice check, that our data agrees with this prediction.

Now, while the overall size and shape of the distribution vary with the charge,
 the degeneracies $\delta_\gamma(j)$ on the tail of the distribution \emph{stabilize} to a common pattern
\be
	\begin{array}{c|l}
	\gamma & \delta{(J_{\rm max})}, \, \delta{(J_{\rm max}-1)},\,\dots \\
	\hline
	4\gamma_1+3\gamma_2 &  1,0,2,2,3,2,2,0,\ldots \\
	7\gamma_1+6\gamma_2 &  1,0,2,2,5,6,13,14,\ldots \\
	8\gamma_1+6\gamma_2 &  1,0,2,2,5,6,13,16,\ldots \\
	8\gamma_1+7\gamma_2 &  1,0,2,2,5,6,13,16,\ldots \label{eq:tails}
	\end{array}
\ee
As (\ref{eq:tails}) shows, the length of the ``saturated'' subsequence 1,0,2,2,5,6,13,16,$\ldots$ increases with $|\gamma|$. Overall, the tail behavior seems to stabilize to the sequence generated by
\be
g(\xi)=\prod\lm_{m=2}^{\infty}(1-\xi^m)^{-2}=1+0\xi+2\xi^2+2\xi^3+5\xi^4+6\xi^5+13\xi^6+\ldots .\label{eq:partitions}
\ee
A slight modification yields the generating function which coefficients are the incremental sum of those in $g(\xi)$
\be
\tilde g(\xi)=\frac{\prod\lm_{m=2}^{\infty}(1-\xi^m)^{-2}}{(1-\xi)}=1+1\xi+3\xi^2+5\xi^3+10\xi^4+16\xi^5+29\xi^6+\ldots , \label{eq:drezet}
\ee
generating the number of planar partitions with at most two rows of the corresponding size, some examples are\\

\begin{tabular}{|c|c|c|}
  \hline
  No boxes &  & 1 empty partition \\
  \hline
  1 box & $1$ & 1 partition \\
  \hline
  2 boxes & $\begin{array}{cc}
               1 & 1
             \end{array}
  $, $2$, $\begin{array}{c}
            1 \\
            1
          \end{array}$
   & 3 partitions \\
   \hline
  3 boxes & $3$, $\begin{array}{ccc}
                   1 & 1 & 1
                 \end{array}$, $\begin{array}{cc}
                                 2 & 1
                               \end{array}$, $\begin{array}{c}
                                                2 \\
                                                1
                                              \end{array}
                               $, $\begin{array}{cc}
                                     1 & 1 \\
                                     1 &
                                   \end{array}
                               $
   & 5 partitions \\
  \hline
\end{tabular}\\

This analogy suggests that the stabilized $\delta_\gamma(j)$ distribution counts some number of \emph{constrained} partitions, only deviating from (\ref{eq:partitions}) at higher orders in $\xi$. This hypothesis is reminiscent of results of \cite{MANSCHOT, REINEKE}.

Of course, $g(\xi)$ is also closely related to the Dedekind $\eta$-function. It is quite curious that the BPS degeneracies have some relation to modular functions. This has been long expected in supergravity \cite{OSV,Denef-Moore,Cheng-et-al} but the appearance in field theory is novel.

In fact, the above   connection to the Dedekind $\eta$ was noted before our work by Reineke \cite{REINEKE03},
who offers a mathematical explanation. But the physical import of this strange behavior remains
mysterious.

\section{Open Problems}\label{sec:OpenProblems}

In conclusion we would like to mention a few open problems
and questions raised by the current work.

\begin{enumerate}

\item It is natural to guess that wild degeneracies will be a common
feature among higher rank theories of class $S$. Strictly speaking,
the only examples we have given are for gauge group $SU(3)$, but we
fully expect that the phenomenon will persist for $SU(K)$ with $K>3$.
This is strongly suggested by the quiver analysis of Section
\ref{subsec:RelateQuivers},  but a fully rigorous proof would
require that one demonstrate that
the path exhibited in the moduli space of stability
parameters of the the Fiol quiver, which leads to wild wall crossing
for $K>3$, actually can be chosen in the moduli space of physical
stability parameters. (While not fully mathematically rigorous,
a compelling physical argument that this is indeed the case is that
we could consider a hierarchy of symmetry breaking where
$SU(K)$ is much more strongly broken to $SU(3) \times U(1)^{K-3}$ than the
$SU(3)$ is broken to $U(1)^2$.)

\item Another open problem along similar lines is
 how the presence of, say, matter multiplets
affects the existence of wild degeneracies.

\item It should be noted that the explicit point on the Coulomb
branch illustrated in  Figure \ref{fig:3-herd-in-su3} is in fact
different from the region explored in Section \ref{subsec:PathCoulomb}.
Nevertheless, using the techniques of Appendix \ref{app:herd-appendix}  we have
checked that the same crucial algebraic equation \eqref{eq:P}
governing the street factors of herds indeed appears in the
spectral networks that arise in this region. These networks are
very similar to but not quite the same as the $m$-herds. One might
ask for a succinct test to see whether a degenerate spectral
network leads to $m$-wild degeneracies.

\item  It would be nice to understand better the
 physics of the curious invariance of the BPS degeneracies under the transformation
by the $g_m$ matrix discussed in Section \ref{sec:mocm} above. To the extent
that the relation to quivers is physical, a physical understanding is indeed
provided by the arguments in Section \ref{subsec:RelateQuivers}. However,
we would like to suggest an alternative interpretation using the
 halo picture of BPS states. If we consider a core particle $\gamma$
with halo particles of charge $\gamma'$ then the replacement of the hypermultiplet of charge
$\gamma$
for the hypermultiplet $\gamma + m \gamma'$ is simply flipping the Fermi sea of the
halo Fock space. (See, e.g. Section 3.5 of  \cite{Andriyash:2010yf}  for a similar transformation.)
Perhaps then a physical derivation of the symmetry could proceed by using Fermi flips
to establish
such a symmetry for framed BPS states and then using recursion relations between framed and
unframed BPS states to deduce it for general degeneracies. This symmetry also raises the
interesting possibility that the mutation method for determining BPS degeneracies can be
extended to higher spin states.

 \item
The $g_{m}$ symmetry of Kronecker quivers makes a surprising prediction about
two well-known formulae: Reineke's formula for Poincar\'e polynomials of quiver
varieties \cite{REINEKE}, and the Manschot-Pioline-Sen wall-crossing formula \cite{MANSCHOT,Manschot:2011xc}.
These formulae involve sums over certain   partitions.
For certain charges, there is rather extensive cancelation between terms in these formulae
implied by  the $g_{m}$ symmetry of the BPS degeneracies. Since the
individual terms in the sum in the MPS formulae have a simple geometrical interpretation
\cite{Manschot:2011xc}    the $g_m$ symmetry together with the MPS formula
imply nontrivial identities on equivariant Dirac indices.
For a simple and dramatic example we can choose $m=3$ and note that that $(1,1)$ has a very simple PSC, but
\be
(g_{3})^{k}\cdot{1 \choose 1} = {{F_{2k-1}} \choose {F_{2k+1}}}
\ee
(where $F_n$ is the $n^{\rm th}$ Fibonacci number)   involves  arbitrarily large charges.
Clearly there are many terms in the MPS formula (\ref{eq:MPS-WCF}) and,
as we just said, their coefficients have a beautiful geometrical interpretation as equivariant
indices of Dirac operators on the Denef moduli spaces. So the
identity\footnote{We use the notation $[n]_{y}:=\frac{y^{n}-y^{-n}}{y-y^{-1}}$}
\be
\Omega(  (F_{2k-1}, F_{2k+1})  ;y)  =  \Omega(  (1,1)  ;y) = [3]_y
\ee
 is a very remarkable set of identities  for these   indices.
It would be interesting to understand better these identities
(and their analogues for $m>3$)   from a geometrical point of view.

\item
Returning to the key algebraic equation \eqref{eq:P}, a natural question is whether
there is a physical interpretation of the other roots of this equation. We
expect that there will be. For example, choose a small path $\wp$ crossing a $c$-street
in an $m$-herd. The corresponding supersymmetric interface has a vev when
wrapped on the circle in $\IR^3 \times S^1$ given by
\be
\langle L_\zeta(\wp)\rangle_m =
\begin{pmatrix} q(m,\zeta) &  0 & 0 \\ 0 & 1 & 0 \\ 0 & 0 & q(m,\zeta) \\ \end{pmatrix}
\ee
where $m$ is a point in Hitchin moduli space ${\cal M}$, $\zeta \in {\mathbb{C}}^*$
has phase $\arg \zeta = \arg Z(\gamma + \gamma')$,
and $q(m,\zeta) = Q(c)\vert_{X_{\gamma_c} \rightarrow {\cal Y}_{\gamma_c}}$,
where ${\cal Y}_{\gamma_c}$ is a function on the twistor space of the Hitchin moduli
space ${\cal M}$ constructed in \cite{GMN1,GMN5}.  It therefore makes sense to ask
about the physical meaning of the analytic behavior of $\langle L_\zeta(\wp)\rangle$,
and this might well involve the other roots of \eqref{eq:P}. Exploring
this point further is beyond the scope of this paper.

\item A closely related point to the previous one is that the exponential growth of
$\Omega$ for certain charges implies a similar exponential growth for $\mu$ and
therefore for $\FOmega$. We expect this will have important implications for the
construction of hyperkahler metrics of associated Hitchin systems proposed
in \cite{GMN1} and for the
definition of the nonabelianization map of \cite{GMN5,GMN6}. Again, we leave this
important point for future work.

\end{enumerate}

\appendix
\section{Protected spin characters of the $3$-Kronecker quiver}\label{app:PSC-tables}

In this subsection we discuss some data for the ``refined BPS degeneracies,'' or, more
properly, the ``protected spin character.''  First, we recall some definitions. Then
we present the data.

\subsection{Spin decompositions}

Short irreducible representations of the ${\cal N}=2$ superalgebra take the general form \cite{WESS-BAGGER,MOORE_PITP}
\be
\rho_{hh}\otimes\mathfrak{h}
\ee
where $\rho_{hh}=(1/2,0)\oplus(0,1/2)$ is the \emph{half-hypermultiplet} representation of $\mathfrak{so}(3)\oplus\mathfrak{su}(2)_{R}$, and $\mathfrak{h}$ is the \emph{Clifford vacuum}. It has been shown recently in \cite{DIACONESCU} that the Clifford vacuum is actually a singlet of $\mathfrak{su}(2)_{R}$, this fact had been previously known as the \emph{no-exotics conjecture}.

In order to extract information about the spin decomposition of the BPS index, we study a refinement known as the \emph{protected spin character} (see e.g. \cite{GMN3,MOORE_FELIX})
\be
\begin{split}
\Omega (\gamma,u;y)&=\Tr_{\mathfrak{h}_{\gamma}}y^{2J_{3}}(-y)^{I_{3}}  \\
 & = \sum_{m}^{} a_{m}(\gamma,u) \,(-y)^{m}, \label{eq:PSC-def}
\end{split}
\ee
where $J_{3},I_{3}$ are Cartan elements of $\mathfrak{so}(3),\,\mathfrak{su}(2)_{R}$ respectively, and the last line defines the coefficients $a_{m}$. The PSC reduces to the BPS index in the limit $y\to -1$.

For a given charge $\gamma$, $\mathfrak{h}_{\gamma}$ has an isotypical decomposition into $\mathfrak{so}(3)$ reps:
\be
	{\mathfrak h}_{\gamma} = \bigoplus_{j} \Big( D_{j} \otimes (j,0) \Big) \label{eq:spin-decomposition}
\ee
where the degeneracy space $D_{j}$ is a complex vector space of dimension $\delta_{\gamma}(j)$. Therefore
\be
	\Omega(\gamma) = \sum_{j} (-1)^{2j}\delta_{\gamma}(j) (2j+1). \label{eq:spin-splitting}
\ee
The numerical evidence given below suggests that the degeneracies $\delta_{\gamma}(j)$ for the $3$-Kronecker quiver satisfy the following property: for fixed $\gamma$, $\delta_{\gamma}(j)\neq 0$ only for $2j$ of a definite parity.\footnote{Indeed the data suggests that $2j$ must be odd for $\gamma=a\gamma_{1}+b\gamma_{2}$ with $a,b$ both even and $2j$ must be even otherwise.}
For such spin degeneracies note that
\be
	(-1)^{2j}\Omega(\gamma) = {\rm{dim}} {\cal H}_{\gamma}.
\ee
Of course, knowing $\Omega(\gamma)$ does not determine the isotypical decomposition. In order to determine that we need to employ a generalization of the KSWCF known in the physics literature as the ``motivic'' KSWCF \cite{KS,KS_MOTIVIC_I,KS_MOTIVIC_II,DIMOFTE-GUKOV-09,GMN3}. We introduce a set of  non-commutative formal variables obeying
\be
\hat Y_{\gamma} \hat Y_{\gamma'}=y^{\langle\gamma,\gamma'\rangle}\hat Y_{\gamma+\gamma'},\qquad\forall\gamma,\gamma'\in\Gamma \, .
\ee
The generalization of (\ref{eq:formal-K}) is then (for details, see \cite{GMN3})
\be
\begin{split}
& \hat\CK^{\Omega(\gamma;y)}_{\gamma_{}}:\, \hat Y_{\gamma_{0}}\mapsto \hat Y_{\gamma_{0}} \  \prod\lm_{m\in\mathbb{Z}}\left(\Phi_{\langle\gamma_{0},\gamma_{}\rangle}((-y)^m \hat Y_{\gamma_{}})\right)^{a_m\, (\text{sign}\langle\gamma_{0},\gamma_{}\rangle)}
\end{split}
\ee
where the $a_{m}$ are defined according to (\ref{eq:PSC-def}), and
\be
	\Phi_n (\xi):=\prod\lm_{s=1}^{|n|}(1+y^{-\mathop{\rm sign}(n)(2s-1)}\xi).
\ee

Let us now apply this formalism to the case at hand, namely the $3$-Kronecker quiver. The motivic version of the wall crossing identity is
\be
	\hat\CK_{\gamma_{2}}\hat\CK_{\gamma_{1}}=\, : \prod_{\mathop{}^{a\gamma_{1}+b\gamma_{2}}_{\ a,b\geq0}} \hat \CK_{\gamma_{}}^{\Omega(a\gamma_{1}+b\gamma_{2};y)}  : \label{eq:motivic-WCF}
\ee
The RHS admits a \emph{unique}
decomposition with the required charge orderings
and hence this equation fully determines the $\Omega(a\gamma_{1}+b\gamma_{2};y)$.

In practical terms the protected spin characters can be extracted from this formula as follows.
 First,   acting with the operator on the LHS of (\ref{eq:motivic-WCF}) on the formal variable $\hat Y_{\gamma_{1}}$, yields
 \footnote{As a side note, this implies that a line defect with charge $\gamma_{1}$
 would support halo configurations of vanilla hypermultiplets, with overall halo charges
 $\gamma_{h}=k\gamma_{2},\,k=0,1,2,3$. The $\mathfrak{so}(3)$ representations of
 the respective framed BPS states would have spin $j=0,1,1,0$ (see \cite{GMN3,MOORE_FELIX}).
  These can also be thought of as the Hall-halo configurations of \cite{DENEF}.}
\be
	\hat Y_{\gamma_{1}}+\left(y^{-2}+1+y^2\right) \hat Y_{\gamma_{1}+\gamma_{2}}+\left(y^{-2}+1+y^2\right) \hat Y_{\gamma_{1}+2\gamma_{2}}+\hat Y_{\gamma_{1}+3\gamma_{2}} .
\ee
with a similar formula for the action on $\hat Y_{\gamma_{2}}$.
Then, we apply an inductive procedure
directly analogous to that used in  (\ref{eq:decomposition-technique}) for the
ordinary BPS indices.

We report the resulting PSCs in \ref{app:PSC-tables-data}, for charges up to $a+b\leq 15$ .

\subsection{The data}\label{app:PSC-tables-data}

The following tables report the content of BPS boundstates corresponding to the $3$-Kronecker quiver, only a quarter of the spectrum is given\footnote{I.e. one half of the \emph{particle} spectrum, namely dimension vectors with non-negative entries.} , the rest is determined by symmetry. For convenience, boundstates are ordered according to the phase of the central charge. Here $j$ labels the $\mathfrak{so}(3)$ irrep of the Clifford vacuum, while $\delta$ counts the number of occurrences of such irreps.\\

\begin{table}[h!]
\centering
\begin{tabular}{ccc}
\begin{minipage}[t]{.3\textwidth}
\small
\be
\begin{array}{l|cc}
\gamma & j & \delta \\
\hline \hline
\gamma_{1} & 0 & 1 \\
\hline
3\gamma_{1}+\gamma_{2} & 0 & 1 \\
\hline
8\gamma_{1}+3\gamma_{2} & 0 & 1 \\
\hline
10\gamma_{1}+4\gamma_{2} & 5/2 & 1 \\
\hline
5\gamma_{1}+2\gamma_{2} & 1 & 1 \\
\hline
7\gamma_{1}+3\gamma_{2} & 3 & 1 \\
 & 5 & 1 \\
\hline
9\gamma_{1}+4\gamma_{2} & 0 & 2 \\
& 1 & 2 \\
& 2 & 3 \\
& 3 & 2 \\
& 4 & 2 \\
& 6 & 1 \\
\hline
2\gamma_{1}+\gamma_{2} & 1 & 1 \\
\hline
4\gamma_{1}+2\gamma_{2} & 5/2 & 1 \\
\hline
6\gamma_{1}+3\gamma_{2} & 3 & 1 \\
 & 5 & 1 \\
\hline
8\gamma_{1}+4\gamma_{2} & 5/2 & 1 \\
& 9/2 & 2 \\
& 11/2 & 1 \\
& 13/2 & 2 \\
& 17/2 & 1 \\
\hline
10\gamma_{1}+5\gamma_{2} & 1 & 1 \\
& 3 & 2 \\
& 4 & 2 \\
& 5 & 4 \\
& 6 & 4 \\
& 7 & 5 \\
& 8 & 4 \\
& 9 & 4 \\
& 10 & 2 \\
& 11 & 2 \\
& 13 & 1 \\
\end{array} %
\nonumber
\ee
\end{minipage}%
&
\begin{minipage}[t]{.3\textwidth}
\small
\be
\begin{array}{l|cc}
\gamma & j & \delta \\
\hline \hline
9\gamma_{1}+5\gamma_{2} & 0 & 7 \\
& 1 & 25 \\
& 2 & 30 \\
& 3 & 38 \\
& 4 & 32 \\
& 5 & 31 \\
& 6 & 23 \\
& 7 & 21 \\
& 8 & 12 \\
& 9 & 11 \\
& 10 & 6 \\
& 11 & 5 \\
& 12 & 2 \\
& 13 & 2 \\
& 15 & 1 \\
\hline
7\gamma_{1}+4\gamma_{2} & 0 & 5 \\
 & 1 & 5 \\
 & 2 & 11 \\
 & 3 & 7 \\
 & 4 & 9 \\
 & 5 & 4 \\
 & 6 & 5 \\
 & 7 & 2 \\
 & 8 & 2 \\
 & 10 & 1 \\
\hline
5\gamma_{1}+3\gamma_{2} & 0 & 2 \\
 & 1 & 2 \\
 & 2 & 3 \\
 & 3 & 2 \\
 & 4 & 2 \\
 & 6 & 1 \\
\end{array} %
\nonumber
\ee
\end{minipage}%
&
\begin{minipage}[t]{.3\textwidth}
\small
\be
\begin{array}{l|cc}
\gamma & j & \delta \\
\hline \hline
8\gamma_{1}+5\gamma_{2} & 0 & 17 \\
& 1 & 32 \\
& 2 & 55 \\
& 3 & 55 \\
& 4 & 61 \\
& 5 & 48 \\
& 6 & 44 \\
& 7 & 30 \\
& 8 & 25 \\
& 9 & 15 \\
& 10 & 12 \\
& 11 & 6 \\
& 12 & 5 \\
& 13 & 2 \\
& 14 & 2 \\
& 16 & 1 \\
\hline
3\gamma_{1}+2\gamma_{2} & 1 & 2 \\
 & 3 & 1 \\
\hline
6 \gamma_{1}+4\gamma_{2} & 1/2 & 4 \\
 & 3/2 & 7 \\
  & 5/2 & 11 \\
   & 7/2 & 7 \\
    & 9/2 & 10 \\
     & 11/2 & 5 \\
      & 13/2 & 5 \\
       & 15/2 & 2 \\
        & 17/2 & 2 \\
         & 21/2 & 1 \\
\end{array} %
\nonumber
\ee
\end{minipage}%
\end{tabular}
\hfill
\end{table}

\clearpage

\begin{table}[h!]
\centering
\begin{tabular}{ccc}
\begin{minipage}[t]{.3\textwidth}
\small
\be
\begin{array}{l|cc}
\gamma & j & \delta \\
\hline \hline
9\gamma_{1}+6\gamma_{2} & 0 & 31 \\
& 1 & 125 \\
& 2 & 173 \\
& 3 & 241 \\
& 4 & 251 \\
& 5 & 279 \\
& 6 & 255 \\
& 7 & 244 \\
& 8 & 201 \\
& 9 & 177 \\
& 10 & 129 \\
& 11 & 109 \\
& 12 & 74 \\
& 13 & 58 \\
& 14 & 37 \\
& 15 & 29 \\
& 16 & 15 \\
& 17 & 13 \\
& 18 & 16 \\
& 19 & 5 \\
& 20 & 2 \\
& 21 & 2 \\
& 23 & 1 \\
\hline
7\gamma_{1}+5\gamma_{2} & 0 & 17 \\
& 1 & 32 \\
& 2 & 55 \\
& 3 & 55 \\
& 4 & 61 \\
& 5 & 48 \\
& 6 & 44 \\
& 7 & 30 \\
& 8 & 25 \\
& 9 & 15 \\
& 10 & 12 \\
& 11 & 6 \\
& 12 & 5 \\
& 13 & 2 \\
& 14 & 2 \\
& 16 & 1 \\
\end{array} %
\nonumber
\ee
\end{minipage}%
&
\begin{minipage}[t]{.3\textwidth}
\small
\be
\begin{array}{l|cc}
\gamma & j & \delta \\
\hline \hline
4\gamma_{1}+3\gamma_{2} & 0 & 2 \\
 & 1 & 2 \\
 & 2 & 3 \\
 & 3 & 2 \\
 & 4 & 2 \\
 & 6 & 1 \\
\hline
8\gamma_{1}+6\gamma_{2} & 1/2 & 94 \\
& 3/2 & 171 \\
& 5/2 & 242 \\
& 7/2 & 263 \\
& 9/2 & 291 \\
& 11/2 & 263 \\
& 13/2 & 252 \\
& 15/2 & 203 \\
& 17/2 & 179 \\
& 19/2 & 128 \\
& 21/2 & 109 \\
& 23/2 & 71 \\
& 25/2 & 58 \\
& 27/2 & 35 \\
& 29/2 & 29\\
& 31/2 & 15 \\
& 33/2 & 13 \\
& 35/2 & 6 \\
& 37/2 & 5 \\
& 39/2 & 2 \\
& 41/2 & 2 \\
& 45/2 & 1 \\
\hline
5\gamma_{1}+4\gamma_{2} & 0 & 5 \\
 & 1 & 5 \\
 & 2 & 11 \\
 & 3 & 7 \\
 & 4 & 9 \\
 & 5 & 4 \\
 & 6 & 5 \\
 & 7 & 2 \\
 & 8 & 2 \\
 & 10 & 1 \\
\end{array}  %
\nonumber
\ee
\end{minipage}%
&
\begin{minipage}[t]{.3\textwidth}
\small
\be
\begin{array}{l|cc}
\gamma & j & \delta \\
\hline
\hline
6\gamma_{1}+5\gamma_{2} & 0 & 7 \\
 & 1 & 25 \\
 & 2 & 30 \\
 & 3 & 38 \\
 & 4 & 32 \\
 & 5 & 31 \\
 & 6 & 23 \\
 & 7 & 21 \\
 & 8 & 12 \\
 & 9 & 11 \\
 & 10 & 6 \\
 & 11 & 5 \\
 & 12 & 2 \\
 & 13 & 2 \\
 & 15 & 1 \\
\hline
7\gamma_{1}+6\gamma_{2} & 0 & 23 \\
& 1 & 95 \\
& 2 & 119 \\
& 3 & 160 \\
& 4 & 150 \\
& 5 & 157 \\
& 6 & 131 \\
& 7 & 124 \\
& 8 & 91 \\
& 9 & 83 \\
& 10 & 57 \\
& 11 & 49 \\
& 12 & 31 \\
& 13 & 26 \\
& 14 & 14 \\
& 15 & 13 \\
& 16 & 6 \\
& 17 & 5 \\
& 18 & 2 \\
& 19 & 2 \\
& 21 & 1 \\
\end{array} %
\nonumber
\ee
\end{minipage}%
\end{tabular}
\hfill
\end{table}
\clearpage

\begin{table}[h!]
\centering
\begin{tabular}{ccc}
\begin{minipage}[t]{.3\textwidth}
\small
\be
\begin{array}{l|cc}
\gamma & j & \delta \\
\hline \hline
8\gamma_{1}+7\gamma_{2} & 0 & 135 \\
& 1 & 353 \\
& 2 & 562 \\
& 3 & 677 \\
& 4 & 765 \\
& 5 & 762 \\
& 6 & 752 \\
& 7 & 679 \\
& 8 & 619 \\
& 9 & 522 \\
& 10 & 455 \\
& 11 & 363 \\
& 12 & 304 \\
& 13 & 231 \\
& 14 & 188 \\
& 15 & 135 \\
& 16 & 109 \\
& 17 & 73 \\
& 18 & 57 \\
& 19 & 36 \\
& 20 & 28 \\
& 21 & 16 \\
& 22 & 13 \\
& 23 & 6 \\
& 24 & 5 \\
& 25 & 2 \\
& 26 & 2 \\
& 28 & 1 \\
\end{array}  %
\nonumber
\ee
\end{minipage}%
&
\begin{minipage}[t]{.3\textwidth}
\small
\be
\begin{array}{l|cc}
\gamma & j & \delta \\
\hline
\hline
\gamma_{1}+\gamma_{2} & 1 & 1 \\
\hline
2\gamma_{1}+2\gamma_{2} & 5/2 & 1 \\
\hline
3\gamma_{1}+3\gamma_{2} & 3 & 1 \\
 & 5 & 1 \\
\hline
4\gamma_{1}+4\gamma_{2} & 5/2 & 1 \\
 & 9/2 & 2 \\
 & 11/2 & 1 \\
 & 13/2 & 2 \\
 & 17/2 & 1 \\
\hline
 5\gamma_{1}+5\gamma_{2} & 1 & 1 \\
 & 3 & 2 \\
 & 4 & 2 \\
 & 5 & 4 \\
 & 6 & 4 \\
  & 7 & 5 \\
   & 8 & 4 \\
    & 9 & 4 \\
     & 10 & 2 \\
      & 11 & 2 \\
       & 13 & 1 \\
\hline
6\gamma_{1}+6\gamma_{2} & 1/2 & 1 \\
& 3/2 & 2 \\
& 5/2 & 5 \\
& 7/2 & 5 \\
& 9/2 & 11 \\
& 11/2 & 9 \\
& 13/2 & 18 \\
& 15/2 & 15 \\
& 17/2 & 20 \\
& 19/2 & 15 \\
& 21/2 & 18 \\
& 23/2 & 9 \\
& 25/2 & 11 \\
& 27/2 & 5 \\
& 29/2 & 5\\
& 31/2 & 2 \\
& 33/2 & 2 \\
& 37/2 & 1 \\
\end{array}  %
\nonumber
\ee
\end{minipage}%
&
\begin{minipage}[t]{.3\textwidth}
\small
\be
\begin{array}{l|cc}
\gamma & j & \delta \\
\hline
\hline
7\gamma_{1}+7\gamma_{2} & 0 & 1 \\
& 1 & 10 \\
& 2 & 12 \\
& 3 & 23 \\
& 4 & 28 \\
& 5 & 41 \\
& 6 & 48 \\
& 7 & 63 \\
& 8 & 68 \\
& 9 & 79 \\
& 10 & 77 \\
& 11 & 79 \\
& 12 & 68 \\
& 13 & 63 \\
& 14 & 48 \\
& 15 & 41 \\
& 16 & 29 \\
& 17 & 23 \\
& 18 & 14 \\
& 19 & 12 \\
& 20 & 6 \\
& 21 & 5 \\
& 22 & 2 \\
& 23 & 2 \\
& 25 & 1 \\
\end{array} %
\nonumber
\ee
\end{minipage}
\end{tabular}
\hfill
\label{tab:etcetc}
\end{table}

\section{The Six-Way Junction} \label{app:six-way}
\begin{figure}
	\begin{center}
		 \includegraphics[scale=0.5]{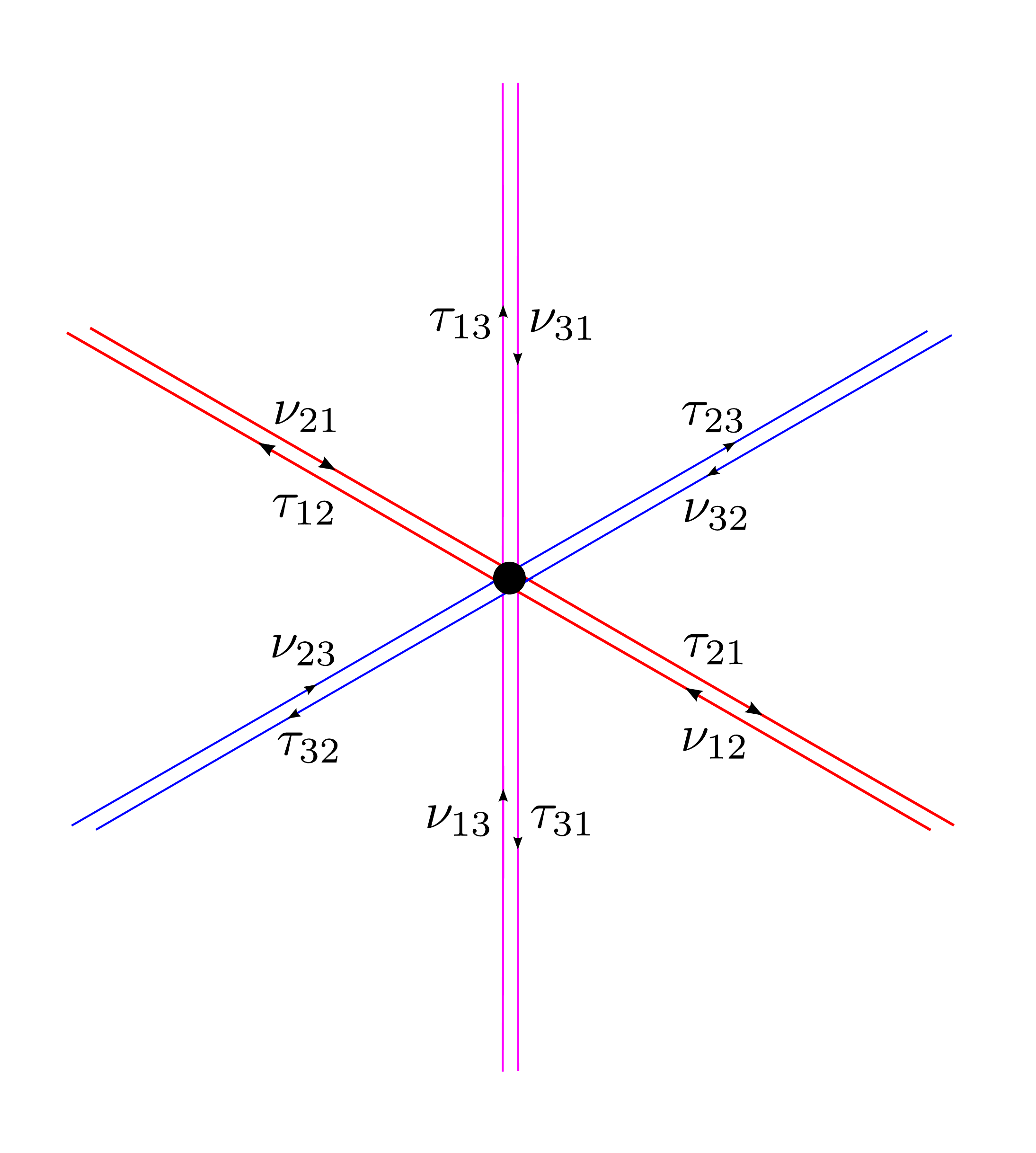}
		\caption{A six-way junction.  Two-way streets are resolved into one-way constituent streets using the British resolution.  Streets of type 12 are red, type 23 are blue, and type 13 are fuchsia.  A soliton generating function attached to a (one-way constituent) street is shown adjacent to its respective street.  Subscripts on the soliton generating functions are ordered pairs $ij \in \{1,2,3\}^2$ denoting the type of solitons that the generating function ``counts". \label{fig:six-way}}
	\end{center}
\end{figure}

For reference, we present some basic conditions on soliton generating functions as enforced by the homotopy invariance of the framed 2D-4D generating functions $F(\wp,\vartheta)$.  First, using the convention described in Section \ref{sec:degen_net}, we assign every two-way street an orientation.  If the network in question is degenerate, we resolve all two-way streets into ``constituent one-way streets" using the \textit{British resolution}: let $p$ be a two-way street; using the orientation on $p$, we resolve $p$ into two one-way streets running in opposite directions, infinitesimally displaced from one another, and such that the street pointing along the orientation of $p$ is to the \textit{left} of the street running against the orientation.  If $p$ is a two-way street of type $ij$ (i.e. composed of coincident streets of type $ij$ and type $ji$), then (after resolving) the street on the left is of type $ij$ and the street on the right is of type $ji$.

Just as with Kirchoff's circuit laws it is most convenient to express our equations locally around each joint (or branch point).  Hence, rather than expressing them in terms of the street-dependent $\Upsilon/\Delta$ notation introduced in (\ref{eq:upsilon_def})-(\ref{eq:delta_def}), we will temporarily adopt a joint-dependent notation.

\begin{definition}
	Let $v \in C$ be a joint or branch point, then $\tau_{ij}$ will denote the soliton generating function attached to a constituent one-way street of type $ij$ running \textit{out} of $v$, and $\nu_{ij}$ will denote the soliton generating function attached to a constituent one-way street of type $ij$ running \textit{into} $v$.
\end{definition}

In a full spectral network, the joint dependent $\tau,\, \nu$ notation can become redundant; so we will eventually revert back to the $\Upsilon/\Delta$ notation in Appendix \ref{app:herd-appendix}.

To define products of soliton generating functions properly we introduce the following.

\begin{definition}
	Let $\eta$ be a formal variable that acts on each formal variable $X_{a}$ in the homology path algebra via
	\begin{align*}
		\eta X_{a} &= X_{a^{\mathrm{tw}}},
	\end{align*}
	where, at the level of 1-chains, $a^{\mathrm{tw}}$ is the 1-chain produced by inserting a half-twist along the circle fiber of $\twid{\Sigma} \rightarrow \Sigma$ at some point\footnote{Up to homotopy (rel endpoints) the insertion point does not matter; hence, it is irrelevant for relative homology.} along $a$.
\end{definition}

\begin{remark}
	It is immediate that $\forall G \in \CA$
	\begin{align*}
		\eta^2 G = X_{H} G = - G.
	\end{align*}
\end{remark}

We now consider a general type of joint, that can occur for a spectral network subordinate to a branched cover with $K \geq 3$ sheets, where six (possibly two-way) streets meet.  The situation is shown in Fig.~\ref{fig:six-way}: the (relevant) sheets of the cover are labeled from $1$ to $3$, and the soliton generating functions attached to a constituent one-way street (under the British resolution of all possible two-way streets) are shown adjacent to their corresponding sheet.  Using homotopy invariance of $F(\wp,\vartheta)$, one arrives at the six-way junction equations:\footnote{In \cite{GMN5} these equations were erroneously written without the factors $\eta$, $\eta^{-1}$.}
\begin{equation}
	\begin{array}{lr}
	\begin{aligned}
	\tau_{12} &= \nu_{12} + \eta \tau_{13} \nu_{32}, \\
	\tau_{23} &= \nu_{23} + \eta \tau_{21} \nu_{13}, \\
	\tau_{31} &= \nu_{31} + \eta \tau_{32} \nu_{21},
	\end{aligned} &
	 \begin{aligned}
	\tau_{21} &= \nu_{21} + \eta^{-1} \nu_{23} \tau_{31}, \\
	\tau_{32} &= \nu_{32} + \eta^{-1} \nu_{31} \tau_{12}, \\
	\tau_{13} &= \nu_{13} + \eta^{-1} \nu_{12} \tau_{23}.
	\end{aligned}
\end{array}
\label{eq:6way}
\end{equation}

At a branch point of type $ij$, we will assume that there is at most one two-way street, of type $ij$, emanating from the branch point; on this two-way street we will take
\begin{align*}
	\tau_{ij} &= X_{a_{ij}}
\end{align*}
where $a_{ij}$ is the charge of a simpleton.\footnote{The coefficient of $\mu(a_{ij}) = 1$ in front of $X_{a_{ij}}$ is a result of the soliton input data (\ref{eq:simpleton_input}).}  As described at the end of Section \ref{sec:wtheta}, fixing a point $z$ near the branch point, such a simpleton is represented by a path which runs from the lift of $z$ on sheet $i$ to the lift of $z$ on sheet $j$.  In \cite{GMN5} one can find a more general rule accommodating the situation of three two-way streets emanating from the branch point; however, we will not need this generalized rule for $m$-herds.

\section{$m$-Herds in Detail}\label{app:herd-appendix}

\subsection{Notational Definitions} \label{app:herd_notation}
We will consider four distinct branch points of a branched cover $\Sigma \rightarrow C$ of
degree $K \geq 3$.  On any local region on $C'$, where the cover may be trivialized, only
three sheets will be relevant and we will label the relevant sheets from $1$ to $3$.  Label the
branch points from 1 to 4 such that branch points 1 and 3 are branch points of type 12, while
 branch points 2 and 4 are branch points of type 23.  For each branch point $i \in \{1, \cdots, 4\}$ we
 will choose a simpleton (cf. the end of Section \ref{sec:wtheta}) $s_{i}$ with endpoints on distinct
 lifts of some $z_{i} \in C'$ close to the $i$th branch point.  $s_{1}$ and $s_{2}$ will be simpletons of
  type $12$ and $23$, respectively, while $s_{3}$ and $s_{4}$ will be of type $21$ and $32$, respectively.
 We denote the charges of these simpletons by
\begin{equation}
	\begin{aligned}
	a_{*} &= [s_{1}] \in \Gamma_{12}(z_{1},z_{1})\\
	b_{*} &= [s_{2}] \in \Gamma_{23}(z_{2},z_{2})\\
	\conj{a}_{*} &= [s_{3}] \in \Gamma_{21}(z_{3},z_{3})\\
	\conj{b}_{*} &= [s_{4}] \in \Gamma_{32}(z_{4},z_{4}).
	\end{aligned}
	\label{eq:charge_def}
\end{equation}
More often, however, computations are performed in the ``$\mathbb{Z}/2\mathbb{Z}$-extended" sets $\twid{\Gamma}(\twid{z},-\twid{z}), \, \twid{z} \in \twid{C'}$ where we define
\begin{equation}
	\begin{aligned}
	a &= [\widehat{s_{1}}] \in \twid{\Gamma}_{12}(\twid{z}_{1}, -\twid{z}_{1})\\
	b &= [\widehat{s_{2}}] \in \twid{\Gamma}_{23}(\twid{z}_{2}, -\twid{z}_{2})\\
	\conj{a} &= [\widehat{s_{3}}] \in \twid{\Gamma}_{21}(\twid{z}_{3}, -\twid{z}_{3})\\
	\conj{b} &= [\widehat{s_{4}}] \in \twid{\Gamma}_{32}(\twid{z}_{4}, -\twid{z}_{4}).
	\end{aligned}
	\label{eq:lifted_charge_def}
\end{equation}
where $\widehat{\left( \cdot \right)}$ denotes the \textit{tangent framing lift} (first discussed in Section \ref{sec:comp_Omega}) and the $\twid{z}_{i} \in \twid{C}$ are the unit tangent vectors at the starting points of the tangent framing lifts.

In a slight abuse of notation, horse streets\footnote{See the definition in Section \ref{sec:herds}.} (which may be two-way), will
be denoted by decorated latin letters: $a_{i}, \, \conj{a}_{i}$ are streets
of type $12$, $b_{i}, \, \conj{b}_{i}$
are of type $23$, and $c$ is
of type $13$.  The subscripts, denoted by $i \in \{1,2,3 \}$, denote which street is in question and the use of overlines are just a notational exploit of the duality operation described below in \ref{sec_duality}.

Furthermore, in contrast with the ``joint-dependent" $\tau,\, \nu$ notation of Appendix \ref{app:six-way}, we will (more naturally) denote soliton generating functions\footnote{We refer to Section \ref{sec:form_var} for the detailed definitions of generating functions and formal variables.} ``streetwise." 	The point $z \in p \subset C$ in the definition of soliton generating functions will be dropped for notational convenience.  As mentioned in a remark at the end of Section \ref{sec:form_var}: for any $z,z' \in p$ the generating functions $\Upsilon_{z}(p)$ and $\Upsilon_{z'}(p)$ are related by parallel transport (similarly for $\Delta_{z}(p)$ and $\Delta_{z'}(p)$).

For the sake of readability, we will modify our notation slightly from Section \ref{sec:form_var} and write streets as subscripts.

\begin{definition}\
	 Let $p$ be a street, the generating function of solitons on $p$ which \textit{agree} with the orientation of $p$ is denoted $\Upsilon_{p}$, the generating function of solitons which \textit{disagree} with the orientation of $p$ is denoted $\Delta_{p}$.  In all figures in this paper streets are oriented in an upward direction (upsilon is for ``up" and delta is for ``down").
\end{definition}

We now wish to associate the street factor (a generating function) $Q_p$ to each street $p$.  To do so, it is convenient to pass through the definition of a closely related auxiliary function.

\begin{definition}
		\item For each street $p$, we define the function
			\begin{equation}
			\mathcal{Q}_p := 1 + \Upsilon_{p} \Delta_{p} \in \CA_{C}
			\label{eq:Q_def}
			\end{equation}
\end{definition}

To produce a formal series in the $X_{\gamma},\, \gamma \in \Gamma$, we use the ``basepoint-forgetting" closure map.

\begin{definition}
	\begin{equation*}
		Q_p := \cl \left[ \mathcal{Q}_{p} \right] \in \formgamma.
	\end{equation*}
\end{definition}

We now make some important technical remarks about the use of $\mathcal{Q}_{p}$ vs. $Q_p$.

\begin{remark}[Remarks]\
	If $p$ is a street of type $ij$, then $\mathcal{Q}_{p}$ is a formal series in formal variables over $\Gamma_{ii}$.  In particular, let $a \in \Gamma_{kl}$, then this means
	\begin{align*}
		\mathcal{Q}_{p} X_{a} &=
		\left\{
		\begin{array}{lr}
			0 & \text{if $k \neq i$} \\
			Q_p X_{a} = X_{a} Q_p & \text{if $k = i$}
		\end{array}
		\right.,\\
		X_a \mathcal{Q}_{p} &=
		\left\{
		\begin{array}{lr}
			0 & \text{if $l \neq i$} \\
			X_{a} Q_p = Q_p X_{a}  & \text{if $l = i$}
		\end{array}.
		\right.
	\end{align*}
	Hence, if the (left or right) action of $\mathcal{Q}_p$ on a soliton function of type $k l$ is nonvanishing, then it can be replaced with the (commutative) action of $Q_p$.  In the following derivations, the action of $\mathcal{Q}_p$ happens to be always nonvanishing; hence, it will almost always be replaced by $Q_p$, except in cases where we resist such replacements for the sake of precision and pedagogy.
\end{remark}

\begin{definition}[Terminology]
	Occassionally we will use the term \textit{spectral data} to refer to the collection of soliton generating functions, street factors, and the functions $\mathcal{Q}_{p}$ supported on a particular collection of streets.
\end{definition}

\subsection{Duality} \label{sec_duality}
As an oriented graph embedded in a disk, Fig.~\ref{fig:horse} is invariant under an involution given by rotating the diagram 180 degrees, and reversing all orientations; we denote this involution on streets $p$ via an overline
\begin{equation}
	p \mapsto \conj{p},
\end{equation}
for $p \in \{a_{i}, \conj{a_{i}}, b_{i}, \conj{b_{i}},c: i=1,2,3 \}$.  As the terminology suggests, this involution satisfies $\conj{\conj{p}}=p$ for every street $p$ and $\conj{c}=c$.  We claim that this geometric involution actually induces a duality operation on all spectral data, i.e. generating functions.  In particular, on any equations involving soliton generating functions, the replacements
\begin{equation}
	\begin{aligned}
	\Upsilon_{p} & \leftrightarrow  \Delta_{\conj{p}}\\
	\eta & \leftrightarrow  \eta^{-1},
	\end{aligned}
\end{equation}
with all products taken in reverse order, will also yield a valid equation.  This claim can be verified by brute-force checking.  Note, in particular, applying the duality operation to the definition of $\mathcal{Q}_p$ in (\ref{eq:Q_def}) will yield $\mathcal{Q}_{\conj{p}}$.

\subsection{The Horse as a Machine} \label{app:horse_machine}

Recall the definition of a horse is given as a condition on the subset of two-way streets of a spectral network in an open disk region (see Section \ref{sec:herds}).  For convenience we restate the definition.

\begin{definition}[Definitions]\
	\begin{enumerate}
	 \item A \textit{horse street} $p \in \{a_{1},a_{2},a_{3},b_{1},b_{2},b_{3},c,\conj{a_{1}},\conj{a_{2}},\conj{a_{3}},\conj{b_{1}}, \conj{b_{2}},\conj{b_{3}}\}$ is one of the streets of Fig.~\ref{fig:horse} (left frame).

	 \item Let $N$ be a spectral network (subordinate to some branched cover $\Sigma \rightarrow C$) and $U \subset C'$ be an open disk region.  Then $U \cap N$ is a horse if a subset of its streets can be identified with Fig.~\ref{fig:horse} in a way such that
	 	\begin{enumerate}
	 		\item Every two-way street is a horse street.

	 		\item There is always a two-way street identified with the street labeled $c$.
	 	\end{enumerate}
	\end{enumerate}
\end{definition}

It may happen, however, that on a horse there are ``background" non-horse streets that cannot be identified with those of Fig.~\ref{fig:horse}; by definition, these are one-way streets.  The following claim ensures that the computation of soliton generating functions on the streets of a horse are independent of the details of the non-horse streets.

\begin{figure}
	\begin{center}
		 \includegraphics[scale=0.4]{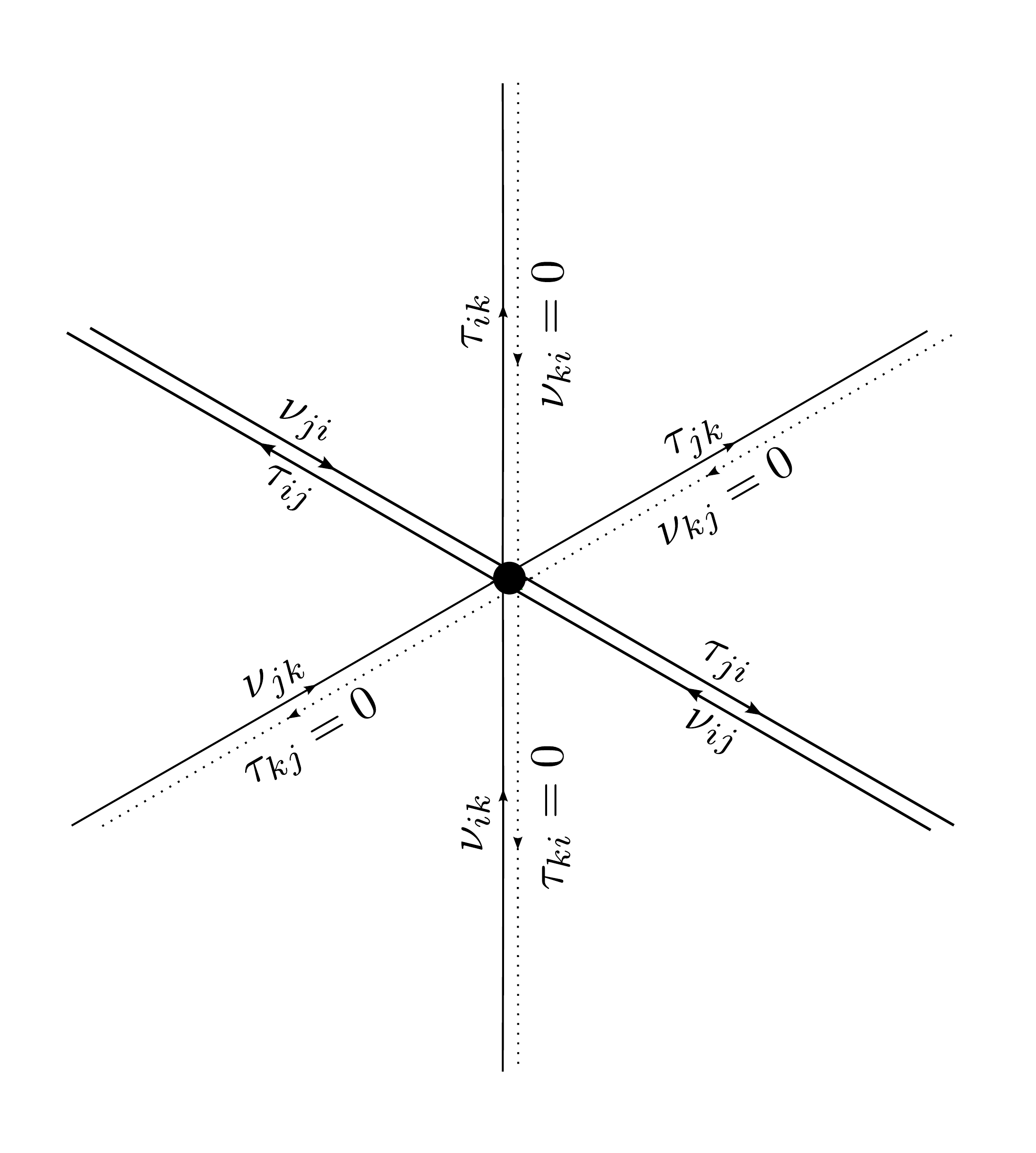}
		\caption{The most general type of joint where non-horse streets of class $(A)$ meets a horse street (which may be two-way).  As in Fig.~\ref{fig:six-way}, streets are resolved into one-way constituents using the British resolution.  Soliton generating functions vanish on the dotted streets. The labels $i,j,k$ are a permutation of the sheets $1,2,3$. \label{fig:background_joint}}
	\end{center}
\end{figure}

\begin{claim}
	 The equations for soliton generating functions on horse streets, induced by (\ref{eq:6way}), close on themselves.  I.e., the equations for the soliton generating functions on a given horse street can be written entirely in terms of the soliton generating functions on horse streets.
\end{claim}
\begin{proof}
	Let us temporarily denote the four joints in Fig.~\ref{fig:horse} (left frame) as \textit{horse joints}.  We split non-horse streets into two classes:
	\begin{enumerate}
		\item[(A)] Streets that have no endpoints on a horse joint.

		\item[(B)] Streets that have a single endpoint on a horse joint.
	\end{enumerate}
	Let us first consider streets of class $(A)$.  The claim is trivial for $(A)$-streets that do not intersect a horse street.  Thus, we turn our attention to a joint where an $(A)$-street meets a horse street.  The most general picture of such a joint\footnote{By the ``most general picture" we mean a six-way junction equipped with the weakest possible constraints on incoming soliton degeneracy functions, compatible with the condition that only the streets of type $ij$ (for some fixed pair $ij$) are two-way.  Using (\ref{eq:6way}), one finds that the most general picture is Fig.~\ref{fig:background_joint}.} is depicted in Fig.~\ref{fig:background_joint}. In this figure: $i,j,k$ label any permutation of the sheets $1,2,3$, the streets of type $jk$ and $ki$ label background one-way streets, and the streets of type $ij$ compose the the horse street (after being split into two streets by the joint).  The soliton generating functions on the horse street are (in the ``joint-wise" notation of Section \ref{app:six-way}) $\tau_{ji}, \nu_{ij},\, \tau_{ij}$, and $\nu_{ji}$.  The claim (for $(A)$-streets) is then equivalent to the statement that $\tau_{ij} = \nu_{ij},\, \tau_{ji} = \nu_{ji}$; we will show this is the case.  Indeed, by the six-way joint equations (\ref{eq:6way}):
	\begin{align*}
		\tau_{ij} &= \nu_{ij} +
		\left\{
		\begin{array}{lr}
			\eta \tau_{ik} \nu_{kj}, & \text{if $ij \in \{12,\,23,\,31\}$}\\
			\eta^{-1} \nu_{ik} \tau_{kj}, & \text{if $ij \in \{21,\,32,\,13\}$}
		\end{array}
		\right.\\
		\tau_{ji} &= \nu_{ji} +
		\left\{
		\begin{array}{lr}
			\eta^{-1} \nu_{jk} \tau_{ki}, & \text{if $ij \in \{12,\,23,\,31\}$}\\
			\eta \tau_{jk} \nu_{ki}, & \text{if $ij \in \{21,\,32,\,13\}$}
		\end{array}
		\right.
	\end{align*}
	but $\nu_{kj} = 0,\, \tau_{k j} =0 , \, \nu_{ki} = 0$, and $\tau_{ki} = 0$.  Hence,
	\begin{align*}
		\tau_{ij} &= \nu_{ij}\\
		\tau_{ji} &= \nu_{ji}.
	\end{align*}

	Now, via inspection of Fig.~\ref{fig:horse}, streets of class $(B)$ are of type 13.  If a $(B)$-street meets a horse street at a non-horse joint, then we apply the same argument used for $(A)$-streets to see that the (equations for) soliton generating functions on horse streets do not depend on the $(B)$-street soliton generating function.  Thus, we focus our attention on the horse joint.
	
	If a $(B)$-street meets a horse joint, then (\ref{eq:6way}) requires the equations for soliton generating functions, on the horse streets meeting the joint, to depend on the soliton generating function of the $(B)$-street.  We will show that the soliton generating function on the $(B)$ street can be rewritten in terms of generating functions on the horse streets.  First, note that if a $(B)$-street meets the horse joint where $a_{1}$ and $b_{1}$ meet, or the joint where $\conj{a}_{1}$ and $\conj{b}_{1}$ meet, then it must be outgoing with respect to the horse joint.  Indeed, the constraint that $c$ is two-way requires the presence of outgoing streets of type $13$ at the horse joints meeting $c$; if the $(B)$-street were incoming, it would combine with one of these outgoing streets to form a two-way street, violating the horse condition.  Without loss of generality, assume the $(B)$ street meets the horse joint where $a_{1}$ and $b_{1}$ meet; denote the soliton generating function on the $(B)$-street by $\Upsilon_{(B)}$.  Then, using (\ref{eq:6way}), it follows that $\Upsilon_{(B)} = \eta^{-1} \Upsilon_{a_{1}} \Upsilon_{b_{1}}$; so its soliton generating function is a function of the soliton generating functions on horse streets.

	If a $(B)$-street meets one of the other two horse joints (where $b_{3}$ and $\conj{a}_{3}$ meet or where $a_{3}$ and $\conj{b}_{3}$ meet), then there are two situations: the horse streets at the horse joint are both two-way, or only one of the horse streets at the horse joint is two-way.  The former situation is equivalent to the situation where the $(B)$-street meets the horse joint where $a_{1}$ and $b_{1}$ meet.  To resolve the latter situation we repeat the same argument used for $(A)$-streets.
\end{proof}

We divide the soliton generating functions supported on horse streets into elements of three subspaces: incoming data, outgoing data, and internal data.

\subsubsection*{Incoming data}
Incoming data is defined as the spectral data which flows into the internal joints of the horse and is supported on the external streets.  Here, the space of such data is composed of four soliton generating functions and their duals:
\begin{equation}
	\incdat = \left\{ \left(
	\begin{array}{cccc}
	\inup_{a_{1}}, & \inup_{b_{1}}, & \indown_{a_{3}}, & \indown_{b_{3}},\\
	 \indown_{\conj{a_{1}}}, & \indown_{\conj{b_{1}}}, &	 \inup_{\conj{a_{3}}}, & \inup_{\conj{b_{3}}}
	  \end{array}
	  \right)	\in \CA_{S}^{\times 8} \right\}.
\end{equation}
It will prove useful to subdivide this space of data further into generating functions of solitons that agree with the orientation of the diagram, $\incdat^{+}$, and those that disagree, $\incdat^{-}$:
\begin{equation*}
	\begin{aligned}
		\incdat^{+} &= \left\{ \left( \inup_{a_{1}}, \inup_{b_{1}}, \inup_{\conj{a_{3}}},  \inup_{\conj{b_{3}}} \right) \in \CA_{S}^{\times 4} \right\}, \\
		\incdat^{-} &= \left\{ \left( \indown_{\conj{a_{1}}}, \indown_{\conj{b_{1}}}, \indown_{a_{3}},  \indown_{b_{3}} \right ) \in \CA_{S}^{\times 4} \right\}.
	\end{aligned}
\end{equation*}

\subsubsection*{Outgoing data}
Similarly, outgoing data is defined as the spectral data which flows out of the internal joints and is supported on external streets.  This consists of the space of soliton generating functions,
\begin{equation}
	\outdat = \left\{\left(
	\begin{array}{cccc}
		\odown_{a_{1}}, & \odown_{b_{1}}, & \oup_{a_{3}}, & \oup_{b_{3}},\\
		\oup_{\conj{a_{1}}}, & \oup_{\conj{b_{1}}}, & \odown_{\conj{a_{3}}}, & \odown_{\conj{b_{3}}}
	\end{array}
	\right) \in \CA_{S}^{\times 8} \right\}.
\end{equation}
As with the incoming data, we can similarly subdivide this data into generating functions of solitons that agree or disagree with the overall orientation:
\begin{equation}
	\begin{aligned}
		\outdat^{+} &= \left\{ \left( \oup_{\conj{a_{1}}}, \oup_{\conj{b_{1}}}, \oup_{a_{3}}, \oup_{b_{3}} \right ) \in \CA_{S}^{\times 4} \right\}, \\
		 \outdat^{-} &= \left\{ \left(\odown_{a_{1}}, \odown_{b_{1}}, \odown_{\conj{a_{3}}}, \odown_{\conj{b_{3}}} \right) \in \CA_{S}^{\times 4} \right\}.
	\end{aligned}
\end{equation}

\subsubsection*{Internal/Bound data}
The internal data of the diagram is composed of the ten soliton generating functions defined on the internal streets $a_{2},\, b_{2},\, \conj{a_{2}},\, \conj{b_{2}}$:
\begin{equation}
	\intdat = \left\{ \left(
	\begin{array}{ccccc}
	\oup_{a_{2}}, & \oup_{b_{2}}, & \oup_{\conj{a_{2}}}, & \oup_{\conj{b_{2}}}, & \oup_{c},\\
	\odown_{\conj{a_{2}}}, & \odown_{\conj{b_{2}}}, & \odown_{a_{2}}, & \odown_{b_{2}}, & \odown_{c}.
	\end{array}
	\right) \in \CA_{S}^{\times 10} \right\}
\end{equation}
However, as far as the results of this paper are concerned, all that is relevant are the street factors $Q_p$, for $p$ an internal street, which are derived from the soliton generating functions above:
\begin{equation}
	\intdat \rightsquigarrow \left\{\left(
	\begin{array}{ccccc}
	\oQ_{a_{2}}, & \oQ_{b_{2}}, & \oQ_{\conj{a_{2}}}, & \oQ_{\conj{b_{2}}}, & \oQ_{c}
	\end{array}
	\right) \in \formgamma^{\times 5} \right\}.
\end{equation}
We then view a horse as a scattering-matrix machine that eats incoming solitons and spits out outgoing solitons + ``bound"/internal solitons:
\begin{equation*}
	\text{Horse}: \incdat \rightarrow \outdat \times \intdat,
\end{equation*}
or in other words, we can determine $\outdat$ and $\intdat$ as a function of $\incdat$; to do so we utilize the six-way junction equations (\ref{eq:6way}) to give\footnote{When using the six-way junction equations on the four relevant joints of a horse, pictured in the left panel of Fig.~\ref{fig:horse}, one must take into account one-way streets of type 13 that flow out of these joints.  However, as shown in the proof of the claim of Section \ref{app:horse_machine}, the soliton generating functions on these one-way streets can be written in terms of soliton generating functions on the horse streets.}
\begin{equation}\label{eq:sixwayhorse}
	\begin{split}
	\oup_{a_{2}} &= \inup_{a_{1}} + \eta \oup_{c} \odown_{b_{2}}\\
	\oup_{a_{3}} &= \oup_{a_{2}} + \eta \left(\eta^{-1} \oup_{a_{2}} \oup_{\conj{b_{2}}} \right) \odown_{\conj{b_{2}}}\\
	& = \oup_{a_{2}}  \osQ_{\conj{b_{2}}}\\
	\oup_{b_{2}} &= \inup_{b_{1}}\\
	\oup_{b_{3}} &= \oup_{b_{2}}\\
	\oup_{c} &= \eta^{-1} \inup_{a_{1}} \inup_{b_{1}}\\
	\odown_{a_{1}} &= \odown_{a_{2}} + \eta^{-1} \inup_{b_{1}} \left( \odown_{c} + \eta \odown_{b_{1}} \odown_{a_{2}} \right)\\
	&= \osQ_{b_{1}} \odown_{a_{2}} + \eta^{-1} \inup_{b_{1}} \odown_{c}\\
	\odown_{a_{2}} &= \indown_{a_{3}} + \eta^{-1} \inup_{\conj{b_{3}}} \left(\eta \odown_{\conj{b_{3}}} \indown_{a_{3}} \right)\\
	&= \osQ_{\conj{b_{3}}} \odown_{\conj{a_{3}}}\\
	\odown_{b_{1}} &= \odown_{b_{2}} + \eta^{-1} \odown_{c} \oup_{a_{2}}\\
	\odown_{b_{2}} &= \indown_{b_{3}}.
	\end{split}
\end{equation}
By applying the duality operation of Section \ref{sec_duality} to each equation above, we produce the rest of the six-way junction equations.

We wish to solve for the outgoing and internal (blue) quantities in terms of the incoming (red) quantities.

\subsubsection{Outgoing Soliton Generating Functions} \label{sec_out_sol}
Starting from $a_{1}$ and moving counter-clockwise around the edge of Fig. \ref{fig:horse}, we have
\begin{equation*}
	\begin{aligned}
	\odown_{a_{1}} &= \left(1 + \inup_{b_{1}} \indown_{\conj{b_{1}}} \indown_{\conj{a_{1}}} \inup_{a_{1}} \right) \left( 1 + \inup_{b_{1}} \indown_{b_{3}}\right) \left(1 + \inup_{\conj{b_{3}}} \indown_{\conj{b_{1}}}\right) \indown_{a_{3}} + \inup_{b_{1}} \indown_{\conj{b_{1}}} \indown_{\conj{a_{1}}} \\
	\odown_{b_{1}} &= \indown_{b_{3}} + \indown_{\conj{b_{1}}} \indown_{\conj{a_{1}}} \inup_{a_{1}} \left( 1+ \inup_{b_{1}} \indown_{b_{3}} \right) \\
	\odown_{\conj{b_{3}}} &= \indown_{\conj{b_{1}}} \\
	\oup_{a_{3}} &= \inup_{a_{1}} \left( 1 + \inup_{b_{1}} \indown_{b_{3}} \right) \left(1+\inup_{\conj{b_{3}}} \indown_{\conj{b_{1}}}\right)  \\
	\oup_{\conj{a_{1}}} &= \inup_{\conj{a_{3}}} \left( 1 + \indown_{\conj{a_{1}}} \inup_{a_{1}} \inup_{b_{1}} \indown_{\conj{b_{1}}} \right) \left( 1 + \inup_{\conj{b_{3}}} \indown_{\conj{b_{1}}}\right) \left( 1 + \inup_{b_{1}} \indown_{b_{3}}\right) + \inup_{a_{1}} \inup_{b_{1}} \indown_{\conj{b_{1}}} \\
	\oup_{\conj{b_1}} &= \inup_{\conj{b_{3}}} + \left( 1 + \inup_{\conj{b_{3}}} \indown_{\conj{b_{1}}} \right) \indown_{\conj{a_{1}}} \inup_{a_{1}} \inup_{b_{1}} \\
	\oup_{b_{3}} &= \inup_{b_{1}}  \\
	\odown_{\conj{a_{3}}} &= \left( 1 + \inup_{\conj{b_{3}}} \indown_{\conj{b_{1}}}\right) \left( 1 + \inup_{b_{1}} \indown_{b_{3}}\right) \indown_{\conj{a_{1}}} .
	\end{aligned}
\end{equation*}

\subsubsection{Outgoing Street Factors} \label{sec_out_deg}
We remark that all outgoing street factors can be expressed in terms of the internal street factors.  Hence, starting from $a_{1}$ and moving counter-clockwise around the edge of the diagram, we have
\begin{equation*}
	\begin{aligned}
		\oQ_{a_{1}} &= \oQ_{c} \oQ_{a_{2}}\\
		\oQ_{b_{1}} &= \oQ_{c} \oQ_{b_{2}} \\
		\oQ_{\conj{b_{3}}} &= \oQ_{\conj{b_{2}}} \\
		\oQ_{a_{3}} &= \oQ_{a_{2}} \\
		\oQ_{\conj{a_{1}}} &=  \oQ_{c} \oQ_{\conj{a_{2}}}  \\
		\oQ_{\conj{b_{1}}} &= \oQ_{c} \oQ_{\conj{b_{2}}}  \\
		\oQ_{b_{3}} &= \oQ_{b_{2}} \\
		\oQ_{\conj{a_{3}}} &= \oQ_{\conj{a_{2}}} .\\
	\end{aligned}
\end{equation*}

\subsubsection{Internal Street Factors} \label{sec_int_deg}
We now state the internal street factors in terms of the incoming soliton generating functions. These equations follow from \eqref{eq:sixwayhorse} and are:

\begin{equation*}
	\begin{aligned}
		\osQ_{c} &= 1 + \inup_{a_{1}} \inup_{b_{1}} \indown_{\conj{b_{1}}} \indown_{\conj{a_{1}}} \\
		\osQ_{a_{2}} &= 1 + \inup_{a_{1}} \oQ_{b_{2}} \oQ_{\conj{b_{2}}}  \indown_{a_{3}} \\
		\osQ_{\conj{b_{2}}} &= 1 + \inup_{\conj{b_{3}}} \indown_{\conj{b_{1}}} \\
		\osQ_{\conj{a_{2}}} &= 1 + \inup_{\conj{a_{3}}} \oQ_{b_{2}} \oQ_{\conj{b_{2}}} \indown_{\conj{a_{1}}} \\
		\osQ_{b_{2}} &= 1 + \inup_{b_{1}} \indown_{b_{3}}.
	\end{aligned}
\end{equation*}
By applying the closure map $\cl$ one produces the corresponding $Q_{p}$ functions.

\begin{remark}
	We note that in all the equations of sections \ref{sec_out_sol} - \ref{sec_out_deg} there is an almost magical cancellation of the half-twists $\eta$; this cancellation will ultimately ensure that the coefficients of the degeneracy generating functions $Q_p$ (as polynomials in some formal variable $X_{\widehat{\gamma_c}}$, yet to be identified) are all positive.

\end{remark}

\subsubsection*{Special Cases}
We now cite two important special cases of incoming data for a horse.

\begin{definition}[Definitions]\
	\begin{enumerate}
		\item A lower-sourced horse is a horse along with exactly ``two-sources from below," i.e. it is a horse restricted to the subset of $\incdat$ where a point in $\incdat^{+}$ is specified:
		\begin{equation}
			\incdat_{\operatorname{LSH}}^{+} = \left\{ \left(
			\begin{array}{l}
			\inup_{a_{1}} = \inX_{a}\\
			\inup_{b_{1}} = \inX_{b}\\
			\inup_{\conj{a_{3}}} = 0\\
			\inup_{\conj{b_{3}}} = 0.
			\end{array}
			 \right) \right\} \subset \incdat^{+}.
			\label{eq:lsh_cond}
		\end{equation} \\

		\item An upper-sourced horse is dual to a lower-sourced horse, i.e. it is a horse restricted to the subset of $\incdat$ where a point in $\incdat^{-}$ is specified:
		\begin{equation}
			\incdat_{\operatorname{USH}}^{-} = \left\{ \left(
			\begin{array}{l}
			\indown_{\conj{a_{1}}} = \inX_{\conj{a}}\\
			\indown_{\conj{b_{1}}} = \inX_{\conj{b}}\\
			\indown_{a_{3}} = 0\\
			\indown_{b_{3}} = 0.
			\end{array}
			 \right) \right \} \subset \incdat^{-}.
			\label{eq:ush_cond}
		\end{equation}

	\end{enumerate}

\end{definition}

\begin{remark}
	Inserting the lower-sourced horse conditions into the equations of Section \ref{sec_int_deg},  the most important of the resulting equations are
	\begin{equation}
		\oQ_{\overline{a_{2}}} = \oQ_{\overline{b_{2}}} = 1;
		\label{eq:lsh_Q_simp_1}
	\end{equation}
	which, furthermore, via (\ref{sec_out_deg}) require
	\begin{equation}
		\oQ_{\overline{a_{3}}} = \oQ_{\overline{b_{3}}} = 1.
		\label{eq:lsh_Q_simp_2}
	\end{equation}
	 The upper-sourced horse conditions yield the dual equations,
	\begin{equation}
		\oQ_{a_2} = \oQ_{a_3} = \oQ_{b_2} = \oQ_{b_3} =  1.
		\label{eq:ush_Q_simp}
	\end{equation}
\end{remark}

With this technology, we can define an $m$-herd on an arbitrary oriented real surface $C$ as a collection of $m$-horses glued together using the relations (\ref{eq:horse_glue}), beginning with a lower-sourced horse coming from a pair of branch points, and ending with an upper-sourced horse near another pair of branch points (which, for the purposes of this paper, we will take to be disjoint from the lower-sourced branch points).

\begin{definition}
	Let $N$ be a spectral network subordinate to some branched cover $\Sigma \rightarrow C$ and let $H \subset N$ be the set of two-way streets of $N$.  Then $N$ is an $m$-herd if the following conditions are satisfied.

	\begin{description}
\labitem{\textit{Horses}}{cond_horses}There exists a collection of open embedded disks $\{U_{l}\}_{l=1}^{m} \subset C'$ forming a covering of $H$, with $U_{l} \cap U_{k} \neq \emptyset$ iff $l = k \pm 1$, and each $N \cap U_{l}$ is:
\begin{itemize}
		\item a lower-sourced horse if $l=1$,
		\item a horse if $1<l<m$,
		\item an upper-sourced horse if $l=m$.
	\end{itemize}

\labitem{\textit{Gluing}}{cond_gluing} Each horse satisfies particular gluing conditions: let $p^{(l)}$ denote a horse street\footnote{Using our previous naming convention: $p \in \{a_{1},a_{2},a_{3},b_{1},b_{2},b_{3},c,\conj{a_{1}},\conj{a_{2}},\conj{a_{3}},\conj{b_{1}}, \conj{b_{2}},\conj{b_{3}}\}$.} on $N \cap U_{l}$.  Then, for $l = 2, \cdots, m -1$, we have the conditions
	\begin{equation}
		\begin{aligned}
			a_{1}^{(l)} &= a_{3}^{(l-1)}\\
			b_{1}^{(l)} &= b_{3}^{(l-1)}\\
			\conj{a_{1}}^{(l)} &= \conj{a_{3}}^{(l+1)}\\
			\conj{b_{1}}^{(l)} &= \conj{b_{3}}^{(l+1)}.
		\end{aligned}
		\tag{\ref{eq:horse_glue}}
	\end{equation}

\labitem{\textit{No Holes}}{cond_noholes} For $l = 1, \cdots, m -1 $, the oriented loops traced out by the words
	\begin{itemize}
		\item $\left(\conj{a}_{2}^{(l)} \right) \left(\conj{b}_{1}^{(l)}
\right) \left(a_{2}^{(l+1)} \right)^{-1} \left(b_{3}^{(l)}\right)^{-1}$,

		\item $\left(\conj{b}_{2}^{(l)} \right) \left(\conj{a}_{1}^{(l)}
\right) \left(b_{2}^{(l+1)} \right)^{-1}\left(a_{3}^{(l)} \right)^{-1}$
	\end{itemize}

are each the oriented boundary of (separate) disks on $C'$ (see Fig.~\ref{fig:noholes}).
	\end{description}

\end{definition}

\begin{remark}[Remarks]\
	\begin{itemize}
		\item Note that a $1$-herd is the spectral network for a saddle: indeed, via the above definition it consists of a single horse which is both lower and upper-sourced.  The picture of a saddle is formed by viewing only the two-way streets remaining after ``removing" the horse streets constrained to be one-way according to (\ref{eq:lsh_Q_simp_1}) - (\ref{eq:ush_Q_simp}).

		\item Let $\incdat^{\pm}(l)$ ($\outdat^{\pm}(l)$) be the domain of incoming (range of outgoing) data associated to the $l$th horse of an $m$-herd. Via the definition, $\incdat^{+}(1)$ and $\incdat^{+}(m)$ are specified by the lower sourced horse conditions (\ref{eq:lsh_cond}) and upper-sourced horse conditions (\ref{eq:ush_cond}) respectively:
		\begin{equation}
			\begin{aligned}
				\incdat^{+}(1) &= \left\{ \left(
					\begin{array}{l}
					\inup_{a_{1}}^{(1)} = \inX_{a}\\
					\inup_{b_{1}}^{(1)} = \inX_{b}\\
					\inup_{\conj{a_{3}}}^{(1)} = 0\\
					\inup_{\conj{b_{3}}}^{(1)} = 0.
					\end{array}
					\right)
					\right \}, \\
					\incdat^{-}(m) &= \left\{ \left(
					\begin{array}{l}
					\indown_{\conj{a_{1}}}^{(m)} = \inX_{\conj{a}}\\
					\indown_{\conj{b_{1}}}^{(m)} = \inX_{\conj{b}}\\
					\indown_{a_{3}}^{(m)} = 0\\
					\indown_{b_{3}}^{(m)} = 0.
					\end{array}
					\right) \right \}.
			\end{aligned}
			\label{eq:data_initial}
		\end{equation}
		 Further, for $l = 2, \, \cdots, m-1$, the gluing conditions (\ref{eq:horse_glue}) force\footnote{We have omitted the parallel transport map (on the RHS of (\ref{eq:data_glue})), detailed in Section \ref{app:global}, that transports spectral data on the $(l-1)$th horse to the $l$th horse.}
		\begin{equation}
			\begin{aligned}
				\incdat^{+}(l) &= \outdat^{+}(l-1), \\
				 \incdat^{-}(l) &= \outdat^{-}(l+1).
			 \end{aligned}
		\label{eq:data_glue}
		\end{equation}
In fact, as we will discover, all spectral data on an $m$-herd can be determined recursively from (\ref{eq:data_glue}) using the initial conditions (\ref{eq:data_initial}).

		\item The technical \ref{cond_noholes} condition excludes cases where there are  ``holes" between adjacent streets when gluing together horses.  This condition is essential for our proof of Prop. \ref{prop_l}, as such holes create obstructions to auxiliary streets introduced in the proof.  Furthermore, the \ref{cond_noholes} condition is utilized in Prop. \ref{prop_Q} in order to produce an explicit expression for the charge $\widehat{\gamma}_{c}$ (defined in (\ref{eq:gammac_def})) that appears in the formal variable $z$, but the condition is not necessary to derive the algebraic equation (\ref{eq:P}).\footnote{In particular, the \ref{cond_noholes} condition is used in the definition of the parallel transport maps $\rho_{*}^{(l,l \pm 1)}$ of Section \ref{app:global}.  One could use a more general notation for parallel transport in a situation without the \ref{cond_noholes} condition and the proof of the algebraic equation would follow similarly, although, the final expression for $z$ would be modified.}
	\end{itemize}
\end{remark}

\begin{figure}
	\begin{center}
		 \includegraphics[scale=0.3]{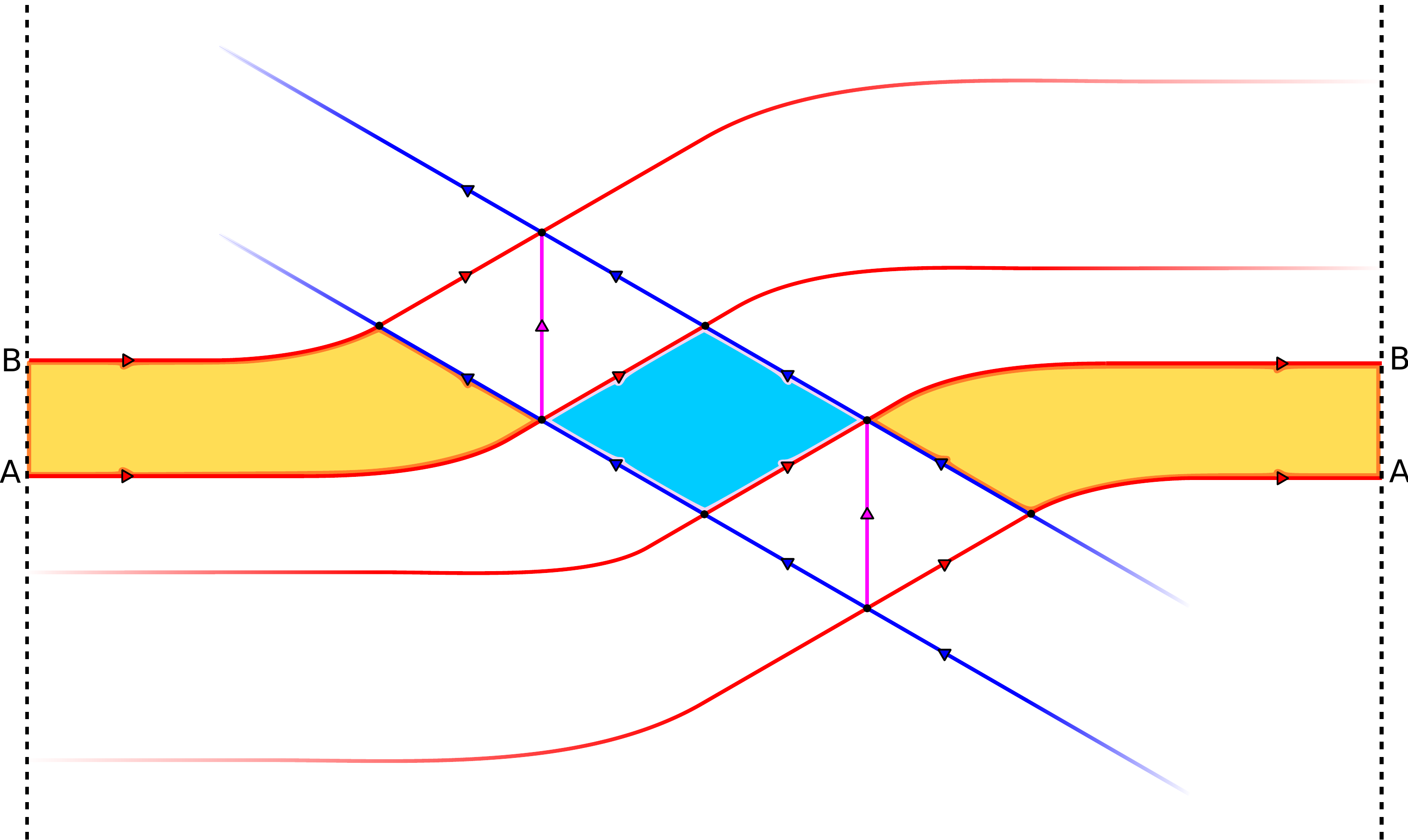}
		\caption{A picture of two horses (cf. Fig.~\ref{fig:horse}) glued together using the \ref{cond_gluing} conditions and satisfying the \ref{cond_noholes} condition; the dotted lines are identified, and we assign the ``horse-indices" $l$ and $l+1$ to the bottom and top horses respectively.  The aqua-blue region is a disk with boundary traced out by the word $\left(\conj{a}_{2}^{(l)} \right) \left(\conj{b}_{1}^{(l)}
\right) \left(a_{2}^{(l+1)} \right)^{-1} \left(b_{3}^{(l)}\right)^{-1}$; the yellow region is a disk with boundary traced out by the word $\left(\conj{b}_{2}^{(l)} \right) \left(\conj{a}_{1}^{(l)}
\right) \left(b_{2}^{(l+1)} \right)^{-1}\left(a_{3}^{(l)} \right)^{-1}$.  Two examples for which the \ref{cond_noholes} condition fails can be pictured by either inserting a puncture, or connect summing with a torus (inserting a ``handle"), inside of the colored regions. \label{fig:noholes}}
	\end{center}
\end{figure}

\subsection{A Global Interlude} \label{app:global}
The following is a technical subsection dedicated to a proper definition of the symbols $\rho_{*}^{(k,l)}$ that appear throughout the proof of Prop. \ref{prop_Q}.  Readers who wish to avoid this technical detour may skip this section and interpret the symbols $\rho_{*}^{(l,l \pm 1)}$ as parallel transport maps along an appropriate path, from the $l$th horse to the $\left(l \pm 1 \right)$th horse,
 along the graph of the $m$-herd living on $C'$; further, the $R^{(k,l)}$ can be replaced by parallel transport maps from the $l$th horse to the $k$th horse.

First, we will define the local system of soliton charges over $\twid{C'}$.

\begin{definition}
	Let $\mathfrak{s}: \bigcup_{\twid{z} \in \twid{C'}}
\twid{\Gamma}(\twid{z},-\twid{z}) \rightarrow \twid{C'}$
be the projection map with fibers $\mathfrak{s}^{-1}(\twid{z})
= \twid{\Gamma}(\twid{z},-\twid{z})$.

\end{definition}

\begin{remark}
	$\mathfrak{s}$ defines a local system of $\twid{\Gamma}$-sets (a locally constant sheaf of $\twid{\Gamma}$-sets) over $\twid{C'}$, when equipped with a parallel transport map defined by a lifted version of the parallel transport of solitons (\ref{eq:soliton_trans}).  More explicitly, for any path $\ell :[0,1] \rightarrow \twid{C'}$, the parallel transport map $\twid{P}_{\ell}: \twid{\Gamma}\left(\ell(0),-\ell(0) \right) \rightarrow \twid{\Gamma} \left( \ell(1), -\ell(1) \right)$ is given by
	\begin{align}
		\twid{P}_{\ell} s &= \left(s' + [ \ell \{j \} ] - [ \ell \{i \} ] \right)\, \operatorname{mod \:} 2H,\, s \in \twid{\Gamma}_{ij}(\ell(0), -\ell(0)).
		\label{eq:local_sys_mon}
	\end{align}
	where
	\begin{itemize}
		\item $s'$ is a lift of $s$ to a relative homology cycle\footnote{Recall from (\ref{eq:twid_gamma_def}): $\twid{\Gamma}(\twid{z}, -\twid{z})$ is defined as a quotient of the the subset $G(\twid{z}, - \twid{z})$ (consisting of relative homology classes on $\twid{\Sigma}$).} on $\twid{\Sigma}$,

		\item  $\ell\{n\}$ is the lift of $\ell$ to a path on $\twid{\Sigma}$ given by lifting $\ell(0)$ to sheet $n$,

		\item $[ \ell\{n\} ]$ is the relative homology class of $\ell\{n\}$

		\item $\left( \cdot \right) \operatorname{mod \:} 2H: G( \ell(1), -\ell(1)) \rightarrow \twid{\Gamma}(\ell(1), -\ell(1))$ is the quotient map (where the subset of relative homology classes $G(\ell(1), -\ell(1))$ is defined in (\ref{eq:G_def})).
		\end{itemize}
By construction, $\twid{P}_{\ell}$ only depends on the homotopy class of $\ell$ (rel endpoints).
\end{remark}

Let $\xi: \twid{C'} \rightarrow C'$ be the unit tangent bundle projection map (previously denoted $\xi^{C'}$).  We now make an important observation.

\begin{remark}[Observation]
	The monodromy
of $\twid{P}_\ell$
around any loop that wraps the circle fibers of $\xi$ is trivial.  I.e., let $z \in C'$ and choose $\ell: S^{1} \rightarrow \left(\xi \right)^{-1}(z) \subset \twid{C'}$ to be a closed loop supported on the circle fiber $\xi^{-1}(z)$, then the monodromy $\twid{P}_{\ell}$ is the identity map.
\end{remark}
\begin{proof}
	The proof is immediate:  if $\ell$ is such a loop, then for any sheet $n$, we have $\cl\left([\ell\{n\}] \right) = H$; the result follows from (\ref{eq:local_sys_mon}).
\end{proof}

\begin{definition}
	Let $S$ be any topological space; $\pi_{1}(S; z_{1}, z_{2})$ is the set of homotopy (rel endpoints) classes of paths $p:[0,1] \rightarrow S$ with $p(0) = z_{1}$ and $p(1) = z_{2}$.
\end{definition}

\begin{corollary} \label{cor_transport_descent}
	Let $\ell: [0,1] \rightarrow \twid{C'}$ be a path.  Then $\twid{P}_{\ell}: \twid{\Gamma}(\ell(0), - \ell(0)) \rightarrow \twid{\Gamma}(\ell(1), - \ell(1))$ is completely specified by the homotopy class (rel endpoints) of the projected path $\xi \circ \ell: [0,1] \rightarrow C'$.
\end{corollary}
	In particular, given $q \in \pi_{1}(C',z_{1}, z_{2})$ along with lifts $\twid{z}_{1} \in \xi^{-1}(z_{1}),\, \twid{z}_{2} \in \xi^{-1}(z_{2})$, we may associate a parallel transport map $\twid{P}_{\ell}: \twid{\Gamma}(\twid{z}_{1}, - \twid{z}_{1}) \rightarrow \twid{\Gamma}(\twid{z}_{2}, - \twid{z}_{2})$ where $\ell: [0,1] \rightarrow \twid{C'}$ is a lift of any path representative of the class $q$ such that $\ell(0) = \twid{z}_{1},\, \ell(1) = \twid{z}_{2}$.  By the corollary this association $(q,\twid{z}_{1}, \twid{z}_{2}) \rightsquigarrow \twid{P}_{\ell}$ is well-defined.

\begin{definition}
	Let $q \in \pi_{1}(C';z_{1}, z_{2})$ and $\twid{z}_{1} \in \xi^{-1}(z_{1}),\, \twid{z}_{2} \in \xi^{-1}(z_{2})$,  then $\twid{P}_{(q,\twid{z}_{1},\twid{z}_{2})}: \twid{\Gamma}(\twid{z}_{1}, -\twid{z}_{1}) \rightarrow \twid{\Gamma}(\twid{z}_{2}, - \twid{z}_{2})$ is the unique parallel transport map assigned to $(q,\twid{z}_{1}, \twid{z}_{2})$.
\end{definition}

	To simplify matters of computation, without ignoring global issues, we will develop a notation, suitable to combinatorics, for parallel transport on an $m$-herd.  As each horse is embedded in a contractible region of $C$, it suffices to keep track of parallel transport of paths \textit{between the horses} of an $m$-herd: our notation need not keep track of parallel transport between points in an individual horse as suggested by the following remark.

	\begin{remark}
	Let $\{ U_{l} \}_{l=1}^{m}$ be an open cover of disks (on $C'$) satisfying the \ref{cond_horses} condition for an $m$-herd, then all paths running between points $z_{1},\,z_{2} \in U_{l}$ and contained within $U_{l}$ are homotopic (rel endpoints).  Thus, by Cor. \ref{cor_transport_descent}, for each pair of points $\twid{z}_{1} \in \xi^{-1}(z_{1}),\, \twid{z}_{2} \in \xi^{-1}(z_{2})$,  there is a unique parallel transport map assigned to all paths running from $\twid{z}_{1}$ to $\twid{z}_{2}$ and contained in $\xi^{-1}(U_{l})$.
	\end{remark}

	\begin{figure}
	\begin{center}
		 \includegraphics[scale=0.4]{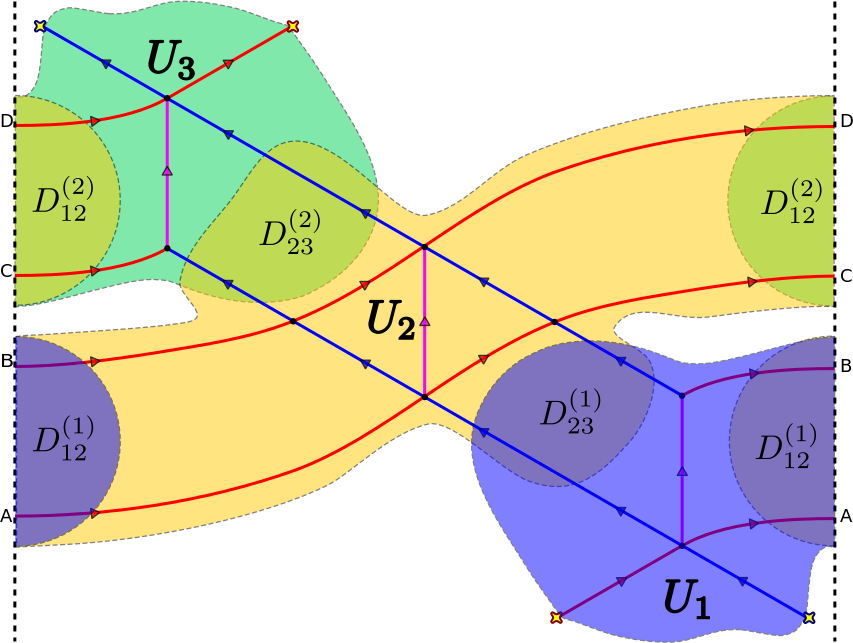}
		\caption{A 3-herd, on $C = \mathbb{R} \times S^{1}$, equipped with a (two-way street) cover $\{U_{l}\}_{l=1}^{3}$ satisfying the \ref{cond_horses} condition along with (\ref{eq:int_cond_1})-(\ref{eq:int_cond_2}). \label{fig:disk_int}}
	\end{center}
\end{figure}

	Now, let us turn our attention to parallel transport of paths
running between horses; in particular, paths contained in
$U_{l} \cup U_{l+1}$ for some $l = 1,\cdots, m-1$.
First, note that each non-vanishing intersection
$U_{l} \cap U_{l + 1},\, l= 1, \cdots, m-1$, will consist of
some number of disconnected disks.  However, on an $m$-herd,
the \ref{cond_noholes} condition
allows us to modify our cover
such that $U_{l} \cap U_{l + 1}$ contains exactly \textit{two} components:
\begin{equation}
	U_{l} \cap U_{l + 1} = D^{(l)}_{12} \sqcup D^{(l)}_{23},
	\label{eq:int_cond_1}
\end{equation}
where the $D^{(l)}_{ij}$ are disks such that for $l=1, \cdots, m - 1$
\begin{equation}
	\begin{aligned}
	\left(\conj{a}_{1}^{(l)} = \conj{a}_{3}^{(l+1)} \right) \cap \left(U_{l} \cap U_{l+1} \right) \subset D_{12}^{(l)}, \, \left(a_{3}^{(l)} = a_{1}^{(l+1)} \right) \cap \left(U_{l} \cap U_{l+1} \right) & \subset D_{12}^{(l)}\\
	\left(\conj{b}_{1}^{(l)} = \conj{b}_{3}^{(l+1)} \right) \cap \left(U_{l} \cap U_{l+1} \right) \subset D_{23}^{(l)}, \, \left(b_{3}^{(l)} = b_{1}^{(l+1)} \right) \cap \left(U_{l} \cap U_{l+1} \right) & \subset D_{23}^{(l)};
	\end{aligned}
	\label{eq:int_cond_2}
\end{equation}
an example of such a cover is shown in Fig.~\ref{fig:disk_int} for the case of a 3-herd on $\mathbb{C} = \mathbb{R} \times S^{1}$.  Thus, fixing a pair of points $\twid{z}_{1} \in \xi^{-1} \left(U_{l} \right),\, \twid{z}_{2} \in \xi^{-1}(U_{l+1})$, our interest lies in two homotopy classes (rel endpoints) of paths that run from $\twid{z}_{1}$ to $\twid{z}_{2}$, and are contained in $U_{l} \cup U_{l+1}$.  In particular, denoting these two classes by $q_{12}, q_{23} \in \pi_{1}(U_{l} \cup U_{l+1}; z_{1}, z_{2})$,
\begin{enumerate}
	\item $q_{12}$ has a path representative given by a simple curve running from $\twid{z}_{1}$ to $\twid{z}_{2}$ and passing through $D_{12}^{(l)}$ (but not $D_{23}^{(l)}$) exactly once,

	\item $q_{23}$ has a path representative given by a simple curve running from $\twid{z}_{1}$ to $\twid{z}_{2}$ and passing through $D_{23}^{(l)}$ (but not $D_{12}^{(l)}$) exactly once.
\end{enumerate}

\begin{definition}
Let $z_{1} \in U_{l},\, z_{2} \in U_{l +1}$, and take $q_{ij}$ ($ij \in \{12,23\}$) to be the homotopy classes described above.  Then, for a choice of lifts $\twid{z}_{1} \in \xi^{-1}(z_{1})$ and $\twid{z}_{2} \in \xi^{-1}(z_{2})$,
	\begin{equation}
		\begin{aligned}
	\rho_{ij}^{(l, l+1)}(\twid{z}_{1}, \twid{z}_{2}) &:= \twid{P}_{(q_{ij},\twid{z}_{1},\twid{z}_{2})}: \twid{\Gamma}(\twid{z}_{1}, -\twid{z}_{1}) \rightarrow \twid{\Gamma}(\twid{z}_{2}, -\twid{z}_{2}),\\
		\rho_{ij}^{(l+1,l)}(\twid{z}_{1}, \twid{z}_{2}) &:= \twid{P}_{(q_{ij}^{-1},\twid{z}_{2},\twid{z}_{1})}: \twid{\Gamma}(\twid{z}_{2}, -\twid{z}_{2}) \rightarrow \twid{\Gamma}(\twid{z}_{1}, -\twid{z}_{1}) = \left[\rho_{ij}^{(l,l+1)} (\twid{z}_{1}, \twid{z}_{2}) \right]^{-1}.
		\end{aligned}
		\label{eq:rho_ij_def}
	\end{equation}
\end{definition}

\begin{definition}[Notation]
	In the following computations we will just write $\rho_{ij}^{(l, l+1)}$, dropping the explicit dependence on the endpoints $\twid{z}_{1} \in \xi^{-1} \left(U_{l} \right)$ and $\twid{z}_{2} \in \xi^{-1} \left(U_{l+1} \right)$; this notation will be sufficiently unambiguous for our purposes.  Indeed, let $\twid{w}_{1} \in \xi^{-1} \left(U_{l} \right),\, \twid{w}_{2} \in \xi^{-1}\left(U_{l+1} \right)$ be another choice of endpoints with projections $w_{i} = \xi(\twid{w}_{i}),\, i =1,2$; then, by a remark above, $\exists!$ homotopy classes $q_{1} \in \pi_{1}(U_{l}; w_{1}, z_{1})$ and $q_{2} \in \pi_{1}(U_{l+1},z_{2},w_{2})$ such that
	\begin{align*}
		\rho_{ij}^{(l,l+1)}(\twid{w}_{1}, \twid{w}_{2}) &=  \twid{P}_{(q_{2}, \twid{z}_{2}, \twid{w}_{2})} \left(\rho_{ij}^{(l,l+1)}(\twid{z}_{1}, \twid{z}_{2}) \right) \twid{P}_{(q_{1}, \twid{w}_{1}, \twid{z}_{1})}  ,\: ij \in \{12,23\}.
	\end{align*}
\end{definition}

Now, on an $m$-herd, (\ref{eq:int_cond_2}) indicates that only solitons of type $12$ or $21$ will be transported via $\rho_{12}^{(l,l+1)}$, and only solitons of type $23$ or $32$ will be transported via $\rho_{23}^{(l,l+1)}$.  With this in mind, for the sake of readability, it will prove convenient to make further notation simplifying definitions.
\\
\begin{definition}[Definitions]\
\begin{enumerate}
	\item Let $\twid{z} \in \xi^{-1}(U_{l})$, then
	\begin{equation}
		\rho_{*}^{(l,l + 1 )}a := \left\{
		\begin{array}{ll}
		\rho_{12}^{(l, l + 1)} a & \text{if $a \in \twid{\Gamma}_{12}(\twid{z},-\twid{z}) \cup \twid{\Gamma}_{21}(\twid{z},-\twid{z})$}\\
		\rho_{23}^{(l, l + 1)} a & \text{if $a \in \twid{\Gamma}_{23}(\twid{z},-\twid{z}) \cup \twid{\Gamma}_{32}(\twid{z},-\twid{z})$}
		\end{array}
		\right.,
		\label{eq:rho_*_def}
	\end{equation}
	and $\rho_{*}^{(l-1,l)} := \left(\rho_{*}^{(l-1,l)} \right)^{-1}$.
	\item
		\begin{equation}
					R^{(k,n)} :=
			\left\{
			\begin{array}{ll}
			\rho_{*}^{(n-1, n)} \cdots \rho_{*}^{(k+1, k+2)} \rho_{*}^{(k, k + 1)} & \text{if $k < n$}\\
			\rho_{*}^{(n+1, n)} \cdots \rho_{*}^{(k-1, k-2)} \rho_{*}^{(k, k-1)} & \text{if $n < k$}.
			\end{array}
			\right.
			\label{eq:R_def}
		\end{equation}
	\end{enumerate}
\end{definition}

\begin{remark}[Remarks]\
\begin{enumerate}
	 \item The $\rho_{*}^{(l,k)}$ extend their action to formal variables $X_{a}$ via
			\begin{equation*}
				\rho_{*}^{(l,k)} X_{a} = X_{\rho_{*}^{(l,k)} a} .
			\end{equation*}
		 \item $R^{(k,n)}$ is a parallel transport map, on the local system $\mathfrak{s}$,
from the $k$th horse
 to the $n$th horse associated to a path that passes through each $l$-horse
 between $k$ and $n$ exactly once.  If $R^{(k,n)}$ acts on a soliton of charge $12$ or $21$, this path
 passes through the sets $D_{12}^{(l)}$ (but never $D_{23}^{(l)}$) for $\min\{k,n\} < l < \max\{k,n\}$;
 if $R^{(k,n)}$ acts on a soliton of charge $23$ or $32$ the path passes through the sets $D_{23}^{(l)}$
 (but never $D_{12}^{(l)}$) for $\min\{k,n\} < l < \max\{k,n\}$.
		\end{enumerate}

We make one final observation that will be of use in Section \ref{app:proof_decomp}.

\begin{remark}
Let $\mathfrak{r}: \bigcup_{z \in C} \Gamma(z,z) \rightarrow C'$ be the projection map with $\mathfrak{r}^{-1}(z) = \Gamma(z,z)$;  this forms a local system over $C'$ when equipped with the parallel transport map
\begin{align}
	P_{q} s_{*} &= s_{*} + [ q \{j \} ] - [ q \{i \} ],\, s_{*} \in \Gamma_{ij}(q(0), q(0)).
\end{align}
The parallel transport on $\mathfrak{r}$ is compatible with the parallel transport (\ref{eq:local_sys_mon}) on the local system $\mathfrak{s}$ in the sense that for any $q \in \pi_{1}(C',z_{1},z_{2})$ and $\twid{z}_{1} \in \xi^{-1}(z_{1}),\, \twid{z}_{2} \in \xi^{-1}(z_{2})$, we have
\begin{align}
	\xi^{\Sigma}_{*} \left( \twid{P}_{(q, \twid{z}_{1}, \twid{z}_{2})} s \right) &= P_{q} \left(\xi^{\Sigma}_{*} s \right),
	\label{eq:mon_compat}
\end{align}
	where, recall, $\xi^{\Sigma}: \twid{\Sigma} \rightarrow \Sigma$ is the unit tangent bundle projection map.
\end{remark}
	Now, we may define the analog of the parallel transport operators $R^{(k,n)}$ for $\mathfrak{r}$.
	\begin{definition}
	Let $s_{*} \in \bigsqcup_{ij \in \{12,21,23,32\}} \Gamma_{ij}(z,z)$ for some $z \in U_{k}$, then
	\begin{equation}
		R_{\mathfrak{r}}^{(k,n)} s_{*} := \xi_{*}^{\Sigma} R^{(k,n)} s,
		\label{eq:R_red_def}
	\end{equation}
	where $s \in \bigsqcup_{ij \in \{12,21,23,32\}} \twid{\Gamma}_{ij}(\twid{z},-\twid{z})$ is any lift of $s_{*}$ (i.e. $s_{*} = \xi^{\Sigma}_{*} s$).
\end{definition}
(\ref{eq:mon_compat}) ensures that (\ref{eq:R_red_def}) is a well-defined (lift-independent) statement.

\end{remark}
\subsection{Identifications of Generating Functions}
Using the notation developed in Section \ref{app:global}, we can express (\ref{eq:data_glue}) explicitly as
\begin{equation}
	\begin{array}{lr}
		\begin{aligned}
			\inup_{a_{1}}^{(l)} &= \rho^{(l-1,l)}_{*} \oup_{a_{3}}^{(l-1)} , \\
			\inup_{b_{1}}^{(l)} &= \rho^{(l-1,l)}_{*} \oup_{b_{3}}^{(l-1)} , \\
			\inup_{\conj{a_{3}}}^{(l)} &= \rho^{(l-1,l)}_{*} \oup_{\conj{a_{1}}}^{(l-1)} , \\
			\inup_{\conj{b_{3}}}^{(l)} &= \rho^{(l-1,l)}_{*} \oup_{\conj{b_{1}}}^{(l-1)} , \\
		\end{aligned} &
		\begin{aligned}
			\indown_{\conj{a_{1}}}^{(l)} &= \rho^{(l+1,l)}_{*} \odown_{\conj{a_{3}}}^{(l+1)} , \\
			\indown_{\conj{b_{1}}}^{(l)} &= \rho^{(l+1,l)}_{*} \odown_{\conj{b_{3}}}^{(l+1)} , \\
			\indown_{a_{3}}^{(l)} &= \rho^{(l+1,l)}_{*} \odown_{a_{1}}^{(l+1)} ,\\
			\indown_{b_{3}}^{(l)} &= \rho^{(l+1,l)}_{*} \odown_{b_{1}}^{(l+1)} .
		\end{aligned}
	\end{array}
	\label{eq:sol_rec}
\end{equation}
In particular,
\begin{equation}
	\begin{array}{lr}
		\begin{aligned}
		\oQ_{a_{1}}^{(l)} &= \oQ_{a_{3}}^{(l-1)},\\
		\oQ_{b_{1}}^{(l)} &= \oQ_{b_{3}}^{(l-1)},
		\end{aligned} &
		\begin{aligned}
		\oQ_{\conj{a_{1}}}^{(l)} &= \oQ_{\conj{a_{3}}}^{(l+1)},\\
		\oQ_{\conj{b_{1}}}^{(l)} &= \oQ_{\conj{b_{3}}}^{(l+1)}.
		\end{aligned}
	\end{array}
	\label{eq:Q_rec}
\end{equation}

\subsection{Proof of Proposition \ref{prop_Q}} \label{app_prop_Q_pf}

\subsubsection{Proof of Equations (\ref{eq:Q_red})}
Using the recursion relations $(\ref{eq:Q_rec})$, in conjunction with the equations listed in Sections \ref{sec_int_deg} and \ref{sec_out_deg}, we first solve for the internal street factors $\oQ_{a_{2}}^{(l)},\, \oQ_{\conj{a_{2}}}^{(l)},\, \oQ_{b_{2}}^{(l)},\,\oQ_{\conj{b_{2}}}^{(l)}$ in terms of street factors on the lower/upper-sourced horses at $l=1$ or $l=m$.  As we noticed in Section \ref{sec_out_deg}, all other street factors can be written in terms of the internal ones.

Now, via (\ref{eq:Q_rec}), and the equations of Section (\ref{sec_out_deg}),
\begin{align}
	\oQ_{a_{2}}^{(l)} &= \oQ_{a_{3}}^{(l)} \nonumber \\
	 &= \oQ_{a_{1}}^{(l+1)} \nonumber\\
	 &= \oQ_{c}^{(l+1)} \oQ_{a_{2}}^{(l+1)}.
\label{eq:Q_a2_rec}
\end{align}
Similarly, we find
\begin{align}
	\oQ_{\conj{a_{2}}}^{(l)} &= \oQ_{c}^{(l-1)}\oQ_{\conj{a_{2}}}^{(l-1)}
	\label{eq:Q_a2d_rec} \\
	\oQ_{b_{2}}^{(l)} &= \oQ_{c}^{(l+1)}\oQ_{b_{2}}^{(l+1)}
	\label{eq:Q_b2_rec} \\
	\oQ_{\conj{b_{2}}}^{(l)} &= \oQ_{c}^{(l-1)}\oQ_{\conj{b_{2}}}^{(l-1)}
	\label{eq:Q_b2d_rec}.
\end{align}
This leads us to the following.
\begin{lemma} \label{lem_Q_red}
	For $l=1,\dots, m$, we have
	\begin{align}
		\oQ_{a_{2}}^{(l)} &= \oQ_{b_{2}}^{(l)} = \prod_{r=l+1}^{m+1} \oQ_{c}^{(r)} \label{eq:claim_1}\\
		\oQ_{\conj{a_{2}}}^{(l)} &= \oQ_{\conj{b_{2}}}^{(l)} = \prod_{r=0}^{l-1} \oQ_{c}^{(r)} \label{eq:claim_2}.
	\end{align}
	with the convention that $\oQ_{c}^{(m+1)} = \oQ_{c}^{(0)} = 1$.
\end{lemma}
\begin{proof}
	From the upper-sourced horse conditions (\ref{eq:ush_cond}) we have
	\begin{equation}
		\oQ_{a_{2}}^{(m)} = \oQ_{b_{2}}^{(m)} = 1;
		\tag{\ref{eq:ush_Q_simp}}
	\end{equation}
	so (\ref{eq:claim_1}) follows via (\ref{eq:Q_a2_rec}) and (\ref{eq:Q_b2_rec}).  Similarly, from the lower-sourced horse conditions (\ref{eq:lsh_cond}) we have
	\begin{equation}
		\oQ_{\conj{a_{2}}}^{(1)} = \oQ_{\conj{b_{2}}}^{(1)} = 1;
		\tag{\ref{eq:lsh_Q_simp_1}}
	\end{equation}
	so (\ref{eq:claim_2}) follows via (\ref{eq:Q_a2d_rec}) and (\ref{eq:Q_b2d_rec}).
\end{proof}

To reduce (\ref{eq:claim_1}) - (\ref{eq:claim_2}) further, we must compute some soliton generating functions.

\subsection*{Computing $\Upsilon_{b_{1}}^{(l)}$}
	Via (\ref{eq:sol_rec})
	\begin{align}
	\inup_{b_{1}}^{(l)} &= \rho^{(l-1,l)}_{*} \oup_{b_{3}}^{(l-1)} \nonumber\\
	 &= \rho^{(l-1,l)}_{*} \inup_{b_{1}}^{(l-1)}.
\label{eq:tau_b1_rec}
\end{align}
Thus, propagating the lower sourced horse conditions
\eqref{eq:lsh_cond} through this recursion relation,
\begin{align}
	\inup_{b_{1}}^{(l)} &= \left( \prod_{r=1}^{l}\rho^{(r-1,r)}_{*} \right) \inX_{b} \\
	&=R^{(1,l)}\inX_{b}.
\label{eq:tau_b1_sol}
\end{align}

\subsection*{Computing $\Delta_{\conj{b_{1}}}^{(l)}$}
The idea is dual to above; indeed
\begin{align}
	\indown_{\conj{b_{1}}}^{(l)} &= \rho^{(l+1,l)}_{*} \odown_{\conj{b_{3}}}^{(l+1)}  \nonumber\\
	 &= \rho^{(l+1,l)}_{*} \indown_{\conj{b_{1}}}^{(l+1)}.
\label{eq:nu_b1d_rec}
\end{align}
Using the upper-sourced horse conditions \eqref{eq:ush_cond},
\begin{align}
	\indown_{\conj{b_{1}}}^{(l)} &= \left(\prod_{r=l}^{m} \rho^{(r+1,r)}_{*} \right) \inX_{\conj{b}} \nonumber\\
	&=R^{(m,l)} \inX_{\conj{b}}.
\label{eq:nu_b1d_sol}
\end{align}

\subsection*{Computing $\Upsilon_{a_{1}}^{(l)}$}
Via (\ref{eq:sol_rec}) and the equation for $\oup_{a_{3}}$ in Section \ref{sec_out_sol},
\begin{align}
	\inup_{a_{1}}^{(l)} &= \rho_{*}^{(l-1,l)}\oup_{a_{3}}^{(l-1)} \nonumber\\
	 &= \rho_{*}^{(l-1,l)} \inup_{a_{1}}^{(l-1)} \left(1+\inup_{\conj{b_{3}}}^{(l-1)} \indown_{\conj{b_{1}}}^{(l-1)} \right) \left(1+ \inup_{b_{1}}^{(l-1)}  \indown_{b_{3}}^{(l-1)} \right) \nonumber\\
	 &= \rho_{*}^{(l-1,l)} \inup_{a_{1}}^{(l-1)} \oQ_{\conj{b_{2}}}^{(l-1)} \oQ_{b_{2}}^{(l-1)} .
	\label{eq:tau_a1_rec}
\end{align}
Using the lower-sourced horse conditions \eqref{eq:lsh_cond},
\begin{align}
\inup_{a_{1}}^{(l)} &= \left(\prod_{r=1}^{l} \rho_{*}^{(r-1,r)}  \inX_{a} \right) \left( \prod_{r=1}^{l} \rho_{*}^{(r-1,r)}  \oQ_{\conj{b_{2}}}^{(r-1)} \oQ_{b_{2}}^{(r-1)} \right) \nonumber \\
 &= R^{(1,l)} \inX_{a} \left(\prod_{r=0}^{l-1}\oQ_{\conj{b_{2}}}^{(r)} \oQ_{b_{2}}^{(r)}\right)
\label{eq:tau_a1_sol_part}.
\end{align}

\subsection*{Computing $\Delta_{\conj{a_{1}}}^{(l)}$}
Again, the computation is dual to that for $\Upsilon_{a_{1}}^{(l)}$,
\begin{align}
	\indown_{\conj{a_{1}}}^{(l)} &= \rho_{*}^{(l+1,l)}\odown_{\conj{a_{3}}}^{(l+1)} \nonumber\\
	 &= \rho_{*}^{(l+1,l)}\left[ \left(1+\inup_{\conj{b_{3}}}^{(l+1)} \indown_{\conj{b_{1}}}^{(l+1)}\right) \left(1+\inup_{b_{1}}^{(l+1)} \indown_{b_{3}}^{(l+1)}\right) \indown_{\conj{a_{1}}}^{(l+1)} \right] \nonumber\\
	 &= \rho_{*}^{(l+1,l)} \oQ_{\conj{b_{2}}}^{(l+1)} \oQ_{b_{2}}^{(l+1)} \indown_{\conj{a_{1}}}^{(l+1)}.
\label{eq:nu_a1d_rec}
\end{align}
So, using the upper-sourced horse conditions \eqref{eq:ush_cond},
\begin{align}
	\indown_{\conj{a_{1}}}^{(l)} &= \left(\prod_{r=l}^{m} \rho_{*}^{(r+1,l)} \oQ_{\conj{b_{2}}}^{(r+1)} \oQ_{b_{2}}^{(r+1)}\right) \inX_{a} \nonumber \\
	 &= \left( \prod_{r=l+1}^{m+1} \oQ_{\conj{b_{2}}}^{(r)} \oQ_{b_{2}}^{(r)} \right)R^{(m,l)} \inX_{a}.
	\label{eq:nu_a1d_sol_part}
\end{align}

These computations lead us to the following key lemma that allows all street factors $Q_p$ to be reduced to powers of a single function.

\begin{lemma} \label{lem_Qc_equiv}
	\begin{equation*}
		\oQ_{c}^{(l)} = \oQ_{c}^{(1)},\,\forall l=1,\cdots, m.
	\end{equation*}
	\end{lemma}
\begin{proof}
	Recall (\ref{eq:tau_b1_rec}), (\ref{eq:nu_b1d_rec}), (\ref{eq:tau_a1_rec}), and (\ref{eq:nu_a1d_rec})
	\begin{align*}
		\inup_{b_{1}}^{(l)} &= \rho^{(l-1,l)}_{*}\inup_{b_{1}}^{(l-1)}\\
		\indown_{\overline{b_{1}}}^{(l)} &= \rho^{(l+1,l)}_{*}\indown_{\overline{b_{1}}}^{(l+1)}\\
		\inup_{a_{1}}^{(l)} &= \rho_{*}^{(l-1,l)}\oQ_{\overline{b_{2}}}^{(l-1)}\oQ_{b_{2}}^{(l-1)} \inup_{a_{1}}^{(l-1)}\\
		\indown_{\overline{a_{1}}}^{(l)} &= \rho_{*}^{(l+1,l)}\oQ_{\overline{b_{2}}}^{(l+1)}\oQ_{b_{2}}^{(l+1)} \indown_{\overline{a_{1}}}^{(l+1)};
	\end{align*}
	we can rewrite the equations for $\indown_{\overline{b_{1}}}^{(l)}$ and $\indown_{\overline{a_{1}}}^{(l)}$ as
	\begin{align*}
		\indown_{\overline{b_{1}}}^{(l)} &= \rho_{*}^{(l-1,l)} \indown_{\overline{b_{1}}}^{(l-1)}\\
		\indown_{\overline{a_{1}}}^{(l)} &= \rho_{*}^{(l-1,l)} \frac{\indown_{\overline{a_{1}}}^{(l-1)}}{\oQ_{\overline{b_{2}}}^{(l)} \oQ_{b_{2}}^{(l)}}.
	\end{align*}
	Using the equation for $\osQ_c$ in Section \ref{sec_int_deg}
	\begin{align*}
		\osQ_{c}^{(l)} &= 1 + \inup_{a_{1}}^{(l)} \inup_{b_{1}}^{(l)} \indown_{\overline{b_{1}}}^{(l)} \indown_{\overline{a_{1}}}^{(l)}\\
		 &= 1 + \left( \rho_{*}^{(l-1,l)} \oQ_{\overline{b_{2}}}^{(l-1)} \oQ_{b_{2}}^{(l-1)} \inup_{a_{1}}^{(l-1)} \right) \left(\rho^{(l-1,l)}_{*} \inup_{b_{1}}^{(l-1)} \right) \left(\rho_{*}^{(l-1,l)} \indown_{\overline{b_{1}}}^{(l-1)}\right) \left(\rho_{*}^{(l-1,l)} \frac{\indown_{\overline{a_{1}}}^{(l-1)}}{\oQ_{\overline{b_{2}}}^{(l)} \oQ_{b_{2}}^{(l)}} \right)\\
		 &= 1 + \left( \osQ_{c}^{(l-1)}-1 \right) \left(\frac{\oQ_{\overline{b_{2}}}^{(l-1)} \oQ_{b_{2}}^{(l-1)}}{\oQ_{\overline{b_{2}}}^{(l)} \oQ_{b_{2}}^{(l)}} \right);
	\end{align*}
	where, on the last line, the cancellation of the $\rho_{*}^{(l-1,l)}$ (parallel transport) actions\footnote{This is consistent with the fact that, according to (\ref{eq:local_sys_mon}), parallel transport acts trivially on charges of type $ii$.} can be seen by working through its definition in equations (\ref{eq:local_sys_mon}), (\ref{eq:rho_ij_def}), and (\ref{eq:rho_*_def}).   Applying the closure map we obtain
	\begin{align*}
		\oQ_{c}^{(l)} &= 1 + \left(\oQ_{c}^{(l-1)}-1 \right) \left(\frac{\oQ_{\overline{b_{2}}}^{(l-1)} \oQ_{b_{2}}^{(l-1)}}{\oQ_{\overline{b_{2}}}^{(l)} \oQ_{b_{2}}^{(l)}} \right).
	\end{align*}
	Using (\ref{eq:claim_1}) and (\ref{eq:claim_2}), then
	\begin{align*}
		\oQ_{c}^{(l)} &= 1 + \left(\oQ_{c}^{(l-1)} - 1 \right) \left( \frac{\prod_{r \neq l-1}\oQ_{c}^{(r)}}{\prod_{r \neq l}\oQ_{c}^{(r)}} \right)\\
		 &= 1 + \left( \oQ_{c}^{(l-1)} - 1 \right) \frac{\oQ_{c}^{(l)}}{\oQ_{c}^{(l-1)}}\\
		 &= 1 + \oQ_{c}^{(l)} - \frac{\oQ_{c}^{(l)}}{\oQ_{c}^{(l-1)}}.
	\end{align*}
	Hence,
	\begin{align*}
		\oQ_{c}^{(l)} &= \oQ_{c}^{(l-1)},\, l = 2, \cdots, m.
	\end{align*}
\end{proof}

The above proposition motivates the following simplified notation.
\begin{definition}
	\begin{equation*}
		P_{m} : = \oQ_{c}^{(1)}.
	\end{equation*}
\end{definition}

Now, when lemmata \ref{lem_Q_red} and \ref{lem_Qc_equiv} are combined, we have
\begin{corollary} \label{cor_Q_red}
	\begin{equation*}
		\begin{aligned}
		\oQ_{a_{2}}^{(l)} &= \oQ_{b_{2}}^{(l)} = \left(P_{m} \right)^{m-l} \\
		\oQ_{\conj{a_{2}}}^{(l)} &= \oQ_{\conj{b_{2}}}^{(l)} = \left(P_{m} \right)^{l-1}.
		\end{aligned}
	\end{equation*}
\end{corollary}
The above corollary, combined with the equations of
Section \ref{sec_out_deg}, is enough to express the remainder of the street factors in terms of $P_{m}$,
\begin{equation*}
	\begin{aligned}
	\left(P_{m} \right)^{m-l} &= \oQ_{a_{3}}^{(l)} = \oQ_{b_{3}}^{(l)} \\
	\left(P_{m} \right)^{l-1} &= \oQ_{\conj{a_{3}}}^{(l)} = \oQ_{\conj{b_{3}}}^{(l)} \\
	(P_{m})^{m-l+1} &= \oQ_{a_{1}}^{(l)} = \oQ_{b_{1}}^{(l)}\\
	(P_{m})^l &= \oQ_{\conj{a_{1}}}^{(l)} = \oQ_{\conj{b_{1}}}^{(l)}.
	\end{aligned}
\end{equation*}
This completes the proof of (\ref{eq:Q_red}) in Prop. 3.1.

\subsubsection{Proof of the Algebraic Equation (\ref{eq:P})}
Via the equation for $\osQ_c$ in Section \ref{sec_int_deg} along with (\ref{eq:tau_b1_sol})-(\ref{eq:tau_a1_sol_part}),
\begin{align*}
	\osQ_{c}^{(l)} &= 1 + \inup_{a_{1}}^{(l)} \inup_{b_{1}}^{(l)} \indown_{\conj{b_{1}}}^{(l)} \indown_{\conj{a_{1}}}^{(l)}\\
	&= 1 + \left[ \left( \prod_{r=0}^{l-1} \oQ_{\conj{b_{2}}}^{(r)} \oQ_{b_{2}}^{(r)} \right) R^{(1,l)} \inX_{a} \right] \left[ R^{(1,l)} \inX_{b} \right] \left[R^{(m,l)} \inX_{\conj{b}} \right] \left[ \left( \prod_{r=l+1}^{m+1} \oQ_{\conj{b_{2}}}^{(r)} \oQ_{b_{2}}^{(r)} \right) R^{(m,l)} \inX_{a} \right]\\
	&= 1 + \left( \prod_{r \neq l} \oQ_{\conj{b_{2}}}^{(r)} \oQ_{b_{2}}^{(r)} \right) \left( R^{(1,l)} \inX_{a} \right) \left(R^{(1,l)} \inX_{b} \right) \left(R^{(m,l)} \inX_{\conj{b}} \right) \left( R^{(m,l)} \inX_{a} \right)\\
	&= 1 + \left( P_{m} \right)^{(m-1)^2} \left( R^{(1,l)} \inX_{a} \right) \left(R^{(1,l)} \inX_{b} \right) \left(R^{(m,l)} \inX_{\conj{b}} \right) \left( R^{(m,l)} \inX_{\conj{a}} \right);
\end{align*}
where, on the last line we utilized Corollary \ref{cor_Q_red}.

\begin{remark}
	We note that,
	\begin{align*}
		R^{(1,l)}a + R^{(1,l)}b+ R^{(m,l)}\conj{b} + R^{(m,l)}\conj{a}
	\end{align*}
	represents a soliton charge of type $11$ on the open set $\xi^{-1}(U_{l}) \subset \twid{C'}$.  Thus, we may apply the map $\cl$ to this expression to produce an element of $\twid{\Gamma}$.
\end{remark}
This leads us to the following definition.

\begin{definition}
	We define
	\begin{align}
		\widehat{\gamma}_{c} :&= \cl \left[R^{(1,l)}a + R^{(1,l)}b + R^{(m,l)}\conj{b}+R^{(m,l)}\conj{a}\right] \in \twid{\Gamma}
		\label{eq:gammac_def}
	\end{align}
	and corresponding formal variable
	\begin{equation}
		z := X_{\widehat{\gamma}_{c}}.
		\label{eq:z_def}
	\end{equation}
\end{definition}
(We will show below that, in fact, \eqref{eq:gammac_def} does not depend on
$l$; thus, this definition is sensible.)

With the above definitions we have
\begin{align*}
	\oQ_{c}^{(l)} &= 1 + z P_{m}^{(m-1)^2}
\end{align*}
hence, by Lemma \ref{lem_Qc_equiv},
 $P_{m}$ satisfies the algebraic equation
\begin{align*}
	P_{m} &= 1 + z P_{m}^{(m-1)^2}.
	\tag{\ref{eq:P}}
\end{align*}

This completes the proof of the algebraic equation in Prop. \ref{prop_Q}.

\begin{remark}
	As we will show in Section \ref{app:proof_decomp}, $\widehat{\gamma}_{c}$ is the sum of two tangent framing lifts of simple closed curves with corresponding homology classes $\gamma,\, \gamma' \in \Gamma$.  In fact, we will show that (\ref{eq:z_def}) can be rewritten in the form stated in Prop. \ref{prop_Q}: $z = (-1)^{m} X_{\twid{\gamma} + \twid{\gamma}'}$, where $\twid{\left( \cdot \right)}: \Gamma \rightarrow \twid{\Gamma}$ is defined in Section (\ref{sec:comp_Omega}) and discussed further in Section \ref{app:sign-rule}.
\end{remark}

\subsection{Proof of the Decomposition of $\widehat{\gamma}_{c}$} \label{app:proof_decomp}

We begin with an example (which may be skipped for the more general proof below).\footnote{The following sections rely on the ideas of Section \ref{app:global}.}

\subsubsection{Example: $\widehat{\gamma}_{c}$ for $m$-herds on the cylinder}
		 We consider generalizations (to arbitrary $m$) of
the herds shown in Fig.~\ref{fig:herds} for $m = 1, \cdots,4$.
Assume we are equipped with a branched $3$-cover of the
cylinder $C = S^{1} \times \mathbb{R}$ with four branch points.
Now, consider an $m$-herd such that it is contained in a
presentation of the cylinder as an identification space
of $[0,1] \times \mathbb{R}$: the streets of type $23$ lie entirely
 in the interior of $(0,1) \times \mathbb{R}$, while the streets of
 type $12$ involved in the identifications (\ref{eq:horse_glue}) pass
 through the identified boundary.  First, each of the charges $a,\,b,\conj{a},\,\conj{b}$ can
 be thought of as flat sections of the local system
 $\mathfrak{s}: \bigcup_{\twid{z} \in \twid{C'}}
 \twid{\Gamma}(\twid{z},-\twid{z}) \rightarrow \twid{C'}$,\,
 locally defined around their respective branch points.
 The two-way streets are contained within the open set
 $U:= \bigcup_{l =1}^{m} U_{l}$, which is homeomorphic
 to $S^{1} \times I$ for $I \cong (0,1)$ an open interval.
 Let $U^{c} \cong (0,1)^2$ be the open set formed by
 removing the vertical line\footnote{Here $\sim$ denotes
 the identification of the boundary of $[0,1] \times \mathbb{R}$ to
 form the cylinder.  The removed vertical line is given by
 the (identified) dotted lines in Fig.~\ref{fig:herds}.}
 $(\{0\} \times \mathbb{R}) \cap U \sim (\{1\} \times \mathbb{R})
 \cap U$ from $U$.  $\mathfrak{s}$ is trivial over
 the open set $\xi^{-1}(U^{c}) \cong (0,1)^2 \times S^{1}$ in $\twid{C'}$;
 so, we can extend $a,\,b,\conj{a},\,\conj{b}$ to flat sections over all of $\xi^{-1}(U^{c})$.

	Now, let $q_{\text{cyl}}: [0,1] \rightarrow U \subset C'$
denote a loop winding once around the $S^{1}$ direction of $U$,
and oriented such that the upper-sourced horse branch points sit
to its ``left," while the lower-sourced horse branch points sit to its ``right";  $\widehat{q}_{\text{cyl}}: [0,1] \rightarrow \twid{C'}$ will denote the tangent framing lift of $q_{\text{cyl}}$.

		Working through the definition of the parallel transport maps $R^{(k,l)}$ in (\ref{eq:local_sys_mon}), (\ref{eq:rho_ij_def})-(\ref{eq:R_def}), we have
	\begin{align*}
		\widehat{\gamma}_{c} &= \cl \left(a + b + \conj{b} + \conj{a} \right) + (m-1) \left( [\hat{q}_{\text{cyc}}\{2\}] -[\hat{q}_{\text{cyc}} \{1\}] \right);
	\end{align*}
the expression in the closure map is defined by evaluating the sections $a,\,b,\conj{a},\,\conj{b}$ at some point $\twid{z} \in \xi^{-1}(U^{c} \cap U_{l})$ and taking their sum to define an element in $\twid{\Gamma}_{11}(\twid{z},-\twid{z})$.

Observe that we can decompose $\widehat{\gamma}_{c}$ as $\widehat{\gamma}_{c} = \widehat{\gamma} + \widehat{\gamma}'$ where,
\begin{align*}
	\widehat{\gamma} &= \cl \left(b + \conj{b} \right)\\
	\widehat{\gamma}' &= \cl \left(a + \conj{a} \right) +  \left(m - 1 \right) \left( [\hat{q}_{\text{cyc}}\{2\}] -[\hat{q}_{\text{cyc}} \{1\}] \right).
\end{align*}
Now, note that we can realize $\widehat{\gamma}$ the tangent
framing lift of a simple closed curve on $\Sigma$.
Indeed, consider an auxiliary street of type $23$, realized as a
straight line on $U \cup \{\text{branch pts.}\}$, running between
the two branch points of type $23$ (beginning at the branch point
emitting the charge $b$ and ending at the branch point emitting the
charge $\conj{b}$).  The lift of this street to $\Sigma$ is a
simple closed curve; the tangent framing lift is a representative
of $\cl \left(b + \conj{b} \right)$.
 Similarly, we can
realize  $\cl \left(a + \conj{a} \right)$ with the tangent
framing lift of a simple closed curve $\ell_{a}$
on $U \cup \{\text{branch pts.}\}$ and so $\widehat{\gamma}'$ can
be realized as a modification of $\ell_{a}$ by smoothly ``detouring"
along the lifts (to sheets 1 and 2) of a curve that winds $m-1$
times around the $S^{1}$ direction of $U$.  The resulting curve is
the tangent framing lift of a simple closed curve.  Furthermore, project these simple-closed curves to $\Sigma$; then letting $\gamma$ and $\gamma'$ be the homology classes of our projections,
with their representative curves it is clear that
$\langle \gamma, \gamma' \rangle = m$.

Now, using different techniques, let us proceed on with the general proof of the decomposition $\widehat{\gamma}_{c} = \widehat{\gamma} + \widehat{\gamma}'$, described in the example above, for an $m$-herd on a general oriented curve $C$.

\subsubsection{General Proof}

Let $\xi^{\Sigma}: \twid{\Sigma} \rightarrow \Sigma$ be the unit tangent bundle projection.
\begin{definition}
	\begin{equation*}
		\gamma_{c} := \xi^{\Sigma}_{*} \widehat{\gamma}_{c} \in \Gamma.
	\end{equation*}
\end{definition}
To derive an explicit expression for $\gamma_{c}$ in terms of simpleton charges (\ref{eq:charge_def}) in $\bigcup_{z \in C} \Gamma(z,z)$, we ``pushforward" the expression (\ref{eq:gammac_def}) via $\xi^{\Sigma}$.
From the definitions (\ref{eq:charge_def}), (\ref{eq:lifted_charge_def}), and (\ref{eq:gammac_def}) it follows that
\begin{equation}
	\gamma_{c} = \cl \left[R_{\mathfrak{r}}^{(1,l)} a_{*} + R_{\mathfrak{r}}^{(1,l)}  b_{*} + R_{\mathfrak{r}}^{(m,l)} \conj{b}_{*} + R_{\mathfrak{r}}^{(m,l)}  \conj{a}_{*}\right]
	\label{eq:gammac_pushed}
\end{equation}
where $R_{\mathfrak{r}}^{(k,n)}$ are the ``pushforward" of the parallel transport operators $R^{(k,n)}$ defined in (\ref{eq:R_red_def}).

We will construct a decomposition $\gamma_{c} = \gamma + \gamma'$ with $\langle \gamma, \gamma' \rangle = m$ roughly by shrinking the $c^{(l)}$ streets of the herd to points.  To be precise, we introduce some definitions.

\begin{figure}
	\begin{center}
		 \includegraphics[scale=0.15]{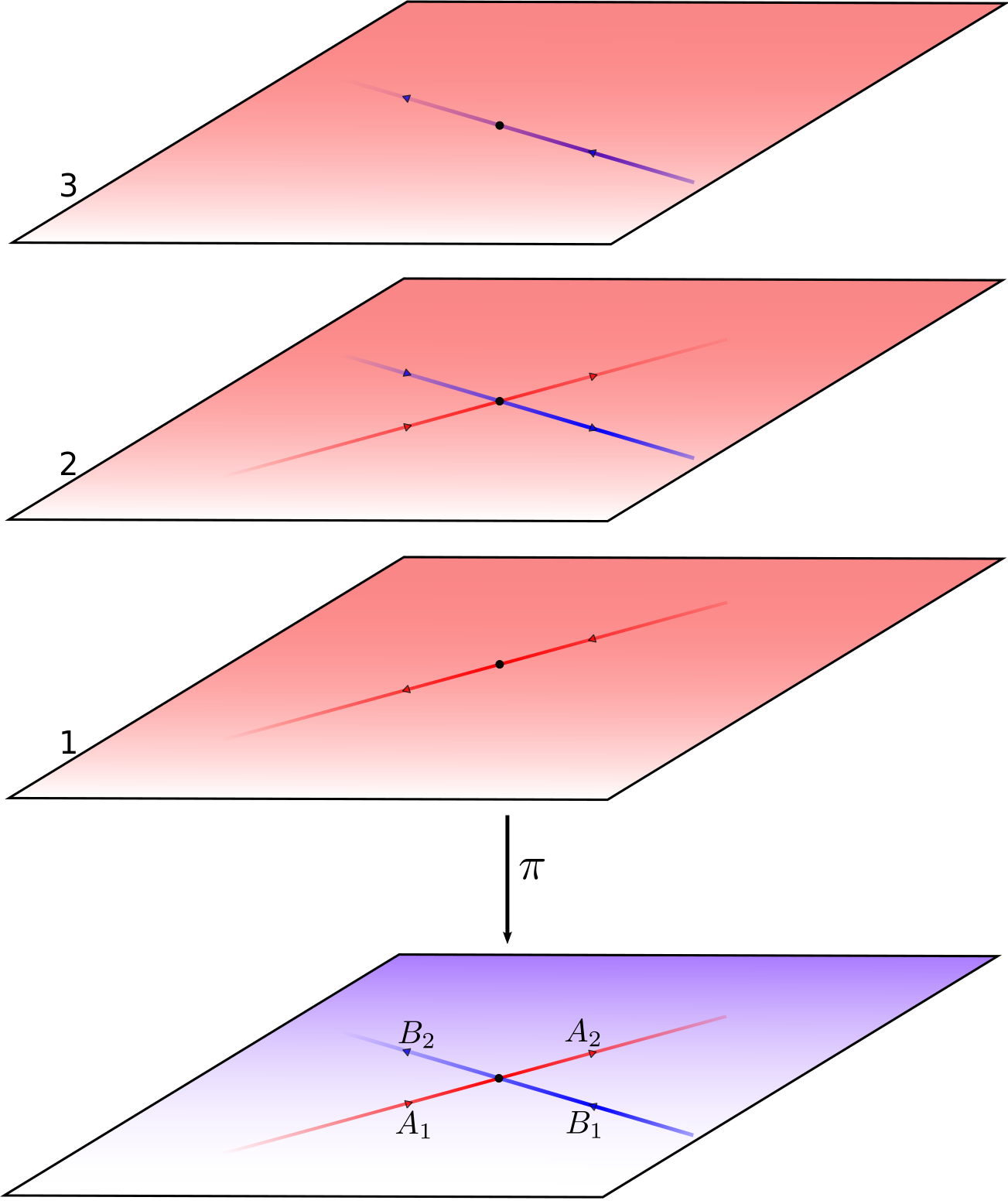}
		\caption{A pony and its lift to $\Sigma$.  Red streets are of type $12$, blue streets are of type $23$.\label{fig:pony}}
	\end{center}
\end{figure}

\begin{definition}[Definitions]\
	\begin{enumerate}
		\item A pony is a partial spectral network as shown in Fig. \ref{fig:pony}.  Upper and lower-sourced ponies are defined similar to upper and lower sourced horses.

		\item The string of ponies $S_{m}$ associated to an $m$-herd $H_{m}$ is the spectral network constructed by placing
		\begin{enumerate}
			\item A lower sourced pony on $U_{1}$\\

			\item Ponies on $U_{l},\,1 < l < m$\\

			\item An upper-sourced pony on $U_{m}$,
		\end{enumerate}
		where the $U_{m}$ are a good open cover satisfying the \ref{cond_horses} condition for $H_{m}$, and forcing the identifications
		\begin{align*}
			A_{1}^{(l+1)} &= A_{2}^{(l)}\\
			B_{1}^{(l+1)} &= B_{2}^{(l)}
		\end{align*}
		on each $U_{l} \cap U_{l+1}$.
	\end{enumerate}
\end{definition}

\begin{remark}[Remarks] \
	\begin{enumerate}
		\item $S_{m}$ is only defined up to homotopy on each disk $U_{l}$.

		\item The interpretation of $S_{m}$ as a spectral network is overkill for our discussion and we introduce it as such mainly for notational convenience: all that will be necessary is the graph of the lift $\operatorname{Lift}(S_{m}) \subset \Sigma$.  However, in the wall-crossing interpretation of $m$-herds discussed in Section \ref{sec:herds}, the spectral network $S_m$ is expected to appear on the wall of marginal stability where two hypermultiplets of intersection number $m$ have coincident central charge phase.  In fact, the procedure of deforming such a picture is what motivated the construction of $m$-herds.
	\end{enumerate}
\end{remark}

\begin{definition}
	 Let $p^{(l)}$ be a street of type $ij$, then $\lift{p}^{(l)} \in C_{1}(\Sigma; \mathbb{Z})$ is the 1-chain on $\Sigma$ representing the lift\footnote{If $p^{(l)}$ connects two joints, this lift has two components.  If $p^{(l)}$ connects a joint to a branch point of type $ij$, then the two components combine to form a connected $1$-chain between sheets $i$ and $j$.} of $p^{(l)}$ as a street of type $ij$ (using the orientation discussed in Section \ref{sec:wtheta}).
\end{definition}

If we define,
\begin{align*}
	\gamma &= \left[\sum_{l=1}^{m} \left( \lift{B}_1^{(l)} + \lift{B}_2^{(l)} \right) \right] \in H_{1}(\Sigma; \mathbb{Z})\\
	\gamma' &= \left[\sum_{l=1}^{m} \left( \lift{A}_1^{(l)} + \lift{A}_2^{(l)} \right) \right] \in H_{1}(\Sigma; \mathbb{Z}),
\end{align*}
then, as shown in Fig. \ref{fig:pony}, $\gamma$ and $\gamma'$ intersect once in each $\pi^{-1}(U_{l}),\, l=1, \cdots m$; hence,
\begin{align*}
	\langle \gamma, \gamma' \rangle & = m.
\end{align*}
Now
\begin{itemize}
	\item $\sum_{l = 1}^{m-1} \left( \lift{A}_1^{(l)} + \lift{A}_2^{(l)} \right)$
	 is a 1-chain representative of the parallel transported charge $ R_{\mathfrak{r}}^{(1,m-1)} a_{*}$.

	\item $\sum_{l = 1}^{m-1} \left( \lift{B}_1^{(l)} + \lift{B}_2^{(l)} \right)$ is a 1-chain representative of $ R_{\mathfrak{r}}^{(1,m-1)} b_{*}$.

	\item $\lift{A}_{1}^{(m)} + \lift{A}_{2}^{(m)}$ is a 1-chain representative of $\conj{a}_{*}$.

	\item $\lift{B}_{1}^{(m)} + \lift{B}_{2}^{(m)}$ is a 1-chain representative of $\conj{b}_{*}$.
\end{itemize}

Hence,
\begin{align*}
	\gamma_{c}  = \gamma + \gamma '.
\end{align*}
Now, each of the 1-chains $\lift{A}_{i}^{(l)},\, \lift{B}_{i}^{(l)}$ have well-defined tangent framing lifts $\widehat{\lift{A}}_{i}^{(l)},\, \widehat{\lift{B}}_{i}^{(l)}$ when thought of as oriented paths on $\operatorname{Lift}(S_{m}) \subset \Sigma$.  Similarly, $\gamma$ and $\gamma'$ have obvious representative curves on $\operatorname{Lift}(S_{m})$ that allow us to produce tangent framing lifts $\widehat{\gamma},\, \widehat{\gamma}'$.  In fact,
\begin{align*}
	\widehat{\gamma} &= \left[\sum_{l=1}^{m} \left( \widehat{\lift{B}}_1^{(l)} + \widehat{\lift{B}}_2^{(l)} \right) \right] \in H_{1}(\twid{\Sigma}; \mathbb{Z})\\
	\widehat{\gamma}' &= \left[\sum_{l=1}^{m} \left( \widehat{\lift{A}}_1^{(l)} + \widehat{\lift{A}}_2^{(l)} \right) \right] \in H_{1}(\twid{\Sigma}; \mathbb{Z}).
\end{align*}
Via similar arguments to above, along with the definition of $\widehat{\gamma}_{c}$ in (\ref{eq:gammac_def}), we have
\begin{equation*}
	\widehat{\gamma}_{c} = \widehat{\gamma} + \widehat{\gamma}'.
\end{equation*}
Alternatively, we can lift $\gamma_{c} = \gamma + \gamma'$ using the map  $\twid{\left( \cdot \right)}: \Gamma \rightarrow \twid{\Gamma}$ defined in (\ref{eq:lift_def}) of Appendix \ref{app:sign-rule}.  Indeed, as the curves representing $\gamma$ and $\gamma'$ intersect $m$ times, we have
\begin{equation*}
	\twid{\gamma}_{c} = \widehat{\gamma} + \widehat{\gamma}' + m H = \widehat{\gamma}_{c} + m H.
\end{equation*}
Thus,
\begin{align*}
	z = X_{\widehat{\gamma}_{c}} = (-1)^{m} X_{\twid{\gamma}_{c}}.
\end{align*}

\subsection{Proof of Proposition \ref{prop_l}} \label{app:prop_l_pf}
We wish to compute the homology class of the $1$-chain $L(n \gamma_{c})$.  First we introduce a few notational definitions that differ slightly from the main body of the paper.
\begin{definition}
	$\lift{p}^{(l,r)} \in C_{1}(\Sigma; \mathbb{Z})$ is the component of $\lift{p}^{(l)} \in C_{1}(\Sigma; \mathbb{Z})$ on the $r$th sheet.  If $p^{(l)}$ is a street of type $ij$, then
		 \begin{align*}
		 	\lift{p}^{(l,r)} &=
		 	\left\{
		 	\begin{array}{ll}
		 		+\left( \text{1-chain representing the lift of $p^{(l)}$ to the $r$th sheet} \right), & \text{if $r = j$} \\
		 		-\left(\text{1-chain representing the lift of $p^{(l)}$ to the $r$th sheet} \right), & \text{if $r = i$}\\
		 		0 & \text{otherwise}
		 	\end{array} .
		 	\right.
		 \end{align*}
\end{definition}

Now,
\begin{equation}
		\begin{aligned}
		L(n \gamma_{c}) &= \sum_{l=1}^{m} \sum_{p^{(l)}} \alpha_{n}(p,l) \liftnb{p}^{(l)} \\
		 &= \alpha_{n} \sum_{l=1}^{m} \left\{ \liftnb{c}^{(l)} + (m-l) \left( \lift{a_{2}}^{(l)} + \lift{a_{3}}^{(l)} + \lift{b_{2}}^{(l)} + \lift{b_{3}}^{(l)} \right) \right. \\
		 & + (l-1) \left( \lift{\conj{a_{2}}}^{(l)} + \lift{\conj{a_{3}}}^{(l)}  + \lift{\conj{b_{2}}}^{(l)} + \lift{\conj{b_{3}}}^{(l)} \right) + \\
		& \left. + (m-l+1) \left(\lift{a_{1}}^{(l)} + \lift{b_{1}}^{(l)} \right) +  l \left( \lift{\conj{a_{1}}}^{(l)} + \lift{\conj{b_{1}}}^{(l)} \right) \right\}.
		\end{aligned}
		\tag{\ref{eq:ln}}
	\end{equation}
after using the results of Prop. \ref{prop_Q} and the definition of $\alpha_n$
given in equation \eqref{eq:P_decomp}.

For the sake of readability we introduce some simplifying notation.

\begin{definition}[Notational Definition]
We denote,
\begin{align*}
	\la_{12} &:= \la_{1}^{(l)} + \la_{2}^{(l)}\\
	\la_{23} &:= \la_{2}^{(l)} + \la_{3}^{(l)}\\
	\la_{123} &:= \la_{1}^{(l)} + \la_{2}^{(l)} + \la_{3}^{(l)};
\end{align*}
and similarly, for $\lac_{i},\, \lb_{i},$ and $\lbc_{i}$.
\end{definition}

Using this notation, we can rewrite our sum in slightly more illuminating form,
\begin{align*}
	L(n \gamma_{c}) &= \alpha_{n} \sum_{l=1}^{m} \left\{(m-l) \left(\la_{123}^{(l)} + \lb_{123}^{(l)} \right) + l \left( \lac_{123}^{(l)} + \lbc_{123}^{(l)} \right) \right.\\
	 & \left. + \lift{a_{1}}^{(l)} + \lift{b_{1}}^{(l)} + \lift{c}^{(l)} -   \lift{\conj{a_{23}}}^{(l)} - \lift{\conj{b_{23}}}^{(l)} \right\}.
\end{align*}

This form suggests we should try to find a homological equivalence taking the terms multiplying the factor $l$, to the terms multiplying the factor $(m-l)$.  We introduce extra 1-chains to aid in our computation.  To define them, it is helpful to think of them as lifts of auxiliary streets.  However, the interpretation as lifts of streets on $C$ is only a notational tool: these streets are not part of any spectral network.

\begin{figure}
	\begin{center}
		 \includegraphics[scale=0.22]{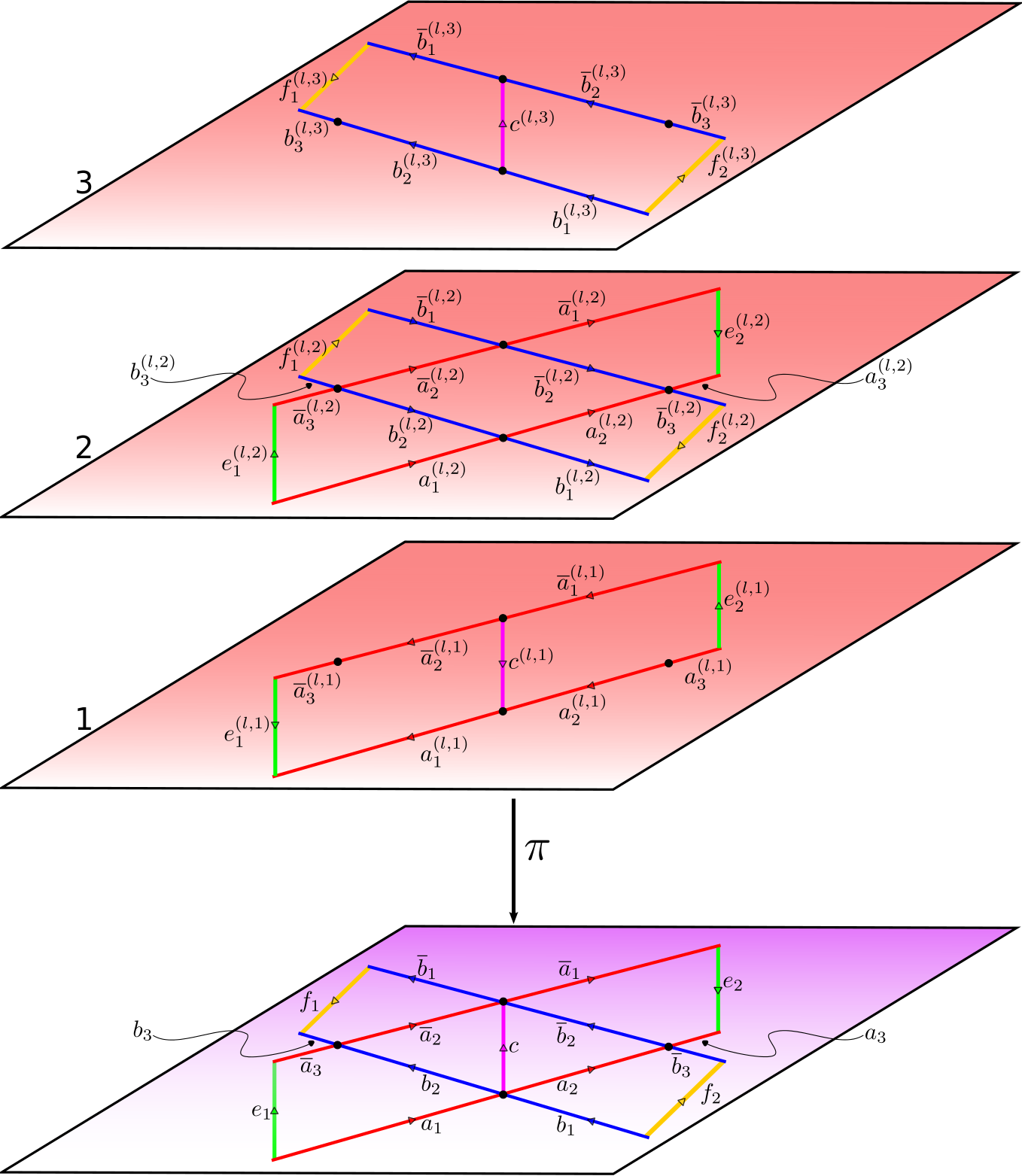}
		\caption{Lift of a horse with extra 1-chains, pictured here as the lift of some auxiliary streets on $C$.  For the sake of readability, the ``horse label" $(l)$ is suppressed on the base $C$.\label{fig:horse_decorated}}
	\end{center}
\end{figure}

\begin{definition}
	 Let $\{U_{l}\}_{l=1}^{m}$ be an open covering satisfying the \ref{cond_horses} condition for an $m$-herd. On each horse we define auxiliary streets as in Fig.~\ref{fig:horse_decorated}: $e_{1}^{(l)}, e_{2}^{(l)} \subset U_{l}$ of type $12$, and $f_{1}^{(l)},\, f_{2}^{(l)} \subset U_{l}$, of type $23$ ; such that,
		\begin{enumerate}
		\labitem{(\textit{C1})}{cond_aux_1}
			\begin{align*}
				e_{1}^{(l+1)} &= - e_{2}^{(l)}\\
				f_{2}^{(l+1)} &= - f_{1}^{(l)},
			\end{align*}
		where ``$-$" indicates orientation reversal.

		\labitem{(\textit{C2})}{cond_aux_2} $e_{1}^{(1)}$ and $f_{2}^{(1)}$ end on the branch points of type $12$ and $23$ (respectively) of the lower-sourced horse, while $e_{2}^{(m)}$ and $f_{1}^{(m)}$ end on the branch points of type $12$ and $23$ (respectively) of the upper-sourced horse.
		\end{enumerate}
\end{definition}

\begin{remark}
	The \ref{cond_noholes} condition removes any obstruction to condition \ref{cond_aux_1}.
\end{remark}

The 1-chains that will aid in our proof are the lifts of the auxiliary streets.

\begin{remark}
	Keeping with the (previously defined) convention for lifts of streets, there are 1-chains (on $\Sigma$) $\lte_{1}^{(l)},\, \lte_{2}^{(l)},\ \ltf_{1}^{(l)},$ and $\ltf_{2}^{(l)}$ (also depicted in Fig.~\ref{fig:horse_decorated}).  It follows that, via \ref{cond_aux_1},
		\begin{equation}
			\begin{aligned}
				\lte_{1}^{(l+1)} &= - \lte_{2}^{(l)}\\
				\ltf_{2}^{(l+1)} &= - \ltf_{1}^{(l)}.
			\end{aligned}
			\label{eq:aux_rel}
		\end{equation}
		for $l=1, \cdots, m-1$.
\end{remark}

\begin{lemma}
Let $\sim$ denote homological equivalence.  Then for each $l = 1, \cdots, m$:
on the first (locally defined) sheet,
	\begin{align}
		0 & \sim \lac_{23}^{(l,1)} + \lte_{1}^{(l,1)} - \la_{1}^{(l,1)} - \lift{c}^{(l,1)} \label{eq:hom_sht1_1}\\
		0 & \sim \lac_{1}^{(l,1)} + \lift{c}^{(l,1)} - \la_{23}^{(l,1)} + \lift{e}_{2}^{(l,1)} \label{eq:hom_sht1_2}.
	\end{align}
On the second sheet,
	\begin{align}
		0 & \sim \lac_{123}^{(l,2)} + \lte_{2}^{(l,2)} - \la_{123}^{(l,2)} + \lte_{1}^{(l,2)} \label{eq:hom_sht2_1}\\
		0 & \sim \lbc_{123}^{(l,2)} + \ltf_{2}^{(l,2)} - \lb_{123}^{(l,2)} +	 \ltf_{1}^{(l,2)} \label{eq:hom_sht2_2}\\
		0 & \sim \lac_{3}^{(l,2)} + \lb_{2}^{(l,2)} - \la_{1}^{(l,2)} + \lte_{1}^{(l,2)} \label{eq:hom_sht2_3}\\
		0 & \sim \la_{2}^{(l,2)} + \lbc_{3}^{(l,2)} + \ltf_{2}^{(l,2)}- \lb_{1}^{(l,2)} \label{eq:hom_sht2_4}\\
		0 & \sim \lb_{2}^{(l,2)} + \la_{2}^{(l,2)} - \lbc_{2}^{(l,2)} - \lac_{2}^{(l,2)} \label{eq:hom_sht2_5} .
	\end{align}
On the third sheet,
	\begin{align}
				0 & \sim \lb_{1}^{(l,3)} + \lift{c}^{(l,3)} - \lbc_{23}^{(l,3)} - \ltf_{2}^{(l,3)} \label{eq:hom_sht3_1} \\
		0 & \sim \lbc_{1}^{(l,3)} + \ltf_{1}^{(l,3)} - \lb_{23}^{(l,3)} + \lift{c}^{(l,3)} \label{eq:hom_sht3_2}.
	\end{align}
\end{lemma}
\begin{proof}
The lemma follows by inspection of Fig.~\ref{fig:horse_decorated}.  Each of the listed sum of 1-chains is the boundary of an oriented disk.
\end{proof}

In particular, it follows from the lemma that
\begin{align*}
	\lac_{123}^{(l)} & \sim \la_{123}^{(l)} - \lte_{1}^{(l)} - \lte_{2}^{(l)}\\
	\lbc_{123}^{(l)} & \sim \lb_{123}^{(l)} - \ltf_{1}^{(l)} - \ltf_{2}^{(l)}.
\end{align*}
Hence,
\begin{align}
		L(n \gamma_{c}) &\sim \alpha_{n} \sum_{l=1}^{m} \left\{(m-l) \left(\la_{123}^{(l)} + \lb_{123}^{(l)} \right) + l \left( \la_{123}^{(l)} + \lb_{123}^{(l)} \right) \right\} + \alpha_{n} R_{1} + \alpha_{n} R_{2} \nonumber \\
		&\sim m \alpha_{n} \sum_{l=1}^{m} \left(\la_{123}^{(l)} + \lb_{123}^{(l)} \right) + \alpha_{n} R_{1} + \alpha_{n} R_{2}
		\label{eq:L_with_remainder}
\end{align}
where
\begin{align*}
	R_{1} &= - \sum_{l = 1}^{m} l \left\{\lte_{1}^{(l)} + \lte_{2}^{(l)} + \ltf_{1}^{(l)} + \ltf_{2}^{(l)} \right\}\\
	R_{2} &= \sum_{l=1}^{m} \left\{\lift{a_{1}}^{(l)} + \lift{b_{1}}^{(l)} + \lift{c}^{(l)} - \lift{\conj{a_{23}}}^{(l)} - \lift{\conj{b_{23}}}^{(l)}  \right\}.
\end{align*}
Using (\ref{eq:aux_rel}), the first of these sums can be simplified,
\begin{align}
	R_{1} &= -\sum_{l = 1}^{m} l \left( \lte_{1}^{(l)} + \ltf_{2}^{(l)} \right) - \sum_{l = 1}^{m} l \left( \lte_{2}^{(l)} + \ltf_{1}^{(l)} \right) \nonumber \\
	&= -\sum_{l = 1}^{m} l \left( \lte_{1}^{(l)} + \ltf_{2}^{(l)} \right) + \sum_{l = 1}^{m-1} l \left( \lte_{1}^{(l+1)} + \ltf_{2}^{(l+1)} \right) -m \left(\lte_{2}^{(m)} + \ltf_{1}^{(m)} \right) \nonumber\\
	&= -\sum_{l = 1}^{m} l \left( \lte_{1}^{(l)} + \ltf_{2}^{(l)} \right) + \sum_{l = 2}^{m} (l-1) \left( \lte_{1}^{(l)} + \ltf_{2}^{(l)} \right) - m \left(\lte_{2}^{(m)} + \ltf_{1}^{(m)} \right) \nonumber\\
	&= -\left(\lte_{1}^{(1)} + \ltf_{2}^{(1)} \right) -  m \left(\lte_{2}^{(m)} + \ltf_{1}^{(m)} \right) - \sum_{l = 2}^{m} \left( \lte_{1}^{(l)} + \ltf_{2}^{(l)} \right) \nonumber\\
	&= -  m \left(\lte_{2}^{(m)} + \ltf_{1}^{(m)} \right) - \sum_{l = 1}^{m} \left( \lte_{1}^{(l)} + \ltf_{2}^{(l)} \right).
	\label{eq:R1_red}
	\end{align}
To reduce $R_{2}$, we use the following lemma.

\begin{lemma}
	\begin{align*}
		\lift{a_{1}}^{(l)} + \lift{b_{1}}^{(l)} + \lift{c}^{(l)} - \lift{\conj{a_{23}}}^{(l)} - \lift{\conj{b_{23}}}^{(l)}  & \sim  \lte_{1}^{(l)} + \ltf_{2}^{(l)}.
	\end{align*}
\end{lemma}
\begin{proof}
	On sheet 1,
	\begin{align*}
		\la_{1}^{(l,1)} + \lb_{1}^{(l,1)} + \lift{c}^{(l,1)} -\lac_{23}^{(l,1)} - \lbc_{23}^{(l,1)} &= \la_{1}^{(l,1)} + \lift{c}^{(l,1)} -\lac_{23}^{(l,1)}.
	\end{align*}
Using (\ref{eq:hom_sht1_1}),
	\begin{align*}
		& \sim \la_{1}^{(l,1)} + \lift{c}^{(l,1)} + \left( \lte_{1}^{(l,1)} - \la_{1}^{(l,1)} - \lift{c}^{(l,1)} \right) \\
		& \sim \lte_{1}^{(l,1)}.
	\end{align*}
 Similarly, on sheet 3, using (\ref{eq:hom_sht3_1}) appropriately,
	\begin{align*}
		\la_{1}^{(l,3)} + \lb_{1}^{(l,3)} + \lift{c}^{(l,3)} -\lac_{23}^{(l,3)} - \lbc_{23}^{(l,1)} &= \lb_{1}^{(l,3)} + \lift{c}^{(l,3)} - \lbc_{23}^{(l,3)}\\
		& \sim \lb_{1}^{(l,3)} + \lift{c}^{(l,3)} + \left( \ltf_{2}^{(l,3)} - \lb_{1}^{(l,3)} - \lift{c}^{(l,3)} \right) \\
		& \sim \ltf_{2}^{(l,3)}.
	\end{align*}
	  On sheet 2
	\begin{align*}
		\la_{1}^{(l,2)} + \lb_{1}^{(l,2)} + \lift{c}^{(l,2)} -\lac_{23}^{(l,2)} - \lbc_{23}^{(l,2)} &= \la_{1}^{(l,2)} + \lb_{1}^{(l,2)} - \lac_{23}^{(l,2)} - \lbc_{23}^{(l,2)}.
	\end{align*}
	Now, via (\ref{eq:hom_sht2_3}) and (\ref{eq:hom_sht2_4})
	\begin{align*}
		\la_{1}^{(l,2)} & \sim \lac_{3}^{(l,2)} + \lb_{2}^{(l,2)} + \lte_{1}^{(l,2)}\\
		\lb_{1}^{(l,2)} & \sim \la_{2}^{(l,2)} + \lbc_{3}^{(l,2)} + \ltf_{2}^{(l,2)}.
	\end{align*}
	Hence,
	\begin{align*}
		\la_{1}^{(l,2)} + \lb_{1}^{(l,2)} + \lift{c}^{(l,2)} -\lac_{23}^{(l,2)} - \lbc_{23}^{(l,2)} & \sim \left( \lac_{3}^{(l,2)} + \lb_{2}^{(l,2)} + \lte_{1}^{(l,2)} \right)\\
		& + \left( \la_{2}^{(l,2)} + \lbc_{3}^{(l,2)} + \ltf_{2}^{(l,2)} \right) - \lac_{23}^{(l,2)} - \lbc_{23}^{(l,2)}\\
		& \sim \lb_{2}^{(l,2)} + \la_{2}^{(l,2)} - \lbc_{2}^{(l,2)} - \lac_{2}^{(l,2)}  + \lte_{1}^{(l,2)} + \ltf_{2}^{(l,2)}\\
		& \sim \lte_{1}^{(l,2)} + \ltf_{2}^{(l,2)},
	\end{align*}
	where the last reduction is due to (\ref{eq:hom_sht2_5}).
\end{proof}
Thus,
\begin{align*}
	R_{2} \sim \sum_{l=1}^{m} \left( \lte_{1}^{(l)} + \ltf_{2}^{(l)} \right);
\end{align*}
so, with (\ref{eq:R1_red}), we have
\begin{align*}
	R_{1} + R_{2} = -m \left(\lte_{2}^{(m)} + \ltf_{1}^{(m)} \right).
\end{align*}
Substituting this result into (\ref{eq:L_with_remainder}),
\begin{align*}
	L(n \gamma) \sim  m \alpha_{n} \sum_{l=1}^{m -1 }
\left(\la_{123}^{(l)} + \lb_{123}^{(l)} \right) + m \alpha_{n} \left[ \left(\la_{123}^{(m)}
+ \lb_{123}^{(m)} \right) - \left(\lte_{2}^{(m)} + \ltf_{1}^{(m)} \right) \right].
\end{align*}
After inspecting Fig.~\ref{fig:horse_decorated}, by deforming slightly on the $m$th horse we can convince ourselves this is precisely a 1-chain representing $\gamma_{c}$.

 To make this claim precise, let $\lift{q}$ be a
  1-chain on $\Sigma$ such that $\partial \lift{q} \subset \pi^{-1}(z)$ for some $z \in C'$,
  and define $[ \lift{q} ]_{R}$ as the corresponding equivalence class in $\bigcup_{z \in C'} \Gamma(z,z)$.  Then, for any $k = 1, \cdots ,m$
\begin{align*}
	R_{\mathfrak{r}}{(1,k)}  a_{*} &= \left[ \sum_{l=1}^{k} \la_{123} \right]_{R} \\
   	R_{\mathfrak{r}}{(1,k)} b_{*} &= \left[ \sum_{l=1}^{k} \lb_{123} \right]_{R}.
\end{align*}
Furthermore, by parallel transporting the endpoints of $\conj{a}$ and $\conj{b}$ along an appropriate path contained in the $m$th horse\footnote{As per our notation motivated in Section \ref{app:global}, we do not write this parallel transport map explictly.}
\begin{align*}
	\conj{a}_{*} &= \left[ \left(\la_{123}^{(m)} - \lte_{2}^{(m)} \right) \right]_{R}\\
	\conj{b}_{*} &= \left[  \left( \lb_{123}^{(m)} - \ltf_{1}^{(m)} \right) \right]_{R}.
\end{align*}
Thus,
\begin{align*}
	[L(n \gamma_{c})]_{R} &= m \alpha_{n} \left[R_{\mathfrak{r}}^{(1,k)} a_{*} +
R_{\mathfrak{r}}^{(1,k)}  b_{*} + \conj{a}_{*} + \conj{b}_{*} \right]_{R}.
\end{align*}
Applying the closure map to both sides, by (\ref{eq:gammac_pushed}) the proposition holds:
\begin{align*}
	[L(n \gamma_{c})] &= m \alpha_{n} \gamma_{c} \in H_{1}(\Sigma ;\mathbb{Z}).
\end{align*}

\subsection{Table of $m$-herd BPS indices $\Omega(n \gamma_{c})$, for low values of $n$ and $m$} \label{app_omega_table}
\begin{table}[h]
\begin{center}
\caption{Values of $\Omega(n \gamma_{c})$ for low $n$ and $m$}
	\begin{tabular}{|c||c|c|c|c|c|c|c|} \hline
	{} & \multicolumn{7}{|c|}{$n$} \\
	\hline
	{} & 1 & 2 & 3 & 4 & 5 & 6 & 7\\
	\hline \hline
	$m = 1$ & 1 & 0 & 0 & 0 & 0 & 0 & 0\\
	\hline
	$m = 2$ & -2 & 0 & 0 & 0 & 0 & 0 & 0\\
	\hline
	$m = 3$ & 3 & -6 & 18  & -84 & 465 & -2808 & 18123 \\
	\hline
	$m = 4$ & -4 & -16 & -144 & -1632 & -21720 & -318816 & -5018328 \\
	\hline
	$m = 5$ & 5 & -40 & 600 & -12400 & 300500 & -8047440 & 231045220 \\
	\hline
	$m = 6$ & -6 & -72 & -1800 & -58800 & -2251500 & -95312880 &	 -4325917260 \\
	\hline
	$m = 7$ & 7  & -126 & 4410 & -208740 & 11579925 & -710338104 & 46716068007 \\
	\hline
	\end{tabular}
\end{center}
\end{table}

\section{Proof of Proposition \ref{prop:asymp}} \label{app:proof_asymp}
Define the sequence
\begin{equation*}
	b_{l} := \binom{(m-1)^2 l}{l};
\end{equation*}
we will show
\begin{align}
	\lim_{n \rightarrow \infty} \frac{\Omega(n \gamma_{c})}{ (-1)^{mn + 1} \left( \frac{m}{(m-1)^2n^2} \right) b_{n}} & =1.
	\label{eq:omega_b_asymp}
\end{align}
Indeed, from (\ref{eq:Omega_sol}),
\begin{align*}
	\frac{\Omega(n \gamma_{c})}{ (-1)^{mn + 1} \left( \frac{m}{(m-1)^2n^2} \right) b_{n}}  &= 1 +  \overbrace{\sum_{\substack{d|n\\ d<n}} (-1)^{m \left( n + d \right)} \mu \left(\frac{n}{d} \right) \left(\frac{b_{d}}{b_{n}} \right)}^{R(n)},
\end{align*}
but
\begin{align*}
	|R(n)| \leq \sum_{\substack{d|n\\ d<n}} \frac{b_{d}}{b_{n}}.
\end{align*}
Now, from the bounds
\begin{equation*}
	\sqrt{2 \pi} n^{n + \frac{1}{2}} e^{-n} < n! \leq n^{n +\frac{1}{2}} e^{1 - n}
\end{equation*}
it follows that
\begin{equation*}
\frac{b_{d}}{b_{n}} < \left(\frac{e}{\sqrt{2\pi}} \right)^{3} \left(\frac{n}{d} \right)^{1/2} e^{c_{m} (d-n)},
\end{equation*}
where $c_{m}$ is the constant defined in (\ref{eq:c_m}).  Hence,
\begin{equation*}
	|R(n)| < \left( \frac{e}{ \sqrt{2\pi}} \right)^{3} \left( n^{1/2} e^{-c_{m} n} \right) \sum_{\substack{d|n\\ d<n}} d^{-1/2} e^{c_{m} d}.
\end{equation*}
Now, the next largest divisor of $n$, other than $n$ itself, is $\leq n/2$.  Using this fact, the observation that $d^{-1/2} e^{c_{m} d}$ is a monotonically increasing function of $d$, and the crude bound that number of divisors of $n$ is $\leq n$, we have
\begin{equation*}
	\sum_{\substack{d|n\\ d<n}} d^{-1/2} e^{c_{m} d} \leq  n \left(\left(\frac{n}{2} \right)^{-1/2} e^{c_{m}n/2} \right) = \sqrt{2n} e^{c_{m} n/2};
\end{equation*}
so
\begin{equation*}
	|R(n)| < \sqrt{2} \left( \frac{e}{ \sqrt{2\pi}} \right)^{3} n e^{-c_{m} n/2},
\end{equation*}
which vanishes as $n \rightarrow \infty$, verifying (\ref{eq:omega_b_asymp}).  In other words, the $n \rightarrow \infty$ asymptotics of $\Omega(n \gamma_{c})$ are given by the asymptotics of the largest term $b_{n}$ of (\ref{eq:Omega_sol}) inside the sum over divisors:
\begin{equation*}
	\Omega(n \gamma_{c}) \sim (-1)^{mn + 1} \left( \frac{m}{(m-1)^2} \right) n^{-2} b_{n}.
\end{equation*}
Equation (\ref{eq:Omega_asymp}) follows by using Stirling's asymptotics on the binomial coefficient $b_{n}$: as $n \rightarrow \infty$,
\begin{equation*}
	b_{n} \sim \frac{1}{\sqrt{2\pi}} \left(\frac{m-1}{\sqrt{m(m-2)}} \right) n^{-1/2} e^{c_{m} n}.
\end{equation*}

\section{A sign rule} \label{app:sign-rule}

In this appendix we discuss a subtle point about signs which was not treated correctly in the first version
of \cite{GMN5}.

The issue concerns the proper way of extracting 4D BPS degeneracy information from the generating functions $Q(p)$ defined in \eqref{eq:def-Q}.
What we want to do is factorize $Q(p)$ as we wrote in \eqref{eq:Q-exp}, but to do so, we need a way of choosing the lifts $\tilde\gamma \in \tilde\Gamma$
of classes $\gamma \in \Gamma$.

We propose the following rule.
First, represent $\gamma$ as a sum of $k$ smooth closed curves $\beta_m$ on $\Sigma$.
Each such
curve has a canonical lift $\hat \beta_m$ to $\tilde\Sigma$ just given by the tangent framing.
Then we define
\begin{equation}
 \tilde\gamma = \sum_{m=1}^k (\hat \beta_m + H) + \sum_{m \le n} \# (\beta_m \cap \beta_n)  H.
 \label{eq:lift_def}
\end{equation}

We need to check that $\tilde\gamma$ so defined is independent of the choice of how we represent
$\gamma$ as a union of $\beta_m$.
First we check that $\tilde\gamma$ is stable
under creation/deletion of a null-homologous loop.
If $\beta$ denotes such a loop then $\hat \beta = H$ modulo $2H$ (indeed, suppose $\beta$
bounds a subsurface $S$; $S$ admits a vector field extending $\hat\beta$,
with $\chi(S)$ signed zeroes in the interior; this vector field gives a 2-chain on $\tilde\Sigma$
which shows $\hat\beta$ is homologous on $\tilde\Sigma$ to $\chi(S) H$;
but $\chi(S)$ is odd since $S$ has a single boundary component.)
Thus the extra term $\hat \beta + H$ added to $\tilde\gamma$ is zero modulo $2H$.
Next we check $\tilde\gamma$ is stable under resolution of an intersection:
indeed this changes $\sum_{m \le n} \# (\beta_m \cap \beta_n)$
by $-1$, and changes $k$
by $\pm 1$, while not changing $\sum \hat \beta_m$; it thus
changes $\tilde\gamma$ by either $0$ or $-2H$, which is in either case trivial mod $2H$.
Finally we note that any representation of $\gamma$ as a union of smooth closed curves can be related to
any other by repeated application of these two operations and their inverses.
It follows that $\tilde\gamma$ is indeed well defined.

Moreover, this rule has the following property:
\begin{equation}
 \tilde\gamma + \tilde\gamma' = \widetilde{\gamma + \gamma'} + \langle \gamma, \gamma' \rangle H.
\end{equation}
It follows that the corresponding formal variables
\begin{equation}
 Y_\gamma = X_{\twid\gamma}
\end{equation}
obey the twisted product rule
\begin{equation}
 Y_{\gamma} Y_{\gamma'} = (-1)^{\langle \gamma, \gamma' \rangle} Y_{{\gamma + \gamma'}}.
\end{equation}
In turn it follows (using the arguments of \cite{GMN3,GMN5}) that, if we use this particular lifting rule
to extract the 4D BPS degeneracies, all the wall-crossing relations (and in particular the KSWCF for the pure 4D degeneracies)
will come out as they should.

\section{Spectral networks and algebraic equations}\label{app:algeq}

It has been noted by Kontsevich that the generating functions of Donaldson-Thomas invariants
are often solutions of algebraic equations.  The equation \eqref{eq:P-intro} is one example.
This equation determines the BPS degeneracies $\Omega(n \gamma_c)$ corresponding to an $m$-cohort.
As we have seen in this paper, this equation can be derived from a close analysis of the
spectral network corresponding to an $m$-herd.

While finding the precise equation \eqref{eq:P-intro} involved some hard work, the bare fact that the BPS generating function obeys
\ti{some} algebraic equation is not so mysterious.  Indeed, this seems to be a general phenomenon, which we expect to occur for
\ti{any} theory of class $S$.  Let us briefly explain why.

The junction equations \eqref{eq:6way} involve variables $\nu$ and $\tau$ attached to each street of the network.
These variables lie \ti{a priori} in the noncommutative algebra $\CA_S$.  However, one can replace them by variables lying in the
commutative algebra $\CA_C$ simply by choosing local trivializations of the torsors $\tilde\Gamma(\tilde z, -\tilde z)$;
indeed such a trivialization gives an embedding of $\CA_S$ into the algebra of $K \times K$ matrices over $\CA_C$;
taking the individual matrix components then gives equations where all of the variables lie in $\CA_C$.
These equations alone do not quite determine
$\nu$ and $\tau$ --- there are not quite enough of them.  However, once one supplements them with the
``branch point'' equations from \cite{GMN5} (which are also algebraic), one then has one equation for each variable.

In principle the spectral network may involve infinitely many streets and joints, so at this stage we may have
an infinite set of algebraic equations in an infinite number of variables.  However, in all examples we have considered,
only finitely many of these equations are relevant for determining any particular BPS generating function.  Indeed, in these
examples the set of ``two-way streets'' is always supported in some compact set $K$
obtained by deleting small discs around punctures on $C$; the intersection $\CW \cap K$
only involves finitely many streets; and there are no streets which enter $K$ from outside.  It seems likely that these properties
hold for \ti{all} spectral networks, although we have not proven it.
In any case, taking these properties for granted, it follows that the finitely many variables $\nu$ and $\tau$ attached
to the finitely many streets in $\CW \cap K$ are indeed determined by a finite set of algebraic equations.

The functions $Q(p)$ in turn are algebraic combinations of the $\nu$ and $\tau$, as are the
BPS generating functions $\prod_{p} Q(p)^{\langle \bar a, p_{\Sigma} \rangle}$.
Thus we expect that the BPS generating functions in any theory of class $S$ always satisfy algebraic equations,
which gives a natural explanation of Kontsevich's observation, at least
in those theories.

\end{document}